\documentclass[12pt, a4paper, titlepage, headsepline, twoside]{book}
\usepackage{graphicx}
\usepackage{style_PM}
\usepackage[usenames,dvipsnames]{xcolor}
\usepackage{amsfonts}       
\usepackage{amsthm,amssymb}         
\usepackage{amsmath}         
\usepackage{mathtools}      
\usepackage{wrapfig}
\usepackage{tkz-graph}
\usepackage{comment}
\usepackage{soul}
\usepackage{nccmath}
\usepackage{balance}
\usepackage{algorithm}
\usepackage{caption}
\usepackage{multirow}
\usepackage{subfigure}

\usepackage{bm}
\usepackage{etoolbox}

\graphicspath{{figures/}}

\usepackage{enumitem}
\usepackage{amsfonts}
\usepackage{algorithmic}
\usepackage{multirow}

\usepackage{pifont}
\usepackage{tikz}
\usepackage{eurosym}
\usepackage{pgfplots}
\usepackage{setspace}
\usepackage{rotating}
\usepackage{circuitikz}

\usepackage{bbm}
\newtheorem{theorem}{Theorem}
\newtheorem{lemma}{Lemma}
\newtheorem{proposition}{Proposition}
\newtheorem{corollary}{Corollary}
          
\newtheorem{definition}{Definition}

\newtheorem{example}{Example}
\tikzset{>=latex}
\usepackage{bm}
\newcommand*{\QEDA}{\hfill\ensuremath{\square}}
\newcommand{\AC}{\mathcal{A}}

\allowdisplaybreaks

\usepackage{etoolbox}
\cslet{blx@noerroretextools}\empty
\usepackage[backend=biber,natbib=true,style=alphabetic,sorting=nyt,giveninits=true, defernumbers=false,citetracker=true,maxnames=99,doi=false,isbn=false,url=true]{biblatex}
\usepackage{xpatch}
\newcommand*{\boldnamedo}[1]{%
  \iffieldequalstr{hash}{#1}
    {\bfseries\listbreak}
    {}}
\newbibmacro*{name:bold}{%
  \forlistloop{\boldnamedo}{\boldnames}}

\newcommand*{\boldnames}{}
\forcsvlist{\listadd\boldnames}
  {{30abb2207c22bf1dcd17f45640248687}}

\usepackage{varioref}
\usepackage[%
colorlinks= true,
raiselinks	= true,
linkcolor	= teal, 
citecolor	= teal,
urlcolor	= ForestGreen,
pdfauthor	= {Karaca Orcun},
pdftitle	= {PhD Thesis},
pdfkeywords	= {},
pdfsubject	= {PhD Thesis},
plainpages= false]{hyperref}
\usepackage[capitalize,nameinlink,compress,noabbrev]{cleveref}

\crefname{appendix}{Appendix}{Appendices} 
\crefname{figure}{Figure}{Figures} 
\crefname{equation}{}{} 
\usepackage{expl3}
\makeatletter
\ExplSyntaxOn
\cs_gset_eq:NN \c_deprecateacro_minus_one \m@ne
\tex_let:D \c_minus_one \c_deprecateacro_minus_one
\ExplSyntaxOff
\makeatother

\usepackage[hyperref=false]{acro}

\newlist{listabbrev}{description}{1}
 \setlist[listabbrev]{
 labelwidth = 5.75em,
 leftmargin = 5.75em,
 labelsep= 0pt,
 labelindent = 0pt,
 format      = \textbf\normalfont}

\NewAcroTemplate[list]{mystyle}{%

  \UseAcroTemplate[list]{description}[0]%
}
\acsetup{ list/template = mystyle,
 list/heading = section*, list/name=Acronyms, first-long-format = {\textbf}, 
list/display=all}

\DeclareAcronym{BNE}{short=BNE, long=Bayes-Nash equilibrium, class=abbrev}
\DeclareAcronym{PSP}{short=PSP, long=Progressive second price, class=abbrev}
\DeclareAcronym{VCG}{short=VCG, long=Vickrey-Clarke-Groves, class=abbrev}
\DeclareAcronym{ATC}{short=ATC, long=Available transmission capacity, class=abbrev}
\DeclareAcronym{FB}{short=FB, long=Flow-based, class=abbrev}
\DeclareAcronym{EC}{short=EC, long=European Commission, class=abbrev}
\DeclareAcronym{TSO}{short=TSO, long=Transmission system operator, class=abbrev}
\DeclareAcronym{ENTSO-E}{short=ENTSO-E, long=European  Network  of  Transmission  System  Operators  for  Electricity, class=abbrev}
\DeclareAcronym{minRAM}{short=minRAM, long=Minimum remaining available margin, class=abbrev}
\DeclareAcronym{IEEE}{short=IEEE, long=Institute of Electrical and Electronics Engineers, class=abbrev}
\DeclareAcronym{DC}{short=DC, long=Direct current, class=abbrev}
\DeclareAcronym{AC}{short=AC, long=Alternative current, class=abbrev}
\DeclareAcronym{KKT}{short=KKT, long=Karush-Kuhn-Tucker, class=abbrev}
\DeclareAcronym{ACER}{short=ACER, long=Agency  for  the  Cooperation  of  Energy  Regulators, class=abbrev}
\DeclareAcronym{OPF}{short=OPF, long=Optimal power flow, class=abbrev}
\DeclareAcronym{LMP}{short=LMP, long=Locational marginal pricing or Lagrange-multiplier-based pricing, class=abbrev}
\DeclareAcronym{IR}{short=IR, long=Individually rational or individual rationality, class=abbrev}
\DeclareAcronym{DSIC}{short=DSIC, long=Dominant-strategy incentive-compatible, class=abbrev}
\DeclareAcronym{MPCS}{short=MPCS, long=Maximum-payment core-selecting, class=abbrev}
\DeclareAcronym{IGCC}{short=IGCC, long=Integrated gasification combined cycle, class=abbrev}
\DeclareAcronym{OCGT}{short=OCGT, long=Open cycle gas turbine, class=abbrev}
\DeclareAcronym{CCGT}{short=CCGT, long=Combined cycle gas turbine, class=abbrev}
\DeclareAcronym{LLG}{short=LLG, long=Local-local-global, class=abbrev}
\DeclareAcronym{CCG}{short=CCG, long=Core constraint generation, class=abbrev}
\DeclareAcronym{MILP}{short=MILP, long=Mixed-integer linear program, class=abbrev}

\usepackage{longtable}

	\DeclareMathAlphabet\mathbfcal{OMS}{cmsy}{b}{n}

\DeclareMathOperator*{\argmin}{arg\,min}
\DeclareMathOperator*{\argmax}{arg\,max}
\newcommand{\CC}{\mathcal{C}}
\newcommand{\BB}{\mathcal{B}}
\newcommand{\R}{\mathbb{R}}

\newcommand{\N}{\mathbb{N}}
\newcommand{\X}{\mathbb{X}}

\newcommand{\norm}[1]{\left\lVert#1\right\rVert}

\begin{document}
\frontmatter
\begin{center}
\thispagestyle{empty}
\large
Dissertation ETH Zürich No. 27287

\vspace*{2cm}

\Large
\textbf{On the theory and applications of mechanism design\\ and coalitional games in electricity markets
}

\vspace*{1.9cm}
\large
A dissertation submitted to attain the degree of\\[0.4cm]
Doctor of Sciences of ETH Zürich\\[0.4cm]
(Dr. sc. ETH Zürich)

\vspace*{1.5cm}

presented by \\[1.5cm]
Orcun Karaca \\[0.2cm]

M. Sc., ETH Zürich, Switzerland\\[1.5cm]

accepted on the recommendation of \\
Prof. Dr. Maryam Kamgarpour, examiner \\
Prof. Dr. Jalal Kazempour, co-examiner \\
Prof. Dr. Sven Seuken, co-examiner \\
Prof. Dr. Neil Walton, co-examiner \\[1cm]
2020
\end{center}


\thispagestyle{empty}
\newpage
\thispagestyle{empty}
\vspace*{\fill}

ETH Zürich

IfA - Automatic Control Laboratory

ETL, Physikstrasse 3

8092 Zurich, Switzerland\\

\copyright\ Orcun Karaca, 2020

All Rights Reserved\\
\chapter*{}
\thispagestyle{empty}
{\large 
\vspace{-2cm}
\emph{~~To my family}

\vspace{5.5cm}

\hspace{8.2cm}{Soundtrack recommendation:}\\

\hfill{\href{https://www.youtube.com/watch?v=6Rh_gvrJ190}{\emph{\color{black}The Birth and Death of the Day~~}}}\\

\hspace{8.2cm}
{ by \emph{Explosions in the Sky}}
}
\newpage
\thispagestyle{empty}
\chapter*{Acknowledgements}
\addcontentsline{toc}{chapter}{Acknowledgements}
\setcounter{page}{1}

First and foremost, I would like to express my greatest gratitude to my PhD advisor, Prof. Maryam Kamgarpour, for giving me the opportunity to undertake this exciting journey and for always believing in me, and for helping me set and pursue my goals. I would like to thank you for providing the right guidance at the right time, and also balancing it perfectly with trust and freedom. Thanks to your counseling---whether it be technical or otherwise---I feel proud to say that in these last four years I have grown greatly as a researcher. I did my best to keep this short!

It was amazing to work with Prof. Neil Walton at an early stage of my PhD. You inspired me by your enthusiasm and critical thinking.  I would also like to thank you for serving as my thesis referee, and for carefully
evaluating this thesis. 

My deepest thanks go also to Prof. Jalal Kazempour for serving as my thesis referee,  for carefully
evaluating this thesis, and for meeting with me to discuss the content of the thesis, regardless of your tight schedule. I especially thank you and your research group for offering invaluable feedback and widening my perspective with research questions on a wide range of topics on many separate occasions. 

I am honored to have met and interacted with Prof. Sven Seuken during my studies. Thank you for serving as my thesis referee, and for carefully
evaluating this thesis as well. Especially thank you and also your research group for invaluable feedback. I benefited greatly from many interactions in your course and in your seminar series.  

It was also a pleasure to collaborate with Dr. Stefanos Delikaraoglou, Prof. Tyler Summers, and Prof. Gabriela Hug. I am grateful for all the time you dedicated to me. I would also like to thank Prof. Rico Zenklusen, Prof. Dan Molzahn, Prof. Patrick Panciatici, Prof. Hung-po Chao, Prof. Pierre Pinson, and Prof. Steven Low for well-timed discussions that shaped my research and my thinking.

Special thanks to Prof. John Lygeros for giving me the opportunity to be a part of IfA during my master's degree. I can't thank you enough for all the time you dedicated to me while I was planning my career steps both during the master's and during PhD. I am also grateful for the trust you had in me while I was the head TA of Linear System Theory. 

I would also like to thank Prof. Roy Smith, Prof. Florian D\"orfler, and Prof. Manfred Morari for creating an amazing atmosphere for research at IfA, full of openness and curiosity, and for inspiring me with their approach to research. It was a great privilege to receive your support and advice throughout this journey.

Prof. Angelos Georghiou was the very first person that introduced me to thinking like a researcher. I don't think I can run out of different aspects that I need to thank you for. All our discussions were extremely valuable for my growth.

Similarly, I don't think I can thank Dr. Paul Beuchat, Dr. Georgios Darivianakis, and Dr. Xiaojing Zhang enough for their support to me as their master's student, and their suggestions and friendships during my PhD studies. 

I would also like to thank Prof. Hamdi Torun, Prof. Ali Emre Pusane, Prof. Kadri Ozcaldiran, Prof. Selim Hacısalihzade, Prof. Heba Yuksel, Prof. Yagmur Denizhan and many others from my bachelor's program for their early career planning advice/guidance and for introducing me to control theory.

I owe a big thanks to Luca, who was a flat-mate, a office-mate, a conference-mate, and the list goes on and on! I am also thankful for all the great time I had discussing research with Pier Giuseppe, Ilnura, and Yimeng especially in our weekly meetings. I am very fortunate for being part of IfA, and I would like to thank all of my colleagues for all the memorable moments together. I would like to express my greatest gratitude to all those people that made this place what it is. I will use my chance to thank all of you in person! I owe a big thanks to Sabrina and Tanja for being extremely helpful with the littlest of problems I had. On a final note, I want to thank my master's students Baiwei, Anna, Petros, Lukas, and Daniel for all the enthusiasm and dedication you brought to your research. 

On a more personal level, I thank all my close friends in Zurich, and also close friends from high school and university. I want to thank you for not giving up on me. I think this sentence would resonate well with you considering especially this last year.

It is also difficult for me to find the right words to thank my family. I wouldn't get to this point without the encouragement and the freedom you gave me.

Last but not least, I am very grateful to my fianc\'ee Jana. At the end of your PhD, you should instead be receiving two doctoral degrees. Thanks for the unconditional love \& support and for bearing with me on my journey to becoming a good researcher.

\begin{flushright}
	Orcun Karaca\\
	Zürich, October 2020
\end{flushright}
\chapter*{Abstract}
\addcontentsline{toc}{chapter}{Abstract}

The liberalization of the energy sector led to the development of electricity markets that would improve economic efficiency and attract new investments to the grid. Although the specific structures and rules of these markets are quite diverse around the world, they were all conceived on the premise of predictable and fully controllable generation with nonnegligible marginal costs. Recent changes, specifically, the increasing renewable energy integration, have progressively challenged such operating assumptions. As this integration became deeper, transmission grids began to experience congestion in unforeseen and uncertain patterns. Moreover, these trends have resulted in substantial out-of-market transactions, questioning now the modeling and the management of electricity markets. In light of this operational paradigm shift, this thesis intends to devise new market frameworks and advance our understanding of the future electricity markets.


In the first part of the thesis, we focus on mechanism design when the market model fully reflects the physics of the grid and the participants. Specifically, we consider an electricity market setting that involves continuous values of different kinds of goods, general nonconvex constraints, and second stage costs. We seek to design the payment rules and conditions under which coalitions of participants cannot influence the market outcome in order to obtain higher collective utility. 
Under the incentive-compatible Vickrey-Clarke-Groves mechanism, our first contribution is to prove that such coalition-proof outcomes are achieved if the submitted bids are convex and the constraint sets are of a polymatroid-type. These conditions, however, do not capture the complexity of the general class of auctions under consideration. By relaxing the property of incentive-compatibility, we investigate further payment rules, called the core-selecting mechanisms, that are coalition-proof without any extra conditions on the submitted bids and the constraint sets. We show that core-selecting mechanisms generalize the economic rationale of the locational marginal pricing (LMP) mechanism. Namely, these mechanisms are the exact class of mechanisms that ensure the existence of a competitive equilibrium in linear/nonlinear prices. This implies that the LMP mechanism is also core-selecting, and hence coalition-proof. In contrast to the LMP mechanism, core-selecting mechanisms exist for a broad class of electricity markets, such as the ones involving nonconvex costs and nonconvex constraint sets. In addition, they can approximate truthfulness without the price-taking assumption of the LMP mechanism. Finally, we show that they are also budget-balanced. 

In the second part of the thesis, we turn our attention to the coordination of regional markets in the spatial domain to exploit the geographic diversification of the renewable resources. In particular, the establishment of a single European day-ahead market has accomplished the integration of the regional day-ahead markets. However, reserves provision and activation remain an exclusive responsibility of regional operators. This limited spatial coordination and the sequential structure hinder the efficient utilization of flexible generation and transmission, since their capacities have to be ex-ante allocated between energy and reserves. {To promote reserve exchange, a recent study proposed a preemptive model that withdraws a portion of the inter-area transmission capacity available from day-ahead energy for reserves by minimizing the expected system cost.} This decision-support tool, formulated as a stochastic bilevel program, respects the current architecture but does not suggest area-specific costs that guarantee sufficient benefits for areas to accept the solution. To this end, our main contribution is to formulate a preemptive model in a framework that allows application of coalitional game theory methods to obtain a stable benefit allocation, that is, an outcome immune to coalitional deviations ensuring willingness of areas to coordinate. We show that benefit allocation mechanisms can be formulated either at the day-ahead or the real-time stages, in order to distribute the expected or the scenario-specific benefits, respectively. For both games, the proposed least-core benefits achieve minimal stability violation, while allowing for a tractable computation with limited queries to the bilevel program. Our case studies, based on an illustrative and a more realistic test case, compare our method with well-studied benefit allocations, namely, the Shapley value and nucleolus. {The upshot of our contribution is to analyze the factors that drive these benefit allocations (e.g., flexibility, network structure, wind correlations).} 

\chapter*{Sommario}
\addcontentsline{toc}{chapter}{Sommario}

La liberalizzazione del settore energetico ha portato allo sviluppo di mercati dell'energia elettrica che migliorerebbero l'efficienza economica e attirerebbero nuovi investimenti nella rete. Sebbene le strutture e le regole specifiche di questi mercati siano piuttosto diverse in tutto il mondo, sono state tutte concepite sulla premessa di una generazione prevedibile e completamente controllabile con costi marginali non trascurabili. I recenti cambiamenti, in particolare la crescente integrazione delle energie rinnovabili, hanno progressivamente messo in discussione tali presupposti operativi. Con l'avanzamento di questa integrazione, le reti di trasmissione hanno iniziato a congestionarsi in maniera imprevedibile ed incerta. Inoltre, queste tendenze hanno portato a sostanziali transazioni fuori mercato, mettendo ora in discussione la modellazione e la gestione dei mercati dell'energia elettrica. Alla luce di questo cambiamento di paradigma operativo, questa tesi intende elaborare nuovi quadri di mercato ed avanzare la nostra conoscenza sui futuri mercati dell'elettricità.

Nella prima parte della tesi, ci concentriamo sulla progettazione di meccanismi nel caso in cui il modello di mercato riflette pienamente la fisica della rete e dei partecipanti. In particolare, consideriamo un quadro di mercato dell'elettricità che comporta valori continui di diversi tipi di beni, vincoli generali non convessi e costi di seconda fase. Cerchiamo di progettare le regole di pagamento e le condizioni in cui le coalizioni di partecipanti non possono influenzare il risultato del mercato per ottenere una maggiore utilità collettiva. Nell'ambito del meccanismo ``incentivo-compatibile" di Vickrey-Clarke-Groves, il nostro primo contributo consiste nel dimostrare che tali risultati a prova di coalizione si ottengono se le offerte presentate sono convesse e i vincoli sono di tipo polimatroide. Queste condizioni, tuttavia, non catturano la complessità della classe generale delle aste in esame. Rilassando la proprietà di meccanismo incentivo-compatibile, investighiamo su ulteriori regole di pagamento, chiamate meccanismi di selezione di core, che sono a prova di coalizione senza condizioni aggiuntive sulle offerte presentate e sui set di vincoli. Dimostriamo che i meccanismi di selezione di core generalizzano la logica economica del meccanismo di determinazione del prezzo marginale per località (LMP). Vale a dire, questi meccanismi sono l'esatta classe di meccanismi che assicurano l'esistenza di un equilibrio competitivo nei prezzi lineari/non lineari. Ciò implica che il meccanismo di LMP è anche un meccanismo di selezione di core, e quindi a prova di coalizione. A differenza del meccanismo LMP, meccanismi di selezione di core esistono per un'ampia classe di mercati dell'elettricità, come quelli che comportano costi non convessi e vincoli non convessi. Inoltre, essi possono approssimare la ``truthfulness" senza l'assunzione del ``price-taker" del meccanismo LMP. Infine, dimostriamo che sono anche equilibrati dal punto di vista del bilancio.

Nella seconda parte della tesi, rivolgiamo la nostra attenzione al coordinamento dei mercati regionali in ambito spaziale per sfruttare la diversificazione geografica delle risorse rinnovabili. In particolare, la creazione di un mercato unico europeo ``day-ahead" ha realizzato l'integrazione dei mercati regionali ``day-ahead". Tuttavia, la fornitura e l'attivazione delle riserve rimangono di esclusiva competenza degli operatori regionali. L'utilizzo efficiente di generazione e trasmissione flessibili è ostacolato da questo coordinamento spaziale limitato e dalla struttura sequenziale, poiché le loro capacità devono essere ripartite ex ante tra energia e riserve.Per promuovere lo scambio di riserve, un recente studio ha proposto un modello preventivo che ritira una parte della capacità di trasmissione inter-area disponibile dall'energia ``day-ahead" per le riserve, minimizzando il costo previsto del sistema.  Questo strumento di supporto alle decisioni, formulato come programma bilivello stocastico, rispetta l'architettura attuale ma non suggerisce costi specifici per area che garantiscano sufficienti benefici alle aree per accettare la soluzione. A tal fine, il nostro contributo principale è quello di formulare un modello preventivo in un quadro che consenta l'applicazione dei metodi della teoria dei giochi coalizionali per ottenere una stabile allocazione dei benefici, cioè un risultato immune da deviazioni coalizionali che garantisca la disponibilità delle aree a coordinarsi. Mostriamo che i meccanismi di allocazione dei benefici possono essere formulati sia nella fase iniziale che in quella in tempo reale, al fine di distribuire i benefici attesi o i benefici specifici dello scenario, rispettivamente. In entrambi i casi, i benefici least-core proposti raggiungono una violazione minima della stabilità, consentendo al tempo stesso un calcolo efficiente con query limitate al programma bilivello.  I nostri casi di studio, basati su un caso di prova illustrativo e più realistico, confrontano il nostro metodo con l'assegnazione di benefici ben studiati, ovvero il valore di Shapley e il nucleolo. Il risultato del nostro contributo è l'analisi dei fattori che guidano queste assegnazioni di benefici (ad esempio, la flessibilità, la struttura della rete, la correlazione col vento).


\tableofcontents

\chapter*{Acronyms and notation}
\addcontentsline{toc}{chapter}{Acronyms and notation}
\printacronyms
\acuseall
For the sake of clarity, some of the acronyms above will be defined again when they appear in the dissertation for the first time.

\subsection*{Mathematical symbols}
\begin{longtable}[l]{l l}
$\R$, $\R_+$ & set of real, and nonnegative real numbers\\
$\mathbb{Z}$, $\mathbb{Z}_+$ & set of integers, and nonnegative integers\\
$[a,b]$ & interval of real numbers $x\in\R$ with $a\le x\le b$\\
$[m]$ & set of all positive integers up to integer $m$\\
$2^A$ & power set of set $A$\\
 $\|x\|_p$ & $p$-norm of $x\in{\R}^n$ \\
 $\mathbb{E}_s[\cdot] $ & expectation calculated over the distribution of the random variable $s$\\
 $f(x)=\mathcal{O}(g(x))$ & big O notation\\
 $\perp$ & complementarity condition\\
\end{longtable}

As a remark, we will reserve additional symbols throughout \cref{part:1,part:2}. These symbols are valid only within their corresponding parts. For instance, calligraphic capital letters denote bid profiles in~\cref{part:1}, whereas they denote sets in~\cref{part:2}. On the other hand, capital letters denote sets in \cref{part:1}, whereas they denote market parameters in~\cref{part:2}. All the necessary clarifications are provided in \cref{sec:p1_2,sec:elecmarkframe}.

\mainmatter 
\chapter{Overview}
\label{ch:overview}
The deregulation of the electricity industry in the '90s led to the development of electricity markets, which has been essential to improve economic efficiency and to attract new investments.
Although the specific structures and rules of these markets are quite diverse around the world, 
they were all conceived on the premise of predictable and fully controllable generation with nonnegligible marginal costs. 
In response to the increasing efforts to reduce CO$_2$ emission, and in order to ensure more sustainable energy use,  there is a rapid transformation in the generation type, namely with the growth of renewable resources. For instance, California mandates 50\% renewable resources by 2025, and 100\% by 2045, and Germany embraced a policy of no active nuclear plants by 2022~\cite{schulte2019100}, see also~\cref{fig:2040ENTSOE}.
These resources pose new challenges both in terms of control and market design, because they are uncertain, intermittent, mostly uncontrollable, and accompanied with zero marginal costs.
In light of this paradigm shift, there is an imperative need to devise new market frameworks to simultaneously incorporate traditional thermal plants (large and inflexible with slow ramp rates), renewable resources, and hydropower plants and other flexible plants such as nuclear power (with high costs). To this end, there has been a surge of interest from academics, industry as well as policymakers~\cite{wilson2002architecture,neuhoff2005large,ahlstrom2015evolution,cramton2017electricity,bose2019some}.

\begin{figure}[h]
    \centering
    \includegraphics[width=0.85\columnwidth]{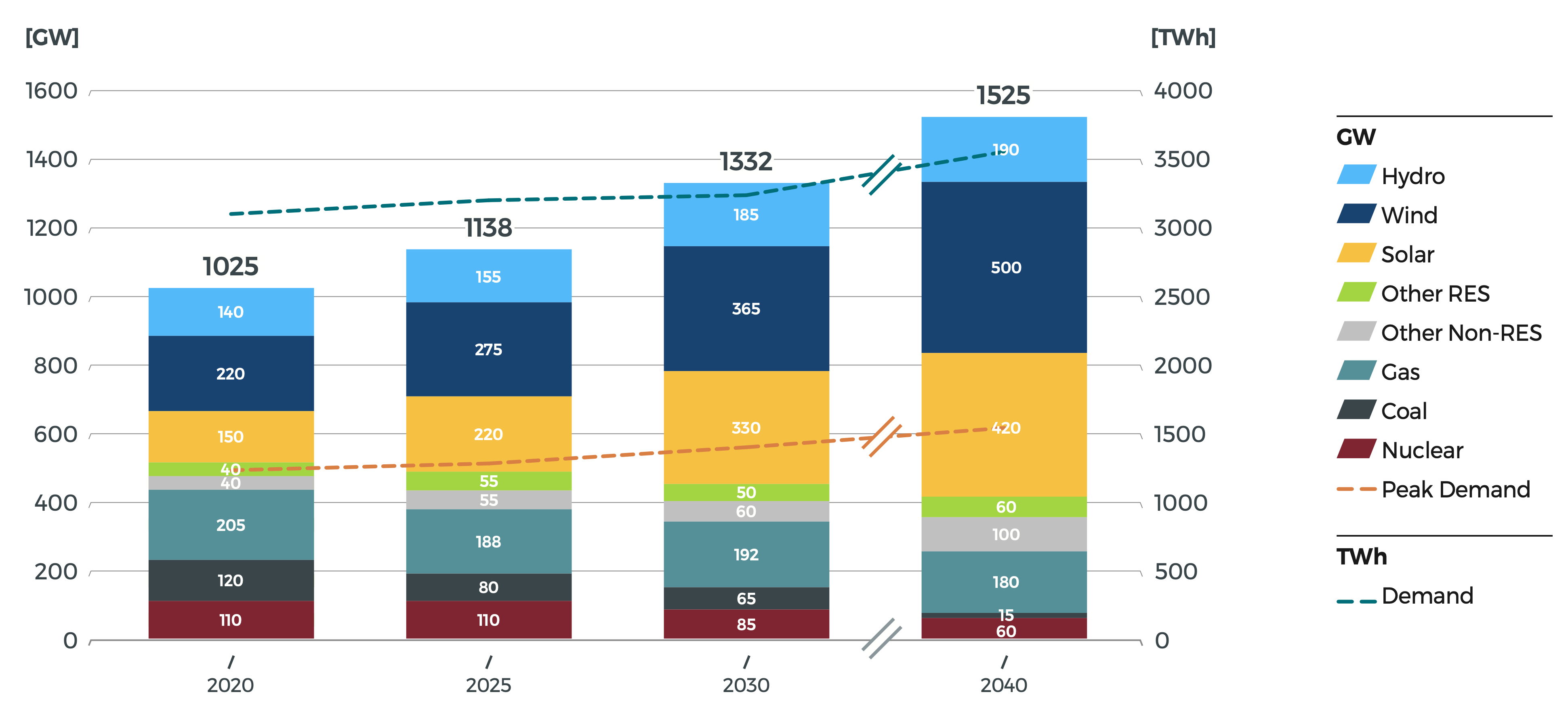}%
    \caption{Installed capacity [GW] and demand [TWh] scenario according to the development plan of ENTSO-E (European Network of Transmission System Operators for Electricity)~\cite{elia20}}
    \label{fig:2040ENTSOE}
\end{figure}

One prevailing issue originates from the fact that the existing electricity markets do not fully reflect the physics of the underlying grid and the technical constraints of the market participants. For instance, power flow equations in their full generality (modeling both transmission losses and reactive power) render the optimization problem nonconvex, which is a hard problem not only to solve but also to compute meaningful prices for. Ignoring such complexities results in inefficient market outcomes and unpriced goods (e.g., reactive power), which in turn necessitates out-of-market transactions. Thanks to the development of better numerical tools, it is now possible to incorporate more accurate models of both the grid and the participants. 
One main challenge is that many existing electricity markets rely on defining payments using the Lagrange multipliers that are well-defined only under convexity assumptions.
To address this challenge, \textit{the first goal of the thesis} is to design payment rules with desirable properties for an electricity market setting, which can be modeled by an auction involving continuous goods, second stage costs, and general nonconvex constraints.\footnote{Auctions, in particular, are effective tools for allocating resources and determining their values among a set of participants. Prominent auction examples, other than those found in electricity markets, are those in spectrum auctions~\citep{cramton2013spectrum,bichler2017handbook}, and auctions for fish harvesting rights and other natural capitals~\citep{bichler2019designing,teytelboym2019natural}.}  This is especially an exciting time to consider and study tools from mechanism design and auction theory. At the time of the submission of this dissertation (October 2020), The 2020 Sveriges Riksbank Prize in Economic Sciences in Memory of Alfred Nobel\footnote{\url{https://www.nobelprize.org/prizes/economic-sciences/2020/summary/}} has just been awarded to Professor Paul Milgrom and Professor Robert Wilson for their groundbreaking contributions to auction theory and inventions of new auction formats. \looseness=-1

Another issue, specifically pertaining to the European electricity markets, is that these markets suffer from partial coordination in space. Although the day-ahead energy markets are jointly cleared, the reserve capacity and real-time balancing markets are still operated on a regional (country) level. If we can remove the existing barriers for cross-border trading between regional operators, geographic diversification of the uncertain renewable resources would smooth out the forecast errors and reduce the need for real-time balancing actions. One main challenge is that, because of the sequential market architecture, the joint-clearing of reserve markets will require allocating a portion of the cross-border transmission capacity from the day-ahead market to the reserve market. To this end, \textit{the second goal of the thesis} is to design benefit allocation mechanisms (that is, novel discount terms for regional costs) such that the resulting region-specific costs make the overall mechanism immune to coalitional deviations ensuring that all regional operators are willing to coordinate via the cross-border transmission capacities we plan to propose.\looseness=-1

\section{Outline and contributions}
\subsection{\cref{part:1}}
In the first part of the thesis, our goal is to study mechanism design to achieve an efficient outcome in an electricity market setting, which involves continuous goods (e.g., electrical power), second stage costs, and general nonlinear constraints (instead of simple constraints, e.g., availability of a fixed number of items in multi-item auctions).

\subsubsection*{Outline}
\cref{ch:p1introduction} provides an in-depth literature review on different aspects of this problem, and motivates our goal.
\cref{sec:p1_2} introduces a general class of electricity markets and discusses desirable properties for mechanisms.  \cref{sec:p1_3} brings in tools from coalitional game theory, namely the {core}. Throughout this chapter, we investigate conditions under which the incentive-compatible VCG mechanism is coalition-proof, that is, immune to coalitional manipulations. Since these conditions do not capture the complexity of the general class of electricity markets, coalition-proof core-selecting payment rules are proposed in \cref{sec:p1_4}. Using tools from coalitional game theory and competitive equilibrium theory, we then prove the equivalence of core and competitive equilibrium. We investigate incentive-compatibility and budget-balance. Finally, \cref{sec:p1_5} presents case studies based on real-world electricity market data. 

\subsubsection*{Overview of contributions}
The contributions of \cref{part:1} are as follows. These contributions are presented in more detail with their connections to the existing results in the literature in \cref{ch:p1introduction}. 
\begin{enumerate}
\item In~\cref{sec:p1_3}, we prove that in the electricity market setting we consider the VCG mechanism is coalition-proof and the VCG utilities lie in the core, if and only if the market objective function~is supermodular. We then derive novel conditions on the bids and the constraint sets based on polymatroid theory such that the VCG mechanism is coalition-proof.

\item In \cref{sec:p1_4}, we show that selecting payments from the core results in a coalition-proof mechanism without any restrictions on~the bids and the constraints. We prove that for electricity markets any competitive equilibrium is efficient. We then establish that a mechanism is core-selecting if and only if it ensures the existence of a competitive equilibrium. This equivalence implies that the LMP mechanism is also a core-selecting mechanism. In the remainder of this chapter, we derive an upper bound on the additional profit a bidder can obtain by a unilateral deviation from its truthful bid, under any core-selecting mechanism. Using this bound, we propose a mechanism that maximizes incentive-compatibility among all core-selecting mechanisms. In addition, we show that any core-selecting mechanism is budget-balanced when we extend our results to the exchange setting. 

\item In \cref{sec:p1_5}, we verify our results with case studies based on real-world electricity market data, including the Swiss reserve market and different optimal power flow models.
\end{enumerate}
\subsection{\cref{part:2}}

In the second part of the thesis, our goal is to propose a coalitional game-theoretic approach to enable coordinated balancing and reserve exchange in a European level and to gain technical and economical insights about such a process.
\subsubsection*{Outline}
\cref{secintro} provides an in-depth literature review on different aspects of this problem, and motivates our goal.
\cref{sec:elecmarkframe} describes the organizational structure of the European electricity markets and introduces a set of necessary assumptions to obtain tractable models. \cref{sec:transcapalloc} discusses the issues related to reserve exchanges {and motivates the formulation of} the preemptive transmission allocation model. \cref{sec:3} {introduces necessary background from} coalitional game theory, whereas \cref{sec:4} focuses on the games arising from this preemptive transmission allocation model, {which provide the basis for the} benefit allocation mechanisms that accomplish the implicit coordination requirements outlined in the previous sections.
The numerical case studies are presented in \cref{sec:CaseStudies}.

\subsubsection*{Overview of contributions}
The contributions of \cref{part:2} are as follows. These contributions are presented in more detail with their connections to the existing literature in \cref{secintro}. 
\begin{enumerate}
\item In~\cref{sec:transcapalloc}, we formulate the coalition-dependent version of the preemptive transmission allocation model such that we can consider coalitional arrangements between only a subset of operators.

\item In~\cref{sec:4}, we first study the coalitional game that treats the benefits as an ex-ante process with respect to the uncertainty realization {and} we provide a condition under which the core is nonempty. Under this condition, it is possible to obtain a stable benefit outcome, that is, an outcome from the core ensuring the willingness of all areas to coordinate. In case this condition is not satisfied, we prove that the least-core, which is an outcome that attains minimal stability violation, also ensures the individual rationality property. We then propose the least-core selecting mechanism as a benefit allocation that achieves minimal stability violation, while enabling the approximation of an additional fairness criterion. In order to {implement} this mechanism with only a few queries to the preemptive model, we formulate a constraint generation algorithm. 
In addition, we formulate a variation of the coalitional game that {allocates} the benefits {in} an ex-post process. For this game, we provide conditions under which the core is empty. Finally, we propose an ex-post version of our least-core benefit allocation mechanism. 

\item In \cref{sec:CaseStudies}, we provide techno-economic insights on the factors that drive benefit allocations first with an illustrative three-area nine-node system and then with a more realistic case study based on a larger IEEE test system.
\end{enumerate}
\section{Publications}
This thesis contains a selected collection of results derived during the author's studies as a Ph.D. candidate. The corresponding articles on which this thesis is based are listed below.

\subsection{\cref{part:1}}
The results on the design of coalition-proof payment rules were developed in collaboration with N. Walton, P. G. Sessa, and M. Kamgarpour. The results on the equivalence of the core and competitive equilibrium were developed in collaboration with M. Kamgarpour. \\
\begin{itemize}
\item[\cite{karaca2019designing}] ``\emph{Designing coalition-proof  reverse  auctions over continuous goods}'', \textbf{O. Karaca}, P. G. Sessa, N. Walton, and M. Kamgarpour, In IEEE  Transactions  on Automatic Control, 2019.
\item[\cite{karaca2018core}] ``\emph{Core-selecting mechanisms in electricity markets}'', \textbf{O. Karaca}, and M. Kamgarpour, In IEEE  Transactions  on  Smart  Grid, 2020.
\end{itemize}

\subsection{\cref{part:2}}
The results on benefit allocation mechanisms for enabling reserve exchanges were developed in collaboration with S. Delikaraoglou, G. Hug, and M. Kamgarpour.\\
\begin{itemize}
    \item[\cite{karaca2019benefits}] ``\emph{Enabling inter-area reserves exchanges through stable benefit allocation mechanisms}'', \textbf{O. Karaca}, S. Delikaraoglou, G. Hug, and M. Kamgarpour, Submitted, 2020.
\end{itemize}

\subsection{Other publications}
The following papers were either submitted or published by the author during his doctoral studies, but are not included in the thesis. \\

The main results of \cite{orcun2018game}, \cite{karaca2018weak}, and \cite{karaca2020regret}, that is, the first three references listed below, are mentioned in \cref{part:1}, but are not treated in detail.\\

\begin{itemize}
    \item[\cite{orcun2018game}] ``\emph{Game theoretic analysis of electricity market auction mechanisms}'', \textbf{O. Karaca}, and M. Kamgarpour, In Proceedings of the IEEE Conference on Decision and Control, 2017,
    \item[\cite{karaca2018weak}] ``\emph{Exploiting weak supermodularity for coalition-proof mechanisms}'', \textbf{O. Karaca}, and M. Kamgarpour, In Proceedings of the IEEE Conference on Decision and Control, 2018,
    \item[\cite{karaca2020regret}] ``\emph{No-regret learning from partially observed data in repeated auctions}'', \textbf{O. Karaca*}, P. G. Sessa*, A. Leidi, and M. Kamgarpour, In Proceedings of the IFAC World Congress, 2020,
\end{itemize}

The results of the following papers are not mentioned. The first three references concern structural control theory and the approximation guarantees for (co)matroid optimization problems. The last two references are on the relaxation hierarchies for polynomial optimization problems.\\

\begin{itemize}    
    \item[\cite{guo2019actuator}] ``\emph{Actuator placement for optimizing network performance under controllability constraints}'', B. Guo, \textbf{O. Karaca}, T. Summers, and M. Kamgarpour, In Proceedings of the IEEE Conference on Decision and Control, 2019,
    \item[\cite{karaca2019}] ``\emph{A comment on performance guarantees of a greedy algorithm for minimizing a supermodular set function on comatroid}'', \textbf{O. Karaca}, B. Guo, and M. Kamgarpour, In European Journal of Operational Research, 2020,
    \item[\cite{guo2019placement}] ``\emph{Actuator placement under structural controllability using forward and reverse greedy algorithms}'', \textbf{O. Karaca*}, B. Guo*, T. Summers, and M. Kamgarpour, In IEEE Transactions on Automatic Control, 2020,
    \item[\cite{wachter2020}] ``\emph{A convex relaxation approach for the optimized pulse pattern problem}'', L. Wachter, \textbf{O. Karaca}, G. Darivianakis, T. Charalambous, Submitted, 2020,
    \item[\cite{karaca2017repop}] ``\emph{The REPOP toolbox: Tackling polynomial optimization using relative entropy relaxations}'', \textbf{O. Karaca}, G. Darivianakis, P. N. Beuchat, A. Georghiou, and J. Lygeros, In Proceedings of the IFAC World Congress, 2017.
\end{itemize}
{\footnotesize{* indicates equal contribution.}}


\part{Designing coalition-proof auctions over continuous goods: The case of electricity markets}

\label{part:1}
\chapter{Introduction}
\label{ch:p1introduction}

A rapid transformation has been underway since the early '90s to replace the tight regulation of the electricity industry with competitive market structures~\cite{wilson2002architecture}. This liberalization has been essential to improve economic efficiency and to attract new investments to the grid~\cite{cramton2017electricity}. Designing electricity markets, however, is a complex task. One inherent complexity is the need to achieve real-time balance of supply and demand because of an inability to store electricity efficiently~\cite{cramton2003electricity}. This task is made more difficult by both intertemporal and network dependencies, and more recently, by high penetration of renewable resources. In contrast to the conventional generators, these renewable resources are uncertain, intermittent, and accompanied with zero marginal costs. 
As a result, there is an essential need to devise new electricity markets to simultaneously incorporate traditional thermal plants (large and inflexible with slow ramp rates), renewable resources, and hydropower plants and other flexible plants such as nuclear power (with high costs).\footnote{These markets are also required to encourage the integration of smart meters and demand response programs as a substitute for flexible generation, e.g., by utilizing the charging profiles of electric vehicles~\cite{rious2015electricity,li2019transactive}. For the ease of discussion and presentation, this part will focus mainly on the complexities originating from the supply/seller side.} To this end, there has been a great amount of interest from academics, industry, as well as policymakers, amassing a significant amount of studies in the field~\cite{ahlstrom2015evolution,bose2019some}. 

The goal of \cref{part:1} of this thesis is to study mechanism design for electricity markets in which generators first submit bids representing their underlying economic costs, and a central operator then optimizes all the resources to secure a reliable grid operation.\footnote{This setting involving multiple sellers but a single buyer (in our case, the central operator) is called a reverse auction. On the other hand, a forward auction involves a single seller and multiple buyers, whereas an exchange involves multiple sellers, buyers, and potentially trading participants selling some goods and buying others simultaneously.} The principal element of such electricity markets is the payment made to each generator since these participants have incentives to strategize around~these payments. In particular, the central operator needs to carefully design the payment rule to ensure an efficient outcome, that is, an outcome maximizing social welfare. This goal {is best} achieved if the central operator solves for the optimal allocation under the condition that all the market participants agreed to reveal their true costs to the central operator.

\section{Related works}
\subsubsection*{Existing works on the design of payment rules}
The locational marginal pricing (LMP) mechanism is a well-studied payment rule used in many existing electricity markets~\cite{schweppe2013spot,alsac1990further,wu1996folk,hogan1992contract}. It is based on using the Lagrange multipliers of nodal balance equations in optimal power flow~(OPF) problems to form linear prices. This proposal, put forth in~\cite{bohn1984optimal}, is ubiquitous in the US market. On the other hand, a zonal abstraction is used in Europe by constraining the inter-zonal power flows with the available transfer capacity or with a flow-based domain.\footnote{The European electricity market idiosyncrasies are discussed in detail later in \cref{sec:elecmarkframe} of \cref{part:2}.} 

If generators assume that the Lagrange multipliers are independent of their bids, then the LMP mechanism is incentive-compatible, that is, generators are incentivized to submit their true costs. This assumption, also called price-taking, arises from competitive equilibrium theory; however, it often does not hold in practice~\cite{mas1995microeconomic}. In particular, empirical evidence has shown that strategic manipulations have increased the LMP payments substantially in electricity markets~\cite{joskow2001quantitative}. For instance, many studies attribute the California electricity crisis of 2000-2001 to Enron's energy traders' manipulations~\cite{mccullough2002congestion}. Moreover, the LMP payments are well-defined only under convexity assumptions on the bids and the constraints~\cite{o2005efficient}. Convexity is a simplifying abstraction of many realistic grid and market models, see for example the models in~\cite{lavaei2012competitive,warrington2012market}.\footnote{Convex electricity market models do not fully reflect the physics of the underlying grid and the discrete nature of typical technical constraints of the market participants (e.g., minimum energy output levels and startup costs). For instance, power flow equations in their full generality (modeling both transmission losses and reactive power) render the optimization problem nonconvex~\cite{lavaei2012zero}. Because convex formulations ignore reactive power (as a supply/good in the auction), the operator has to acquire this service through out-of-market transactions. PJM Interconnection spent $\$342$ million in 2018 for this purpose alone~\cite{winnicki2019convex}. Given that high penetration of solar
generation requires more and more reactive power, we need to
ensure reactive power capabilities exist. } Without such restrictions, it is not possible to guarantee the existence of meaningful Lagrange multipliers~\cite{bikhchandani1997competitive}. To circumvent this difficulty, some real-world electricity markets compute linear prices by some convexification/approximation method and then complement these linear prices with side payments. These side payments go under the name of {uplift}, following a nomenclature established during the UK electricity market restructuring~\cite{hogan2003minimum}. Essentially, the uplift is equivalent to either the deficit or the opportunity cost of the participant when using the linear prices derived from the approximate models. However, when such approximations are introduced the incentive-compatibility issue becomes even more concerning, since these uplift payments are known to create additional incentives for manipulative behavior~\cite{chao2019incentives,liberopoulos2016critical}.

In contrast to the LMP mechanism, the Vickrey-Clarke-Groves (VCG) mechanism ensures that truthful bidding is the dominant-strategy Nash equilibrium~\cite{vickrey1961counterspeculation,clarke1971multipart,groves1973incentives}. Consequently, several recent works have proposed the use of this payment rule in a broad class of electricity market problems~\cite{samadi2012advanced, pgs,xu2017efficient}.  However, the VCG mechanism is often deemed undesirable for practical applications since coalitions of generators can strategically bid to increase their collective utility. As a result, it is susceptible to different kinds of manipulations such as collusion and shill bidding~\cite{hobbs2000evaluation,ausubel2006lovely,yokoo2004effect}.\footnote{Since the same participants are involved in similar market transactions day after day, electricity markets can particularly be exposed to collusion and shill bidding~\cite{anderson2011implicit}.} These shortcomings are decisive in the practicality of the VCG~mechanism since in a larger context the VCG~mechanism is not truthful. As a result, practical applications of the VCG~mechanism in real commerce are rare at best~\cite{rothkopf2007thirteen}.

The shortcomings described above occur when the VCG utilities are not in the~\textit{core}, as it is outlined in combinatorial auction literature \cite{ausubel2006lovely,milgrom2004putting,milgrom2017discovering}.\footnote{In the remainder, we use the terms combinatorial auction and multi-item auction interchangeably to refer to the case with discrete goods/items.} The core is a concept from coalitional game theory where the participants have no incentives to leave the grand coalition, that is, the coalition of all participants \cite{osborne1994course,peleg2007introduction}.
 Recently, coalitional game theory has received attention from different communities, e.g., for aggregating power generators~\cite{baeyens2011wind}, deriving control policies for multi-agent systems~\cite{maestre2014coalitional}, and sharing storage devices~\cite{chakraborty2018sharing}. In this part, we use coalitional game theory to ensure that the VCG mechanism is coalition-proof, in other words, collusion and shill bidding are not profitable. To this end, we derive conditions on the submitted bids and the constraint sets of the market that ensure core VCG utilities by utilizing some recent advances from combinatorial optimization literature. We show that under separable convex bids (or marginally increasing in the discrete/quantized case) and~polymatroid-type constraints the VCG mechanism is coalition-proof.

The restricted market setting for core VCG utilities, however, does not capture the complexity of general auctions arising in electricity markets. Specifically, these markets may involve nonconvex bids (e.g., startup costs), and complex constraint sets that are not polymatroids (e.g., DC or AC optimal power flow constraints). Hence, it may not be possible to ensure core VCG utilities. To this end, we focus on payment rules that are coalition-proof without any extra conditions on the bids and the constraints. These payment rules are referred to as core-selecting mechanisms, and they were first proposed for multi-item auctions~\cite{day2008core}. Computational difficulties in finding a core outcome was addressed in \cite{day2007fair}. Further studies have shown that different core outcomes can be chosen for robustness and fairness criteria \cite{erdil2010new, day2012quadratic}. A closed-form analytical Bayes-Nash equilibrium (BNE) analysis was carried out in a three-bidder two-item auction in \cite{ausubel2010core}. This study was extended with novel algorithmic frameworks for computational BNE analyses and the design of large multi-item auctions in the works of~\cite{bosshard2017computing,bunz2018designing,bosshard2018computing,bosshard2018non}. In this part of the thesis, we generalize the coalition-proofness of core-selecting mechanisms, developed and analyzed for multi-item auctions, to an electricity market setting, which involves continuous goods (e.g., electrical power), second stage costs, and general nonlinear constraints (instead of simple constraints, e.g., availability of a fixed number of items in multi-item auctions).

We then show that core-selecting mechanisms are the exact class of mechanisms that ensure the existence of a competitive equilibrium under linear/nonlinear prices in the electricity market setting described above. Our result implies that the LMP mechanism is also core-selecting, and hence coalition-proof whenever strong duality holds (for instance, in the convex setting). The equivalence of competitive equilibrium and the core was first shown in~\cite{shapley1971assignment} for item exchanges under unit demand and unit supply, for example, house allocation problems. The work in~\cite{bikhchandani2002package} characterizes competitive equilibria in multi-item auction problems with a simple supply-demand balancing equality constraint. However, similar to the past work in~\cite{shapley1971assignment}, the proof does not readily apply to the electricity market setting we consider, since this result utilizes linear-programming duality in application to appropriate linear-programming reformulations of the multi-item auction problems with simple constraints (see also the proof in~\cite{parkes2002indirect}). Moreover, previous works for core-selecting mechanisms in multi-item auctions such as \cite{day2007fair,day2008core,day2012quadratic,erdil2010new} do not explicitly address the connection between the core and the competitive equilibrium. 

Complementing the existing works on the LMP mechanism for electricity markets, we further highlight that core-selecting mechanisms are applicable to a broad class of electricity markets, such as the ones featuring nonconvex costs and/or nonconvex constraint sets, whereas the LMP mechanism has applicability only for the case of convex bids and constraints.
Naturally, core-selecting mechanisms relax the incentive-compatibility property of the VCG mechanism. In order to alleviate this issue, we prove that core-selecting mechanisms can approximate incentive-compatibility without relying on the price-taking assumption of the LMP mechanism, while achieving the budget-balance property in an exchange setting.\footnote{Note that the benefits of core-selecting mechanisms are accompanied by nonlinear pricing which might be regarded as a big shift for some existing electricity markets which have been cleared by linear nodal prices throughout the last decade~\cite{bose2019some}.}
 
\subsubsection*{Other related works from mechanism design theory}
Let us contrast our work with other existing mechanism design research.
The authors in~\cite{lazar2001design} design a forward VCG auction for continuous goods by restricting each participant to submitting a single price-quantity pair to the operator. They show that this mechanism, called progressive second price (PSP) mechanism, has a truthful $\epsilon$-Nash equilibrium, which can be attained by best-response dynamics. The work in~\cite{jia2010analysis} studies the PSP mechanism subject to quantized pricing assumptions, and proposes a fast algorithm converging to a quantized Nash equilibrium with a high probability. More recently, the work in \cite{zou2017efficient} addresses cross-elasticity in PSP design arising from having a multi-period problem for charging electric vehicles. This work is generalized to decentralized procedures and double-sided auctions in \cite{zou2017resource}. In regard to these studies, the PSP mechanism cannot be implemented in dominant-strategies, and truthfulness is only in the price dimension for a given quantity. Specifically, in the PSP mechanism, there is no single true quantity to declare, and the optimal
		quantity depends on the bids of other participants~\cite[\S 3.2]{lazar2001design}. Moreover, the convergence analysis of the best response dynamics is limited to strongly concave true valuations and simple market constraints, for example, availability of a fixed amount of a single continuous good~\cite[\S 3.1]{lazar2001design}.
On the other hand, the work in~\cite{xu2017efficient} applies the VCG mechanism to the wholesale electricity markets and shows that it results in larger payments than the locational marginal pricing mechanism. There are also recent auction theory applications from control community in provisioning of a distributed database~\cite{sanghavi2008new}, and selecting a host for a noxious resource (e.g., trash disposal facility)~\cite{wang2017ex}.
Nevertheless, none of the aforementioned works consider coalitional manipulations. Finally, the works in \cite{marden2013overcoming,marden2014generalized,li2014decoupling} study the design of the participants' utilities such that the selfish behavior of the participants results in a social welfare maximizing outcome. By contrast, in our case, the true valuations of the participants are \textit{a priori} unknown and they are not part of the design. Instead, we are guiding the participants to a social welfare maximizing outcome by designing meaningful incentives through the payment rule. 


\section{Summary of goals and contributions}

The contributions of \cref{part:1} are as follows. 
\begin{enumerate}
\item We prove that in the electricity market setting we consider the VCG mechanism is coalition-proof and the VCG utilities lie in the core, if and only if the market objective function~is supermodular. These results are direct extensions of the results in \cite{ausubel2002ascending,ausubel2006lovely} from the multi-item forward auction setting.

\item Considering again the special setting of continuous goods, second stage costs, and complex constraints, we derive novel conditions on the bids and the constraint sets under which the VCG mechanism is coalition-proof.

\item We then show that selecting payments from the core results in a coalition-proof mechanism without any restrictions on~the bids and the constraints, extending results of~\cite{day2008core} from the multi-item~forward auction setting to the electricity market setting. 

\item We prove that for electricity markets any competitive equilibrium is efficient. This result extends the well-known first fundamental theorem of welfare economics stating that competitive markets tend towards an efficient outcome. 

\item We establish that a mechanism is core-selecting if and only if it ensures the existence of a competitive equilibrium. This result is novel for the electricity market setting we consider, and it also applies to the exchange setting. This equivalence implies that the LMP mechanism is also a core-selecting mechanism. 

\item We derive an upper bound on the additional profit a bidder can obtain by a unilateral deviation from its truthful bid, under any core-selecting mechanism. Using this bound, we propose a mechanism that maximizes incentive-compatibility among all core-selecting mechanisms. This result directly extends the previous proposals from the multi-item forward auction setting, such as \cite{day2007fair,day2012quadratic}. 

\item In addition, we show that any core-selecting mechanism is budget-balanced when we extend our results to the exchange setting. This result is novel since multi-item exchange literature defines the core without the central operator resulting in general in an empty core~\cite{hoffman2010practical,day2013division,milgrom2007package,bichler2017core}. 

\item Finally, we verify our results with case studies based on real-world electricity market data.
\end{enumerate}

\subsubsection*{Organization}

\cref{sec:p1_2} introduces a general class of electricity markets and discusses desirable properties for mechanisms.  \cref{sec:p1_3} brings in tools from coalitional game theory, namely the {core}. Throughout this chapter, we investigate conditions under which the VCG mechanism is coalition-proof. Since these conditions do not capture the complexity of the general class of electricity markets, alternative payment rules are proposed in \cref{sec:p1_4}. Using tools from coalitional game theory and competitive equilibrium theory, we then prove the equivalence of core and competitive equilibrium. We investigate incentive-compatibility and budget-balance. Finally, \cref{sec:p1_5} presents case studies based on real-world electricity market data. 
\chapter{Mechanism framework for electricity markets}\label{sec:p1_2}

We start with a generic (one-sided) electricity market reverse auction.\footnote{Our results can be generalized to exchanges, see the discussion provided in~\cref{sec:p1_4}.} The set of participants consists of the central operator $l=0$ and the bidders $L=\{1,\ldots,\lvert L\rvert\}$. Let there be $t$ types (or kinds) of power supplies in the auction. These types can include control reserves, also known as ancillary services~\cite{abbaspourtorbati2016swiss}, or active and reactive power injections differentiated by their nodes, durations, and scheduled times. Supplies of the same type from different bidders are fungible (that is, interchangeable) to the central operator. 

We assume that each bidder~$l$ has a private true cost~function $c_l: \X_l \rightarrow \mathbb R_+$, $\X_l\subseteq\R_+^t$. We further assume that $0\in \X_l$ and~$c_l(0)=0$. This assumption holds for many electricity markets, for instance, control reserve markets and day-ahead markets that include generators' start-up costs.  Each bidder~$l$ then submits a bid function to the central operator, denoted by $b_l:\hat \X_l \rightarrow \mathbb R_+$, where $0\in \hat \X_l\subseteq\R_+^t$ and~$b_l(0)=0$.\footnote{There are markets that include shut-down costs (that is, $b_l(0)>0$) or minimum output levels (that is, $b_l(0)$ is infinitely large or $0\notin \hat{\X}_l$). To address these markets, throughout this part we draw attention to the properties and the results for which the assumptions, $c_l(0)=0$, $b_l(0)=0$, are pivotal.}$^,$\footnote{Assume that both functions lie in a function space defined by the market rules.} 

Given the bid profile $\BB=\{b_l\}_{l\in L}$, \textit{a mechanism} defines an allocation rule $x_l^*(\BB)\in \hat \X_l$ and a payment rule $p_l(\BB)\in\R$ for each bidder $l$. 
In electricity markets, the allocation rule is generally determined by the economic dispatch, that is, minimizing the procurement cost subject to some security constraints
\begin{equation}\label{eq:main_model}
\begin{split}
J(\BB)=&\min_{x\in \hat \X,\,y}\,\, \sum\limits_{l\in L} b_l(x_l) + d(x,y)\\
&\ \ \mathrm{s.t.}\ \ h(x,y)= 0,\, g(x,y)\leq 0,\\
\end{split}
\end{equation}
where $\hat \X=\prod_{l\in L}\hat \X_{l}$. In the case of a two-stage electricity market model, the operator can buy the goods from another market at a later stage. The variables~$y\in\R^p$ may correspond to these second stage variables and the function~$d:\R^{t\rvert L\rvert}\times\R^p\rightarrow \R$ could represent the second stage cost. In~\cref{sec:p1_5}, we provide a real-world electricity market example where the function $d$ incorporates expected daily market prices in a weekly market. The function~$h:\R^{t\rvert L\rvert}\times\R^{p}\rightarrow \R^{q_1}$ defines the equality constraints and the function~$g:\R^{t\rvert L\rvert}\times\R^{p}\rightarrow \R^{q_2}$ defines the inequality constraints.\footnote{Since supplies of the same type from different bidders are fungible to the central operator, the functions $d$, $g$, and $h$ are in fact functions of $\sum_{l\in L} x_l$ and not $x$. For the sake of simplicity, we keep this general form and draw attention to the derivations whenever this assumption is utilized.} These constraints may correspond to the network balance constraints, and voltage and line limits in OPF problems. Alternatively, they may also correspond to procurement of the required amounts of power supplies, for instance, in the Swiss control reserve markets accepted reserves must have a deficit probability of less than 0.2\%. Thus, problem (\ref{eq:main_model}) defines a general class of electricity market problems, including energy-reserve co-optimized markets~\cite{xu2017efficient,carlson2012miso, cheung1999energy, chow2005electricity,amjady2009stochastic,kargarian2014spider,reddy2015joint}, stochastic markets~\cite{abbaspourtorbati2016swiss,conejo2010decision,bouffard2005market}, and AC-OPF problems~\cite{lavaei2012competitive,lavaei2012zero,molzahn2014moment,molzahn2015sparsity,winnicki2019convex}.\footnote{Several stochastic electricity market works, such as \cite{pritchard2010single,zakeri2018pricing}, study single settlement mechanisms for two-stage markets. Such mechanisms tie the first stage and the second stage markets together by asking the bidders to provide bid functions for both stages of the market simultaneously. These works then set prices for both stages of the market, where the second stage prices are generally uncertainty-dependent. On the other hand, the mechanisms we study distribute payments (and accept bids) only for the first stage of the market. We assume that the settlement and the actual solution of the second stage is separate from that of the first stage primarily motivated by the markets in \cite{abbaspourtorbati2016swiss, conejo2010decision}.\label{footnote:stoch}} As a remark, if the problem~\eqref{eq:main_model} is infeasible, the objective value is unbounded, $J(\BB)=\infty$.

Let the optimal solution of \eqref{eq:main_model} be denoted by $x^*(\BB)\in \hat \X$ and $y^*(\BB)\in\R^{p}$. {We assume that in case of multiple optima there is a tie-breaking rule.} We assume that the utility of bidder $l$ is linear in the payment received: 
$$u_l(\BB)=p_l(\BB)-c_l(x^*_l(\BB)),$$ (that is, quasilinear utilities~\cite{mas1995microeconomic}).  A bidder whose bid is not accepted, $x_l^*(\BB)=0$, is not paid and $u_l(\BB)=0$. 
The utility of the operator $u_0(\BB)$ is defined by the total payment, namely, $$u_0(\BB)=-\sum_{l\in L} p_l(\BB) - d(x^*(\BB),y^*(\BB)).$$
This total payment can be an expected value when the function~$d$ is an expected second stage cost. If the problem \eqref{eq:main_model} is infeasible, the utility of the operator is given by $u_0(\BB)=-\infty$.

There are several fundamental properties we desire for the mechanism~\cite{milgrom2004putting,krishna2009auction}. A mechanism is \textit{in\-dividually rational} (IR) if bidders do not face negative utilities, $u_l(\BB)\geq 0$ for all $l\in L$. This property is also often referred to as voluntary participation or cost recovery. A mechanism is \textit{efficient} if the sum of all the utilities $\sum_{l=0}^{\lvert L\rvert} u_l(\BB)$ is maximized. From the definition of the utilities, we have $$\sum_{l=0}^{\lvert L\rvert} u_l(\BB)=-\sum_{l\in L} c_l(x^*_l(\BB))- d(x^*(\BB),y^*(\BB)).$$ Notice that this value is maximized if we are solving for the optimal allocation of the market in~\eqref{eq:main_model} under the condition that the bidders submitted their true costs~$\{c_l\}_{l\in L}$. As a result, we can attain efficiency by eliminating potential strategic manipulations.

Several definitions are in order. Let $\mathcal{B}_{-l}$ be the bid profile of all the bidders, except bidder $l$.
The bid profile $\mathcal{B}$ is a \textit{Nash equilibrium} if for every bidder~$l$, $u_l(\mathcal{B}_l\cup\mathcal{B}_{-l})\geq u_l(\tilde{\mathcal{B}}_l\cup\mathcal{B}_{-l})$, $\forall\tilde{\mathcal{B}}_l$. 
The bid profile $\mathcal{B}$ is a \textit{dominant-strategy Nash equilibrium} if for every bidder $l$,
$u_l(\mathcal{B}_l\cup\hat{\mathcal{B}}_{-l})\geq u_l(\tilde{\mathcal{B}}_l\cup\hat{\mathcal{B}}_{-l})$, $\forall\tilde{\mathcal{B}}_l$, $\forall \hat{\mathcal{B}}_{-l}$.

Connected with the observation above related to the efficiency property, we say that a mechanism is \textit{dominant-strategy incentive-compatible} (DSIC) if the truthful bid profile $\mathcal C=\{c_l\}_{l\in L}$ is the dominant-strategy Nash equilibrium. In other words, every bidder finds it more profitable to bid truthfully, regardless of what others bid. However, as it will become clear later, unilateral deviations are not the only strategic manipulations we need to consider in order to ensure that the bidders reveal their true costs.

As the last desirable property, we consider immunity to collusion and shill bidding and this is the main topic of this part of the thesis. 
Bidders $K\subseteq L$ are \textit{colluders} if they obtain higher collective utility by changing their bids from $\mathcal C_K=\{c_l\}_{l\in K}$ to $\BB_K=\{b_l\}_{l\in K}$. In other words, this would imply $\sum_{l\in K}u_l(\mathcal B_K\cup\BB_{-K})>\sum_{l\in K}u_l(\mathcal C_K\cup\BB_{-K})$. A bidder $l$ is a \textit{shill bidder} if there exists a set $S$ and bids $\BB_S=\{b_k\}_{k\in S}$ such that the bidder~$l$ finds participating with bids $\BB_S$ more profitable than participating with a single truthful bid~$\CC_l$. In other words, this is given by $\sum_{k\in S}u_k(\mathcal B_S\cup\BB_{-l})> u_l(\CC_l\cup\BB_{-l})$. 

Finally, by \textit{coalition-proof}, we mean that a group of bidders whose~bids are not accepted when bidding their true costs, $x^*_l(\mathcal C_K\cup\BB_{-K})=0,$ $\forall l\in K$, cannot profit from collusion, and no bidder can profit from using shill bids. We remark that it is not possible to achieve immunity to collusion from all sets of bidders. For instance, no mechanism can eliminate the situation where all bidders inflate their bid prices simultaneously, see also the collusion examples in~\cite{beck2009revenue}.

Since the bidders strategize around the payment rule, payment design plays a crucial role in attaining the aforementioned properties. In light of the discussions above, we discuss well-studied payment rules that fail to attain some of these properties for the general class of electricity markets in~\eqref{eq:main_model}.
\newpage
\section{Pay-as-bid mechanism}

In the \textit{pay-as-bid mechanism}, the payment rule is
$$p_l(\mathcal{B})=b_l(x^*_l(\mathcal{B})).$$ For instance, several European balancing markets are settled under a pay-as-bid
	mechanism, see~\cite{mazzi2018price} and further references therein.
It follows that each bidder's utility is $u_l(\mathcal{B})=b_l(x^*_l(\mathcal{B}))-c_l(x^*_l(\mathcal{B}))$. A rational bidder would overbid to ensure positive utility. Consequently, under the pay-as-bid mechanism, the central operator calculates the optimal allocation for the inflated bids rather than the true costs. Furthermore, the bidders need to spend resources to learn how to bid to maximize their utility. There are many Nash equilibria arising from the pay-as-bid mechanism, none of which are incentive-compatible~\cite{bernheim1986menu}. This issue was analyzed in our previous work in~\cite{orcun2018game}, which is not included in this thesis because the pay-as-bid mechanism will not be treated in detail.

\section{Lagrange-multiplier-based payment mechanisms}
The \textit{LMP mechanism} is adopted in markets where polytopic DC-OPF constraints and nondecreasing convex bids are considered. For simplicity in notation, assume there is a single bidder at each node of the network. Under this assumption, each bidder is supplying a one-dimensional power supply of a unique type. Then, the payment rule is $$p_l(\BB)=\lambda_l^*(\BB)\, x^*_l(\BB),$$ where $\lambda_l^*(\BB)\in\R$ is the Lagrange multiplier of the $l^{\text{th}}$ nodal balance equality constraint.\footnote{With a slight notational abuse we ignore the previously defined $t$-dimensional form of $ x^*_l(\BB)$.} See \cite{wu1996folk} for an exposition on the~calculation of the LMP payments from the Karush-Kuhn-Tucker (KKT) conditions of DC-OPF problems, we also kindly refer to the DC power flow constraints in \eqref{mod:A-DA} in \cref{part:2}.

Assume that each bidder is a price-taker, in other words, each bidder considers the Lagrange multiplier of its node to be independent of its bid. Then, in addition to being IR, the LMP mechanism is DSIC.\footnote{IR requires $c_l(0)=0$ and $b_l(0)=0$ for all $l\in L$.} 
However, this economic rationale of the LMP mechanism involves a strong assumption not found in practice~\cite{joskow2001quantitative}. Under the LMP mechanism, a~bidder can in fact maximize its utility by both inflating its bids and withholding its maximum supply~\cite{ausubel2014demand,TangJ13,tang2013game}. On the positive side,~in~\cref{sec:p1_4}, we show that this mechanism is coalition-proof.

Another aspect to consider is that the economic rationale of the Lagrange multipliers follows from strong duality~\cite{bikhchandani1997competitive,lavaei2012competitive,warrington2012market}.\footnote{There are extensions of LMP through uplift payments to address the duality gap specifically arising from the unit commitment costs. Under convexified power flow equations, the works in \cite{o2005efficient,chao2019incentives,hogan2003minimum} compute linear prices from an integer restriction, an integer relaxation, and a convex hull approximation, respectively (for a comparison kindly refer to~\cite{gribik1993market}). These prices are then complemented with bidder-dependent uplift side payments, which are equivalent to either the deficits or the opportunity costs of the participants when using approximated linear prices. These bulk payments are known to create additional incentives for manipulative behavior. Finally, note that in Europe the uplift idea is completely rejected. Some European day-ahead markets even execute a suboptimal solution (as opposed to~\eqref{eq:main_model}). For instance, they disregard an optimal solution when a participant requires an uplift for its individual rationality when using approximated linear prices~\cite{van2011linear}.} For DC-OPF problems, strong duality is implied by the convexity of the bid profile and the linearity of the constraints~\cite{bertsekas1999nonlinear}. Strong duality, however, may not hold for the optimization problem~\eqref{eq:main_model}, and hence the Lagrange multipliers may not be meaningful in an economic sense. For instance, nonlinear AC-OPF constraints (modeling also the reactive power) are known to yield a non-zero duality~gap for many practical problems~\cite{lesieutre2011examining}, and sufficient conditions for zero duality gap are in general restrictive~\cite{low2014convex}. 

\section{The Vickrey-Clarke-Groves mechanism}
As an alternative, the \textit{VCG mechanism} is characterized by
$$p_l(\BB)=b_l(x^*_l(\BB))+(H(\BB_{-l})-J(\BB)),$$ where $\BB_{-l}=\{b_k\}_{k\in L\setminus l}$.
The function $H(\BB_{-l})\in\R$ must be chosen carefully to ensure the IR property. A well-studied choice is the \textit{Clarke pivot rule} $
H(\BB_{-l})=J(\BB_{-l}),  $
where $J(\BB_{-l})$ is the optimal value of (\ref{eq:main_model}) with the constraint $x_l=0$, removing the bidder $l$ from both the objective and the constraints.\footnote{The general form is referred to as the Groves mechanism. The Clarke pivot rule is known to generate the minimum total payment ensuring the IR property~\cite{krishna1998}.\label{footnote:minima}} This mechanism is well-defined under the assumption that a feasible solution exists when a bidder is removed. This is a practical assumption in electricity markets~\cite{xu2017efficient}, for instance, it holds for almost all IEEE test systems~\cite{christie2000power}. 
The~VCG mechanism can be shown to satisfy IR, DSIC, and efficiency (if bidders pick their dominant strategies) for the market in~\eqref{eq:main_model}.\footnote{IR requires $c_l(0)=0$ and $b_l(0)=0$ for all $l\in L$.} This result is relegated to the appendix in \cref{sec:appforVCG}. It is a generalization of the works in~\cite{vickrey1961counterspeculation,clarke1971multipart,groves1973incentives} which do not consider continuous goods, second stage cost and general nonlinear constraints.  

Despite these theoretical properties, the VCG mechanism can suffer from collusion and shill bidding which can result in a loss of efficiency. To illustrate these issues, we study two simple energy market examples. For these examples, we consider the VCG mechanism with the Clark-pivot rule and with $d(x,y)\equiv 0$ in the central operator's objective.  Later, we come back to these examples in \cref{sec:p1_3}, in order to discuss conditions to eliminate collusion and shill-bidding. 

The first example is a reverse auction of a single type of power supply. In these markets, each bidder is allowed to submit mutually exclusive bids that can equivalently be represented as bid curves, see~\cref{sec:3a}.

\begin{example}[Simple Market]\label{ex:first_simple_example}
	Suppose the central operator has to procure $800$ MW of power supply from bidders $1$, $2$ and $3$ who have the true costs $\$100$ for $400$ MW, $\$400$ for $400$ MW and $\$600$ for $800$ MW, respectively. Under the VCG mechanism, bidders $1$ and $2$ win and receive $p_1^{\text{VCG}} = 100 + (600-500)= \$200$ and $p_2^{\text{VCG}} = 400 + (600-500)= \$500$. Suppose bidders $1$ and $2$ collude and change their bids to $\$0$ for $400$ MW. Then, bidders $1$ and $2$ receive a payment of $\$600$ each for the same allocation. In fact, bidders~$1$~and~$2$ could represent multiple identities of a single losing bidder (that is, a bidder with the true cost greater than $\$600$ for $800$ MW). Entering the market with two shill bids, this bidder receives a payment of $2\times\$600$ for $800$ MW.
\end{example}
\vspace{.1cm}

The second example is a power market where the operator is procuring a set of different types of power supplies.

\begin{example}[Power Market]\label{ex:second_simple_example}
	We consider a reverse auction with three different types of supplies, types A, B and C. Here, type A can replace types~B~and~C simultaneously.\footnote{This is an abstraction for power reserve markets where secondary reserves can replace both negative and positive tertiary reserves, simultaneously \cite{abbaspourtorbati2016swiss}.} Suppose the central operator has to procure $100$ MW of type B and $100$ MW of type C (or equivalently only $100$ MW of type~A) from bidders $1$ to $5$. Truthful bid profiles of $5$ bidders are provided in~\cref{tab:bidprof}. 
	\begin{table}[h]
		\caption{Bid profile in the Power Market}
		\label{tab:bidprof}
		\begin{center}
			\begin{tabular}{|l||l||l||l||l||l|}
				\hline
				Bidders (Types) & $1$ (A) & $2$ (B) &  $3$ (B) & $4$ (C) & $5$ (C) \\
				\hline 
				MW & $100$ & $100$  & $100$ & $100$& $100$ \\
				\hline
				\$ & $500$  & $350$  & $400$ & $250$ & $400$ \\
				\hline
			\end{tabular}
		\end{center}
	\end{table}
	
	The constraint set in \eqref{eq:main_model} is given by
	\begin{equation}	\label{newc2}
	\begin{split}
	\big\{ x\in\{0,100\}^5\, \rvert\, & x_1 + x_2 + x_3 \geq M(\{\text{A},\text{B}\})=100,\\
	& x_1 + x_4 + x_5 \geq M(\{\text{A},\text{C}\})=100\big\}.  
	\end{split}
	\end{equation}	
	Under the VCG mechanism, bidder $1$ wins and receives $p_1^{\text{VCG}} = 500 + (600-500)= \$600$. Suppose losing bidders $2$ and $4$ collude and change their bid prices to $\$0$. Then, bidders $2$ and $4$ receive $\$400$ each and they obtain a collective VCG profit of $\$200$. The total payment of the operator increases from $\$600$ to $\$800$. This is unfair towards bidder $1$ who is willing to offer the same supply for $\$500$.
\end{example}
\vspace{.1cm}

It is troubling that the VCG mechanism can result in large payments through coalitional manipulations. In all these examples, there exists a group of bidders who is willing to offer the same amount of good by receiving less payment. From the central operator's perspective, the operator would instead want to renegotiate the payments with only a subset of participants.



\section{Appendix}
\subsection{Properties of the Vickrey-Clarke-Groves mechanism}\label{sec:appforVCG}
For the model introduced in \eqref{eq:main_model}, our first result shows that the VCG mechanism first derived in  \cite{vickrey1961counterspeculation,clarke1971multipart,groves1973incentives} satisfies all three fundamental properties. This result is a straightforward generalization of the works in \cite{vickrey1961counterspeculation,clarke1971multipart,groves1973incentives}, which do not consider continuous values of goods, second stage costs, and general constraints.

\begin{theorem}
	\label{thm:incentive_comp}
	Given the market model \eqref{eq:main_model}, 
	\begin{enumerate}
		\item[(i)] The Groves mechanism is DSIC. 
		\item[(ii)] The Groves mechanism is \text{efficient}.
		\item[(iii)] The Groves mechanism ensures nonnegative payments and IR when the Clarke pivot rule is utilized, $
		h(\mathcal{B}_{-l})=J(\mathcal{B}_{-l})$.
	\end{enumerate}
\end{theorem}
\begin{proof}
	(i) We distinguish between bidder $l$ placing a generic bid $\mathcal{B}_l = b_l$ and bidding truthfully $\mathcal C_l= c_l$.  For the set of bids $\mathcal{B}$, the utility of bidder $l$ is given by:
	\begin{align*}
	u_l({\mathcal{B}}) = H(\mathcal{B}_{-l})\, -\Big(\sum\limits_{k\neq l} b_k(x_k^*(\mathcal{B})) + c_l(x_l^*(\mathcal{B})) + d(x^*(\mathcal{B}),y^*(\mathcal{B}))\Big),
	\end{align*}
	where the term in brackets is the market objective of~\eqref{eq:main_model} under the bids $\hat{\mathcal C} =\mathcal C_l\cup \mathcal{B}_{-l}$ but evaluated at $(x^*(\mathcal{B}),y^*(\mathcal{B}))$. For $\hat{\mathcal C}$ note that $ u_l(\hat{\mathcal C})= H(\mathcal{B}_{-l}) - J(\hat{\mathcal C})$. Then, we have the following:
	\begin{equation*}
	J(\hat{\mathcal C}) \leq \sum\limits_{k\neq l} b_k(x_k^*(\mathcal{B})) + c_l(x_l^*(\mathcal{B})) + d(x^*(\mathcal{B}),y^*(\mathcal{B})).
	\end{equation*}
	We can now show that $u_l(\hat{\mathcal C}) \geq u_l({\mathcal{B}})$ because $(x^*(\mathcal{B}),y^*(\mathcal{B}))$ is a feasible suboptimal allocation for the auction under the bids $\hat{\mathcal C}$. Therefore, bidding truthfully is a best response strategy, regardless of other bidders' strategies $\mathcal{B}_{-l}$.
	
	(ii) By the definition of the payment rule and incentive-compatibility, we have $p_l(\mathcal C)= u_l(\mathcal C) + c_l(x_l^*(\mathcal C))$ where $\mathcal C=\{c_l\}_{l\in L}$. We then have: $u_0(\mathcal C) = - \sum_{{l\in L}} c_l(x_l^*(\mathcal C)) - d(x^*(\mathcal{C}),y^*(\mathcal{C})) - \sum_{{l\in L}}{u_l(\mathcal C)}$. The sum of utilities, $\sum_{l=0}^{\lvert L\rvert}  u_l(\mathcal C)=- \sum_{{l\in L}} c_l(x_l^*(\mathcal C)) - d(x^*(\mathcal{C}),y^*(\mathcal{C})) $ is maximized since $(x^*(\mathcal{C}),y^*(\mathcal{C}))$ is the minimizer to the optimization problem \eqref{eq:main_model} under true costs.
	
	(iii) Nonnegative payments can be verified substituting Clarke pivot rule for $H(\mathcal{B}_{-l})$:
	\begin{equation*}
	p_l(\mathcal{B})=b_l(x^*_l(\mathcal{B}))+(J(\mathcal{B}_{-l})-J(\mathcal{B}))\geq0,
	\end{equation*}
	for all set of bids $\mathcal{B}$. For individual rationality, we have to assume bidders are not bidding less than their true costs\footnote{Otherwise, the IR property may not hold for the utilities, $u_l(\mathcal{B})$, however, it would still hold for the revealed utilities, which are the utilities with respect to the submitted bids, see the definitions and the discussions in \cref{sec:p1_4}. Moreover, note that both results rely on the assumption that $c_l(0)=0$ and $b_l(0)=0$ for all $l\in L$.}, that is, $b_l(x)\geq c_l(x),\, \forall x\in\X_l$. We have
	\begin{equation*}
	u_l(\mathcal{B}) = b_l(x_l^*(\mathcal{B}))-c_l(x_l^*(\mathcal{B}))+J(\mathcal{B}_{-l}) - J(\mathcal{B}) \geq 0,
	\end{equation*}
	for all set of bids $\mathcal{B}$.
\end{proof}

In summary, all bidders have  incentives to reveal their true costs in a VCG mechanism. Dominant-strategy incentive-compatibility makes it easier for entities to enter the auction, without spending resources in computing optimal bidding strategies. This can promote participation in the market. As a remark,~\cref{thm:incentive_comp}-(ii) repeats the fact that solving for the optimal allocation with the true costs yields an efficient mechanism. In the remainder, we consider the Clarke pivot rule for the VCG mechanism since it ensures individual rationality.

\chapter{Ensuring coalition-proof Vickrey-Clarke-Groves outcomes}\label{sec:p1_3}

In coalitional game theory, the \textit{core} defines the set of utility allocations that cannot be improved upon by forming coalitions~\cite{osborne1994course,peleg2007introduction}.\footnote{Throughout this chapter, we use the term utility allocation and the auction outcome interchangeably.} In this chapter, we start by showing that if the truthful VCG outcome always lies in the core, then the VCG mechanism eliminates {{any}} incentives for {collusion} and {shill bidding}. Keeping this in mind, our main goal is to derive sufficient conditions on \eqref{eq:main_model} to ensure that the VCG outcome lies in the core, and hence the VCG mechanism is coalition-proof.

For every $S\subseteq L$,  let $J(\mathcal{B}_S)$ be the objective function under any set of bids $\mathcal{B}_S=\{b_l\}_{l\in S}$ from the coalition $S$. It is defined by the following expression:
\begin{equation}\label{eq:33}
\begin{split}
J(\BB_S)=&\min_{x\in \hat \X,\,y}\,\, \sum\limits_{l\in S} b_l(x_l) + d(x,y)\\
&\ \ \mathrm{s.t.}\ \ h(x,y)= 0,\, g(x,y)\leq 0,\, x_{-S}=0,\\
\end{split}
\end{equation}
where the stacked vector $x_{-S}\in\R_+^{t(\lvert L\rvert -\lvert S\rvert)}$ is defined by omitting the subvectors from the set $S$. It is straightforward to see that this function is nonincreasing, that is, $J(\mathcal{B}_R) \geq J(\mathcal{B}_S)$ for $R\subseteq S$.\footnote{This holds since $b_l(0)=0$ for all $l\in L$.} 

Next, we define the core with respect to the truthful bids, $\mathcal{C}_R=\{c_l\}_{l\in R}$, and refer to this definition solely as the \textit{core}. 
\begin{definition}\label{def:core_def}
	For every set of bidders $R\subseteq L$, the core $Core(\mathcal{C}_R)\in\R\times\R^{\rvert R\rvert}_+$ is defined as follows
	\begin{equation}\label{eq:mcoredef}
	\begin{split}
	Core(\mathcal{C}_R)=\Big\{u\in\R\times\R^{\rvert R\rvert}_+ \,|\, &u_0+\sum\limits_{l\in R} u_l=-J(\mathcal{C}_R),\\
	&u_0+\sum\limits_{l\in S}u_l\geq-J(\mathcal{C}_S),\, \forall S \subset R \Big\}.
		\end{split}
	\end{equation}
\end{definition}
Note that there are $2^{\rvert R\rvert}$ linear constraints that define a utility allocation in the core for the set of bidders $R$. The core is always nonempty in an auction because the utility allocation $u_0=-J(\mathcal{C}_R)$  and $u_l=0$ for all $l\in R$ always lies in the core. This allocation corresponds to the utility allocation of the pay-as-bid mechanism under the truthful bidding $\mathcal{C}_R$.\footnote{If $c_l(0)\neq0$, the function $J$ may not be nonincreasing, the pay-as-bid utilities under the truthful bidding may not lie in the core, and the core may be empty. In this case, to guarantee that the core exists, IR constraints for such bidders, $u_{l}\geq0$, can be completely removed, or can be replaced by $u_l\geq -c_l(0)$ whenever $c_l(0)$ is finite. It can easily be verified that this new core is always nonempty.\label{footnote:corempty}}

For the properties considered in mechanism design, we highlight the implications of the constraints in~\eqref{eq:mcoredef}. Restricting the utility allocation to the nonnegative orthant yields the IR property for the bidders. The equality constraint implies that the mechanism is efficient, since the term on the right is maximized by the optimal allocation. We say that a utility allocation is unblocked if there is no set of bidders that could make a deal with the operator from which every member can benefit, including the operator. This condition is satisfied by the inequality constraints. 

The truthful VCG outcome attains the maximal utility in the core for every bidder. Note that under the VCG mechanism each bidder's utility is given by $u_l^{\text{VCG}}=J(\mathcal{C}_{-{l}})-J(\mathcal{C})$. Then, for every bidder~$l$, $u_l^{\text{VCG}}=\max\left\{u_l\,\rvert\, u\in Core(\mathcal{C})\right\}$, see \cite[Theorem~5]{ausubel2002ascending},~\cite[Theorem~2]{orcun2018game}. {In general, this maximal point may not lie in the core. The~following example gives visual insight about the core in terms of payments and illustrates the dominant-strategy Nash equilibrium of the VCG mechanism corresponding to~\cref{ex:first_simple_example}. In this example, shill bidding and collusion are shown to be profitable under the VCG mechanism. 
	\begin{example}\label{ex:core_illust}
		We revisit~\cref{ex:first_simple_example}. Without loss of generality, assume that in case of a tie the central operator prefers bidders $1$ and $2$ over bidder~$3$. We can visualize the core outcomes in terms of the payments for the bidders $1$ and $2$ by removing the losing bidder~$3$, $p_3^{\text{VCG}} = 0$, and the operator. Core outcomes and the VCG payments ($p_i^{\text{VCG}}$) are given in~\cref{fig:core_picture_1}. 
		\begin{figure}[th]
			\begin{center}
				\begin{tikzpicture}[scale=0.89, every node/.style={scale=0.55}]
				\coordinate (r0) at (1,5);
				\coordinate (s0) at (2,4);
				\coordinate (si) at (1,4);
				\coordinate (s1) at (1,5);
				\draw[->] (0,0) -- (7,0) node[right] {\huge $p_1$};
				\draw[->] (0,0) -- (0,6.75) node[above] {\huge $p_2$};
				\foreach \x in {1,2,6}
				\draw (\x cm,1pt) -- (\x cm,-1pt) node[anchor=north] {\LARGE\pgfmathparse{100*\x} \pgfmathprintnumber[    
					fixed,
					fixed zerofill,
					precision=0
					]{\pgfmathresult}};
				\foreach \y in {4,5,6}
				\draw (1pt,\y cm) -- (-1pt,\y cm) node[anchor=east]  {\LARGE\pgfmathparse{100*\y} \pgfmathprintnumber[    
					fixed,
					fixed zerofill,
					precision=0
					]{\pgfmathresult}};;
				\draw[scale=1, line width=.22mm, domain=0:6,smooth,variable=\x,black!80!blue] plot ({\x},{6-\x}) node[below] at(7.1,1.2){\huge $p_1+p_2\leq \$600$};
				\draw[scale=1,line width=.22mm, domain=0:7,smooth,variable=\x,black!80!blue]  plot ({\x},{4}) node[right] at(6,3.4) {\huge $p_2\geq \$400$};
				\draw[scale=1, line width=.22mm,domain=0:6.7,smooth,variable=\y,black!80!blue]  plot ({1},{\y}) node[above] at(2.2,6.5) {\huge $p_1\geq  \$100$};
				\draw[scale=1, line width=.22mm, domain=0:2,dashed,variable=\x,black!40!red]  plot ({\x},{5});
				\draw[scale=1, line width=.22mm, domain=0:5,dashed,variable=\y,black!40!red]  plot ({2},{\y}) node[right]  at(1.9,5.6){\huge $(p_1^{\text{VCG}},p_2^{\text{VCG}})$};
				\filldraw[draw=black!80!blue,line width=.22mm, fill=gray!40] (r0) -- (s0) -- (si) -- (s1) -- cycle;
				\node[black!40!red] at (2,4.975) {\Huge\textbullet};
				\node[black!40!blue] at (1,3.975) {\Huge\textbullet};
				\draw[->,line width=.32mm] (2.5,2) -- (1.3,4.3) ;
				\node at (3,1.6) {\Huge $\substack{\text{Core}}$};
				\draw[->,line width=.32mm,black!40!blue] (-1.5,2) -- (0.9,3.9) ;
				\node[black!40!blue] at (-2.25,1.5) {\Huge $\substack{\text{Pay-as-bid}\\ \text{under truthful bidding}}$};
				\draw[scale=1, line width=.6mm, domain=1:2,smooth,variable=\x,black!60!green] plot ({\x},{6-\x});
				\end{tikzpicture}
				\caption{Core outcomes and the VCG payments under truthful bidding}\label{fig:core_picture_1}
			\end{center}
		\end{figure}
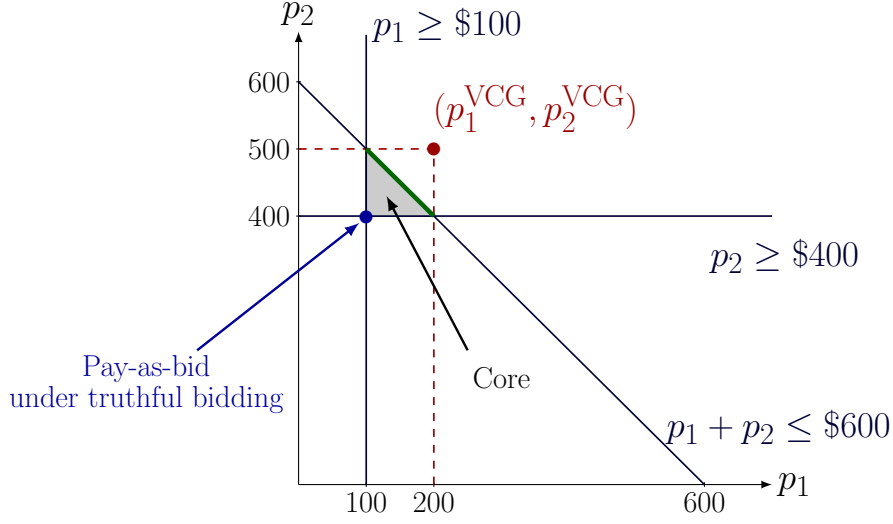
		\end{example}

	Considering that the core will characterize coalition-proof outcomes, we are ready to investigate the conditions under which the VCG outcome lies in the core. To this end, we {provide three} sufficient conditions that ensure core VCG outcomes for the auction model~(\ref{eq:main_model}). 
	
	Notice that there are $2^{|L|}$ linear core constraints in $Core(\mathcal{C})$, see~\eqref{eq:mcoredef}. First, we derive the following equivalent characterization with significantly lower number of constraints.
\begin{lemma}
	\label{lem:lemma_core}
	Let $W\subseteq L$ be the winners of the reverse auction \eqref{eq:main_model} for the set of bidders $L$, that is, each bidder $l\in W$ is allocated a positive quantity. Let $ u\in\R\times\R^{\rvert L\rvert}_+$ be the corresponding utility allocation. Then, $ u \in Core(\mathcal{C})$ if and only if $ u_0=-J(\mathcal{C})-\sum_{l\in L} u_l$ and
	\begin{equation}
	\label{eq:core_constraints}
	\sum_{l \in K }  u_l\leq J(\mathcal{C}_{-K}) - J(\mathcal{C}),\ \forall K\subseteq W.
	\end{equation}
\end{lemma}

 The proof is relegated to the appendix in \cref{appforwin}, and it is utilized in the proofs of the results of this chapter. The following proposition is our first sufficient condition for core VCG outcomes.

\begin{proposition}\label{lem:removal_of_two}
	The truthful VCG outcome is in the core, $ u^{\text{VCG}} \in Core(\mathcal{C})$, if the market in \eqref{eq:main_model} is infeasible whenever any {two} winners $l_1,l_2\in W$ are removed from the set of bidders~$L$.
\end{proposition}

The proof is relegated to the appendix in \cref{appforlose}. The above condition can only be present in some specialized instances of reverse auctions. It is not possible in general to guarantee that this condition will hold in a market by enforcing restrictions on the bidders and the market model.

{We will soon show that supermodularity provides an equivalent condition for core VCG outcomes. To this end, we bring in its definition. }

\begin{definition}\label{def:supms}
	A function $f:2^{L}\rightarrow\R$ is \textit{supermodular} if $$f(S)-f(S_{-l})\leq f(R)-f(R_{-l}),$$
for all $S\subseteq R\subseteq L$ and for all $l\in S$. Or, equivalently, for all~$S, R\subseteq L$, $f(S\cup R)+f(S\cap R)\geq f(S)+f(R)$ must hold. A function $f:2^{ L}\rightarrow\R$ is \textit{submodular} if $-f$ is supermodular. Furthermore, a function is nondecreasing if $f(S')\leq f(S)$, for all $S'\subseteq S.$
\end{definition}

  For the remainder of this chapter, the objective function $J$ in~\eqref{eq:33} is said to be supermodular if supermodularity condition holds under any bid profile. Our main result of this section proves that supermodularity of the objective function is necessary and sufficient for ensuring core VCG outcomes.

{
	\begin{theorem}\label{thm:iff_supermodularity}	For any bid profile $\mathcal{C}$ and for any set of participating auction bidders $R\subseteq L$, the truthful VCG outcome is in the core, \text{if and only if} the objective function $J$ in~\eqref{eq:33} is supermodular.  
	\end{theorem} 
}
Note that a similar result but for submodularity was proven in~\cite[Theorem~6]{ausubel2006lovely}, for a forward multi-item auction without any other constraints. Our result is an extension of this result to our reverse auction setting in \eqref{eq:main_model}. The proof is relegated to the appendix in \cref{appforsup}.
There are extensions of supermodularity for nonsupermodular functions via the notion of weak supermodularity~\cite{bian2017guarantees,das2011submodular}. Using this condition, our previous work~in~\cite{karaca2018weak} ensures outcomes that are not far away from the core. The results of~\cite{karaca2018weak} on weak coalition-proofness will not be treated in detail.

As previously anticipated, we now prove that the outcomes from the core, hence the supermodularity condition, make collusion and shill bidding unprofitable in a VCG mechanism.

\begin{theorem}
	\label{thm:no_collusions}
			For the set of bidders $L$, consider a VCG auction mechanism modeled by~\eqref{eq:main_model}. If the objective function $J$ is supermodular, then, 
	\begin{itemize}
	\item[(i)] {A group of bidders who lose when bidding their true values cannot profit by a joint deviation.}
	\item[(ii)] Bidding with multiple identities is unprofitable for any bidder.\end{itemize}
\end{theorem}

The proof is relegated to the appendix in \cref{appforsupcoal}.~\cref{thm:no_collusions} shows that if the operator has a supermodular objective (or equivalently, if the VCG outcomes are always in the core), then the VCG mechanism is coalition-proof. Given this result, we next investigate sufficient conditions on the bids and the constraint sets {in order to ensure supermodularity and thus coalition-proof~outcomes.} In the following we derive conditions for two classes of markets.

\section{Markets for a single type of good}\label{sec:3a}
We start by considering simpler reverse auctions where the operator has to procure a fixed amount $M\in\R_+$  of a single type of good. 
Each bidder~$l$ has a private true cost function $c_l:\mathbb R_+ \rightarrow \mathbb R_+$ that is nondecreasing with $c_l(0)=0$. These {types} of auctions are mainly characterized by single-stage decisions with mutually exclusive bids. This means that a bidder can offer a set of bids, of which only one can be accepted. We first show that such discrete bids fit into our model~\eqref{eq:main_model}.  Here, bidder~$l$ submits truthful bids for $n_l$ discrete amounts as $\{(c_{l,i},x_{l,i})\}_{i=1}^{n_l}$
{where $c_{l,i}\in\R_+$ and the amounts offered by each bidder $x_{l,i} \in \R_+$
	must be equally spaced by some increment $m$  which is a divisor of $M$, that is,
	\begin{equation}\label{eq:simpler_clearing_model_cond}
	x_{l,i}= im,\; \text{for some }{i\in\mathbb Z_+}.
	\end{equation}
}Note that, there is an equivalent representation of the form $c_l(x)\in\R_+$ for $x>0$ as follows: 
{	\begin{equation}\label{cdef}
	c_l(x)= \min_{i=1,...,n_l} \left\{ c_{l,i} \,\rvert\, x_{l,i} \geq x\right\},
	\end{equation}
}where $c_l(0)=0$. This form equivalently represents that all the amounts up to the size of the winning bid are available to the operator. Furthermore, bid prices of this form are piecewise constant and continuous from the left.

We consider auctions cleared by
\begin{equation}	\label{eq:simpler_clearing_model}
\begin{split}
J(\mathcal{C}_S) =  &\min_{x\in\R_+^{\rvert S \rvert}}  \;\sum_{l \in S}c_l(x_l)\\ &\ \ \mathrm{s.t. }\ \,   \sum_{l \in S}x_l \geq M,
\end{split}
\end{equation}
{for} $S\subseteq L$.} 
Note, we can equivalently assume that $x_l$, above, takes values in $\{x_{l,i} \,\rvert\, i \in \mathbb{Z_+} \}\subseteq \mathbb{R}_+$ (cf. \eqref{cdef} and~\cref{fig:marginally_inc}) and doing so, we let $x^*=\{x^*_l\}_{l \in S}$ be the optimal values in these restricted sets.

The model \eqref{eq:simpler_clearing_model} is within the auction model \eqref{eq:main_model}. We can now derive conditions on bidders' true costs to ensure supermodularity of $J$. Thus, we derive conditions under which the truthful VCG outcome from~\eqref{eq:simpler_clearing_model} would lie in the core. 
\begin{theorem}
	\label{thm:conditions_on_bids}
	Given \eqref{eq:simpler_clearing_model_cond}, if the true costs are marginally increasing, namely,
	$x_{l,b} - x_{l,a} = x_{l,d} - x_{l,c}$
	implies that 
	$ c_{l,b} - c_{l,a} <  c_{l,d} - c_{l,c} $
	for each bidder $l \in L$ and for each $0 \leq x_{l,a}<x_{l,c}< x_{l,d}$, 
	then the objective function $J$ in \eqref{eq:simpler_clearing_model} is supermodular.
\end{theorem}

This setting includes reverse auctions of multiple
identical items as a subset. However, we highlight that our proof does not share similarities with that of \cite[Theorem~8]{ausubel2006lovely}, which achieves submodular objective functions in forward auctions of multiple single-unit items utilizing a substitutes condition on different items. Our proof relies on an important lemma we prove showing that the allocations, $x_l^*$, of every bidder is nondecreasing when a bidder is removed from the auction \eqref{eq:simpler_clearing_model}. The proof is relegated to the appendix in~\cref{appforsimpl}.
As a corollary of this result, marginally increasing costs imply coalition-proof VCG outcomes for~\eqref{eq:simpler_clearing_model}\, and thus eliminate incentives for collusion and shill bidding. This condition is visualized in~\cref{fig:marginally_inc}. 

\begin{figure}[t]
	\begin{center}
		\begin{tikzpicture}[scale=0.85, every node/.style={scale=0.65}]
		\draw[->] (0,0) -- (7,0) node[right] {\Large Quantity};
		\draw[->] (0,0) -- (0,7.2) node[above] {\Large Cost};
		\draw[black!60!blue,fill=black!60!blue] (0,0) circle (.075cm);
		\draw[black!60!blue,fill=black!60!blue]  (1,0.1837) circle (.075cm);
		\draw[black!60!blue,fill=black!60!blue]  (2,0.769) circle (.075cm);
		\draw[black!60!blue,fill=black!60!blue]  (3,1.65) circle (.075cm);
		\draw[black!60!blue,fill=black!60!blue] (4,2.937) circle (.075cm);
		\draw[black!60!blue,fill=black!60!blue]  (5,4.5)circle (.075cm);
		\draw[black!60!blue,fill=black!60!blue] (6,6.9) circle (.075cm);
		\draw[black!60!blue]  (0,0.1837) circle (.075cm);
		\draw[black!60!blue]  (1,0.769) circle (.075cm);
		\draw[black!60!blue]  (2,1.65) circle (.075cm);
		\draw[black!60!blue] (3,2.937) circle (.075cm);
		\draw[black!60!blue]  (4,4.5) circle (.075cm);
		\draw[black!60!blue]  (5,6.9) circle (.075cm);
		\draw[scale=1, line width=.22mm, domain=-1.5:6,dashed,variable=\x,black!40!blue]  plot ({\x},{6.9});
		\draw[scale=1, line width=.26mm, domain=5.035:6,variable=\x,black!60!blue]  plot ({\x},{6.9});
		\draw[scale=1, line width=.22mm, domain=0:6.9,dashed,variable=\y,black!40!blue]  plot ({6},{\y});
		\draw[scale=1, line width=.22mm, domain=-1.5:3,dashed,variable=\x,black!40!blue]  plot ({\x},{1.65});
		\draw[scale=1, line width=.26mm, domain=2.035:3,variable=\x,black!60!blue]  plot ({\x},{1.65});
		\draw[scale=1, line width=.22mm, domain=0:1.65,dashed,variable=\y,black!40!blue]  plot ({3},{\y});	
		\draw[scale=1, line width=.22mm, domain=-0.5:5,dashed,variable=\x,black!40!blue]  plot ({\x},{4.5});
		\draw[scale=1, line width=.26mm, domain=4.035:5,variable=\x,black!60!blue]  plot ({\x},{4.5});
		\draw[scale=1, line width=.22mm, domain=0:4.587,dashed,variable=\y,black!40!blue]  plot ({5},{\y});	
		\draw[scale=1, line width=.22mm, domain=-0.5:2,dashed,variable=\x,black!40!blue]  plot ({\x},{.769});
		\draw[scale=1, line width=.26mm, domain=1.035:2,variable=\x,black!60!blue]  plot ({\x},{.769});
		\draw[scale=1, line width=.22mm, domain=0:.737,dashed,variable=\y,black!40!blue]  plot ({2},{\y});
		\draw[scale=1, line width=.26mm, domain=3.035:4,variable=\x,black!60!blue]  plot ({\x},{2.937});
		\draw[scale=1, line width=.26mm, domain=0.035:1,variable=\x,black!60!blue]  plot ({\x},{.1837});
		\draw[<->, line width=.25mm] (2,-0.3) -- (5,-0.3) node at (3.5,-0.65) {\LARGE $x_{l,b} - x_{l,a}$};		
		\draw[<->, line width=.25mm] (3,-1.3) -- (6,-1.3) node at (4.5,-1.65) {\LARGE $x_{l,d} - x_{l,c}$};		
		\draw[<->, line width=.25mm] (-0.5,0.9) -- (-0.5,4.4) node[rotate=90] at (-1,3) {\LARGE$c_{l,b} - c_{l,a}$};
		\draw[<->, line width=.25mm] (-1.5,1.75) -- (-1.5,6.8) node[rotate=90,] at (-2,4.4) {\LARGE $c_{l,d} - c_{l,c}$};									
		\end{tikzpicture}
		\caption{Marginally increasing piecewise constant bid prices}\label{fig:marginally_inc}
	\end{center}
\end{figure}
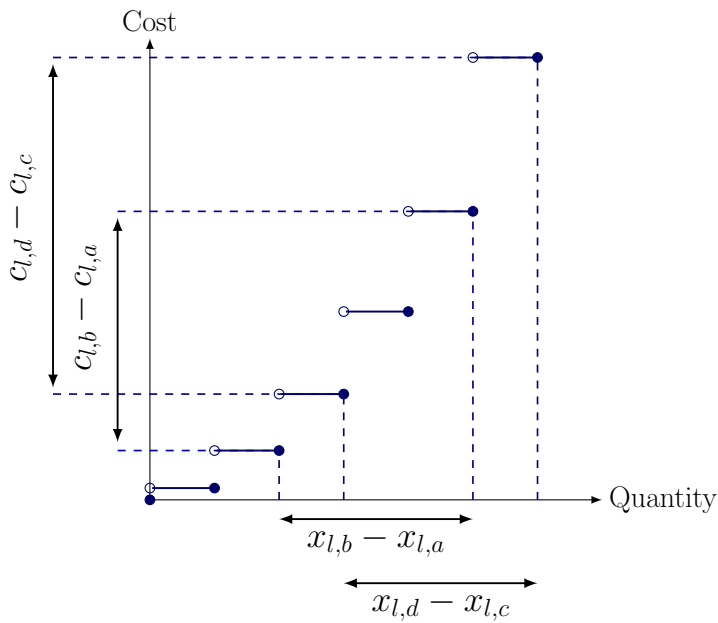

We note that the analogue of~\cref{thm:conditions_on_bids} holds for continuous bids and for strictly convex bid curves. This could also be seen as the limiting case, by taking the limit where the increment $m$ goes~to~$0$.

We illustrate the conditions in~\cref{thm:conditions_on_bids}  by revisiting \cref{ex:first_simple_example}~from~\cref{sec:p1_2}.

\begin{example}[Simple Market]
	Revisiting \cref{ex:first_simple_example} and \cref{ex:core_illust}, we observe that the bid from bidder $3$ does not satisfy the conditions in~\cref{thm:conditions_on_bids}, because the bid price for $400$ MW is not submitted. Assume instead bidder $3$ provides the mutually exclusive bids of $\$300$ for $400$ MW and $\$600$ for $800$ MW. Suppose bidders $1$ and $2$ change their bids to 
	$\$0$ for $400$ MW. Then, bidders $1$ and $2$ are the winners and each receives a payment of $\$300$. If they were multiple identities of a single bidder, after shill bidding, sum of their payments would decrease from $\$700$ to $\$600$.
\end{example}

Next, we consider a more general setting with different types of goods where it is still possible to derive conditions to ensure supermodularity and coalition-proof VCG outcomes.

\section{Markets for different types of goods}

We now consider reverse auctions where the central operator is procuring a set of different types of goods. Each bidder has a private true cost function $c:\mathbb R_+^t \rightarrow \mathbb R_+$ that is nondecreasing with $c(0)=0$. We assume that this cost has an additive from, $c(x)=\sum_{\tau=1}^{t}c_\tau(x_\tau)$. Typically, in these markets, bids are submitted separately for each type with an upper-bound on the amount to be procured, $\bar{X}_\tau\in\R_+$ \cite{abbaspourtorbati2016swiss}. The operator treats these bids as bids from different identities, and then distributes the payments accordingly. In this case, the set~$L$ is the extended set of bidders such that the bid profile $\mathcal C$ is given by the bids of the form $c_l:\mathbb R_+ \rightarrow \mathbb R_+$, $\bar{X_l}\in\R_+$,  for all $l\in L$.\footnote{This assumption is in fact without loss of generality. Note that any supermodular function $f:2^{L}\rightarrow\R$ preserves its supermodularity when we reduce its ground set $L$ by combining some of the elements to obtain a new ground set $L'$ with $|L'|\le |L|$, see the discussions in~\cite[Appendix B]{bach2011learning}.}

Let $[t]= \{1,\ldots,t\}$. Define the set $\{A^\tau\}_{\tau=1}^t$ to be a partition of the set $L$ where each set $A^\tau\subseteq L$ is the set of bidders submitting a bid for goods of type $\tau$. Specifically, we consider auctions cleared by the optimization problem:
\begin{equation}	\label{eq:multiple}
\begin{split} 
J(\mathcal{C}_S) =&  \min_{x\in\R^{|L|}_+  } {\ \sum_{l \in S}c_l( x_l)} \\
&\ \ \mathrm{s.t. } \ \sum_{l\in A^T} x_l \geq M(T) ,\, \forall T \subseteq [t], \\
&\quad \quad\ \  x_l \leq \bar{X}_l,\, \forall l\in L, \\
&\quad \quad\ \ x_l = 0,\, \forall l\in L\setminus S, 
\end{split}
\end{equation}
where $A^T=\bigcup_{\tau\in T}A^\tau$. We minimize the sum of declared costs subject to constraints on the subsets of types $[t]$. Here, the function $M:2^{[t]}\to\R_+$ defines the amount the operator wants to procure from possible combinations of different types of goods. We assume that $M(\emptyset)=0$ (normalized). We remark that the optimization problem in~\eqref{eq:multiple} contains the case in \cref{ex:second_simple_example}. 

In the following result we see that if $c_l$ and $M$  satisfy convexity and supermodularity conditions respectively, then $J$ is supermodular. 

\begin{theorem}\label{thm:re2_proof}
	{The objective function $J$ given by \eqref{eq:multiple} is supermodular when} $c_l$ is increasing and convex for all $l \in L$, $M$ is supermodular 
and nondecreasing.
\end{theorem}

The proof is relegated to the appendix in~\cref{app:finproof}, and it builds upon recent advances in polymatroid optimization.\footnote{Polymatroid is a polytope associated with a submodular function. We highlight that under the conditions in~\cref{thm:re2_proof}, the first set of constraints in~\eqref{eq:multiple} is a contra-polymatroid~\cite{schrijver2003combinatorial}.}$^,$\footnote{Later on during the doctoral studies, we proved that the central operator exhibits strong-substitute property as it is defined in~\cite{milgrom2007substitute}, if it has a contra-polymatroid procurement constraint. \cite[Theorem 22]{milgrom2007substitute} states that if all bidders have strong-substitute valuations, a forward auction has a submodular objective function. Hence, it could in fact be possible to combine our observations with \cite[Theorem 22]{milgrom2007substitute} to extend our result. We believe that this extension can potentially relax the separable convexity requirement of \cref{thm:re2_proof} on the bid functions to supermodularity and component-wise convexity, see \cite[Theorem 11]{milgrom2007substitute} and the comparisons of these functions classes in~\cite[\S 3]{bian2019provable}). However, this requires further investigation. Finally, we highlight that our proof does not share similarities with that of \cite[Theorem 22]{milgrom2007substitute}, which relies on an analysis based on the conjugates of the bid functions. \label{footnote:extension}} As a corollary of this result, the conditions on $c_l$, $l\in L$ and $M$ in~\cref{thm:re2_proof} imply core VCG outcomes for \eqref{eq:multiple}. The conclusions of~\cref{thm:no_collusions} on shill bidding and collusion are further corollaries of this~result. 

Next, we illustrate that the VCG outcome may not lie in the core for a general polyhedral constraint set by revisiting~\cref{ex:second_simple_example}, and then we illustrate how the conditions in~\cref{thm:re2_proof} imply core outcomes {in this example. In the numerics, we also show that the polymatroid condition holds for optimal power flow problems with no line limits.}

\begin{example}[Power Market]
	Revisiting~\cref{ex:second_simple_example}, for the constraint set \eqref{newc2}, we highlight that $M$ is not supermodular: $M(\{\text{A},\text{B},\text{C}\}) + M(\{\text{A}\}) < M(\{\text{A},\text{B}\}) + M(\{\text{A},\text{C}\}) $ where $M(\{\text{A}\})=0$ and $M(\{\text{A},\text{B},\text{C}\})=0$. Then, the VCG outcome under collusion is blocked by a deal between the operator and bidder~$1$. {Now instead consider the following constraint set:}
	\begin{equation*}	\label{newcz}
	\begin{split}
	\{ x\in\{0,100\}^5  \,\rvert\, & x_1 + x_2 + x_3 \geq \tilde{M}(\{\text{A},\text{B}\})=100,\\
	& x_1 + x_4 +  x_5 \geq \tilde{M}(\{\text{A},\text{C}\})=100,\\
	& x_1 + x_2 + x_3 + x_4 + x_5 \geq \tilde{M}(\{\text{A},\text{B},\text{C}\})=200\}, 
	\end{split}
	\end{equation*}	
	which is an inner approximation of the constraints found in~\cref{ex:second_simple_example}. 
	Under this new constraint set, type A can still replace types B and C, but it cannot replace both types simultaneously. In other words, types B and C cannot complement each other to replace type~A as well.
	Note that the function $\tilde{M}$ is supermodular, normalized and nondecreasing, which satisfies all the requirements of~\cref{thm:re2_proof}.  {So, here the VCG outcome lies in the core.} Specifically,
	under the VCG mechanism, bidders $2$ and $4$ become winners and receive $p_2 = 350 + (650-600)= \$400$ and $p_4 = 250 + (750-600)= \$400$. This outcome is not blocked by any other coalition and collusion is not profitable for bidders. This example illustrates that marginally increasing cost curves alone are not enough to conclude the supermodularity of the reverse auction objective function in~(\ref{eq:main_model}). 
\end{example}

\section{Appendix}
\subsection{Proof of~\cref{lem:lemma_core}}\label{appforwin}
The utility allocation
	$ u^{}$ is unblocked by every $S\subseteq L$ if and only if
	\begin{equation*}
	-J(\mathcal{C}_S) \leq \sum_{l \in S} u^{}_l + u_0 = \sum_{l \in S} u_l - \sum_{l \in W}  u^{}_l - J(\mathcal{C}),\,\forall S\subseteq L,
	\end{equation*}
	since the losing bidders are not allocated and they obtain zero payment. Thus, the core can equivalently be parametrized as $$\sum_{l \in W\setminus S} u_l \leq J(\mathcal{C}_S)- J(\mathcal{C}),\ \forall S\subseteq L,$$ for the bidders. Moreover, fixing the set $K = W\setminus S$, the dominant constraints are those corresponding to minimal $J(\mathcal{C}_S)$, in particular, when we have $S = L\setminus K$ (since this is the maximal set with $K$ not taking part in the coalition~$S$). \QEDA

\subsection{Proof of~\cref{lem:removal_of_two}}\label{appforlose}
Notice that if the optimization problem \eqref{eq:main_model} is infeasible then the objective value is $J(\mathcal{B})=\infty$. Given this property of the market, inequality constraints \eqref{eq:core_constraints} in~\cref{lem:lemma_core} simplify to constraints on the utilities of single bidders, that is,
	$u^{\text{VCG}}_l\leq J(\mathcal{C}_{-l}) - J(\mathcal{C}),$ for all $l\in W$.
	This inequality follows directly from the definition of the VCG utility, $u^{\text{VCG}}_l= J(\mathcal{C}_{-l}) - J(\mathcal{C}),$ for all $l\in W$. The equality constraint in~\cref{lem:lemma_core} is satisfied by definition. \QEDA
	
	\subsection{Proof of~\cref{thm:iff_supermodularity}}\label{appforsup}
We prove that supermodularity is sufficient for the truthful VCG utilities to lie in the core, $Core(\mathcal{C}_R)$, for all $R\subseteq L$. Notice that we have $u_{l,R}^{\text{VCG}}=J(\mathcal{C}_{R\setminus{l}}) - J(\mathcal{C}_R)$ (that is, the VCG utility of bidder $l$ when $R$ is participating), we drop the dependence on $R$ for the sake of simplicity in notation. We recall~\cref{lem:lemma_core}. Hence, we need to show that
	\begin{equation}
	\label{eq:core_constraints2}
	\sum_{l \in K }  u_l\leq J(\mathcal{C}_{R \setminus K}) - J(\mathcal{C}_R),\ \forall K\subseteq W,
	\end{equation}
	where $W$ are the winner subset of $R$.
	Let $K=\{ l_1, \dots, l_k\}$. 
	{Notice that, by supermodularity,
		$
		J(\mathcal{C}_{R\setminus{l_\kappa}}) 
		- J(\mathcal{C}_R) 
		\leq 
		J(\mathcal{C}_{R \setminus \{ l_\kappa ,..., l_k \} })
		-
		J(\mathcal{C}_{R\setminus \{ l_{\kappa+1} ,..., l_k \}}) 
		$. 
		Thus 
		\begin{align*}
		\sum_{l\in K} u_l^{\text{VCG}} 
		&= 
		\sum_{\kappa=1}^k
		J(\mathcal{C}_{R\setminus{l_\kappa}}) 
		- J(\mathcal{C}_R) 
		\\
		&\leq 
		\sum_{\kappa =1}^k
		J(\mathcal{C}_{R \setminus \{ l_\kappa ,..., l_k \}} )
		- J(\mathcal{C}_{R\setminus  \{ l_{\kappa+1} ,..., l_k \}}) 
		\\
		&=
		J(\mathcal{C}_{R \setminus K} ) - J(\mathcal{C}_R).
		\end{align*}
		The last equality holds by a telescoping sum.}
	Thus, we see that \eqref{eq:core_constraints2} holds and so, by~\cref{lem:lemma_core}, the VCG outcome belongs to the core. The same argument can be repeated to obtain core VCG outcomes for any set of participating auction bidders $R\subseteq L$ and for any profile $\mathcal{C}$.

	To prove that supermodularity is also necessary for outcomes to lie in the core, {we proceed by contradiction.} Suppose that the supermodularity condition does not hold for a bidder~${l}$. Then, there exist sets ${S}\subseteq{R}$ where $J(\mathcal{C}_{R\setminus{l}})\! -\! J(\mathcal{C}_R)\!>\!  J(\mathcal{C}_{S\setminus{l}})\! -\! J(\mathcal{C}_S)$. We may,{ without loss of generality,} choose $R= S\cup \{i\}$ for some $i$. To see this, take $S^0=S$ and
	$S^\kappa=S^{\kappa-1}\cup \{ l_\kappa\}$ with $S^k=R$, then,
		\begin{align}
		\sum_{\kappa =1}^k J(\mathcal{C}_{S^\kappa\setminus{l}})-J(\mathcal{C}_{S^{\kappa-1}\setminus{l}}) &= J(\mathcal{C}_{R\setminus{l}})-J(\mathcal{C}_{S\setminus{l}}) \notag\\
		& >  J(\mathcal{C}_R)-J(\mathcal{C}_S) = \sum_{\kappa=1}^k J(\mathcal{C}_{S^\kappa})-J(\mathcal{C}_{S^{\kappa-1}}).\notag
		\end{align}
	The strict inequality above must hold for one of the summands  $J(\mathcal{C}_{S^\kappa\setminus {l}})-J(\mathcal{C}_{S^{\kappa-1}\setminus {l}}) > J(\mathcal{C}_{S^\kappa})-J(\mathcal{C}_{S^{\kappa-1}})$.  Therefore, we may consider sets ${S} \subseteq {R}$ that differ by one bidder, say $i$. Thus, by this observation, we have
	$ J(\mathcal{C}_{{R}\setminus {l}}) - J(\mathcal{C}_{R}) > J(\mathcal{C}_{{S}\setminus {l}}) - J(\mathcal{C}_{S}) = J(\mathcal{C}_{{R}\setminus{\{i,l\}}}) - J(\mathcal{C}_{{R}\setminus {i}})$ for $i \in {R}\setminus {S}$.
	Further, after rearranging the above inequality we obtain
	{
		\begin{equation}\label{eq:JR}
		J(\mathcal{C}_{R\setminus{i}})-J(\mathcal{C}_{R}) > J(\mathcal{C}_{R\setminus{\{i,l\}}})-J(\mathcal{C}_{R\setminus{l}})\geq 0.
		\end{equation}	
	}That is, both bidder $i$ and $l$ are winners of the VCG auction with bidders $R$.  
	Considering the auction with the set of bidders ${R}$ and with  $i,{l}\in W$, and ${K}=\{i, {l}\}$, we have: 
	{
		\begin{align*}
		\sum_{l' \in {K}} u_{l'}^{\text{VCG}}=&\sum_{l' \in {K}}\!\! {J(\mathcal{C}_{{R}\setminus{l'}}) - J(\mathcal{C}_{R})} \\
		=& 
		J(\mathcal{C}_{{R}\setminus{\{i\}}}) - J(\mathcal{C}_{{R}}) + J(\mathcal{C}_{{R}\setminus{l}}) - J(\mathcal{C}_{R}) 
		\\
		>& 
		J(\mathcal{C}_{{R}\setminus{\{i,l\}}}) - J(\mathcal{C}_{{R}\setminus{l}}) + J(\mathcal{C}_{{R}\setminus{l}}) - J(\mathcal{C}_{R}) 
		\\
		=& J(\mathcal{C}_{{R} \setminus {K}}) - J(\mathcal{C}_{R}),
		\end{align*}
		where in the inequality above we apply \eqref{eq:JR}.
		Thus, given \cref{lem:lemma_core}, \eqref{eq:core_constraints2} does not hold, consequently, $u^{\text{VCG}} \notin Core(\mathcal{C}_R)$. Thus the outcome of the VCG mechanism is not in the core for the subset of bidders $R\subseteq L$.
	}
	\QEDA

\subsection{Proof of~\cref{thm:no_collusions}}\label{appforsupcoal}
		(i) Let $K$ be a set of colluders who would lose the auction when bidding their true values $\mathcal{C}_K=\{c_l\}_{l\in K}$, when bidding $\mathcal{B}_K=\{b_l\}_{l\in K}$ they become winners, that is, they are all allocated a positive quantity. We define $ {\mathcal{C}} = \mathcal{C}_{K}\cup \mathcal{C}_{-K}$ and ${\mathcal{B}}={\mathcal{B}}_{ K}\cup \mathcal{C}_{-K}$ where $\mathcal{C}_{-K}=\{c_l\}_{l\in L\setminus K}$ denotes the bidding profile of the remaining bidders. As a remark, the profile $\mathcal{C}_{-K}$ is not necessarily a truthful profile. The VCG utility that each player $l$ in $K$ receives under ${\mathcal{B}}$ is
		\begin{align*} 
		u_l^{\text{VCG}}({\mathcal{B}}) &\leq u_l^{\text{VCG}}({\mathcal{B}}_{-l}\cup \mathcal{C}_l)\\  &= J({\mathcal{B}}_{-l} ) - J({\mathcal{B}}_{-l}\cup \mathcal{C}_l) \\  
		& \leq J( \mathcal{C}_{-K}) - J( \mathcal{C}_{-K} \cup {\mathcal{C}}_l) \\
		& =  J( \mathcal{C}_{-l}) - J( \mathcal{C}) \\ &=u_l^{\text{VCG}}({\mathcal{C}})\\&=0,
		\end{align*}
		where the first inequality follows from the dominant-strategy incentive-compatibility of the VCG mechanism. The first equality comes from the definition of the VCG mechanism and the second inequality applies the supermodularity of the function~$J$. The next equality comes from the fact that the set $K$ originally was a group of losing bidders.
		
		So we see that, for all $l \in K$, the utility $u_l^{\text{VCG}}({\mathcal{B}}) $ is upper bounded by the utility that bidder $l$ would receive if every colluder was bidding truthfully. However, by the initial assumption, these bidders were losers while bidding truthfully, and hence $u_l^{\text{VCG}}({\mathcal{C}})=0,$ for all $l \in K$. Thus, there is no benefit for losers from colluding by jointly deviating from their truthful bids.
		
		(ii) Similar to part (i), define $\mathcal{C}=\mathcal{C}_{-l}\cup\mathcal{C}_{l}$. The profile $\mathcal{C}_{-l}$ is not necessarily a truthful profile. Shill bids of bidder~$l$ are given by $\mathcal{B}_{S}=\{b_k\}_{k\in S}$.~We define a merged bid $\tilde{\mathcal{B}}_l$ as $$\tilde{b}_l(x_l)=\min_{x_k\in\hat{\X}_k,\,\forall k}\, \sum_{k\in S}b_k(x_k)\ \mathrm{s.t. }\sum_{k\in S}x_k=x_l.$$ We then define ${\tilde{\mathcal{B}}}=\mathcal{C}_{-l}\cup{\tilde{\mathcal{B}}}_{l}$.
		The VCG utility from shill bidding under ${\mathcal{B}}=\mathcal{C}_{-l}\cup{\mathcal{B}}_{S}$, $\sum_{k\in S}{u}_k^{\text{VCG}}({\mathcal{B}})$, is given by
		\begin{align*} &= \sum_{k\in S}[J({\mathcal{B}_{-k}})-J({\mathcal{B}})+b_k(x_k^*({\mathcal{B}}))]-{c}_l(\sum_{k\in S}x_k^*({\mathcal{B}}))\\
		&\leq [J({\mathcal{B}}_{-S} ) - J({\mathcal{B}})]+\sum_{k\in S} b_k(x_k^*({\mathcal{B}}))-{c}_l(\sum_{k\in S}x_k^*({\mathcal{B}}))\\ 
		&= [J({\mathcal{C}}_{-l} ) - J({\tilde{\mathcal{B}}})]+\tilde{b}_l(\sum_{k\in S}x_k^*({\mathcal{B}}))-{c}_l(\sum_{k\in S}x_k^*({\mathcal{B}}))\\ 
		&=  {u}_l^{\text{VCG}}({\tilde{\mathcal{B}}}) \\ 
		&\leq  {u}_l^{\text{VCG}}({\mathcal{C}}).
		\end{align*}
		The first inequality follows from the supermodularity of~$J$. The second equality holds since we have $J({\tilde{\mathcal{B}}})=J({{\mathcal{B}}})$. This follows from the definition of the merged bid and the following implication. Since the goods of the same type are fungible for the central operator, the functions $h$, $g$ and~$d$ in fact depend on $\sum_{l\in L} x_l$. 
		The third equality follows from the definition of VCG utility. The second inequality is the DSIC property of the VCG mechanism. Therefore, the total VCG utility that $l$ receives from shill bidding is upper bounded by the utility that $l$ would receive by bidding truthfully as a single bidder. Making use of shills, hence, is not profitable.\QEDA 

\subsection{Proof of~\cref{thm:conditions_on_bids}}\label{appforsimpl}
A preliminary version of this proof was first presented in a conference paper in~\cite{pgs}. The following proof simplifies this result specifically in the steps of the main proof and corrects notational inconsistencies. 

To prove~\cref{thm:conditions_on_bids}, the following lemma is needed.

\begin{lemma}\label{lem:increasing_quantities}
	Under the market model \eqref{eq:simpler_clearing_model}, for an auction with bidders $S$ and $R=S\cup\{j\}$  with corresponding allocations  $x$ and $x'$,  marginally increasing costs imply that 
	$\forall l\in S, \
	x_l'\leq x_l.
	$
\end{lemma}
\begin{proof}
	The proof  follows by contradiction. That is, we will show that when $x'$ is such that ${x}'_{{l}} > x_{{l}}$, for some $l\in S$, then $x'$ can be modified to provide a lower cost via the allocation $q$ for bidders~$S$ (thus contradicting optimality of $x$). 
	
	First, we notice that since bids are equally spaced by the amount $m$ and with marginally increasing cost,  $\sum_{l \in S}x_l = \sum_{l \in R}{x'_l} = M$ holds. Now, in order to procure exactly $M$ amount from bidders $R$, some bidders' allocations must decrease, that is, the set $K = \{ l \in S \,\rvert\, x_l'< x_l  \}$ is nonempty. Consider a feasible allocation $q'$ for the auction with bidders $R$ where the amount $M$ is being procured and
	\[
	q'_l = 
	\begin{cases}
	x'_j, &\text{ for }l=j,\\
	x_l, & \text{ for } l\in S\backslash K, \\
	q'_l, &\text{ for }l\in K \text{ where }x_l' \leq q'_l \leq x_l.
	\end{cases}
	\]
	Hence, $q'$ is constructed from $x'$ by transferring the amount $m'~=~\sum_{l\in S\backslash K} x'_l - x_l $ from bidders in $S\backslash K$ to bidders in $K$. In doing so, the inequality $x'_l \leq q_l' \leq x_l$ can be satisfied:
	\[
	m' \leq \sum_{l\in K} (x_l -x'_l).
	\]
	The above inequality holds because when summing over $l\in S$, $x_l$'s sum to $M$ and $x'_l$'s sum to $M-x_j$.
	
	Since {$x'$} is optimal for bidders $R$ and $q'$ is not:
	\begin{align}
	J(\mathcal{C}_R) =& c_j{(x'_j)} + \sum_{l\in S\backslash K } c_{{l}}{(x'_l)} + \sum_{l \in K }{c_l}{(x'_l)}\notag\\
	\leq & c_j{(x'_j)} + \sum_{l\in S\backslash K} c_l{(x_l)} + \sum_{l\in K} c_l{(q'_l)}=\bar{J}(q'),  \label{eq:q'}
	\end{align}
	where we used $\bar{J}(q')$ as a short-hand-notation for the cost corresponding to choosing the allocation~$q'$ under truthful bidding. 
	
	Now, we use the marginally increasing true costs to replace the summations over $K$ in \eqref{eq:q'}. 
	In particular, define $q=(q_l : l\in S)$ so that 
	$$ 
	q_l := x_l +(x'_l- {q'_l})=
	\begin{cases}
	x'_l & \text{for } l \in S \backslash K,\\
	x_l + x'_l -q'_l & \text{for } l \in K.
	\end{cases}
	$$
	Note that $q$ is feasible since $x'$ and $q'$ have the same sum over $S$ (and thus cancel) and $x_l$ is feasible. 
	Further, since $(q_l - x_l)  = (x'_l- {q'_l})$,
	\begin{equation} 
	\label{eq:increas}
	\sum_{l \in K}{ c_l{(q_l)}  - c_l{({x'_l})}} < \sum_{l \in K}{c_l{(x_l)} - c_l{({q'_l})}}.
	\end{equation}
	Adding \eqref{eq:increas} to both side of \eqref{eq:q'} (and canceling $c_j{(x'_j)}$) gives
	\begin{align*}
	\bar{J}(q) &= \sum_{l\in S\backslash K} c_l{(x'_l)} + \sum_{l\in K} c_l{(q_l)}\\
	& < \sum_{l\in S\backslash K} c_l{(x_l)} + \sum_{l\in K} c_l{(x_l)} = J(\mathcal{C}_S),
	\end{align*}
	which contradicts the optimality of $x$. This concludes the proof.
\end{proof}
\begin{corollary}\label{proofcor}
	Given the conditions in~\cref{thm:conditions_on_bids} and the optimal allocation to procure the amount $M$, for any lower amount $\tilde M\leq M$(while still being a multiple of $m$)  to be procured, the allocation for each bidder does not increase.
\end{corollary}

This corollary follows directly  from~\cref{lem:increasing_quantities}. Now, we are ready to prove~\cref{thm:conditions_on_bids}. 

\begin{proof}
	We prove that $J$ is supermodular. We adopt the same notation used in~\cref{lem:increasing_quantities} and we identify with $W\subseteq S$ the set of winners under the set $S$.
	For each $l \notin W $, we have by definition $x_{l} = 0$. 
	Thus $J(\mathcal{C}_{S\setminus{l}})-J(\mathcal{C}_S)=0$, (since the optimal solution is unchanged when $l$ is removed from S). By~\cref{lem:increasing_quantities},  ${x}'_{l}=0$ and so $J(\mathcal{C}_{R\setminus{l}})-J(\mathcal{C}_R)=0$ also. Thus, supermodularity holds for $l \notin W$.
	
	For each winning bidder $w \in W$, for the sake of notational compactness, we denote $u_{w}^{\text{VCG}}(\mathcal{C}_S) = J(\mathcal{C}_{S\setminus{w}})- J(\mathcal{C}_S)$. Note that these values are in fact the truthful VCG utilities under $\mathcal{C}_S$, which justifies the VCG term. Adopting the same notation of~\cref{lem:increasing_quantities}, we can indicate it as: 
	\[
	u_{w}^{\text{VCG}}(\mathcal{C}_S):= - c_{w}{(x_{w})}+\sum_{l \in S_{-w}} { ( c_l{(\varrho_l)} - c_l{(x_l)} ) }, 
	\]
	where $\varrho_l$ are the optimal allocations of each $l \in S_{-w}$, when $w$ leaves the auction. By~\cref{lem:increasing_quantities}, $\varrho_l\geq x_l$. Similarly, after bidder $i$ enters the auction, $u_{w}^{\text{VCG}}(\mathcal{C}_R) = J(\mathcal{C}_{R\setminus{w}})- J(\mathcal{C}_R)$. That is,
	\begin{equation*}
	u_{w}^{\text{VCG}}(\mathcal{C}_R):=- c_{w}{({x}'_{w})}  + \sum_{l \in S_{-w}} { ( c_l{(\varrho_l')} - c_l{({x'_l})} ) } +  c_i{(\varrho_i')} - c_i{({x'_i})},
	\end{equation*}
	where $\varrho_i'$ are the amounts accepted from $ l \in R_{- w}$ when $w$ leaves the new auction. By~\cref{lem:increasing_quantities}, we again have $\varrho_l'\geq {x'_l}$.
	
	Notice that so far we applied Lemma~\ref{lem:increasing_quantities} to justify the increase of the accepted amounts, first, from each $l\in S_{- w}$ and now from $l \in R_{-w}$, due to the exit of $w$ from the auctions. We can also apply~\cref{lem:increasing_quantities} again and affirm that $ x'_l \leq x_l\ \forall l \in S$, and in particular $ x'_{w} \leq x_{w}$, because of the entrance of $i$.
	
	We now find suitable lower and upper bounds to ensure  inequality  $u_{w}^{\text{VCG}}(\mathcal{C}_R) \leq u_{w}^{\text{VCG}}(\mathcal{C}_S)$, which would confirm supermodularity. 
	
	First, note that 
	$ J(\mathcal{C}_S)= \sum_{ l \in S_{-w}}{ c_l{(x_l)} } + c_{w}{(x_{w})} \leq \sum_{ l \in S_{-w}}{ c_l{(q_l)} } + c_{w}{(x'_{w})}$, 
	where $q_l$'s are from the cheapest allocation to procure the amount $(M - x'_w)$ among $S_{-w}$.
	By~\cref{lem:increasing_quantities} (and~\cref{proofcor}) we  have $\varrho_l \geq q_l \geq x_l, \hspace{1em} \forall l \in S_{-w}$. Note that $x_l$'s sum to $(M - x_w) \leq (M- x'_w) $ (due to $x'_w \leq x_w$), and $\varrho_l$'s sum to $M$. 
	
	Moreover, since every $c_l$ is increasing, $q_l$'s are such that $\sum_{l \in S_{-w}} ({\varrho_l}- q_l) = x'_{w}$, because exactly the amount $M$ is purchased.
	Using the above suboptimal allocation, we  have a lower bound for $u_{w}^{\text{VCG}}(\mathcal{C}_S)$:  
	\begin{equation}
	\label{lower_bound}
	u_{w}^{\text{VCG}}(\mathcal{C}_S) \geq \sum_{l \in S_{-w}} { ( c_l{(\varrho_l)} - c_l{(q_l)} ) } - c_{w}{(x'_{w})} .
	\end{equation}
	Defining now $\delta_l =({\varrho_l} - q_l), \forall l \in S_{-w}$ we must have $\sum_{ l\in S_{-w}} \delta_l = x'_{w}$. Notice that $$ J(\mathcal{C}_{R\setminus{w}})= \sum_{l \in S_{-w}} {c_l{(\varrho_l')} } + c_i{(\varrho_i')} \leq   \sum_{l \in S_{-w}} {c_l{({x'_l} + \delta_l)} } + c_i{(x'_i)},$$  since the right hand side is a feasible cost to procure the amount $M$ among the bidders $\{S,i\}\setminus w$. Indeed, $\sum_{l \in S}{x'_l} + x'_i = M$ and $\sum_{l \in S_{-w}} {\delta_l} = x'_{w}$. Hence, we have: 
	\begin{equation}
	\label{upper_bound}
	\begin{split}
	u_{w}^{\text{VCG}}(\mathcal{C}_R) \leq \sum_{l \in S_{-w}}& {(c_l{(x'_l + \delta_l)} - c_l{(x'_l)})} + (c_i{(x'_i)} - c_i{(x'_i)}) - c_{w}{(x'_{w})} .
	\end{split}
	\end{equation}
	Observe that the terms in the parentheses cancel each other.
	Moreover, via marginally increasing costs (also via strictly convex costs), we have: 
	\begin{equation} 
	\label{increasing_marginal}
	( c_l{(x'_l + \delta_l)} - c_l{(x'_l)} ) \leq  ( c_l{(\varrho_l)} - c_l{(q_l)} ),\, \forall l \in S_{-w}.
	\end{equation}
	The above holds because $\forall l \in S_{-w}$, $(\varrho_l - q_l) =( x'_l + \delta_l - x'_l) = \delta_l$ and $x'_l \leq q_l$. In particular, $x'_l$ are the amounts accepted to procure the amount $(M-x'_{w})$ among $\{S,i\} \setminus w$, while $q_l$ are those to procure the same amount among $S_{-w}$. Then, combining equations~\eqref{upper_bound}, \eqref{increasing_marginal} and \eqref{lower_bound}, we finally obtain $u_{w}^{\text{VCG}}(\mathcal{C}_R) \leq u_{w}^{\text{VCG}}(\mathcal{C}_S).$ As a result, we obtain supermodularity and this concludes the proof.
\end{proof}

\subsection{Proof of~\cref{thm:re2_proof}}\label{app:finproof}

	We prove that $J$ is supermodular. To this end, we first need to reparametrize the problem \eqref{eq:multiple} into another class of optimization problem. We define the function $f:2^{ L }\to\R_+$ as follows
	\begin{equation*}
		f(A)=\max_{\substack{\forall T \subseteq \{1,\ldots,t\}, \\ A\supseteq A^T}}\, M(T).
	\end{equation*}
	
	We consider the following optimization problem:
	\begin{equation}	\label{eq:multiple_type}
	\begin{split} 
	J(\mathcal{C}_S) =&  \min_{x\in\R^{|L|}_+  } {\ \sum_{l \in S}c_l( x_l)} \\
	&\ \ \mathrm{s.t. } \ \sum_{l\in A} x_l \geq f(A) ,\, \forall A \subseteq L, \\
	&\quad \quad\ \  x_l \leq \bar{X}_l,\, \forall l\in L, \\
	&\quad \quad\ \ x_l = 0,\, \forall l\in L\setminus S.
	\end{split}
	\end{equation}
	
	First, notice that $f(A^T)=M(T)$. Since $x_l\geq0,\,\forall l$, the constraints added by the definition of the function $f$ are all redundant constraints. Then, these constraints in \eqref{eq:multiple_type} are feasible, once the constraints in \eqref{eq:multiple} are satisfied. Hence, these problems are equivalent.  We also remark that $f(\emptyset)=0$ and $f$ is nondecreasing.
	
	Before we proceed, we need to show that supermodularity of the function $M$ implies supermodularity of the function $f$. Suppose $f(A)=M(T^A)$ and $f(B)=M(T^B)$ for some $T^A$, $T^B$. Then,
	\begin{equation*}
		\begin{split}
		f(A)+f(B)=\,&M(T^A)+M(T^B)\\
					\leq\,&M(T^A\cup T^B)+M(T^A\cap T^B)\\
					\leq\,&f(A\cup B)+f(A\cap B).
		\end{split}
	\end{equation*}
	First inequality follows from supermodularity of the function~$M$. The last inequality holds since these sets are feasible suboptima for $f(A\cup B)$ and $f(A\cap B)$. We conclude that the function $f$ is supermodular.
	
	Note that, given supermodularity of the function $f$, the first set of constraints in \eqref{eq:multiple_type} defines a \textit{contra-polymatroid}, see \cite[Section~44]{schrijver2003combinatorial} for a detailed exposition. {This class of problems are important in combinatorial optimization because they can often be solved in polynomial time. For the remainder of the proof, we extend the work of \cite{he2012polymatroid} on replenishment games to reverse auctions over contra-polymatroids and box constraints as in~\eqref{eq:multiple_type}.}
	
	We first show that the constraint  $\sum_{l\in L} x_l = f(L)$ is redundant and can be added to the original constraint set. We denote $S^{c}=L\setminus S$ and denote $x^*$ as the optimal allocation for \eqref{eq:multiple_type}. 
	It can be shown that for every $x_k^*$, $k\in S$ , there exists a set $k\in A_k \subseteq L$ such that  $\sum_{l\in A_k} x^*_l = f(A_k)$ is tight at the optimal solution. We can prove this via contradiction. Assume, for  $x_k^*$, $k\in S$, there does not exist a set $k\in A_k\subseteq L$ such that $\sum_{l\in A_k} x^*_l \geq f(A_k)$ is tight. Then, one can simply decrease the value of $x^*_k$ and get a lower objective value. Furthermore, note that any constraint corresponding to $A\not\supset S^c$ is redundant to $A\cup S^c$ because $f$ is nondecreasing and $x_l = 0$, for all $l\in S^c$. Then, this set $A_k$ has to be a superset of $S^c$, $A_k \supset S^c $.
	
	Next, we show that if the constraints for $A$ and $B$ are tight, so is the constraint for $A\cup B$: 
	\begin{subequations}
		\begin{align}
		f(A\cup B) + f(A\cap B) &\geq f(A) + f(B)\label{11a} \\ &= \sum_{l\in A} x^*_l + \sum_{l\in B} x^*_l\label{11b} \\ &= \sum_{l\in A\cup B} x^*_l + \sum_{l\in A\cap B} x^*_l\label{11c} \\ &\geq f(A\cup B) + f(A\cap B). \label{11d}
		\end{align}
	\end{subequations}	
	Inequality~\eqref{11a} follows from supermodularity of $f$, $f(A\cup B)  - f(A) \geq f(B) - f(A\cap B)$. Equality~\eqref{11b} follows from $A$ and $B$ being tight. We arranged the terms in the equality~\eqref{11c}. Inequality~\eqref{11d} follows from the feasibility of $x^*$ for the problem in~\eqref{eq:multiple_type}. Then, it is easy to see that  \eqref{11d} is in fact an equality and we can conclude that $\sum_{l\in A\cup B} x^*_l = f(A\cup B) $. 
	
	Recall that the constraint corresponding to the set $ A_l\supset S^c\cup \{l\}$ is tight for $x_l^*$. Hence, we can conclude that the constraint corresponding to the set $\bigcup_{l\in S}A_l = L$ is also tight and $\sum_{l\in L} x_l^* = f(L)$ holds for optimal solution.
	
	Next, we reformulate the first set of constraints in \eqref{eq:multiple_type} as follows:
	\begin{equation}\label{eq:ref1}
	\mathcal{P} = \Big\{ x\in\R_+ \,\rvert\,  -f(A) \geq \sum_{l\in A^c} x_l -\sum_{l\in L} x_l ,\, \forall A \subseteq L\Big\}.
	\end{equation}	
	Define $h(A)=-f(A^c)$ where $h(\emptyset)= -f(L)= -\sum_{l\in L} x_l $ and reorganize the constraint set~\eqref{eq:ref1}. 
	\begin{equation*}\label{eq:ref}
	\begin{split}
	\mathcal{P} &= \Big\{ x\in\R_+ \,\rvert\, h(A^c) +\sum_{l\in L} x_l \geq \sum_{l\in A^c} x_l ,\, \forall A \subseteq L\Big\}\\
	& = \Big\{ x\in\R_+ \,\rvert\,  h(A) +\sum_{l\in L} x_l \geq \sum_{l\in A} x_l ,\ ,\forall A \subseteq L\Big\}\\
	&  = \Big\{ x\in\R_+ \,\rvert\,  h(A) - h(\emptyset) \geq \sum_{l\in A} x_l ,\, \forall A \subseteq L,\, h(\emptyset) = -\sum_{l\in L} x_l  \Big\}. 
	\end{split}	
	\end{equation*}	
	We see that $k(A)=h(A)-h(\emptyset)$ is a nondecreasing submodular function and it is normalized. In literature, the function  $k$ is often called a \textit{rank function}~\cite{schrijver2003combinatorial}. Note that, $k(S)=h(S)-h(\emptyset)=-g(L\setminus S)+g(L)$. From the feasibility of the problem $\eqref{eq:multiple_type}$, we have that $g(L\setminus S)=0$ and $k(S)=\sum_{l\in L} x_l  $. We reorganize the constraint set  and obtain the following:
	\begin{equation*}
	\mathcal{P} = \Big\{ x\in\R_+  \,\rvert\, \sum_{l\in A} x_l \leq k(A) ,\, \forall A \subseteq L,\, k(S) = \sum_{l\in L} x_l  \Big\}. \\
	\end{equation*}	
	Finally, given that $k$ is nondecreasing and $x_l = 0,\, \forall l\in L\setminus S$, we can show that the constraints corresponding to $A\supset S$ are all redundant upper bounds. Then, we obtain the following
	\begin{equation*}	
	\mathcal{P} = \Big\{ x\in\R_+  \,\rvert\, \sum_{l\in A} x_l \leq k(A) ,\, \forall A \subseteq S,\, k(S) = \sum_{l\in S} x_l  \Big\}.
	\end{equation*}	
	The following result is known, and we refer to \cite[Theorem~6.1]{yao1997stochastic} and \cite{fujishige1980lexicographically}. The set $\mathcal{P}\cap \{x \,\rvert\,   x_l \leq \bar{X}_l,\, \forall l\}$ is equivalent to the set,
	\begin{align*}
	&\mathcal{P}'= \Big\{ x\in\R_+  \,\rvert\, \sum_{l\in A} x_l \leq \bar f(A) ,\, \forall A \subseteq S,\,k(S) = \sum_{l\in S} x_l  \Big\}, 
	\end{align*}	
	where $\bar f(A)=\min_{B\subseteq A}\{ k(A\setminus B)+\sum_{l\in B}\bar{X}_l  \}$ and this function is also a rank function \cite{yao1997stochastic}. We also assert that $\bar f(S)=k(S)$. To verify this equality, notice that: 
	\begin{equation*}\label{eq:deriv}
	\begin{split}
	k(S)=\sum_{l\in S} x_l &= \sum_{l\in S\setminus B} x_l +\sum_{l\in B} x_l\\
	&\leq k(S\setminus B)+\sum_{l\in B}\bar{X}_l,\, \forall B\subseteq S, \\
	\end{split}	
	\end{equation*}	
	hence, $\bar f(S)=\min_{B\subseteq S}\{ k(S\setminus B)+\sum_{l\in B}\bar{X}_l  \}= k(S)$.
	Finally, we obtain:
	\begin{equation}	\label{eq:final_form}
	\begin{split} 
	-J(\mathcal{C}_S) =&  \max_{\substack{x_l\in\R_+, \forall l \in S}}  \,\sum_{l \in S}-c_l( x_l) \\
	&\ \ \quad  \mathrm{s.t. } \  \ \sum_{l\in A} x_l \leq \bar f(A) ,\, \forall A \subseteq S, \\
	&\quad\quad\quad\ \ \sum_{l\in S} x_l=\bar f(S). \\
	\end{split}
	\end{equation} 
	
	Now, we are ready to bring in the results from the work of~\cite{he2012polymatroid}. The first set of constraints in \eqref{eq:final_form} defines a \textit{polymatroid}, see \cite[Section 44]{schrijver2003combinatorial}. In \cite[Theorem~3]{he2012polymatroid} (which builds upon the well-known optimality of the greedy algorithm over a polymatroid~\cite{edmonds2003submodular} and \cite[Theorem~2.1]{schulz2010sharing} by forming a linearization argument to the concave objective functions), it is proven that maximizing a separable concave function over a polymatroid results in a submodular objective function. The result is in fact proven for optimizing over the base polymatroid where the polymatroid constraint set is intersected with the equality $\sum_{l\in S} x_l=\bar f(S)$. Then, invoking this result, we conclude that $-J$ is submodular, and $J$ is supermodular. \QEDA

\chapter{Coalition-proof core-selecting mechanisms}\label{sec:p1_4}
In the previous chapter, we showed that the VCG mechanism can be guaranteed to achieve coalition-proofness only in restricted market problems.
In this chapter, we study core-selecting mechanisms. We first show that they are coalition-proof, in the sense that a group of bidders whose bids are not accepted when bidding truthfully cannot profit from collusion, and a shill bidder cannot profit more than its VCG utility corresponding to its truthful bidding. 
We then show that, in addition to being coalition-proof, core-selecting mechanisms generalize the economic rationale of the LMP mechanism. Namely, they are also the exact class of mechanisms that ensure the existence of a competitive equilibrium. 
This result implies that the LMP mechanism is core-selecting, and hence coalition-proof. Since the LMP mechanism may not exist, we then define a class of core-selecting mechanisms, applicable to any market modeled by~\eqref{eq:main_model}, that also approximates DSIC without the price-taking assumption. Finally, we prove its budget-balance.

\section{Coalition-proofness via core}

To design coalition-proof payments, we first define the \textit{revealed utilities}, that is, the utilities with respect to the submitted bids. This is crucial since the true utilities of the bidders are unknown for the design of the payment rules. We then bring in the definition for the \textit{revealed core}. 

The revealed utility of bidder $l$ is defined by $\bar u_l(\BB)=p_l(\BB)-b_l(x^*_l(\BB)),$ and the revealed utility of the operator is the same as its utility, $\bar u_0(\BB)=-\sum_{l\in L} p_l(\BB) - d(x^*(\BB),y^*(\BB))$.  
\begin{definition}\label{def:r_core_def}
	For every set of bidders $R\subseteq L$, the revealed core $Core(\mathcal{B}_R)\in\R\times\R^{\rvert R\rvert}_+$ is defined as follows
	\begin{equation}\label{eq:coref}
	\begin{split}
	Core(\mathcal{B}_R)=\Big\{\bar u\in\R\times\R^{\rvert R\rvert}_+ \,|\, &\bar u_0+\sum\limits_{l\in R}\bar u_l=-J(\mathcal{B}_R),\\
	&\bar u_0+\sum\limits_{l\in S}\bar u_l\geq-J(\mathcal{B}_S),\, \forall S \subset R \Big\}.
	\end{split}
	\end{equation}	
\end{definition}

Both the revealed utilities and the revealed core above would correspond to the true utilities and the true core as in~\cref{def:core_def} if we assume that the submitted bids are the true costs.
A mechanism is said to be \textit{core-selecting} if its payments ensure that the revealed utilities lie in the revealed core. Then, the payment rule is given by
$$p_l(\BB)=b_l(x^*_l(\BB))+\bar u_{l}(\BB),$$ 
where $\bar u(\BB)\in Core(\BB)$. For instance, it can be verified that the pay-as-bid mechanism is a core-selecting mechanism where $\bar u_l(\BB)=0$ for all $l\in L$, and $\bar u_0(\BB)=-J(\BB)$. This implies that the core is nonempty and these mechanisms exist.\footnote{If $b_l(0)\neq0$, the function $J$ may not be nonincreasing, the pay-as-bid revealed utilities may not lie in the revealed core, and the revealed core may be empty. We kindly refer to \cref{footnote:corempty} of \cref{sec:p1_3} for a way of addressing this problem.}  
Furthermore, core-selecting mechanisms are IR for the revealed utilities, since they are restricted to the nonnegative orthant for the bidders in \eqref{eq:coref}.

Our main result of this section shows that these mechanisms give rise to coalition-proof outcomes. 

\vspace{.1cm}
\begin{theorem}\label{thm:no_collusions_bocs}
	Consider a core-selecting auction mechanism modeled by \eqref{eq:main_model}.
	\begin{itemize}
		\item[(i)] {A group of bidders who lose when bidding their true values cannot profit by a joint deviation.}
		\item[(ii)] \hspace{-.05cm}{Bidding with multiple identities is unprofitable for all bidders with respect to the VCG utilities.}\end{itemize}
\end{theorem}
\vspace{.1cm}

The proof is relegated to the appendix in~\cref{app:Gcore}, and it extends a similar observation from~\cite[Theorem 1]{day2008core} in multi-item auctions to the electricity market setting under consideration. This proof can be used as an alternative approach to prove~\cref{thm:no_collusions}. We remark that the proof method differs from~\cref{thm:no_collusions} since this proof does not require supermodularity. 

We conclude that core-selecting mechanisms are coalition-proof. As a remark, the revealed core ensures this property since the inequality constraints in \eqref{eq:coref} restrict the revealed utilities such that they cannot be improved upon by forming coalitions. In fact, as it is discussed in~\cref{sec:p1_3}, the VCG mechanism fails to attain coalition-proofness since it is in general not a core-selecting mechanism.

\section{Coalition-proofness via competitive equilibrium}

This section shows that the core-selecting mechanisms offer an economic rationale similar to that of the LMP mechanism.
To state this result, we bring in tools from competitive equilibrium theory~\cite{mas1995microeconomic}. 
\begin{definition}\label{def:compeq}
	An allocation $x^*\in\R^{t|L|}_+$ and a set of price functions $\{\psi_l\}_{l\in L}$, where $\psi_l:\R_+^t\rightarrow\R$ and $\psi_l(0)=0$, constitute a competitive equilibrium if and only if the following two conditions hold
	\begin{align}
	&\text{(i)}\ x_l^* \in \argmax_{x_l\in \X_l}\, \psi_l(x_l)-c_l(x_l),\quad \forall l\in L,\label{eq:firstcondition} \\
	&\text{(ii)}\ x^* \in \argmin_{x\in \R_+^{t|L|}} \bigg\{\min_{\substack{y:\, h(x,y)= 0\\ g(x,y)\leq 0 }}\,\, \sum\limits_{l\in L} \psi_l(x_l) + d(x,y) \bigg\}. \label{eq:secondcondition}
	\end{align}
\end{definition}

As a remark, price functions~$\{\psi_l\}_{l\in L}$ map from power supplies to the payment received by each bidder, whereas the payments~$\{p_l\}_{l\in L}$ map from the bid profile to the payments.
Notice that the functions $\{\psi_l\}_{l\in L}$ can be nonlinear and bidder-dependent. As a result, this definition extends beyond the traditional Walrasian competitive equilibrium which requires a linear price for each type of supply \cite{mas1995microeconomic,bikhchandani2002package,parkes2002indirect}.

The central assumption of competitive equilibrium conditions is that participants take the price functions~$\{\psi_l\}_{l\in L}$ as given and they do not anticipate their effects on them. Consequently, condition~\eqref{eq:firstcondition} implies that bidders are willing to supply their allocations since these allocations maximize their utilities given the price functions. Furthermore, these allocations also minimize the total payment of the central operator given the price functions by condition~\eqref{eq:secondcondition}. We highlight that a competitive equilibrium is not a game-theoretic solution, under neither the cooperative nor noncooperative approaches. Instead, it is a set of consistency conditions that models how payments would be formed from economic interactions. These conditions are considered to be a powerful benchmark in economic~analysis of electricity markets~\cite{mas1995microeconomic,schweppe2013spot,bose2019some}.

{Next, we show that for competitive equilibrium analysis we can restrict our attention to the optimal allocation of~\eqref{eq:main_model} under the truthful bid profile. To this end, the following lemma proves that a competitive equilibrium is efficient. \begin{lemma}\label{lem:effce}
 If an allocation $x^*$ and price functions~$\{\psi_l\}_{l\in L}$ constitute a competitive equilibrium, then~$x^*=x^*(\CC)$, that is, the allocation is the optimal allocation of~\eqref{eq:main_model} under true costs.
\end{lemma}}

The proof is relegated to the appendix in \cref{compeff}. We say that a mechanism \textit{ensures the existence of a competitive equilibrium} if under any true cost profile $\CC$ from the bidders there exists a set of price functions $\{\psi_l\}_{l\in L}$ such that these price functions constitute a competitive equilibrium with $x^*(\CC)$, and $\psi_l(x^*_l(\CC))=p_l(\CC)$ for all $l\in L$. Under such mechanisms, condition in~\eqref{eq:firstcondition} implies the following. Suppose bidders treat their price functions to be independent of their bids. Then, $x^*_l(\CC)$ maximizes the utility of each bidder. In this case, bidders would be willing to bid truthfully to ensure they are allocated the optimal quantity $x^*_l(\CC)$.
 
 As an example, consider the DC-OPF problem with a single bidder at each node. Recall that $\lambda_l^*(\CC)\in\R$ is the Lagrange multiplier of the $l^{\text{th}}$ nodal balance equality. Let the stacked vector $\lambda^*(\CC)\in\R^t$ be the concatenation of $\lambda_l^*(\CC)$, $\forall l$ where the number of supply types $t$ is equal to the number of nodes. If strong duality holds, then the LMP mechanism results in a competitive equilibrium with the allocation $x^*(\CC)$, and the set of price functions $\psi_l^{\text{LMP}}(x)=\lambda^*(\CC)\,x$ for all $l\in L$~\cite{schweppe2013spot,wu1996folk}. As a remark, in case there are several bidders on a single node in a DC-OPF problem, these bidders would be supplying the same type of supply to the grid. These bidders would then face the same LMP price $\lambda_l^*(\CC)$ since the price functions are formed by the Lagrange multipliers of the nodes. Hence, the LMP price functions are considered to be bidder-independent.

We now prove that core-selecting mechanisms coincide with those ensuring the existence of a competitive equilibrium.\begin{theorem}\label{thm:lmpiscs}
	A mechanism is core-selecting if and only if it ensures the existence of a competitive equilibrium, that is, for any true cost profile $\CC$ there exists a set of price functions $\{\psi_l\}_{l\in L}$ such that these functions constitute a competitive equilibrium with $x^*(\CC)$, and $\psi_l(x^*_l(\CC))=p_l(\CC)$ for all $l\in L$. 
\end{theorem}

The proof is relegated to the appendix in~\cref{app:lmpiscs}.
As a corollary, the LMP mechanism is core-selecting, and hence it is also coalition-proof. Coalition-proofness can be another motivation for using LMP in DC-OPF problems. {As a side note, using~\cref{thm:lmpiscs}, we can prove that the LMP payments are upper bounded by the VCG payments, we kindly refer to the appendix in~\cref{app:lmpvcg}.} 

\cref{thm:lmpiscs} shows that core-selecting mechanisms generalize the economic rationale of the LMP prices to core prices. The latter prices can be nonlinear and bidder-dependent. 
As mentioned earlier, Lagrange multipliers may not constitute a competitive equilibrium since strong duality may not hold for the general class of markets in~\eqref{eq:main_model}. In this case, the market may not have a competitive equilibrium in linear price functions. For instance, deriving a meaningful payment mechanism to accompany nonconvex AC-OPF dispatch and/or nonconvex unit commitment decisions is an open problem~\cite{bose2019some,FERC}. As we discussed in \cref{sec:p1_2}, there are various approximation techniques being investigated and implemented for the latter~\cite{gribik1993market,o2005efficient,chao2019incentives}.\footnote{We remark that condition~\eqref{eq:secondcondition} of a competitive equilibrium is ignored whenever one derives linear prices from an approximated problem, and complements it with uplift payments. Hence, the proposals in~\cite{gribik1993market,o2005efficient,chao2019incentives} are not core-selecting mechanisms, and they cannot be guaranteed to achieve coalition-proofness. Notice that, if we ignore the uplift payments, the linear prices derived in~\cite{gribik1993market,o2005efficient,chao2019incentives} satisfy the condition~\eqref{eq:secondcondition} for the approximated problems, however, in this case, it is easy to see that these linear prices would violate condition~\eqref{eq:firstcondition}.}
As core-selecting mechanisms exist even under nonconvex bids and nonconvex constraint sets, they are viable payment mechanisms for any market modeled by~\eqref{eq:main_model}.

Nevertheless, the VCG mechanism is the unique DSIC mechanism computing the optimal allocation under the submitted bids~\cite{green1979incentives,green1977characterization} (and IR in a minimal sense, see \cref{footnote:minima} from \cref{sec:p1_2}). Since the VCG mechanism may not be core-selecting, the DSIC property is in general violated under core-selecting mechanisms, and unilateral deviations can be profitable. Furthermore, competitive equilibrium theory relies on bidders treating their price functions as independent from their bids. This assumption does not take into account the full set of strategic behaviors. 
In the next section, we address the design of coalition-proof mechanisms to approximate DSIC without the price-taking assumption. We then address the case where the demand-side also submits bids.

\section{Designing coalition-proof mechanisms}\label{sec:3b}

\subsection{Approximating incentive-compatibility}\label{sec:3b1}
In this section, we characterize the class of core-selecting mechanisms that approximates DSIC by minimizing the sum of potential profits of each bidder from a unilateral deviation. Invoking the claim in~\cite{parkes2002achieving}, approximating DSIC provides us also with a meaningful method to approximate the efficiency property.

First, we quantify the violation of the DSIC property under any core-selecting mechanism.\begin{lemma}\label{lem:lie}
	Let $u(\BB)$ denote the utilities for any bid profile~$\BB$ under a core-selecting mechanism. The additional profit of bidder~$l$ by a unilateral deviation from its truthful bid, that is, $u_l(\hat \BB_l\cup \BB_{-l})-u_l({\CC_l\cup \BB_{-l}})$ for a nontruthful bid~$\hat \BB_l$, is at most $u^{\text{VCG}}_l(\CC_l\cup \BB_{-l})-u_l({\CC_l\cup \BB_{-l}})$, where $u^{\text{VCG}}_l(\CC_l\cup \BB_{-l})=J(\BB_{-l})-J(\CC_l\cup \BB_{-l})$. 
\end{lemma}

The proof is relegated to the appendix in~\cref{app:lemlie}, and it extends~\cite[Theorem 3.2]{day2007fair} to our electricity market setting.
This lemma provides a measure for the loss of incentive-compatibility under core-selecting mechanisms. The bound on the additional profit is exactly given by the difference between the VCG payment and the core-selecting payment in the case in which the bidder is truthful in both mechanisms. Hence, the closer the payments are to the VCG payments, the better the bound is.
Note that calculating the optimal deviation by the bidder given in the proof of~\cref{lem:lie} requires full information of the bids. Attempting this optimal deviation involves a risk, since bidding any amount higher to obtain a larger profit is not possible, and it would result in zero allocation and zero~utility.

Next, we design an approximately DSIC core-selecting mechanism.
A mechanism is said to be \textit{maximum payment core-selecting} (MPCS) if its revealed utilities are given by,
\begin{equation}\label{eq:mpcs}
	\bar u^{\text{MPCS}}(\BB) = \argmax_{u \in \text{Core}(\BB)}\ \sum_{l\in L} u_l - \epsilon\, \norm{u_l-\bar u_l^{\text{VCG}}(\BB)}_2^2\,, \vspace{-.1cm}
\end{equation}
where $\epsilon$ is a small positive number. This problem maximizes sum of the revealed utilities of the bidders over the revealed core constraints. The second term in the objective of~\eqref{eq:mpcs} is used as a tie-breaker. This term ensures that the optimizer is unique by picking the revealed utilities that are nearest to the VCG revealed utilities. See~\cref{fig:core_picture} for an illustration of the revealed utilities for two bidders under different mechanisms.\footnote{By definition, under the MPCS mechanism sum of the revealed utilities of the bidders are higher than the one under the LMP mechanism. However, we underline that every bidder may not receive an MPCS utility higher than its LMP utility. In the numerics, we provide one such instance.} Observe that problem~\eqref{eq:mpcs} is a convex quadratic program since the core is given by a set of linear equality and inequality constraints in~\eqref{eq:coref}. In the appendix in~\cref{appcomps}, we discuss its computational aspects and provide computationally efficient approaches.  

\begin{figure}[t]
	\begin{center}
		\begin{tikzpicture}[scale=0.79, every node/.style={scale=0.58}]
		\coordinate (r0) at (0,6);
		\coordinate (s0) at (6,0);
		\coordinate (si) at (0,0);
		\coordinate (s1) at (0,6);
		\draw[->] (0,0) -- (7,0) node[right] {\huge $\bar u_1$};
		\draw[->] (0,0) -- (0,6.75) node[above] {\huge $\bar u_2$};
		\filldraw[draw=black!80!blue,line width=.22mm, fill=gray!20] (r0) -- (s0) -- (si) -- (s1) -- cycle;
		\draw[scale=1, line width=.22mm, domain=0:6,smooth,variable=\x,black!80!blue] plot ({\x},{6-\x}) node[below] at(-2.75,1.5){\Huge $\text{Core}$};
		\draw[scale=1, line width=.22mm, domain=0:6,dashed,variable=\x,black!40!red]  plot ({\x},{6}) node[below] at(8.4,6.9){\Huge $(\bar u_1^{\text{VCG}},\,\bar u_2^{\text{VCG}})$};
		\draw[scale=1, line width=.22mm, domain=0:6,dashed,variable=\y,black!40!red]  plot ({6},{\y});
		\node[black!40!red] at (6,5.95) {\Huge\textbullet};
		\node[black!60!blue] at (3,2.95) {\Huge\textbullet};
		\node[black!] at (3.3,2.97) {\footnotesize\textbullet};
		\node[black!60!green] at (1.2,3.85) {\Huge\textbullet};
		\draw[scale=1, line width=.22mm, domain=0:3,smooth,variable=\x,black!60!blue] plot ({6-\x},{6-\x});
		\draw[scale=1, line width=.22mm, domain=0:0.312,smooth,variable=\x,black!60!blue] plot ({3.3+\x},{3.3-\x});
		\draw[scale=1, line width=.22mm, domain=-0.012:0.3,smooth,variable=\x,black!60!blue] plot ({3.6-\x},{3.0-\x});
		\draw[scale=1, line width=.22mm, domain=0:3,dashed,variable=\x,black!60!blue]  plot ({\x},{3});
		\draw[scale=1, line width=.22mm, domain=0:3,dashed,variable=\y,black!60!blue]  plot ({3},{\y});
		\draw[scale=1, line width=.22mm, domain=0:1.2,dashed,variable=\x,black!60!green]  plot ({\x},{3.9});
		\draw[scale=1, line width=.22mm, domain=0:3.9,dashed,variable=\y,black!60!green]  plot ({1.2},{\y});
		\node[black!40!red] at (-0.95,6) {\Huge $\bar u_2^{\text{VCG}}$};
		\node[black!60!blue] at (-1,2.9) {\Huge $\bar u_2^{\text{MPCS}}$};
		\node[black!60!green] at (-0.95,4.0) {\Huge $\bar u_2^{\text{LMP}}$};
		\node[black!40!red] at (6,-.6) {\Huge $\bar u_1^{\text{VCG}}$};
		\node[black!60!blue] at (3,-.6) {\Huge $\bar u_1^{\text{MPCS}}$};
		\node[black!60!green] at (1.2,-.6) {\Huge $\bar u_1^{\text{LMP}}$};
		\node[black] at (-.75,-.5) {\Huge $(0,0)$};
		\draw[->,line width=.22mm] (-1.7,1.1) -- (2.1,1.1) ;
		\end{tikzpicture}
		\vspace{-.1cm}
		\caption{Illustration of the revealed utilities under different mechanisms}\label{fig:core_picture}
		\vspace{-.1cm}
	\end{center}
\end{figure}
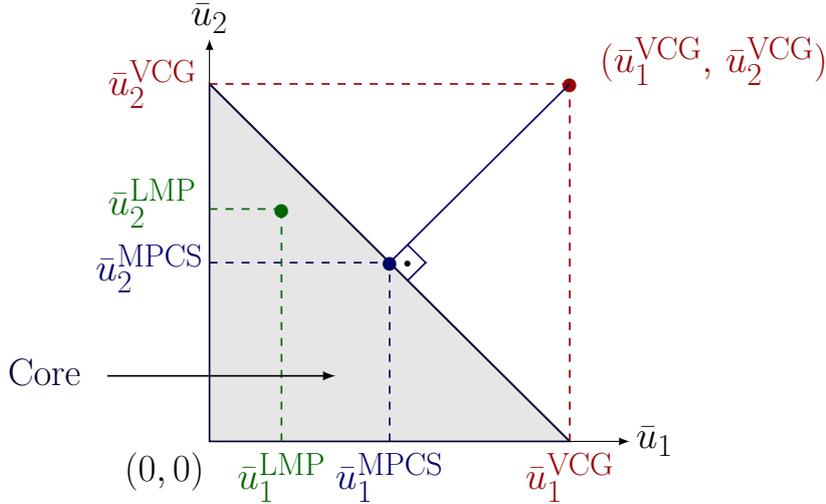

The following theorem shows that the MPCS mechanism approximates the DSIC property.

\begin{theorem}\label{thm:mpcsmax}
	The MPCS mechanism minimizes the sum of maximum additional profits of each bidder by a unilateral deviation from a truthful bid profile $\CC$, that is, the sum of the maximum values of $u_l(\hat \BB_l\cup \CC_{-l})-u_l({\CC})$ for a nontruthful bid $\hat \BB_l$ for all bidders $l\in L$, among all core-selecting mechanisms.
\end{theorem}
The proof is relegated to the appendix in~\cref{app:mpcsmax}.
Under the MPCS mechanism, the total incentives to deviate from truthful bidding are minimal, and hence it is approximately DSIC. This mechanism extends the previous proposals from the multi-item forward auctions computing minimum revenue core in \cite{day2007fair}, and the VCG-nearest core in \cite{day2012quadratic}.\footnote{We kindly refer to \cite{erdil2010new,bunz2018designing} for discussions and analyses on different core outcome choices in multi-item auctions.} Notice that when the VCG utilities lie in the core, they constitute the optimizer to the problem \eqref{eq:mpcs}. This follows since $\bar u_l^{\text{VCG}}({{\BB}})=\max\{\bar u_l\ \rvert\ \bar u\in Core({{\BB}})\}$. Hence, for such instances, the MPCS mechanism is equivalent to the VCG mechanism. The LMP mechanism does not have this~property.

The MPCS mechanism does not rely on the price-taking assumption to approximate DSIC. In return, it generally yields nonlinear and bidder-dependent payments. Nonlinearity might be regarded as a big shift for some existing markets~\cite{bose2019some}. Moreover, the bidders could find it hard to accept bidder-dependency, or this property might even be precluded by law~\cite{cramton2003electricity}. Nevertheless, the MPCS mechanism can still provide an elegant economic rationale when meaningful linear prices do not exist. The coalition-proofness property of the MPCS mechanism discourages bidders from entering the market with multiple identities to try and exploit the bidder-dependency.

\subsection{Market design considerations in exchanges}\label{sec:444}

In this section, we extend our model to exchange markets. 
We assume that each bidder~$l$ has a private true cost~function $c_l: \X_l \rightarrow \mathbb R$, where $0\in \X_l\subseteq\R^t$ and $c_l(0)=0$. Furthermore, the bid function is denoted by $b_l:\hat \X_l \rightarrow \mathbb R$, where $0\in \hat \X_l\subseteq\R^t$  and $b_l(0)=0$.  Note that the domains of these functions are now relaxed to $\R^t$. As a remark, an exchange market is more general than a two-sided market since these bid functions can also represent a bidder interested in buying and selling different types of supplies, simultaneously.
The remaining definitions for the one-sided auctions naturally extend to exchanges. 
Moreover, results on IR, DSIC, coalition-proofness, and competitive equilibrium further hold in exchanges when we relax the power supply domains from $\R_+^t$ to $\R^t$.
For instance, the core is again defined by~\eqref{eq:coref}. Then, the proof of~\cref{thm:no_collusions_bocs} applies to exchanges, proving the connection between coalition-proofness and core-selecting.\footnote{In the combinatorial exchange literature, the core is usually defined by intersecting \eqref{eq:coref} with $\bar u_0=0$~\cite{hoffman2010practical,day2013division,milgrom2007package}. However, this new core is in general empty. Furthermore, such core definitions implicitly assume that a transaction can occur even without the involvement of the central operator. This is not true for the electricity market problems since the dispatch has to be secure via the transmission network.}

In addition to the properties we studied so far, an exchange requires that the operator obtains a revenue, adequate to cover its total payment to balance its budget. In one-sided markets, this might be less of an issue since the demand-side is assumed to be inelastic to the price changes. We say that a mechanism is \textit{budget-balanced} if the operator has a nonnegative utility under any bid profile~$\BB$, $u_0(\BB)\geq0$. We say that it is \textit{strongly budget-balanced} if this utility is exactly zero, $u_0(\BB)=0$. 
Under DC-OPF exchange problems, the LMP mechanism is budget-balanced~\cite[Fact~4]{wu1996folk}. 
On the other hand, the VCG mechanism is not always budget-balanced.\footnote{In the appendix in \cref{app:deficvcg}, we characterize the instances in which the VCG mechanism has a deficit. In~\cref{app:coalpvcg}, we also prove that the VCG mechanism cannot be guaranteed to achieve coalition-proofness in exchanges.} This also follows from the Myerson-Satterthwaite impossibility theorem, which shows that no exchange can always solve for the optimal allocation under the submitted bids, and attain DSIC, budget-balance, and IR simultaneously~\cite{myerson1983efficient} (even in the more general setting of Bayesian implementation, we kindly refer to the extensive discussions in \cite[\S 3.2]{parkes2002achieving}). Fortunately, under the core selecting mechanisms, we can guarantee budget-balance.

\begin{theorem}\label{thm:corewb}
	Any core-selecting mechanism is~budget-balanced, that is, $\bar u_0\geq 0$.
\end{theorem}
The proof is relegated to the appendix in~\cref{app:thmcore}, and it is a straightforward manipulation of the revealed core constraints.
It follows that the MPCS mechanism is budget-balanced in addition to the properties discussed in~\cref{sec:3b1}. This also provides an alternative proof to the budget-balance of the LMP mechanism. 

Note that the LMP mechanism provides methods to reallocate its budget surplus through financial transmission and flowgate congestion rights~\cite{wu1996folk,hogan1992contract}. These rights are important tools to provide market signals to incentivize investment in transmission capacity. In order to compensate the owners of these rights in the MPCS mechanism, a plausible solution is to include them as an additional constraint to \eqref{eq:mpcs}, that is, $\sum_{l \in W } \bar u_l\leq -J(\mathcal{B})-\Delta_{\text{r}}$, where $\Delta_{\text{r}}\geq 0$ is the total rights to be distributed. Defining ways to incorporate these rights, and providing the correct investment signals for transmission capacity expansion, is part of the future research directions.

\section{Appendix}
\subsection{Proof of~\cref{thm:no_collusions_bocs}}\label{app:Gcore}
First, we need the following lemma.
\begin{lemma}\label{lem:implicore}
	Let $\bar u\in\R\times\R^{\rvert L\rvert}_+$ be a revealed utility allocation in $Core(\mathcal{B})$. Then, for every set of bidders $K\subseteq L$ we have $\sum_{l\in K}\bar{u}_l({\mathcal{B}})\leq J({\mathcal{B}}_{-K} ) - J({\mathcal{B}})$.
\end{lemma}
\begin{proof} Since $\bar u_0=-J(\mathcal{B})-\sum_{l\in L}\bar u_l$, we reorganize the inequality constraint as follows $-J(\mathcal{B})-\sum_{l\in L\setminus S}\bar u_l\geq-J(\mathcal{B}_S),\, \forall S \subseteq L$.
	Setting $K=L\setminus S$ yields the statement. 
\end{proof}

Next, we prove that core-selecting mechanisms are coalition-proof.
\begin{proof}
	(i) Let $K $ be a set of colluders who would lose the auction when bidding their true values $\mathcal C_l=c_l$, when bidding ${\mathcal{B}}_l= b_l $ they become winners, that is, they are all allocated a positive quantity. We define $\hat{\mathcal{C}} = C_{K}\cup\mathcal{B}_{-K}$ and ${\mathcal{B}}={\mathcal{B}}_{ K}\cup \mathcal{B}_{-K}$ where $\mathcal{B}_{-K}=\{b_l\}_{l\in L\setminus K}$ denotes the bidding profile of the remaining bidders. As a remark, the profile $\mathcal{B}_{-K}$ is not necessarily a truthful or a strategic profile. We denote the \text{utility} that each bidder $l$ receives as $u_l$. The total utility that colluders receive under ${\mathcal{B}}$ is
	\begin{align*} 
	\sum_{l\in K}{u}_l({\mathcal{B}}) &= \sum_{l\in K}\bar{u}_l({\mathcal{B}})+b_l(x_l^*({\mathcal{B}}))-{c}_l(x_l^*({\mathcal{B}}))\\
	&\hspace{-0.05cm}\leq J({\mathcal{B}}_{-K} ) - J({\mathcal{B}})+\sum_{l\in K} b_l(x_l^*({\mathcal{B}}))-{c}_l(x_l^*({\mathcal{B}}))\\ 
	&\hspace{-0.05cm} = J(\hat{\mathcal{C}}) - \Big[\sum_{l\in L} b_l(x_l^*({\mathcal{B}})) + d(x^*({\mathcal{B}}),y^*({\mathcal{B}})) -\sum_{l\in K} b_l(x_l^*({\mathcal{B}}))+\sum_{l\in K}{c}_l(x_l^*({\mathcal{B}}))\Big] \\
	&\hspace{-0.05cm}  = J(\hat{\mathcal{C}}) - \Big[\sum_{l\in K} {c}_l(x_l^*({\mathcal{B}})) + \sum_{l\in L\setminus K} {b}_l(x_l^*({\mathcal{B}})) +d(x^*({\mathcal{B}}),y^*({\mathcal{B}}))\Big] \\ &\hspace{-0.05cm}\leq 0 = \sum_{l\in K}{u}_l({\hat{\mathcal{C}}}).
	\end{align*}
	The first equality follows from the core-selecting payment rule, where $\bar{u}({\mathcal{B}})$ is the revealed utility allocation. {The first inequality follows from~\cref{lem:implicore}.} The second equality comes from the fact that the set $K$ originally was a group of losers, {so $J(\mathcal{B}_{-K})=J(\hat{\mathcal{C}})$}. 
	After substituting these terms, we see that  the term in brackets is the cost $\bar{J}$ of $\hat{\mathcal{C}}$ but evaluated at a feasible suboptimal allocation $(x^*(\mathcal{B}),y^*(\mathcal{B}))$. Then, $\sum_{l\in K}{u}_l({\mathcal{B}}) $ is upper bounded by $0$ which is the total utility that the colluders would receive bidding truthfully. 
	
	As a result, there is at least one colluder not facing any benefit in collusion. Moreover, they cannot increase their collective utility by a joint deviation. Hence, collusion is not profitable for the losing bidders.
	
	(ii) Define $\mathcal{C}=\mathcal{C}_{-l}\cup\mathcal{C}_{l}$, where $\mathcal{C}_{-l}$ denotes the bidding profile of the remaining bidders. The profile $\mathcal{C}_{-l}$ is not necessarily a truthful profile. Shill bids of bidder~$l$ are given by $\mathcal{B}_{S}=\{b_k\}_{k\in S}$.~We define a merged bid $\tilde{\mathcal{B}}_l$ as $$\tilde{b}_l(x_l)=\min_{x_k\in\R_+^t,\,\forall k}\, \sum_{k\in S}b_k(x_k)\ \mathrm{s.t. }\sum_{k\in S}x_k=x_l.$$ We then define ${\tilde{\mathcal{B}}}=\mathcal{C}_{-l}\cup{\tilde{\mathcal{B}}}_{l}$.
	The total utility obtained from shill bidding under ${\mathcal{B}}=\mathcal{C}_{-l}\cup{\mathcal{B}}_{S}$, $\sum_{k\in S}{u}_k({\mathcal{B}})$, is given by
	\begin{align*} &= \sum_{k\in S}[\bar{u}_k({\mathcal{B}})+b_k(x_k^*({\mathcal{B}}))]-{c}_l(\sum_{k\in S}x_k^*({\mathcal{B}}))\\
	&\leq [J({\mathcal{B}}_{-S} ) - J({\mathcal{B}})]+\sum_{k\in S} b_k(x_k^*({\mathcal{B}}))-{c}_l(\sum_{k\in S}x_k^*({\mathcal{B}}))\\ 
	&= [J({\mathcal{C}}_{-l} ) - J({\tilde{\mathcal{B}}})]+\tilde{b}_l(\sum_{k\in S}x_k^*({\mathcal{B}}))-{c}_l(\sum_{k\in S}x_k^*({\mathcal{B}}))\\ 
	&=  {u}_l^{\text{VCG}}({\tilde{\mathcal{B}}}) \\ 
	&\leq  {u}_l^{\text{VCG}}({\mathcal{C}})
	\end{align*}
	The first inequality follows from the core-selecting payment rule and~\cref{lem:implicore}. The second equality holds since we have $J({\tilde{\mathcal{B}}})=J({{\mathcal{B}}})$. This follows from the definition of the merged bid and the following implication. Since the goods of the same type are fungible for the central operator, the functions $g$, $h$ and~$d$ in fact depend on $\sum_{l\in L} x_l$. 
	The third equality follows from the definition of the VCG utility. The second inequality is the dominant-strategy incentive-compatibility of the VCG mechanism. Therefore, the total utility that $l$ receives from shill bidding is upper bounded by the utility that $l$ would receive by bidding truthfully as a single bidder in a VCG auction. Making use of shills, hence, is not profitable with respect to the VCG utilities. 
\end{proof}
\subsection{Proof of~\cref{lem:effce}}\label{compeff}
		We show that if an allocation~$x^*$ and price functions $\{\psi_l\}_{l\in L}$ constitute a competitive equilibrium, then $x^*$ is the optimal solution to the optimization problem defined by $J(\CC)$ in \eqref{eq:main_model}. By the first condition in \eqref{eq:firstcondition}, we have
		\begin{equation}\label{eq:cond1trick}
		\sum_{l\in L}\psi_l(x_l^*)-c_l(x_l^*)\geq \sum_{l\in L} \psi_l(x_l)-c_l(x_l),\,\forall x\in \X,
		\end{equation}
		where $\X=\prod_{l\in L} \X_l$. Define $(x^*,y^*)$ as the optimal solution pair to the optimization problem in \eqref{eq:secondcondition}, and $(x^*(\CC),y^*(\CC))$  as the optimal solution pair to the problem defined by $J(\CC)$. 
		Then, we obtain the following
		\begin{align*}
		\sum\limits_{l\in L} c_l(x_l^*) + d(x^*,y^*)& \leq  \sum_{l\in L}\psi_l(x_l^*) - \psi_l(x_l^*(\CC)) + c_l(x_l^*(\CC))+ d(x^*,y^*)  \\
		& =  J(\CC) + \sum_{l\in L}\psi_l(x_l^*) + d(x^*,y^*)\\&\hspace{1cm}- \big(\sum_{l\in L} \psi_l(x_l^*(\CC)) + d(x^*(\CC),y^*(\CC))\big) \leq J(\CC).
		\end{align*}
		The first inequality follows from~\eqref{eq:cond1trick}. We then obtain the equality by adding and subtracting the term $d(x^*(\CC),y^*(\CC))$, and substituting $J(\CC)=\sum_{l\in L}c_l(x_l^*(\CC))+d(x^*(\CC),y^*(\CC))$. The second inequality follows since $(x^*(\CC),y^*(\CC))$ is a suboptimal feasible solution to the optimization problem in \eqref{eq:secondcondition}. Note that $(x^*,y^*)$ is originally feasible for the problem defined by $J(\CC)$, otherwise, it would not satisfy the constraints in both \eqref{eq:firstcondition} and~\eqref{eq:secondcondition}. Hence, from the last inequality, it follows that $(x^*,y^*)$ is an optimal solution pair to the optimization problem defined by $J(\CC)$. Since previously we assumed that the optimal solution is unique according to some tie-breaking rule, we obtain the desired result $x^*=x^*(\CC)$.\QEDA
		
\subsection{Proof of~\cref{thm:lmpiscs}}\label{app:lmpiscs}
	We generalize the arguments from~\cite{bikhchandani2002package,parkes2002indirect} that characterize competitive equilibria of multi-item auction problems with a simple supply-demand balancing equality constraint to continuous goods, second stage cost, and general nonlinear constraints. Our proof (in the ``$\impliedby$" direction) is different from these works since it does not rely on linear-programming duality and weak duality arguments that are available to the multi-item setting with simple constraints.
	Invoking~\cref{lem:effce}, for both directions of the proof we restrict our attention to the optimal allocation under truthful bids.
	
	($\impliedby$) For the market \eqref{eq:main_model}, we first prove that if a mechanism ensures the existence of a competitive equilibrium, then it is a core-selecting mechanism. To do so, we show that the revealed utilities lie in the revealed core under any bid profile. 
	
	Given the bid profile $\BB=\{b_l\}_{l\in L}$, allocation $x^*(\BB)$ and price functions $\{\psi_l\}_{l\in L}$, we have;
	\begin{align}
	x^*_l(\BB) &\in \argmax_{x_l\in \hat \X_l}\, \psi_l(x_l)-b_l(x_l), \forall l\in L, \label{eq:feasubtot1} \\
	x^*(\BB) &\in \argmin_{x\in \R_+^{t|L|}} \left\{\min_{\substack{y:\, h(x,y)= 0\\ g(x,y)\leq 0 }}\, \sum\limits_{l\in L} \psi_l(x_l) + d(x,y) \right\}. \label{eq:feasubtot}
	\end{align}
	These conditions must hold because the mechanism does not know the true costs $\CC$, and it has to ensure the existence of an efficient competitive equilibrium in case the true costs are given by $\BB=\{b_l\}_{l\in L}$. Notice that the price functions depend on the bid profile, $\psi_l(x_l)=\psi_l(x_l;\BB)$, and we drop this dependence for the sake of simplicity in notation. Using $\psi_l(x^*_l(\BB))=p_l(\BB)$, the revealed utilities are defined by  $\bar u_l(\BB)=\psi_l(x^*_l(\BB))-b_l(x^*_l(\BB))$, and $\bar u_0(\BB)=-\sum_{l\in L}\psi_l(x^*_l(\BB))- d(x^*(\BB),y^*(\BB))$, where $y^*(\BB)$ is the optimal solution to \eqref{eq:feasubtot}.
	
	Next, we show that $\bar u(\BB)\in Core(\BB)$, which would conclude that the mechanism is core-selecting. First, observe that the individual-rationality constraints are satisfied; $\bar u_l(\BB)\geq 0$ since $0\in \hat \X_l,\, \psi_l(0)=0$, and $b_l(0)=0$. Second, we have the equality constraints in~\eqref{eq:coref}: $\sum_{l\in L} \bar u_l(\BB)+\bar u_0(\BB) = -\sum_{l\in L}b_l(x^*_l(\BB)) - d(x^*(\BB),y^*(\BB))=-J(\BB)$. Third, we show that the inequality constraints in~\eqref{eq:coref} hold, that is, \begin{equation}\label{eq:stcorein}
	-\bar u_0(\BB)\leq J(\BB_S)+\sum\limits_{l\in S}\bar u_l(\BB),\, \forall S \subset L.
	\end{equation}
	Define the following restricted problem for any subset $S\subset L$,
	\begin{equation}\label{eq:restrict}
	\begin{split}
	\mu_0(S) =  -&\min_{\substack{x\in\hat \X,\,y\\ x_{-S}=0}}\ \sum\limits_{l\in S} \psi_l(x_l) + d(x,y)\\
	&\ \ \mathrm{  s.t. } \ \, h(x,y)= 0,\, g(x,y)\leq 0,
	\end{split}
	\end{equation}
	where the optimal allocation is denoted by $x^*(S)$. Because $x^*(S)$ is a feasible solution to \eqref{eq:feasubtot}, we obtain $-\bar u_0(\BB) \leq -\mu_0(S)$. 
	We then let  $x^*(\BB_S)$ be the optimal solution to $J(\BB_S)$. This solution is a suboptimal feasible solution to \eqref{eq:restrict} and hence $-\mu_0(S)\leq\sum_{l\in S} \psi_l(x^*_l(\BB_S))+d(x^*(\BB_S),y^*(\BB_S)).$ 
	Then, it suffices to show that 
	\begin{equation}\label{eq:deseq}
	\sum_{l\in S} \psi_l(x^*_l(\BB_S))+d(x^*(\BB_S),y^*(\BB_S))\leq J(\BB_S)+\sum_{l\in S}\bar u_l(\BB),\,\forall S\subset L,\end{equation}since this would imply the inequality in \eqref{eq:stcorein}. 
	Via the condition in~\eqref{eq:feasubtot1}, we have 
	\begin{equation*}
	\psi_l(x^*_l(\BB_S))-b_l(x^*_l(\BB_S))\leq \psi_l(x^*_l(\BB))-b_l(x^*_l(\BB))=\bar u_l(\BB).
	\end{equation*}
	Summing the inequality above over all $l\in S$, we obtain 
	$$\sum_{l\in S} \psi_l(x^*_l(\BB_S))-b_l(x^*_l(\BB_S))\leq\sum_{l\in S}\bar u_l(\BB).$$
	By adding $d(x^*(\BB_S),y^*(\BB_S))$ on both sides and reorganizing, the above inequality yields \eqref{eq:deseq}.	
	Consequently, we have $-\bar u_0(\BB)\leq J(\BB_S)+\sum_{l\in S}\bar u_l(\BB)$, for any $S\subset L$. Hence, the revealed utilities lie in $Core(\BB)$. 
	
	($\implies$) We now prove that any core-selecting mechanism ensures the existence of an efficient competitive equilibrium. In other words, we show that, for the truthful optimal allocation $x^*(\CC)\in \X$, there exists a set of price functions $\psi_l:\R_+^t\rightarrow\R$,~$\forall l$ such that the conditions in~\cref{def:compeq} are satisfied, and $\psi_l(x^*_l(\CC))=p_l(\CC)$.  Consider the utility allocation $u\in Core(\CC)$ of a core-selecting mechanism under truthful bidding. Define the price functions $\{\psi_l\}_{l\in L}$ as follows
	\begin{equation*}
	\psi_l(x) = 
	\begin{cases}
	0 & x=0,\\
	c_l(x)+u_l& x\in \X_l \setminus \{0\},\\
	\infty& \text{otherwise}. \\
	\end{cases}
	\end{equation*}
	For these price functions, the first condition in~\eqref{eq:firstcondition} holds by construction. To show that, we study two possible cases. If $x^*_l(\CC)$ is nonzero, then $x^*_l(\CC)\in \argmax_{x_l\in \X_l}\, \psi_l(x_l)-c(x_l)$ since $u_l\geq 0$. On the other hand, if $x^*_l(\CC)=0$, then $u_l=0$ and $0\in \argmax_{x_l\in \X_l}\, 0$.

	We prove the second condition in \eqref{eq:secondcondition} by contradiction. Assume there exists $x,\,y$ such that $h(x,y)=0$, $g(x,y)\leq0$,~and{\medmuskip=2mu\thinmuskip=2mu\thickmuskip=2mu\begin{equation}\label{eq:contrastate}
		\sum_{l\in L} \psi_l(x_l^*(\CC)) + d(x^*(\CC),y^*(\CC))>\sum_{l\in L} \psi_l(x_l) +d(x,y).
		\end{equation}}Define a subset $S\subseteq L$ such that $x_l=0$ for all $l\in L\setminus S$. Observe that if $x_l>0,\,\forall l$, then $S=L$. 
	Then, the core implies
\begin{equation}\label{eq:coretr}\begin{split}
		-u_0&=J(\CC)+\sum\limits_{l\in L} u_l\\&= \sum_{l\in L} \psi_l(x_l^*(\CC)) + d(x^*(\CC),y^*(\CC)).
		\end{split}
		\end{equation}The second equality follows from the definition of the price functions. Using the second equality in~\eqref{eq:coretr}, the inequality in~\eqref{eq:contrastate} is equivalent to
	\begin{equation*}
	J(\CC)+\sum\limits_{l\in L} u_l>\sum_{l\in S} \psi_l(x_l) + d(x,y).
	\end{equation*}
	By the definition of $\psi_l$, we obtain 
	\begin{equation*}
	J(\CC)+\sum\limits_{l\in L\setminus S} u_l>\sum_{l\in S} c_l(x_l) + d(x,y) \geq J(\CC_S),
	\end{equation*}
	where the last inequality is from the feasible suboptimality of $(x,y)$ for the problem defined by $J(\CC_S)$. This is because $h(x,y)=0$, $g(x,y)\leq0$, and $x\in \X$ (otherwise, the price functions are unbounded). 
	
	Using the first equality in \eqref{eq:coretr}, we have $u_0+\sum_{l\in S}u_l<-J(\CC_S)$ for $S\subseteq L$. This contradicts $u\in Core(\CC)$, and thus the mechanism cannot be a core-selecting mechanism. We conclude that $x^*(C)$ is optimal for the central operator given the price functions. As a result, $x^*(\CC)$ and $\{\psi_l\}_{l\in L}$ constitute a competitive equilibrium. Finally, from $p_l(\CC)=c_l(x^*_l(\CC))+u_l$, we obtain $\psi_l(x^*_l(\CC))=p_l(\CC)$ for each bidder $l\in L$. This concludes the proof.\QEDA
	
	\subsection{Comparison of the LMP and the VCG payments}\label{app:lmpvcg}
\begin{proposition}\label{prop:upb}
	Given any bid profile $\BB$, for every bidder~$l$ the payment under the LMP mechanism is upper bounded by the payment under the VCG mechanism.
\end{proposition}
\begin{proof}
	A similar result was proven in~\cite{xu2017efficient}, using convex analysis in the context of DC-OPF markets. We provide a simple and more general proof applicable to any setting where the LMP mechanism ensures the existence of a competitive equilibrium.
	The proof is an application of \cite[Theorem~5]{ausubel2002ascending} that compares the utilities of iterative ascending auctions with that of the VCG mechanism. Since the LMP utilities lie in the revealed core, it suffices to show that the VCG revealed utilities are greater than any other revealed utility in the revealed core. 
	
	Observe that the VCG payment of bidder $l$ is given by $p_l^{\text{VCG}}=b_l(x^*_l(\BB))+(J(\BB_{-l})-J(\BB)),$  whereas the revealed VCG utility is $\bar u_l^{\text{VCG}}=J(\BB_{-l})-J(\BB)$.
	Assume there exists a revealed utility allocation $\tilde{u}\in Core(\BB)$ where $\tilde{u}_l> \bar u_l^{\text{VCG}}$. These utilities are blocked by the coalition $L_{-l}$; $-J(\BB_{-l})> -J(\BB)-\tilde{u}_l=\tilde{u}_0+\sum_{k\in L_{-l}}\tilde{u}_k,$ where the equality follows from the definition of the revealed core. This contradicts that $\tilde{u}\in Core(\BB)$. We conclude that the revealed core utilities are upper bounded by the ones under the VCG mechanism. Moreover, it is not hard to show that this upper bound is tight for some core-selecting mechanism since $\bar u_l^{\text{VCG}}=\max\left\{\bar u_l\ \rvert\ \bar u\in Core(\BB)\right\}$~\cite[Theorem~5]{ausubel2002ascending}. We obtain the proposition.
\end{proof}

Proof of~\cref{prop:upb} simplifies the arguments in~\cite{xu2017efficient} greatly. 
\subsection{Proof of~\cref{lem:lie}}\label{app:lemlie}

	Assume there exists a bid $\hat b_l:\hat \X_l \rightarrow \mathbb R_+$ such that 
	\begin{equation*}
	\begin{split}
	\big[\hat{b}_l(x_l^*(\hat{{\BB}}_l\cup \BB_{-l})) -& c_l(x_l^*(\hat{{\BB}}_l\cup \BB_{-l})) + \bar u_l(\hat{{\BB}}_l\cup \BB_{-l})\big]\\ &- u_l({\CC_l\cup \BB_{-l}}) > \big[J(\BB_{-l})-J(\CC_l\cup \BB_{-l})\big]-u_l({\CC_l\cup \BB_{-l}}),
	\end{split}
	\end{equation*}
	where $\hat{\BB}_l=\{\hat b_l\}$, $x^*(\hat{{\BB}}_l\cup \BB_{-l})$ is the optimal allocation of the market problem corresponding to $J(\hat{{\BB}}_l\cup \BB_{-l})$, and $\bar u(\BB)$ denote the revealed utilities for any bid profile~$\BB$ under core-selecting mechanism.
	The inequality above is equivalent to the existence of a deviation that is more profitable than the given upper bound. 
	Notice that the following holds, $$\bar u_l({\hat{{\BB}}_l\cup \BB_{-l}})\leq \bar u^{\text{VCG}}_l(\hat{{\BB}}_l\cup \BB_{-l})=J(\BB_{-l})-J(\hat{{\BB}}_l\cup \BB_{-l}),$$ since $\bar u_l^{\text{VCG}}(\hat{{\BB}}_l\cup \BB_{-l})=\max\{\bar u_l\ \rvert\ \bar u\in Core(\hat{{\BB}}_l\cup \BB_{-l})\}$~\cite[Theorem~5]{ausubel2002ascending}.
	Combining the inequalities above, we have
		\begin{equation*}
	\begin{split}
	\hat{b}_l(x_l^*(\hat{{\BB}}_l\cup \BB_{-l})) - c_l(x_l^*(\hat{{\BB}}_l\cup \BB_{-l}))& + J(\BB_{-l})-J(\hat{{\BB}}_l\cup \BB_{-l}) \\
	&> J(\BB_{-l})-J(\CC_l\cup \BB_{-l}).
	\end{split}
	\end{equation*}Observe that the first term is the VCG utility under a non-truthful bid, whereas the second term is the VCG utility under a truthful bid. The strict inequality above contradicts the DSIC property of the VCG mechanism. We conclude that $ u^{\text{VCG}}_l(\CC_l\cup \BB_{-l})-u_l({\CC_l\cup \BB_{-l}})$ is an upper bound on the additional profit obtained from a unilateral deviation. 
	
	We remark that, for any bidder, there exists a bid achieving the profit in the upper bound. Define $\epsilon$ to be a small positive number, which is required to avoid ties. It is straightforward to show that the following bid achieves exactly the truthful VCG utility, and hence the exact upper bound in $ u^{\text{VCG}}_l(\CC_l\cup \BB_{-l})- u_l({\CC_l\cup \BB_{-l}})$:
	\begin{equation*}
	\hat{b}_l(x) = 	\begin{cases}
	0 & x=0,\\
	c_l(x)+ u^{\text{VCG}}_l(\CC_l\cup \BB_{-l})-\epsilon& x\in \X_l \setminus 0,\\
	\infty& \text{otherwise}. \\ \end{cases}
	\end{equation*}
	Since the market solves for the optimal allocation in \eqref{eq:main_model}, if bidder~$l$ is allocated a positive quantity while bidding truthfully, then this bidder is also allocated a positive quantity while bidding $\hat{b}_l$. Moreover, under any core-selecting mechanism, we have $\epsilon\geq\bar{u}_l^{\text{VCG}}(\hat{\BB}_l\cup \BB_{-l})\geq\bar{u}_l(\hat{\BB}_l\cup \BB_{-l})\geq0$ where $\hat{\BB}_l=\{\hat b_l\}$. As a result, by bidding $\hat b_l$, the bidder~$l$ obtains its truthful VCG utility, ${u}_l(\hat{\BB}_l\cup \BB_{-l})= u^{\text{VCG}}_l(\CC_l\cup \BB_{-l})-\epsilon+\bar{u}_l(\hat{\BB}_l\cup \BB_{-l})$, and achieves exactly the upper bound in the lemma as $\epsilon\rightarrow 0^+$. \QEDA
	
\subsection{Computing the MPCS payments}\label{appcomps}

The MPCS payments are computationally difficult for auctions involving many bidders, because one needs to solve the auction problem \eqref{eq:main_model} for $2^{|L|}$ different subsets to define the revealed core constraints in~\eqref{eq:coref}. Furthermore, problem \eqref{eq:main_model} can be NP-hard in some cases. Invoking~\cref{lem:lemma_core}, we can reduce the number of constraints to $2^{\rvert W\rvert}$, which grows exponentially only in the number of winners. We call this approach with the reduced number of core constraints \textit{the direct approach}. Note that we have $\bar u \in Core(\mathcal B)$, if and only if $\bar u_0= -J(\mathcal B)-\sum_{l\in L}\bar u_l$ and $\bar u_{-0}$ lies in
\begin{equation}\label{eq:simp_core}
\mathbb{K}(\mathcal B)=\Big\{\bar u_{-0}\in\R^{\rvert L\rvert}_+\,|\,\sum_{l \in K } \bar u_l\leq J(\mathcal{B}_{-K}) - J(\mathcal{B}), \forall K\subseteq W\Big\}.
\end{equation}
Unfortunately, this approach may still not be a computationally feasible one since there can be many winners to~\eqref{eq:main_model}. Therefore, we study iterative approaches where core constraints are generated on demand.

As suggested in \cite{day2007fair,bunz2015faster}, the state of the art approach for calculating a revealed core outcome is to use constraint generation and in practice, this algorithm requires the generation of only several revealed core constraints. The method was initially used in the '50s in order to solve linear programs that have too many constraints \cite{dantzig1954solution}. Instead of directly solving the large problem, one solves a primary problem with only a subset of its original constraints. From this primary solution, one can formulate a secondary problem that adds another constraint to the first step. The algorithm iterates between these two problems and converges to the optimal solution of the large problem. 
Next, we formulate the core constraint generation algorithm for the electricity market problem in \eqref{eq:main_model}, similar to the way it was previously utilized in forward multi-item auctions in~\cite{day2007fair}.

Bidders' revealed utilities at the first step of our algorithm are given by $\bar u^0_{-0}=\bar u^{\text{VCG}}_{-0}(\mathcal  B)$, where $\bar u^{\text{VCG}}_{-0}$ is the revealed utilities of the VCG mechanism. Convergence of the algorithm does not require this choice. However, this choice is intuitive since $\bar u^{\text{VCG}}_{-0}(\mathcal  B)$ is the solution to \eqref{eq:mpcs} if the VCG outcome lies in the core. 

As an iterative method, at each step $k$, we find the blocking coalition that has the largest violation for the revealed utility allocation~$\bar u^k_{-0}$.\footnote{We remark that the coalition $C$ is a blocking coalition if $J(\mathcal{B}_C)+\sum_{l\in C}\bar u_l<-\bar u_0$, see the constraints in \eqref{eq:coref}. } If a blocking coalition exists, the coalition with the largest violation for the revealed utility allocation~$\bar{u}^k_{-0}$ is given by \begin{equation}\label{eq:simpleform}C^k=\argmin_{C\subseteq L}\ J(\mathcal{B}_C)+\sum_{l \in C } \bar u_l^k.\end{equation} This follows from the constraints in $Core(\mathcal B)$ in \eqref{eq:coref}. 

Let $W$ be the set of winners. It is straightforward to see that the existence of this~blocking~coalition~$C^k$ is equivalent to the violation of the constraint $\sum_{l \in W\setminus C^k} \bar u_l^k\leq J(\mathcal  B_{L\setminus\{W\setminus C^k\}})- J(\mathcal  B)$ from the set in \eqref{eq:simp_core}. This follows from the equivalent characterization of the core in~\cref{lem:lemma_core}. We call the set $W\setminus C^k$ the blocked winners, and the problem \eqref{eq:simpleform} generates the revealed core constraint on the revealed utilities of this set of bidders.


Next, we reformulate the problem in \eqref{eq:simpleform} using inflated bids. Given a revealed utility allocation $\bar u_{-0}^k\in\R^{\rvert L\rvert }_+$, we define the central operator's objective at step~$k$ as 
\begin{equation*}
\bar{J}^k(x,y;\mathcal B)=\sum\limits_{l\in L} b_l^{\bar u_l^k}(x_l) + d(x,y),
\end{equation*}
where the inflated bid $b_l^{\bar u_l^k}(x_l)\in\R_+$ is given by
\begin{equation*}
b_l^{\bar u_l^k}(x_l) = 
\begin{cases}
0 & x_l=0,\\
b_l(x_l)+\bar u_l^k& \text{otherwise}. 
\end{cases}
\end{equation*}
Note that even if the bid $b_l$ is a convex bid curve,  the inflated bid $b_l^{\bar u_l^k}$ is not convex if $\bar u_l^k\neq 0$, because of the discontinuity at $0$. As a remark, the discontinuity at $0$ can be described by binary variables.
Then, the optimization problem \eqref{eq:simpleform} for finding the blocking coalition with the largest violation is reformulated as follows
\begin{equation}\label{eq:main_algorithm_ccg}
\begin{split}
z(\bar u^k_{-0})&=\min_{x\in\hat\X,y}\ \bar{J}^k(x,y;\mathcal B)\ \mathrm{s.t.}\ h(x,y)= 0,\, g(x,y)\leq 0,\\
x^*(\bar  u^k_{-0})&=\argmin_{x\in \hat \X} \left\{\min_{\substack{y:\, h(x,y)= 0\\ g(x,y)\leq 0 }}\,\bar{J}^k(x,y;\mathcal B) \right\}, \\
C^k&=\{l\in L\,\rvert\,  x_l^*(\bar  u^k_{-0})\neq0\},\\
\end{split}
\end{equation}
where $C^k$ is the blocking coalition with the largest violation for the revealed utility allocation~$\bar{u}^k_{-0}$. Notice that the problem \eqref{eq:main_algorithm_ccg} essentially solves the auction where winners' bids are inflated by their revealed utilities from the earlier step. 

After obtaining the blocking coalition and the corresponding central operator cost $z(\bar{u}^k_{-0})$, we solve the following two problems to obtain another candidate for a MPCS revealed utility allocation. First, we take the subset of revealed utilities that are maximizing the total revealed utility of bidders as follows
\begin{equation}	\label{eq:step_of_ccg}
\begin{split} 
\nu^k =&  \max_{\substack{\bar u_{-0}\in\R^{\rvert L\rvert }_+} } \bm 1^\top \bar u_{-0} \\
&\ \ \ \ \mathrm{s.t. } \sum_{l \in W\setminus{C^t} } \bar u_l\leq J(\mathcal B_{L\setminus\{W\setminus{C^t}\}})- J(\mathcal B),\, \forall t\leq k,\\
& \quad\quad\quad\ \bar u_{-0}\leq\bar u_{-0}^{\text{VCG}}(\mathcal  B),\\
\end{split}
\end{equation}
where $J(\mathcal B_{L\setminus\{W\setminus{C^t}\}}) =z(\bar u_{-0}^t)-\sum_{l \in W\cap C^t } \bar u_l^t $. We highlight that this term, $J(\mathcal B_{L\setminus\{W\setminus{C^t}\}})$, does not require any further solution to the reverse auction problem. As a result, it is straightforward to see that the problem \eqref{eq:step_of_ccg} essentially solves the problem \eqref{eq:mpcs} with only a subset of its original constraints and without the tie-breaker and the central operator's revealed utility.

We then implement the tie breaker in the objective of \eqref{eq:mpcs}. We have $\bar{u}_{-0}^{k+1}$, as the solution to the following quadratic program:
\begin{equation}	\label{eq:next_step_of_ccg}
\begin{split} 
\bar{u}_{-0}^{k+1} =&  \argmin_{\substack{\bar u_{-0}\in\R^{\rvert L\rvert }_+ } }\ \norm{\bar u_{-0}-\bar u_{-0}^{\text{VCG}}(\mathcal  B)}_2^2=\sum_{l\in L}(u_l-\bar u_l^{\text{VCG}}(\BB))^2\\
&\  \ \ \mathrm{s.t. }\ \ \sum_{l \in W\setminus{C^t} } \bar u_l\leq J(\mathcal  B_{L\setminus\{W\setminus{C^t}\}})- J(\mathcal  B),\, \forall t\leq k,\\
&\ \ \ \quad\quad\quad \bar u_{-0}\leq\bar u_{-0}^{\text{VCG}}(\mathcal  B),\  \bm 1^\top \bar u_{-0} = \nu^k. \\
\end{split}
\end{equation}

The entire process for determining core-selecting payments is summarized in~\cref{alg:ccg_algorithm}.
\begin{algorithm}[H]
	\caption{Core Constraint Generation (CCG) Algorithm}
	\begin{algorithmic}[1]\label{alg:ccg_algorithm}
		\renewcommand{\algorithmicrequire}{\textbf{Initialize:} Solve the optimization problem \eqref{eq:main_model}.}
		\renewcommand{\algorithmicensure}{  \textbf{Iteration step $k$:}                           }
		\REQUIRE 
		Calculate the VCG utilities and set  $\bar u^0=\bar u^{\text{VCG}}(\mathcal  B)$ and $k=0$. \\
		\STATE Solve the optimization problem \eqref{eq:main_algorithm_ccg}. 
		\WHILE {$z(\bar u^k_{-0})< J(\mathcal B)+\sum_{l \in W} \bar u_l^k$ (or check also $W\setminus{C^k}\neq\emptyset$)}
		\STATE Obtain $\bar u_{-0}^{k+1}$ by solving \eqref{eq:step_of_ccg} and then \eqref{eq:next_step_of_ccg}.
		\STATE Update $k = k+1$.
		\STATE Solve the optimization problem \eqref{eq:main_algorithm_ccg}. 
		\ENDWHILE 
		\RETURN $\bar u_{-0}^{k}$.
	\end{algorithmic}
\end{algorithm}

\cref{alg:ccg_algorithm} converges to the solution of \eqref{eq:mpcs} for the MPCS mechanism. The proof of convergence is straightforward, since the stopping criterion certifies that no revealed core constraints are violated by the solution $\bar u^k$ at that iteration, that is, $$J(\mathcal{B}_S)+\sum_{l\in S}\bar u_l^k\geq-\bar u_0^k=J(\mathcal B)+\sum_{l \in W} \bar u_l^k,$$ for all $S\subset L$, see \cite[Theorem~4.2]{day2007fair}. We note that this algorithm may still require the generation of all possible revealed core constraints, which is equivalent to solving the problem \eqref{eq:main_algorithm_ccg} $2^{\rvert W \rvert}$ times. In practice, even when there are many winners, the algorithm requires the generation of only several core constraints.

\subsection{Proof of~\cref{thm:mpcsmax}}\label{app:mpcsmax}

		For the proof, we ignore the second term in the objective of the MPCS mechanism since it is required only for tie-breaking purposes. Theorem considers deviations from the case in which the bidders are revealing their true costs $\CC=\{c_l\}_{l\in L}$. We can reformulate problem \eqref{eq:mpcs} as follows,
	\begin{align}
	\bar u^{\text{MPCS}}(\CC) &= \argmax_{u \in \text{Core}(\CC)}\ \sum_{l\in L} u_l - \sum_{l\in L} u_l^{\text{VCG}}(\CC) \nonumber \\
	&=\argmin_{u \in \text{Core}(\CC)}\ \sum_{l\in L} (u_l^{\text{VCG}}(\CC)-u_l)\ . \label{eq:mpcs3}
	\end{align}
	First equality follows since  $\sum_{l\in L} u_l^{\text{VCG}}(\CC)$ is a constant. Note that $u_l(\CC) = \bar u_l(\CC)$ under truthful bidding.
	Invoking~\cref{lem:lie}, the optimization problem in \eqref{eq:mpcs3} implies that the MPCS mechanism minimizes the sum of additional profits of each bidder by a unilateral deviation from a truthful bid profile $\CC$ among all other core-selecting mechanisms. \QEDA

\subsection{Characterizing deficit under the VCG mechanism}\label{app:deficvcg}
Notice that the Myerson-Satterthwaite impossibility theorem does not rule out the possibility of having realizations of the VCG mechanism that are budget-balanced. 
Next, we extend this observation by showing that in a DC-OPF market with no line limits the VCG mechanism is at most strongly budget-balanced.
\begin{proposition}\label{prop:bbvcg}
	The VCG mechanism never yields a positive utility for the operator in a DC-OPF market with no line limits.
\end{proposition}

We need the following lemma for our proof.
\begin{lemma}\label{lem:templem}
	Assume $J$ is modeled by
	\begin{equation}\label{eq:simplej}
		J(\CC)=\min_{x_l\in \X_l,\,\forall l}\, \sum_{l\in L} c_l(x_l)\ \mathrm{s.t.}\, \sum_{l\in L } x_l=0,
	\end{equation}
	where $c_l,\,\forall l$ are convex increasing and $\X_l,\,\forall l$ are polytopic constraints. Define $q\,\CC$ as the bid profile consisting of $\CC$ replicated $q$ times. Then, for any $q\in\N_+$, we have $qJ(\CC) = J(q\,\CC)$.
\end{lemma}
\begin{proof}
	We prove this by showing that the optimal solution to $J(q\,\CC)$ is given by concatenating the decision variables in the optimal solution of $J(\CC)$ $q$ times. 
	Since the problem defined by $J$ satisfies constraint qualification conditions, the KKT conditions are both necessary and sufficient for the optimality of a solution~\cite{bertsekas1999nonlinear}. We see that the KKT conditions of $J(q\,\CC)$ are satisfied by a primal solution which is the concatenation of the optimal solution of $J(\CC)$ in~\eqref{eq:simplej}~$q$~times, and a dual solution which is the Lagrange multiplier of the equality constraint of $J(\CC)$ in~\eqref{eq:simplej}. This concludes that $qJ(\CC) = J(q\,\CC)$. 
\end{proof}	

Note that this market model includes the DC-OPF markets where the network graph is connected and there are no line limits. 
We are now ready to prove~\cref{prop:bbvcg}.
\begin{proof}(Proof of~\cref{prop:bbvcg})	
We prove this by contradiction. Under the VCG mechanism, assume the operator has a positive utility:
	\begin{equation*}
	0< u_0(\CC) = - J(\CC) - \sum_{l\in W} (J(\CC_{-l})-J(\CC)), 
	\end{equation*}
	where $W\subseteq L$ is the set of bidders whose allocations are not zero. By reorganizing, we obtain
	\begin{equation*}
	|W| J(\CC)>J(\CC_W)+\sum_{l\in W} J(\CC_{-l})\geq J(|W|\,\CC),
	\end{equation*}
	where $|W|\,\CC$ is a bid profile consisting of $\CC$ replicated $|W|$ times.
	The last inequality follows because the allocation of the problems on the left is a suboptimal feasible allocation to the problem on the right. Note further that $|W|J(\CC) = J(|W|\,\CC)$. This follows from~\cref{lem:templem}. We obtain a contradiction $J(|W|\,\CC)< J(|W|\,\CC)$. Hence, the VCG mechanism achieves at most strong budget-balance, and it never yields a positive utility for the operator in this case. 
\end{proof}

As a remark, the LMP mechanism is strongly budget-balanced for the same market under any bid profile~\cite[Fact 5]{wu1996folk}. 
Under the VCG mechanism, it is straightforward to create a two-bidder example of the market in \eqref{eq:simplej} that yields a negative utility for the operator. 

\begin{example}
	Suppose there are two bidders in the market~\eqref{eq:simplej}. The cost function of bidder 1 is given by $c_1(x_1)=x_1,\, 0\leq x_1\leq1$. The cost function of bidder 2 is given by $c_2(x_2)=3x_2,\, -1\leq x_1\leq0$. Under the VCG mechanism, bidder $1$ receives the payment $\$3$ since $p_1^{\text{VCG}} = 1 + (0-(-2))= \$3$. Whereas bidder 2 makes the payment $\$1$ since $p_2^{\text{VCG}} = -3 + (0-(-2))= -\$1$. Hence, the central operator has a $\$2$ deficit.
\end{example}
\subsection{Coalition-proofness of the VCG mechanism in an exchange market}\label{app:coalpvcg}
We remind the reader the definition of supermodularity from~\cref{def:supms}.
\begin{definition}\label{def:sup}
	A function $J$ is \textit{supermodular} if  
	$J(\BB_S)-J(\BB_{S\setminus{l}})\leq J(\BB_R)-J(\BB_{R\setminus{l}})$
	for all coalitions $S\subseteq R\subseteq L$ and for each bidder $l\in S$ under any bid profile. 
\end{definition}

We show that the VCG mechanism cannot be guaranteed to achieve coalition-proofness in an exchange market by showing that the market objective function~$J$ can never be supermodular.
For the following, we assume that if there is at most one bidder in the exchange, no exchange occurs, that is, $J(\BB_l)=0,$ for all $l\in L$ and $J(\emptyset)=0$.
\begin{proposition}\label{pro:cpvcg} 
	In an exchange market, if the function $J$ in~\eqref{eq:main_model} is supermodular, then no exchange occurs, $J(\BB_S)=0,\, \forall S\subseteq L$.
\end{proposition}
\begin{proof}
	Assume that the function $J$ is supermodular. Then, for bidders $i,j,k\in L$, we have
	\begin{equation*}
	J(\BB_{\{i,j\}})-J(\BB_j)\leq J(\BB_{\{i,j,k\}})-J(\BB_{\{j,k\}}).
	\end{equation*}
	This concludes that $J(\BB_{\{j,k\}})\leq 0$, since $J(\BB_j)=0$ and $J(\BB_{\{i,j\}})\geq J(\BB_{\{i,j,k\}})$. Supermodularity further implies that
	\begin{equation*}
	J(\BB_k)-J(\emptyset)\leq J(\BB_{\{j,k\}})-J(\BB_j).
	\end{equation*}
	This yields $J(\BB_{\{j,k\}})\geq 0$ since $J(\emptyset)=0$. Hence, we obtain $J(\BB_{\{j,k\}})= 0$. Moreover, this holds for any bidder pair $j,k$. By using this result, we repeat the steps above to obtain $J(\BB_{\{i,j,k\}})= 0,\, \forall i,j,k$. We conclude that similar steps can be repeated until we obtain $J(\BB_S)=0,$ for every subset $S\subseteq L$.	
	Note that this holds under any bid profile $\BB$. 
\end{proof}

Considering~\cref{thm:no_collusions}, we conclude that the VCG mechanism is not guaranteed to achieve coalition-proofness. We highlight that, to the best of our knowledge, this impossibility result for exchanges is novel. An intuition behind it is that an exchange allows for different kind of manipulations than the ones in a one-sided auction. For instance, a bidder can enter the market both as a buyer and as a seller. Then, it can manipulate the outcome by changing the amount it buys and it sells, knowing that its final allocation is given by their difference.

\subsection{Proof of~\cref{thm:corewb}}\label{app:thmcore}
		Consider the revealed utility $\bar u\in Core(\BB)$ of a core-selecting mechanism. We remind the reader the proof of~\cref{lem:lemma_core}. Using $\bar u_0=-J(\BB)-\sum_{l\in L}\bar u_l$, we can derive an equivalent characterization of the inequalities in the core as follows
		$\sum_{l\in L\setminus S}\bar u_l\leq J(\BB_S)-J(\BB),\, \forall S \subseteq L.$
		Setting $K=L\setminus S$, these inequalities are equivalent to $\sum_{l\in K}\bar{u}_l({\BB})\leq J({\BB}_{-K} ) - J({\BB})$, for every set of bidders $K\subseteq L$. 
		By keeping the most binding constraints, we obtain
		\begin{equation}\label{eq:bind}
		\sum_{l\in K}\bar{u}_l\leq J({\BB}_{-K} ) - J({\BB}),\, \forall K\subseteq W,
		\end{equation}
		where $W\subseteq L$ is the set of bidders whose allocations are not zero.
		Next, we assume that the pay-as-bid mechanism is budget-balanced, $u_0(\BB)=-J(\BB)\geq 0$ under any bid profile. This assumption is satisfied in many exchange markets, for instance, combinatorial exchanges~\cite{hoffman2010practical}, DC-OPF exchanges~\cite{wu1996folk} and two-sided electricity markets~\cite{xu2017efficient}.  
		Using this assumption, the inequality in \eqref{eq:bind} implies that $\sum_{l \in W } \bar u_l\leq J(\mathcal{B}_{-W}) - J(\mathcal{B})\leq- J(\mathcal{B}),$ since $J(\mathcal{B}_{-W})\leq0$. By reorganizing, we have $- J(\mathcal{B})-\sum_{l \in L } \bar u_l\geq 0$, since bidders who receive zero allocation are not paid. Using the equality constraint of the core in \eqref{eq:coref}, we conclude that $\bar u_0=u_0\geq 0$. Hence, the central operator ends up with a nonnegative utility, and the mechanism is budget-balanced. \QEDA
\chapter{Numerical case studies}\label{sec:p1_5}

The goal of this chapter is to compare the effectiveness of the mechanisms we have discussed, based on electricity market examples. We first consider several different DC-OPF problems where the demand is inelastic.
We then consider an AC-OPF problem from \cite{bukhsh2013local}. This problem yields a nonzero duality gap, and hence the Lagrange multipliers cannot be guaranteed to have an economic rationale. As an alternative, we show that the MPCS mechanism coincides with the VCG mechanism, and hence it is both coalition-proof and DSIC. We then study the two-stage Swiss reserve procurement auction from~\cite{abbaspourtorbati2016swiss} which also fails to attain strong duality. For this example, the VCG mechanism is not core-selecting, and hence it is not coalition-proof. Instead, the MPCS mechanism yields coalition-proof outcomes. Finally, we consider a two-sided market with DC-OPF constraints to compare the budget-balance under different mechanisms. We show that we can compute the MPCS mechanism in a reasonable time. All problems are solved on a computer equipped with 32 GB RAM and a 4.0 GHz quad-core Intel i7 processor. 
\section{Four-node three-generator network model}
We consider a dispatch problem with polytopic DC power flow constraints in \cref{fig:three_node}, based on the models considered in \cite{wu1996folk}, we also kindly refer to the constraints in \eqref{mod:A-DA} in \cref{part:2}. Cost curves are quadratic polynomials. All lines have the same susceptance. In \cref{fig:three_node}, line limit from node $i$ to node $j$ is denoted by $C_{i,j}=C_{j,i}\in\R_+$. Under the VCG mechanism, payments and utilities are given in the first column of~\cref{tableNET}. Suppose via coalition bidders $1$ and~$2$ change their bids to $b_l(x)=0$ for all $x\in\R_+$. Then, bidders $1$ and $2$ are the only winners of the dispatch problem and their payments and utilities are given in the second column of \cref{tableNET}. Collusion is profitable and the total payment of the operator increases from \$$260$ to \$$280$. 
\vspace{.2cm}
\begin{figure}[ht]
	\begin{center}
		\begin{tikzpicture}[scale=0.8, every node/.style={scale=0.5}]
		\draw[-,black!80!blue,line width=.32mm] (-0.3,0) -- (-0.3,0.42);
		\draw[-,black!80!blue,line width=.32mm] (0.3,0) -- (0.3,0.42);
		\draw[-,black!80!blue,line width=.32mm] (0.3,3.58) -- (0.3,4);
		\draw[-,black!80!blue,line width=.32mm] (-0.3,3.58) -- (-0.3,4);
		\draw[-,black!80!blue,line width=.32mm] (-3.59,1.8) -- (-4,1.8);
		\draw[-,black!80!blue,line width=.32mm] (-3.59,2.2) -- (-4,2.2);
		\draw[-,black!80!blue,line width=.32mm] (3.59,1.8) -- (4,1.8);
		\draw[-,black!80!blue,line width=.32mm] (3.59,2.2) -- (4,2.2);
		
		\draw[-,black!80!blue,line width=.3mm] (-3.6,1.8) -- (-0.3,0.4) node[anchor=west]  at(-4,0.25) {\LARGE $C_{3,2}=10\text{ MW}$};
		\draw[-,black!80!blue,line width=.3mm] (-0.3,3.6) -- (-3.6,2.2) node[anchor=west]  at(-4, 3.65) {\LARGE $C_{3,1}=10\text{ MW}$};
		\draw[-,black!80!blue,line width=.3mm] (0,4) -- (0,0) node[anchor=west]  at(0.1,2) {\LARGE $C_{1,2}=10\text{ MW}$};
		\draw[-,black!80!blue,line width=.3mm] (0.3,3.6) -- (3.6,2.2) node[anchor=west]  at(1.75,3.65) {\LARGE $C_{1,4}=10\text{ MW} $};
		\draw[-,black!80!blue,line width=.3mm] (0.3,0.4) -- (3.6,1.8) node[anchor=west]  at(1.75,0.25) {\LARGE $C_{2,4}=10\text{ MW} $};
		\draw[-,line width=.85mm] (-1,0) -- (1,0) node[anchor=west]  {\LARGE $2$};
		\draw[-,line width=.85mm] (-1,4) -- (1,4) node[anchor=west]  {\LARGE $1$};
		\draw[-,line width=.85mm] (-4,3) -- (-4,1) node[anchor=west]  {\LARGE $3$};
		\draw[-,line width=.85mm] (4,3) -- (4,1) node[anchor=west]  {\LARGE $4$};	
		\draw[<-,line width=.3mm] (0,-0.05) -- (0,-1);
		\draw[<-,line width=.3mm] (0,4.05) -- (0,5);
		\draw[<-,line width=.3mm] (-4.05,2) -- (-5,2);
		\draw[->,line width=.3mm] (4,2) -- (5,2);
		\draw (0,-1.5) circle (.5cm) node {\LARGE $G_2$} node at(2.7,-1.5) {\huge $c_2(x)=.1x^2 + 12x$};
		\draw (0,5.5) circle (.5cm)node {\LARGE $G_1$} node at(2.7, 5.6) {\huge $c_1(x)=.1x^2 + 12x$};
		\draw (-5.5,2) circle (.5cm) node {\LARGE $G_3$}node at(-8,2) {\huge $c_3(x)=.1x^2 + 5x$};
		\draw (5.5,2) circle (.5cm)node {\LARGE $D$} node at(7.6, 2) {\huge $D=20\text{ MWh}$};
		\end{tikzpicture}
		\caption{Four-node three-generator network}\label{fig:three_node}
	\end{center}
\end{figure}
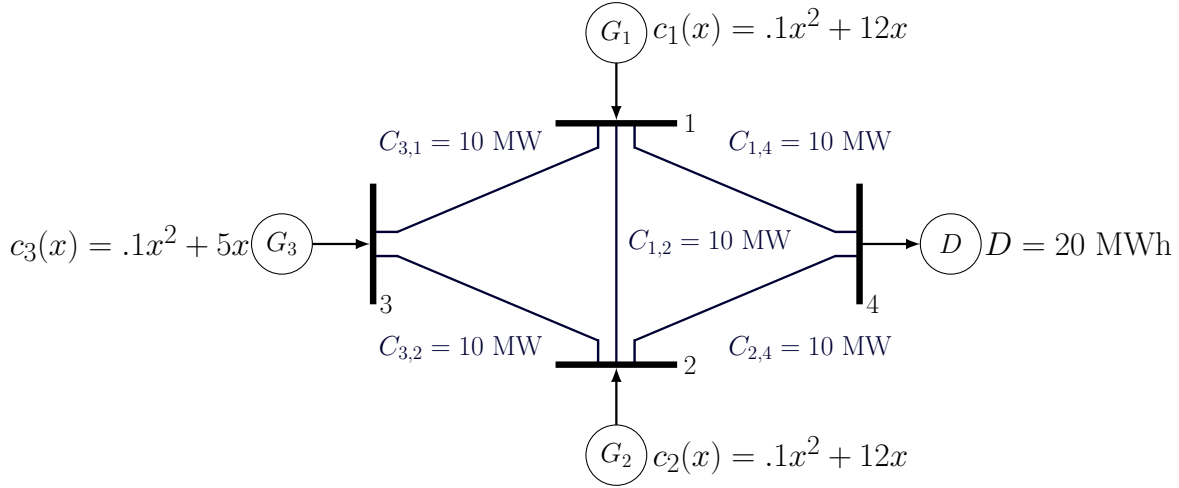

\begin{table}[h]
	\caption{The VCG outcomes for the network model (\$, MWh)}
	\label{tableNET}
	\begin{center}
		\begin{tabular}{|l||l||l||l||l|}
			\hline
			& \multicolumn{2}{l||}{Truthful Bidding} & \multicolumn{2}{l|}{Collusion ($1$, $2$)}\\
			\hline
			& payment (utility)  & $x$ & payment (utility)  & $x$  \\
			\hline
			Bidder $1$ & $0$ $(0)$ & 0 & $140$ $(10)$& 10 \\
			\hline
			Bidder $2$  & $0$  $(0)$ & 0 & $140$ $(10)$ & 10 \\
			\hline
			Bidder $3$ & $260$  $(120)$ & 20& $0$ $(0)$ & 0\\
			\hline
		\end{tabular}
	\end{center}
\end{table}

Next, we consider the MPCS mechanism. Under the MPCS mechanism, payments and utilities are given in the first column of~\cref{tableNET2}. Suppose, via coalition, bidders $1$ and $2$ change their bids to $b_l(x)=0$ for all $x\in\R^t_+$. Then, bidders $1$ and $2$ are the only winners of the dispatch problem and their payments and utilities are given in the second column of~\cref{tableNET2}. We observe that the collective utility of bidders $1$ and~$2$ reduces to $-\$120$. Hence, in this case, collusion is not profitable for bidders~$1$~and~$2$. Furthermore, after the collusion, the total payment of the operator also reduces from \$$260$ to \$$140$. This example shows that the core-selecting mechanisms can eliminate collusion of the losing bidders.\footnote{
In the example above, it can be verified that the constraint set can be reformulated as $x_1 = x_2$, $x_1 + x_2 + x_3 = 20$, $x_l \leq 20$ for all $l$.
Hence, this example is in fact an extension of the well-studied Local-Local-Global (LLG) multi-item auction model from~\cite{ausubel2010core,krishna1996simultaneous} to a DC-OPF problem involving continuous goods. In the LLG model, three bidders compete for two distinct discrete items. Bidders 1 and 2 are considered local bidders, each wanting only one of the items. Bidder 3 is global in the sense that it only wants both items at the same time. Consequently, the LLG model can have only two scenario outcomes, namely, either the local bidders are allocated one item each, or the global bidder receives both items. On the other hand, our example can be considered as having infinitely many continuous allocation outcomes that lie in between these two scenarios. } We note that it is in general not possible to provide guarantees on the collusion of the winning bidders. Specifically, core-selecting mechanisms are known to not exhibit revenue monotonicity~\cite{lamy2010core}.

\begin{table}[h]
	\caption{The MPCS outcomes for the network model (\$, MWh)}
	\label{tableNET2}
	\begin{center}
		\begin{tabular}{|l||l||l||l||l|}
			\hline
			& \multicolumn{2}{l||}{Truthful Bidding} & \multicolumn{2}{l|}{Collusion ($1$, $2$)}\\
			\hline
			& payment (utility)  & $x$ & payment (utility)  & $x$  \\
			\hline
			Bidder $1$ & $0$ $(0)$ & 0 & $70$ $(-60)$& 10 \\
			\hline
			Bidder $2$  & $0$  $(0)$ & 0 & $70$ $(-60)$ & 10 \\
			\hline
			Bidder $3$ & $260$  $(120)$ & 20& $0$ $(0)$ & 0\\
			\hline
		\end{tabular}
	\end{center}
\end{table}
\section{IEEE test systems with DC power flow models}
The following simulations are based on the IEEE test systems with polytopic DC power flow constraints again adopting the models considered in \cite{wu1996folk}.
\subsection{14-bus, 30-bus and 118-bus test systems}

We consider the IEEE $14$-bus \cite{christie2000power}, $30$-bus \cite{alsac1974optimal, ferrero1997transaction} and $118$-bus test systems \cite{christie2000power}. We assume all bidders are truthful and the true cost curves are convex quadratic polynomials, provided in the references. In practice, truthfulness may only hold under the VCG mechanism since it is dominant-strategy incentive-compatible.\footnote{Our previous work in~\cite{karaca2020regret} developed no-regret learning algorithms that can take advantage of the partially observed data in electricity markets under several different payment rules including the pay-as-bid, the core-selecting, and the VCG mechanisms. Such algorithms are known to converge to a \emph{coarse-correlated equilibrium}. This equilibrium concept generalizes Nash equilibrium to the case where all players are endowed with a probability distribution over the state of the game, see~\citep{cesa2006prediction}. This work is not included in this thesis because bidding algorithms will not be treated. We kindly refer to~\cite{karaca2020regret}.} The corresponding total payments of the mechanisms are shown in \cref{tablenet1}. All the mechanisms lead to the same winner allocation as expected. 
\begin{table}[ht]
	\caption{Total payments of the IEEE test systems}
	\label{tablenet1}
	\begin{center}
		\begin{tabular}{|l||l||l||l|}
			\hline
			Mechanism & $14$-bus & $30$-bus & $118$-bus \\
			\hline
			Pay-as-bid & \$$7642.6$  & \$$565.2$  & \$$125947.8$  \\
			\hline
			LMP & \$$10105.1$  &  \$$716.9$  & \$$167055.8$  \\  
			\hline
			MPCS& \$$10513.4$ & \$$746.4$  & \$$169300.4$ \\  
			\hline
			VCG & \$$10513.4$  & \$$746.4$  & \$$169300.4$  \\
			\hline
		\end{tabular}
	\end{center}
\end{table}

For all three test systems, we observe that the VCG mechanism has a slightly larger total payment than the LMP mechanism. Moreover, as we expect, the VCG payment of every bidder is larger than its LMP payment. 

Another observation is that the VCG outcomes are in the core for all systems. Next, we provide an explanation for each system. $14$-bus and $118$-bus systems do not have any line limits, hence, they have the form of \eqref{eq:simpler_clearing_model}. Invoking~\cref{thm:conditions_on_bids}, we conclude that supermodularity condition holds. Despite the fact that $30$-bus system has line limits, the VCG outcome is in the core. This result can be explained in two ways. First, none of the line limit constraints are tight. Second, we observe that removing two bidders can yield to an infeasible problem, similar to~\cref{lem:removal_of_two}. These test systems are specialized instances and they do not necessarily conclude that the VCG mechanism is coalition-proof for the DC power flow models. We examine this shortcoming of the VCG mechanism in our next simulation.

Computation times are provided only for the $118$-bus case, because the other problems are trivially small. For this electricity market, the direct solution approach is not computationally feasible, because there are $19$ winners out of $54$ bidders and the optimal cost calculation takes $451$~milliseconds (with GUROBI 7.5~\cite{gurobi} called through MATLAB via YALMIP~\cite{lofberg2005yalmip}). This approach would require $66$ hours. Computation times for the VCG mechanism and the CCG algorithm are $24.8$ and $31.4$ seconds respectively. After the VCG mechanism, the CCG algorithm converges only in a single iteration. This iteration takes 6.6 seconds, because it involves binary variables whereas the optimal cost calculation for the market model does not. 

\subsection{Effect of line limits}

We consider the IEEE $14$-bus test system, with a line limit on lines exiting node $1$, connecting node~$1$ to nodes $2$ and $5$. We set this line limit to be $10$ MW.  We again assume all bidders are truthful. The corresponding total payments of the mechanisms are shown in~\cref{tablenet2}. 

\begin{table}[ht]
	\caption{Total payments of the IEEE 14-bus test system with line limits}
	\label{tablenet2}
	\begin{center}
		\begin{tabular}{|l||l|}
			\hline
			Mechanism & 14-bus with line limits \\
			\hline
			Pay-as-bid & \$$9715.2$  \\
			\hline
			LMP & \$$10361.0$  \\  
			\hline
			MPCS& \$$11220.1$   \\  
			\hline
			VCG & \$$11432.1$  \\
			\hline
		\end{tabular}
	\end{center}
\end{table}

We observe that the VCG outcome does not lie in the core. Line limits are tight and the problem does not have the form of \eqref{eq:simpler_clearing_model}. Hence, shill bidding and collusion can be profitable for bidders. Moreover, we observe that the MPCS mechanism yields a larger total payment than the LMP mechanism. 

With this result, we also reiterate that convex bid curves are not enough to ensure that the VCG outcomes are in the core. Similar results can be obtained for $118$-bus test system by fixing $50$~MW limits on two lines, one connecting nodes~$5$ and~$6$, another connecting nodes~$9$ and~$10$.  

For the 14-bus example, computation times for the VCG mechanism and the CCG algorithm are $3.6$ and $8.2$ seconds, respectively. After the VCG mechanism, the CCG algorithm converges in $4$~iterations. 

\section{AC-OPF problem with a duality gap}

The following simulation is based on a 5-bus network model given in~\cite{bukhsh2013local}. Apart from having additional generators, the network model is the same as the model found in~\cite{bukhsh2013local}. We provide the true generator costs for active power in~\cref{table:actable}. Note that it would be straightforward to include also the costs for reactive power. For this problem, strong duality does not hold. We verified this by showing that the semidefinite programming relaxation is not tight for this polynomial optimization problem~\cite{molzahn2014moment}. Consequently, Lagrange multipliers may not be meaningful in an economic sense. We highlight that we can solve this problem to global optimality via the second level of moment relaxations (sum-of-squares hierarchy)~\cite{molzahn2014moment}. 

Payments under the pay-as-bid and the MPCS mechanisms are provided in~\cref{table:actable}. Note that the pay-as-bid would actually not lead to truthful behavior. It is provided for comparison since the LMP mechanism is not applicable. The comparison of these two core-selecting mechanisms under truthful bidding can be regarded as a result of competitive equilibrium assumptions. These assumptions imply that bidders would be willing to bid truthfully since truthful allocations maximize their utilities, see the discussion provided after~\cref{lem:effce}.  

\begin{table}[h]\caption{Generator data for 5-bus AC-OPF problem}\label{table:actable}
	\begin{center}
	{\begin{tabular}{|l||l||l||l||l||l|}
		\hline
		Gen. & Node& Cost & $x_l^*$ MW & Pay-as-bid & MPCS\\
		\hline 
		1& 1  & $.1x_1^2+4x_1$ & 246.0 & \$7038.0& \$12772.3\\
		\hline 
		2 & 5  & $.1x_2^2+1x_2$ & 98.2 & \$1061.5 & \$2435.6 \\
		\hline 
		3 & 1  & $.1x_3^2+30x_3$ & 0 & 0 & 0 \\
		\hline 
		4 & 5  & $.1x_4^2+15x_4$ & 0 & 0 & 0\\
		\hline
	\end{tabular}}
	\end{center}
\end{table}

For this example, even though the constraints are not polymatroid-type, the MPCS mechanism happens to coincide with the VCG mechanism, attaining the DSIC property. 
Moreover, we can ensure that losing bidders 3 and 4 cannot profit from collusion, and no bidder can profit from bidding with multiple identities. Here, we briefly illustrate one collusion instance for bidders 3 and 4. Suppose their bids are given by $b_3(x_3)=.1x_3^2+2x_3$ and $b_4(x_4)=.1x_4^2+.5x_4$. Then, their new allocations would be $x_3^*=128.0$ MW and $x_4^*=50.3$ MW respectively. However, these bidders would now be making losses under the MPCS mechanism since we can compute their utilities as $u_3^*=-\$306.7$ and $u_4^*=-\$147.7$. Hence, collusion would not be profitable for bidders 3 and 4.

As a remark, a rational bidder would overbid under the pay-as-bid mechanism. Suppose all the bidders increase their cost coefficients by $87.7\%$ under the pay-as-bid mechanism. Then, the total payment under the MPCS would be less than the one under the pay-as-bid since the MPCS mechanism incentivizes truthful bidding via both coalition-proofness and~DSIC.

For the problem under true costs, we can also show that there is no linear price function that would yield a competitive equilibrium. Since the bids are strictly convex, the condition in \eqref{eq:firstcondition} requires assigning bidder~1 a linear price equal to its marginal cost at $246$ MW. This price is given by $\$53.2/\text{MW}$, yielding the payment $\$13087.2$. Under this payment mechanism, utility of the bidder~1 cannot be in the core since its utility is greater than the one under the VCG mechanism. This concludes the nonexistence of a competitive equilibrium in linear prices for this problem.

This AC-OPF problem was solved in 1.85 seconds using the method in~\cite{molzahn2015sparsity} with MOSEK 9 \cite{mosek2015mosek} called through MATLAB. 

\section{Swiss reserve procurement auctions}

The following simulation is the Swiss reserve procurement auction in the 46th week of 2014 which is based on a pay-as-bid payment rule~\cite{abbaspourtorbati2016swiss}. This auction involves 21 plants bidding for secondary reserves, 25 for positive tertiary reserves and 21 for negative tertiary reserves. The bids are discrete, that is, they are given by sets of reserve size and price pairs. We remind the reader that the formulation in~\cref{sec:p1_2} can capture such bids. The objective also includes a second stage cost corresponding to the uncertain daily auctions. Moreover, the market involves complex constraints arising from nonlinear cumulative distribution functions. These~constraints imply that the deficit of reserves cannot occur with a probability higher than 0.2\%, and they include coupling between the first and the second stage decision variables. 

Since this problem does not attain strong duality, meaningful linear prices cannot be derived. 
The total payments of the pay-as-bid, MPCS, and VCG mechanisms are shown in~\cref{tablen3}. Notice that the VCG utilities do not lie in the core since otherwise the MPCS mechanism would coincide with the VCG mechanism. As a result, the MPCS mechanism is coalition-proof, but it does not attain the DSIC property. As is discussed in~\cref{lem:lie}, we can still quantify the loss of the DSIC property by the difference between the MPCS and the VCG payments.
\begin{table}[h]
	\caption{Total payments of the reserve market (million CHF) }
	\label{tablen3}
\begin{center}
		\begin{tabular}{|l||l||l|}
			\hline
Pay-as-bid & MPCS & VCG\\
			\hline
 $2.293$ & $2.437$ & $2.529$ \\  
			\hline
		\end{tabular}
		\end{center}
\end{table}

For this electricity market, the direct MPCS computation approach is not computationally feasible, because there are $28$ winners and the optimal cost calculation takes $8$ seconds (with GUROBI 7.5~\cite{gurobi} called through MATLAB via YALMIP~\cite{lofberg2005yalmip}). This approach would require $68$ years. Computation times for the VCG mechanism and the CCG  algorithm are $580.6$ and $659.2$~seconds respectively. After the VCG mechanism, the CCG algorithm converges in $4$~iterations. This shows that the MPCS mechanism can be computed in a reasonable time even when there are many~participants. 

\section{Two-sided markets with DC power flow models}

We consider the DC-OPF problem in~\cref{fig:three_node_two}. Bidders are truthful, and the true costs are quadratic polynomials. All lines have the same susceptance. Line limit from node $i$ to node $j$ is denoted again by $C_{i,j}=C_{j,i}\in\R_+$. The optimal allocation is computed as $x^*=[0.58,\, 0.58,\, 4,\, -5.16]$ MW. The pay-as-bid mechanism yields a positive budget of~$\$48.3$. The LMP mechanism also results in a positive budget since the limits $C_{3,1}$ and $C_{3,2}$ are tight at the optimal solution~\cite[Fact~5]{wu1996folk}. This balance is $\$2.8$. For this problem, the MPCS mechanism achieves strong budget-balance with~$\$0$. Finally, the VCG mechanism has a deficit of $-\$34.8$. 

By definition, under the MPCS mechanism sum of the utilities of the bidders are higher than the one under the LMP mechanism. However, we underline that under the MPCS mechanism, not every bidder receives a utility higher than its LMP utility. In this DC-OPF problem, the total utility of the supply-side reduces by $\$4.1$, whereas the utility of the demand-side increases by $\$6.9$ when we compare the MPCS outcome with the LMP outcome.  
Finally, we compute the bound in~\cref{lem:lie} for bidder~$4$. The pay-as-bid, the LMP, and the MPCS mechanisms yield the maximum additional profits  $\$48.3$, $\$21.7$, and $\$14.9$ respectively. Under the VCG mechanism, there is no room for unilateral deviation, and hence we obtain $\$0$ as the maximum additional profit.
This DC-OPF problem was solved in 0.32 seconds with GUROBI 7.5 \cite{gurobi} called through MATLAB via YALMIP~\cite{lofberg2005yalmip}. To calculate the MPCS payments, we solved the market problem under $2^4-1$ different coalitions. 
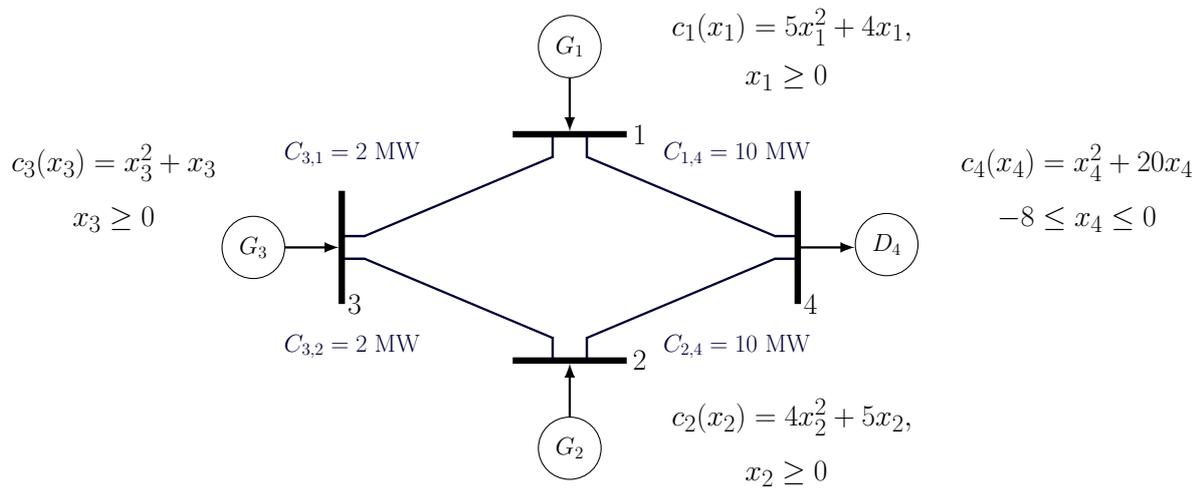
\begin{figure}[h]
	\begin{center}
		\begin{tikzpicture}[scale=0.75, every node/.style={scale=0.45}]
		\draw[-,black!80!blue,line width=.32mm] (-0.3,0) -- (-0.3,0.42);
		\draw[-,black!80!blue,line width=.32mm] (0.3,0) -- (0.3,0.42);
		\draw[-,black!80!blue,line width=.32mm] (0.3,3.58) -- (0.3,4);
		\draw[-,black!80!blue,line width=.32mm] (-0.3,3.58) -- (-0.3,4);
		\draw[-,black!80!blue,line width=.32mm] (-3.59,1.8) -- (-4,1.8);
		\draw[-,black!80!blue,line width=.32mm] (-3.59,2.2) -- (-4,2.2);
		\draw[-,black!80!blue,line width=.32mm] (3.59,1.8) -- (4,1.8);
		\draw[-,black!80!blue,line width=.32mm] (3.59,2.2) -- (4,2.2);
		\draw[-,black!80!blue,line width=.3mm] (-3.6,1.8) -- (-0.3,0.4) node[anchor=west]  at(-5.1,0.25) {\LARGE $C_{3,2}=2\text{ MW}$};
		\draw[-,black!80!blue,line width=.3mm] (-0.3,3.6) -- (-3.6,2.2) node[anchor=west]  at(-5.1, 3.65) {\LARGE $C_{3,1}=2\text{ MW}$};
		\draw[-,black!80!blue,line width=.3mm] (0.3,3.6) -- (3.6,2.2) node[anchor=west]  at(1.55,3.65) {\LARGE $C_{1,4}=10\text{ MW} $};
		\draw[-,black!80!blue,line width=.3mm] (0.3,0.4) -- (3.6,1.8) node[anchor=west]  at(1.55,0.25) {\LARGE $C_{2,4}=10\text{ MW} $};
		\draw[-,line width=.85mm] (-1,0) -- (1,0) node[anchor=west]  {\huge $2$};
		\draw[-,line width=.85mm] (-1,4) -- (1,4) node[anchor=west]  {\huge $1$};
		\draw[-,line width=.85mm] (-4,3) -- (-4,1) node[anchor=west]  {\huge $3$};
		\draw[-,line width=.85mm] (4,3) -- (4,1) node[anchor=west]  {\huge $4$};	
		\draw[<-,line width=.3mm] (0,-0.05) -- (0,-1);
		\draw[<-,line width=.3mm] (0,4.05) -- (0,5);
		\draw[<-,line width=.3mm] (-4.05,2) -- (-5,2);
		\draw[->,line width=.3mm] (4,2) -- (5,2);
		\draw (0,-1.55) circle (.55cm) node {\LARGE $G_2$} node at(3.9,-1) {\huge $c_2(x_2)= 4x_2^2 + 5x_2,$};
		\draw (0,5.55) circle (.55cm)node {\LARGE $G_1$} node at(3.9, 5.9) {\huge $c_1(x_1)=5x_1^2 + 4x_1,$};
		\draw (-5.55,2) circle (.55cm) node {\LARGE $G_3$}node at(-8,3.5) {\huge $c_3(x_3)=x_3^2 + x_3$};
		\draw (5.56,2.05) circle (.55cm)node {\LARGE $D_4$} node at(8.9, 3.5) {\huge $c_4(x_4)=x_4^2 + 20x_4$};
		\draw node at(8.9,2.5) {\huge $-8\leq x_4\leq0$};
		\draw node at(3.8,5) {\huge $x_1\geq0$};
		\draw node at(3.8,-2) {\huge $x_2\geq0$};
		\draw node at(-8,2.5) {\huge $x_3\geq0$};
		\end{tikzpicture}
		\vspace{-.3cm}
		\caption{Two-sided DC-OPF model}\label{fig:three_node_two}
	\end{center}
\end{figure}


\part{Enabling inter-area reserve exchange through stable benefit allocation mechanisms}\label{part:2}
\chapter{Introduction}\label{secintro}

{The existing electricity market architectures and the predominant models for power system balancing were designed in a time when fully controllable generators with nonnegligible marginal costs were prevalent.}
{However, as the increasing shares of variable and partly predictable renewable resources displace controllable generation the dispatch flexibility decreases, while the operational uncertainty characteristics become increasingly more complex.}
{In light of this new operational paradigm, there is an imperative need to re-evaluate the current electricity market design}, and there has been a surge of interest in proposing new market frameworks~\cite{neuhoff2005large,ahlstrom2015evolution}.

According to the European electricity market design, the bulk volume of energy trading takes place in the day-ahead market, which is typically cleared 12-36 hours before the actual delivery, based on single-valued point forecasts of the stochastic power output of renewable resources. In turn, a balancing market is cleared  close to the hour of delivery in order to compensate any deviations from the day-ahead schedule. Apart from these energy-only trading floors, a reserve capacity auction is organized, usually prior to the day-ahead market, in order to ensure that sufficient capacity is set aside for the provision of real-time balancing services. 
Following this sequential clearing and single-valued forecast approach, the current market structure attains only limited coordination between the day-ahead and the balancing stages.
Aiming to enhance this temporal coupling, recently proposed dispatch models employ scenario-based stochastic programming in order to co-optimize the day-ahead and the reserve capacity markets \cite{pritchard2010single,morales2012pricing}. {However, these approaches cannot be directly applied to any of the existing electricity markets, since they would require significant restructuring of the current market frameworks.}
\looseness=-1

In terms of geographical considerations, the European electricity market is fully coordinated only at the day-ahead stage, while reserve capacity and balancing markets are still operated on {a country/regional level}~\cite{dominguez2019reserve}. If we can remove the existing barriers for cross-border trading, geographic diversification of the uncertain renewable resources would smooth out the forecast errors and reduce the need for balancing actions. In order to improve the security of supply and the efficiency of the balancing system, the European Commission (EC) regulation has already established a detailed guideline \cite{EC1} that lays out the rules also for the integration of the balancing markets. 
In this ongoing process that is expected to be completed by 2023 \cite{ENTSOE1}, several questions remain open regarding the specific structure and final coordination arrangement among the ENTSO-E (European Network of Transmission System Operators for Electricity) member countries, since there is no binding legislation that enforces transmission system operators (TSOs) to enter such collaborations. The recent study in \cite{EC2} to investigate the benefits of different organizational models for the integration of balancing markets, shows that the 10-year net present value of coordinated balancing ranges from 1,7 to 3,8 B\euro, depending on the degree of coordination in the inter-area exchanges, see also~\cite{newbery2016benefits}.

{Nonetheless, any regional coordination arrangement in the procurement and activation of reserves depends upon the availability of inter-regional transmission capacity.
Given that the energy and reserve capacity markets are cleared separately and sequentially, a portion of this {transmission} capacity made available has to be withdrawn from day-ahead energy {market} to be allocated to reserve {exchange}. 
The exact methodology for the process of allocating inter-area transmission capacity to reserves remains to be discussed and approved by the ENTSO-E members \cite{ENTSOE2}~\cite[Art. (14)]{eur19}. Moreover, an application can be filed by even two or more TSOs.} {In coalitional game theory, such arrangements would be called coalitional deviations, since they involve only some of the members.}  {Motivated by this, in order to gain technical and economical insights about such a process, the goal of \cref{part:2} is to propose a coalitional game-theoretic approach for the design of a transmission allocation mechanism.}

\section{Related works}

\subsubsection*{Existing methodologies for deciding on transmission allocations}

{The reservation of inter-area interconnections for reserve exchange withdraws transmission resources from the day-ahead market, where the main volume of electricity is being traded. Currently, these cross-border capacities for the day-ahead market are decided by the TSOs and computed while respecting the minimum remaining available margin (minRAM) rule of $70\%$ following the requirements of the Clean Energy Package~\cite[Article 16(8a-8b)]{eur19}, see~\cite[\S 2]{elia20} for their computation. A sub-optimal reservation of transmission capacity for reserve exchange from these day-ahead quantities may lead to significant efficiency losses at the day-ahead stage. To prevent this situation, the Agency  for  the  Cooperation  of  Energy  Regulators  (ACER)~\cite{ACER} mandates to perform detailed analyses demonstrating that such reservation from day-ahead market would increase overall social welfare. Up to this date, inter-area transmission capacity is typically not removed from day-ahead energy exchange for reserves.} One notable exemption is the Skagerrak interconnector between Western Denmark and Norway, in which $15\%$ of the day-ahead cross-border transmission capacity is permanently set aside for reserve exchange~\cite{energinet}. Nevertheless, this allocation is static, while the true optimum varies dynamically depending on generation, load, and system uncertainties. As such, \cite{delikaraoglou2018optimal} developed a preemptive transmission allocation model that defines the optimal inter-area transmission capacity allocation to improve both spatial and temporal coordination {at the reserve procurement stage}.

{
The recently proposed preemptive transmission allocation model of~\cite{delikaraoglou2018optimal} focuses on the minimization of the expected system cost, assuming implicitly full coordination among all the regional operators and their market participants. This assumption is in line with the current state of the day-ahead market, which is fully integrated across Europe, or even for the balancing markets in certain regions, e.g., in the Nordic system all reserve activation offers are pooled into a common merit-order list and are available to all TSOs~\cite{bondy2014operational}. {This assumption allows us to model each of these three trading floors by one respective optimization problem, which would not be possible in case of partial coordination.} However, the initial version of the preemptive model} does not suggest an area-specific cost allocation which guarantees that all areas have sufficient benefits to accept the proposed solution.\footnote{An area (country or region) as a whole includes consumers and generators pertaining to that area and area operators (and potentially the transmission owners). Notice that the transmission capacity allocated to the reserve exchange affects the incentive structure of all these market participants. Then, the benefits are the reductions in the total cost allocated to an area from all three stages of the sequential market, whereas the cost refers to minus the social welfare, which is given by the sum of the consumers' and generators' surplus pertaining to that area and the congestion rents collected by the corresponding area operator~\cite{kristiansen2018mechanism}.} {This would be concerning for the fairness of the future integrated balancing markets.} 

Similar issues regarding fair settlements are already under investigation by Swissgrid for the simpler setting of imbalance netting~\cite{avramiotis2018investigations}. Moreover, the stakeholder document from the International Grid Control Cooperation \cite[\S 6]{IGCC}, developed by ten European operators, describes analytically a fair settlement scheme for the imbalance netting process. Hence, such incentives will be an actual issue while moving towards fully integrated European markets. To address this, we {integrate} the preemptive model of \cite{delikaraoglou2018optimal} in a {mathematical} framework that allows the application of tools from coalitional game theory in order to obtain {a stable benefit allocation, that is, sufficient benefits providing immunity to coalitional deviations ensuring that all areas are willing to coordinate via the preemptive model.}

\subsubsection*{Related works on coalitional game theory}

The {concepts} from coalitional game theory have recently been widely used in the energy community. The Shapley value has been employed in problems regarding the distribution of social welfare among TSOs participating in an imbalance netting cooperation~\cite{avramiotis2018investigations} as well as the benefit allocation in transmission network expansion~\cite{ruiz2007effective} and in cross-border interconnection development~\cite{kristiansen2018mechanism}.
Other applications of the Shapley value in the energy field include cooperation problems in the Eurasian gas supply system \cite{nagayama2014network} and the $\text{CO}_2$ emissions abatement in mainland China~\cite{he2018estimation}.
However, the Shapley value is in general not within the core, that is, the set of stable outcomes.
\cite{baeyens2013coalitional} shared the expected profits from aggregating wind power generation using the core benefit allocations, whereas a similar concept was applied for prosumer cooperation in a combined heat and electricity system in~\citep{mitridatidesign}, and for cross-border transmission expansion in Northeast Asia in~\citep{churkin2019can}.

In contrast to the case studies in the aforementioned works and \cref{part:1}, realistic instances of our problem in \cref{part:2} can exhibit an empty core. To this end, we utilize the least-core as a solution concept, since it achieves minimal stability violation, that is, minimal benefit improvements from coalitional deviations~\cite{maschler1979geometric}. To obtain a unique outcome, we propose the approximation of a fairness criterion, which {is at the discretion of the regulatory authorities to define.
We propose two variations of the benefit allocation mechanism that can be executed either at the day-ahead or the real-time stage to distribute the expected or the actual benefits (that is, when the uncertainty is revealed), respectively.} 
We illustrate how these formulations establish {a trade-off between allocating the risk of facing scenario-dependent benefit outcomes to the regulator of this organization or to the areas.}
To overcome the exhaustive enumeration of all coalitional deviations, we show that the least-core selecting allocations in this work can be computed efficiently via an iterative constraint generation algorithm. Similar algorithms were utilized to compute an outcome from the core in combinatorial auctions~\cite{day2007fair} and electricity markets~in~\cref{sec:p1_4} in \cref{part:1} of this thesis. In contrast, this part shows that this algorithm can also be extended to the least-core in a general nonconvex problem. 

\section{Summary of goals and contributions}

The contributions of \cref{part:2} are as follows. 
\begin{enumerate}
\item We formulate the coalition-dependent version of the preemptive transmission allocation model such that we can consider coalitional arrangements between only a subset of operators. This is a novel extension of the model proposed by \cite{delikaraoglou2018optimal}. 

\item We then study the coalitional game that treats the benefits as an ex-ante process with respect to the uncertainty realization {and} we provide a condition under which the core is nonempty. Under this condition, it is possible to obtain a stable outcome. In case this condition is not satisfied, we prove that the least-core, which is an outcome that attains minimal stability violation, also ensures the individual rationality property. These two results are obtained for coalitional games where the coalitional value function is given by a stochastic bilevel optimization problem.

\item We then propose the least-core selecting mechanism as a benefit allocation that achieves minimal stability violation, while enabling the approximation of an additional fairness criterion. In order to {implement} this mechanism with only a few queries to the preemptive model, we formulate a constraint generation algorithm. 

\item In addition, we formulate a variation of the coalitional game that {allocates} the benefits {in} an ex-post process, which can be applied only after {the uncertainty realization is known}. 
For this game, we provide conditions under which the core is empty. 

\item We propose an ex-post version of our least-core benefit allocation mechanism. 
The ex-ante and ex-post versions of this mechanism can achieve the same fundamental properties for the areas either for every uncertainty realization or in expectation, respectively. (However, we note that the former requires the regulator to have a financial reserve to buffer the fluctuations in the budget.)

\item Finally, we provide techno-economic insights on the factors that drive benefit allocations first with an illustrative three-area nine-node system and then with a more realistic case study based on a larger IEEE test system.
\end{enumerate}
\looseness=-1

\subsubsection*{Organization}

\cref{sec:elecmarkframe} describes the organizational structure and introduces a set of necessary assumptions to obtain tractable models. \cref{sec:transcapalloc} discusses the issues related to reserve exchanges {and motivates the formulation of} the preemptive transmission allocation model. \cref{sec:3} {introduces necessary background from} coalitional game theory, whereas \cref{sec:4} focuses on the games arising from the preemptive model, {which provide the basis for the} benefit allocation mechanisms that accomplish the implicit coordination requirements outlined in the previous section.
The numerical case studies are presented in \cref{sec:CaseStudies}.

\chapter{European electricity market framework}\label{sec:elecmarkframe}
In this chapter, we first describe the basic organizational structure of the European electricity market, which impels essentially the need for the benefit allocation mechanisms proposed in this part. In turn, we present a set of necessary assumptions that must be imposed to obtain a tractable formulation of the market-clearing problems and finally we provide the mathematical formulations that we employ for modelling different trading floors.

\section{Sequential electricity market design and modeling assumptions}\label{sec:seqassm}

The existing market design based on the sequential and independent clearing of reserves, day-ahead, and balancing markets (illustrated in~\cref{fig:seq_mark}) suffers from two main caveats that become increasingly pronounced as we move towards larger shares of renewable energy production. On the one hand, the day-ahead schedule is optimized based on purely deterministic inputs, that is, single-valued forecasts of renewable resources.
As a result, the day-ahead market is not responsive to the uncertainty associated with the forecast errors and thus it is weakly coordinated with the real-time balancing. {On the other hand, the decoupling of energy and upward/downward reserve capacity trading into two independent auctions ignores the substitution and complementary properties of these two services and leads to inefficient reserve procurement and energy schedules. Eliminating this issue requires that the participants are perfectly capable of accounting for these properties internally in their trading strategies. However, quantifying such opportunity costs is a challenging problem for the participants, see~\cite{swider2007bidding} and references therein.} As we discussed in the previous chapter, in order to enable the inter-area exchange of reserves given this decoupling, the operator has to withdraw a certain share of the interconnection capacities from the day-ahead energy trading and then use this headroom for the interconnections in the reserve capacity market.

\begin{figure}[h]
    \centering
    \includegraphics[width=0.34\columnwidth]{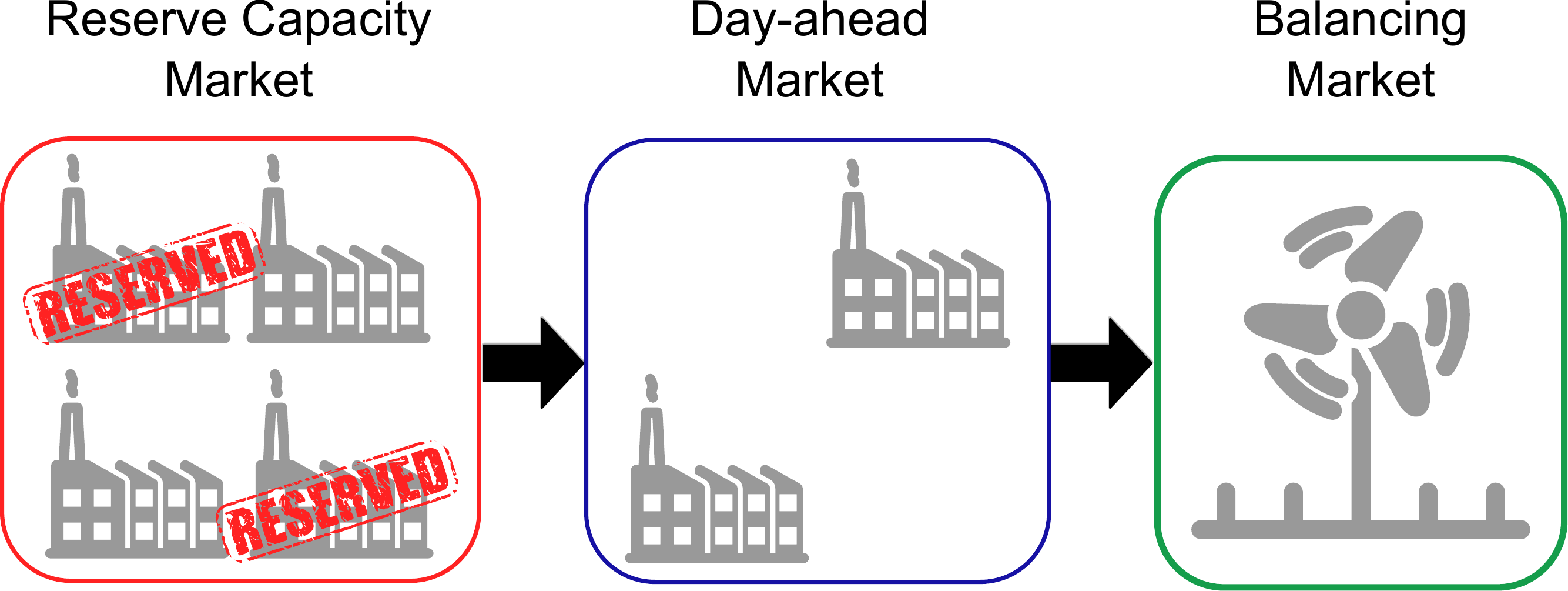}%
    \includegraphics[width=0.62\columnwidth]{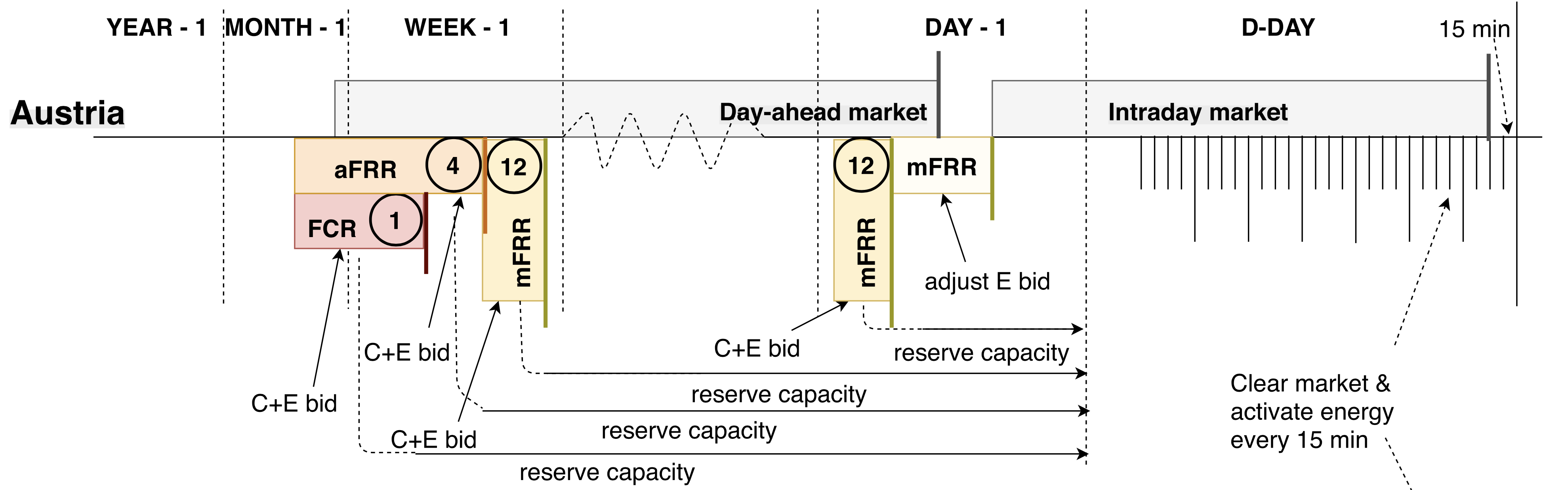}
    \caption{On the left: An illustration of the sequential European electricity market framework. On the right: The current Austrian market structure described in~\cite{poplavskaya2019distributed}: Intraday market corresponds to real-time balancing; aFRR, FCR, and mFRR are different types of reserves depending on how fast they are to react to imbalances.}
    \label{fig:seq_mark}
\end{figure}

From a theoretical perspective, these two issues can be contained if reserve capacity procurement, day-ahead energy schedules, and real-time re-dispatch actions are jointly optimized based on a probabilistic description of the uncertainty, see \cite{pritchard2010single,morales2012pricing}. {However, the adoption of a (scenario-based) stochastic dispatch model as a market-clearing algorithm requires significant restructuring of the current market framework, and poses several computational challenges when it is applied to real-life scale power systems.
Owing to the restructuring and the computational issues, we restrict ourselves to the status-quo market architecture and we embody its design attributes in our methodology aiming to mitigate the resulting inefficiencies.}

In the following sections of the chapter, we build the mathematical models of the different trading floors based on a set of assumptions that allows us to capture the main attributes of the European market, while maintaining tractability. 

In line with the current practice, we consider a zonal network representation during the reserve procurement. However, the full network topology is taken into account in the day-ahead and the balancing market-clearing models, using a DC power flow approximation. Note that {our network model can be readily adapted to a zonal day-ahead market (as it is today in Europe), for instance, where the inter-zonal transmission energy flows are constrained by the available transfer capacity (ATC) or by a flow-based domain. The challenge would be additional complexity originating from parameter choices for modelling the zonal day-ahead market, {as well as} the unscheduled flows and congestion to be tackled by {counter-trading or} re-dispatching in the balancing stage. {Using a full network representation, that is, nodal pricing, both in the day-ahead and the balancing markets} eliminates {potential discrepancies that may arise due to} idiosyncratic congestion effects and {allows us to concentrate on issues related to reserves exchange and transmission capacity reservation}. Moreover, for the flow-based domain, \cite{mltw13} show that the parameter choices {(e.g., the selection of a base case for the zonal injections, the determination of Power Transfer Distribution Factors and Generation Shift Keys)} can lead to very different market exchanges and prices. We also refer to \cite{solisanalysis} for numerical results illustrating this paradigm, and {to} \cite{en17,CRE17} for discussions on this sensitivity.
Since this part focuses on transmission allocation issues concerning primarily the regional operators, we believe that allowing for a more complete network representation 
and judiciously abstracting certain details of the real-world operation do not undermine our main goal, that is, {the development of a decision-support tool that provides} techno-economical insights on reserves exchange {and transmission capacity allocation}. Due to these reasons, many studies that focus on renewable integration and coordination analyses commonly ignore the zonal day-ahead market design. {Instead, }they either assume a nodal market similar to our work~\cite{lw08,lw09,dominguez2019reserve}, or they implement a simple transportation network ignoring zonal congestion~\cite{sw14,kem14}. Nonetheless, for the sake of completeness}, we will also provide a zonal model when discussing the day-ahead market. Our conclusions would still generalize to such models, since our model is replicating the main imperfections of the current design: the reaction of the market to its parametric decisions originating from the separate trading of energy and reserves.

On the generation side, we consider that all market participants are perfectly competitive (that is, they bid truthfully). Day-ahead energy offers are submitted in price-quantity pairs that internalize the marginal production cost as well as the unit commitment and inter-temporal constraints, e.g., ramping limits, in accordance with the portfolio bidding practice in the European market.
{Moreover, we assume that the reserve capacity offer prices provide adequate incentives to the flexible generators for the provision of real-time balancing services such that the prices of up and down re-dispatches are the same as the marginal prices in the day-ahead stage, similar to the previous works of~\cite{ahmadi2013multi}.}
In terms of stochastic renewable in-feed, we focus on wind power generation and we model forecast errors using a finite set of scenarios. Assuming null production costs, the corresponding offer price and the spillage cost are set equal to zero. 
On the consumption side, we consider inelastic demand with a large penalty on lost load and thus the social welfare maximization becomes equivalent to the cost minimization.

Finally, we assume that the current implementation of the sequential market provides a budget balanced method to allocate the system cost to all the areas, that is, the costs of the reserve capacity market, the day-ahead market and the balancing market are allocated to the areas without any deficit or surplus. In the numerics, for the case of no inter-area exchange of reserves, we provide and discuss one such allocation method based on the zonal and nodal prices that assigns producer and consumer surpluses to their corresponding areas, and divides the congestion rent equally between the adjacent areas, see~\cite{kristiansen2018mechanism} for the reasoning and \cref{sec:illu_ex} for its computation.

\section{Mathematical formulation}
\label{sec:seq}
\subsection{Reserved symbols for \cref{part:2}}\label{sec:part2not}

Main notation for this part is stated below. An illustration is provided in \cref{fig:notation}. Additional symbols are defined throughout the part when~needed.

\begin{itemize}

\item[] \textbf{Sets and indices}
	\item [$\AC$] Set of areas indexed by $a$
	\item [$\mathcal E$] Set of inter-area links indexed by $e$
	\item [$\mathcal L$] Set of transmission lines indexed by~$\ell$
	\item [$a_r(e/\ell)$] Receiving-end area of link $e$/line $\ell$ 
	\item [$a_s(e/\ell)$] Sending-end area of link $e$/line $\ell$ 
	\item [$ \Lambda_e$] Set of tie-lines across link~$e$
	\item [$\mathcal I$] Set of dispatchable power plants indexed by $i$
	\item [$\mathcal J$] Set of stochastic power plants indexed by $j$
	\item [$\mathcal N$] Set of network nodes (buses) indexed by $n$.
	\item [$ \mathcal{M}_{n}^{\mathcal I}$] Set of dispatchable power plants  $i$ located at node $n$
	\item [$ \mathcal{M}_{n}^{\mathcal J}$] Set of stochastic power plants $j$ located at node $n$
	\item [$ \mathcal{M}_{a}^{\mathcal I}$] Set of dispatchable power plants  $i$ located in area $a$
	\item [$ \mathcal{M}_{a}^{\mathcal J}$] Set of stochastic power plants $j$ located in area $a$
	\item [$ \mathcal{M}_{a}^{\mathcal N}$] Set of nodes $n$ located in area $a$
	\item [$ \mathcal{S}$] Set of stochastic power production scenarios indexed by~$s$

\item[] \textbf{Parameters}
\item [$\overline{{W}}_{j}$] Expected power production of stochastic power plant~$j$ [MW]
\item [$W_{js}$] Power production by stochastic power plant $j$ in scenario~$s$ [MW]
\item [$\pi_{s}$] Probability of occurrence of scenario $s$
\item [$A_{\ell n}$] Line-to-bus incidence matrix
\item [$B_{\ell}$] Absolute value of the susceptance of AC line $\ell$
\item [$D_{n}$] Demand at node $n$ [MW]
\item [$C_{i}$] Energy offer price of power plant $i$ [\euro/MWh]
\item [$C^{+/-}_{i}$] Up/down reserve capacity offer price of power plant $i$ [\euro/MW]
\item [$C^{\text{sh}}$] Value of involuntarily shed load [\euro/MWh]
\item [$P_{i}$] Capacity of dispatchable power plant $i$ [MW]
\item [$R^{+/-}_{i}$] Up/down reserve capacity offer quantity of power plant $i$ [\euro/MW]
\item [$RR_a^{+/-}$] Up/down reserve capacity requirements of area $a$ [MW]
\item [$T_{e/\ell}$] Transmission capacity of link $e$/line $\ell$ [MW]

\item[] \textbf{Variables}

\item [$\delta_n$] Voltage angle at node $n$ at day-ahead stage [rad]
\item [${\delta}_{ns}$] Voltage angle at node $n$ in scenario $s$ [rad]
\item [$\chi_{e/\ell}$] Transmission allocation of link $e$/line $\ell$, that is, percentage of inter-area interconnection capacity of link $e$/line $\ell$ allocated to reserves exchange
\item [$f_{\ell}$] Power flow in line $\ell$ at day-ahead stage [MW]
\item [$f_{e}$] Power flow in link $e$ at day-ahead stage [MW]
\item [${f}_{\ell s}$] Power flow in line $\ell$ in scenario $s$ [MW]
\item [$l^{\text{sh}}_{ns}$] Load shedding at node $n$ in scenario $s$ [MW]
\item [${p}_{i}$] Day-ahead schedule of dispatchable power plant $i$ [MW]
\item [${p}_{is}^{+/-}$]Up/down regulation provided by dispatchable power plant $i$ in scenario $s$ [MW]
\item [$r^{+/-}_{e}$] Up/down reserve capacity (`exported') from area $a_s(e)$ to area $a_r(e)$ [MW] (equivalently imported to area $a_r(e)$ from area $a_s(e)$)
\item [${w}_{j}$] Day-ahead schedule of stochastic power plant $j$ [MW]
\item [$w^{\text{spill}}_{js}$] Power spilled by stochastic power plant $j$ in scenario $s$ [MW]
\end{itemize}

\begin{figure}[h]
	\centering
	\begin{tikzpicture}[scale=1.25, every node/.style={scale=0.6}]
	\draw[-,line width=.5mm] (-2,.5) -- (-2,-.5) node[anchor=west]  {\LARGE $n_s$};
	\draw[-,line width=.5mm] (2,.5) -- (2,-.5) node[anchor=west]  {\LARGE $n_r$};	
	\draw[-,line width=.1mm] (-2,0.35) -- (2,0.35) node at (.4, .6) {\LARGE $\ell_1$};
	\draw[-,line width=.1mm] (-1.35,-1.25) -- (1.6,-1.25) node at (.4, -1) {\LARGE $\ell_2$}; 
	\draw[-,line width=.1mm] (-1.35,-2.1) -- (1.6,-2.1) node at (0, -1.6) {\LARGE $\vdots$} node at (.4, -1.85) {\LARGE $\ell_q$}; 
	\draw[-,line width=.1mm] (-2,-0.15) -- (-1.5,-0.15); 
	\draw[-,line width=.1mm] (-2,0.1) -- (-2.65,0.1) node[anchor=west] at (2.8, 0.075)  {\LARGE $\cdots$}; 
	\draw[-,line width=.1mm] (2.65,0.1) -- (2,0.1) node[anchor=west] at (-3.4, 0.075) {\LARGE $\cdots$} ; 
	\draw[-,line width=.1mm] (-2,-0.25) -- (-2.5,-0.25); 
	\draw[-,line width=.1mm] (2.5,-0.25) -- (2,-0.25); 
	\draw[-,line width=.1mm] (-2.5,-0.65) -- (-2.5,-0.25) node at (-2.5, -.9) {\LARGE $\vdots$} ;
	\draw[-,line width=.1mm] (2.5,-0.25) -- (2.5,-0.65) node at (2.5, -.9) {\LARGE $\vdots$} ;
	\draw (-1.25,-0.25) circle (.25cm) node {\LARGE $i$} node at (-2.6, -1.8) {\LARGE $i\in\mathcal M^{\mathcal I}_{a_s(e)}$} ;
	\draw[->,line width=.3mm] (-2.1,-1.35) -- (-1.5,-.75);
	\draw[dashed, draw=black] (-4,-2.5) rectangle (-0.85,0.95) node at (-2.5, 1.3) {\LARGE $\text{Area}\ a_s(e)$} node at (-3.4, -2.9) {\LARGE $\mathcal{H}(e,a_s(e))=-1$};
	\draw[dashed,draw=black] (4,-2.5) rectangle (1.1,0.95) node at (2.5, 1.3) {\LARGE $\text{Area}\ a_r(e)$} node at (3.4, -2.9) {\LARGE $\mathcal{H}(e,a_r(e))=1$};
	\draw[dashed,draw=black] (-0.5,0.95) rectangle (.75,-2.5) node at (.1, 1.3) {\LARGE $\text{Link}\ e$} node at (.1, -2.9) {\LARGE $\Lambda_e=\{\ell_1,\ldots,\ell_q\}$};
	\end{tikzpicture}
	\caption{An illustration of the main notation used for graph $(\mathcal A,\mathcal E)$ }\label{fig:notation}
\end{figure}
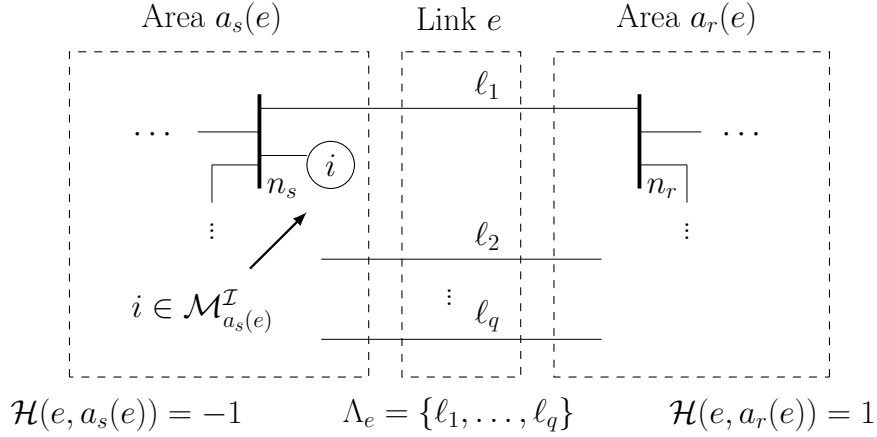

\subsection{Reserve capacity market} 

Having as a fixed input the upward/downward reserve requirements $RR_a^{+}/RR_a^{-}$ in each area~$a$ and a pre-defined share $\chi_e$ of the transmission capacity of each inter-area link~$e$ allocated to reserves, the reserve market clearing is formulated as:
\looseness=-1
\begin{subequations}
	\label{mod:A-Res}
	\begin{flalign}
		 &\underset{\Phi_{\text R}}{\min} \quad  \sum_{i \in \mathcal I}\left(C^{+}_{i} r^{+}_{i} + C^{-}_i r^{-}_{i}\right) \label{eq:A-Res-obj}	 
	\end{flalign}
	\begin{flalign}
	    \mathrm{s. t.}\ & \sum_{i \in \mathcal{M}_{a}^{\mathcal I}} r_{i}^{+} + \sum_{e\in\mathcal E} \mathcal{H}(e,a) r_{e}^{+} \geq RR_a^{+}, \quad \forall a\in\AC, \label{eq:A-Res-req-up}\\
	    & \sum_{i \in \mathcal{M}_{a}^{\mathcal I}} r_{i}^{-} + \sum_{e\in\mathcal E} \mathcal{H}(e,a) r_{e}^{-} \geq RR_a^{-}, \quad \forall a\in\AC,  \label{eq:A-Res-req-dn}\\
		& 0 \le r_{i}^{+} \le  {R}^{+}_{i}, \quad \forall i\in \mathcal I, \quad 0 \le r_{i}^{-} \le  {R}^{-}_{i}, \quad \forall i\in\mathcal I, \label{eq:A-Res-cdn-max} \\%
		& -\chi_{e} T_{e}\le r_{e}^{+} \le  \chi_{e} T_{e}, \quad \forall e \in \mathcal E, \quad
		-\chi_{e} T_{e} \le r_{e}^{-} \le \chi_{e} T_{e}, \quad \forall e \in \mathcal E, \label{eq:A-Res-X2} 
	\end{flalign}
\end{subequations}
where $\Phi_{\text R} = \{ r^{+}_{i},r^{-}_{i}, \forall i ; r_{e}^{+}, r_{e}^{-}, \forall e  \} $ is the set of optimization variables. The objective function \eqref{eq:A-Res-obj} to be minimized is the cost of reserve procurement. 
Constraints \eqref{eq:A-Res-req-up} and \eqref{eq:A-Res-req-dn} ensure, respectively, that the upward and downward reserve requirements of each area are satisfied either by procuring reserve capacity from intra-area generators or via inter-area reserves exchange that is modeled using the incidence matrix $\mathcal{H}(e,a)$. 
As shown in \cref{fig:notation}, for each link $e$ with sending and receiving ends in areas $a_s(e)$ and $a_r(e)$, respectively,
$\mathcal{H}(e,a)$ is equal to 1 (-1) if reserve import (export) is considered from (to) area $a=a_s(e)$ ($a=a_r(e)$) and zero for any other area. {With this definition, availability of cross-border reserves within the neighboring areas for each area~$a$ is modeled by \eqref{eq:A-Res-req-up} and \eqref{eq:A-Res-req-dn}.} We underline that directed links are used as a notational convention, and both $r_e^+$ and $r_e^-$ are free of sign.
Upward and downward capacity offers of dispatchable power plants are enforced by constraints \eqref{eq:A-Res-cdn-max}. {In the numerics, these capacities will be chosen such that both upward/downward reserves can be procured from the power plants in a feasible manner when needed in real-time.} 
The set of constraints \eqref{eq:A-Res-X2} models the bounds on reserves exchange between two areas across link $e$. 

Following the current practice, we consider a zonal network representation for the reserve capacity markets and thus the transmission capacity ${T}_e$ of link $e$ is defined as the aggregated flow limit of all tie-lines $\ell \in \Lambda_{e}$ across link~$e$ calculated as
	$${T}_e = \sum_{\ell \in \Lambda_{e}} \ {T}_{\ell},$$ for all $e\in\mathcal E$. Setting the transmission capacity allocation $\chi_e$ to any value different than zero, establishes practically a reserve exchange mechanism between the areas located at the two ends of the link and consequently it enables the exchange of balancing services during real-time operation. On the contrary, setting $\chi_e=0$ implies that there would be no reserve exchange at the procurement stage, that is, the cross-border transmission capacity is fully allocated to day-ahead energy exchanges. In that case, we also prevent the exchange of balancing services and the imbalance netting between the adjacent areas, as we will formally describe in the balancing market model formulation below.

\subsection{Day-ahead market}

Given the optimal reserve procurement $\hat \Phi_{\text R} = \{ \hat r^{+}_{i},\hat r^{-}_{i}, \forall i ; \hat r_{e}^{+}, \hat r_{e}^{-}, \forall e  \}$, the day-ahead schedule is the solution to the following optimization problem:  
\begin{subequations}\vspace{.5cm}
	\label{mod:A-DA}
	\begin{flalign}
		&\underset{\Phi_{\text D}}{\min} \quad  \sum_{i\in \mathcal I}{C}_{i}{p}_{i} \label{eq:DA-Obj}
	\end{flalign}
	\begin{flalign}	
		\mathrm{s. t.}\ & \sum_{j \in \mathcal{M}_n^{\mathcal J}} w_j + \sum_{i\in \mathcal{M}_n^{\mathcal I}}{p}_{i} - \sum_{\ell \in \mathcal L} A_{\ell n} f_{\ell}   = D_n, \quad \forall n\in \mathcal N,  \label{eq:A-CV-da-bal-constr} \\
		&  \hat{r}_{i}^{-} \le p_i \leq \ {P}_{i}- \hat{r}_{i}^{+},   \quad \forall i\in\mathcal I,\quad 0 \le w_j \leq \overline{{W}}_{j},  \quad \forall j\in\mathcal J,  \label{eq:A-CV-da-conv-max}\\
		& f_{\ell} = B_{\ell} \sum_{n\in\mathcal N} A_{\ell n} \delta_n, \quad \forall \ell \in \mathcal L,\label{eq:A-CV-da-flowAC-def}\\  &-(1- \chi_{\ell}) \ {T}_{\ell} \le f_{\ell} \le (1- \chi_{\ell}) \ {T}_{\ell}, \quad \forall \ell \in \mathcal L , \label{eq:A-CV-da-flowAC-constr}\\
		& \delta_1 = 0, \quad \delta_n \; \text{free}, \quad \forall n \in \mathcal{N}, \label{eq:A-CV-da-RefBus-constr}\vspace{.5cm}
	\end{flalign}
\end{subequations}where $\Phi_{\text D} = \{p_i, \forall i; w_{j}, \forall j; \delta_n, \forall n; f_{\ell}, \forall \ell \}$ is the set of variables. We define $\chi_\ell= \chi_e$ for all tie-lines $\ell \in\Lambda_e$ and $ \chi_\ell=0$ for all intra-area lines. For the remainder of this part, we strictly follow this notation. The objective is the day-ahead cost of energy production. Constraints \eqref{eq:A-CV-da-bal-constr} enforce the day-ahead power balance for each node. The upper and lower production limits of dispatchable power plants are enforced by~\eqref{eq:A-CV-da-conv-max}, taking into account the reserve schedule from the previous trading floor. Constraints \eqref{eq:A-CV-da-conv-max} also limit the stochastic production to a point forecast, typically the expected value of the stochastic process. Power flows are first computed in \eqref{eq:A-CV-da-flowAC-def} and then restricted in \eqref{eq:A-CV-da-flowAC-constr} by the capacity limits considering that $(1- \chi_{\ell})$ percent of the capacity is available for day-ahead energy trade. {This is because a $\chi_{\ell}$ portion of the cross-border transmission resource made available is withdrawn from day-ahead energy to instead be allocated to reserves, following the regulations \cite{ENTSOE2}.

As previously discussed, the current practice is in fact a zonal hybrid where the flows are constrained by the available transmission capacity or a flow-based domain. In the case of a zonal market, this could be a portion of the cross-border transmission capacities previously computed for the day-ahead market respecting the minRAM rule of $70\%$ in~\cite[Article 16(8a-8b)]{eur19}~\cite[\S 2]{elia20}. (Note that the minRAM rule may also be applied after the reservation of a transmission share for reserves exchange.) Our method can also be readily adapted to the zonal scheme. We present in detail the integration of a zonal market into our method in the appendix in~\cref{app:zonal}. Finally, the voltage angle at node~$1$ is fixed to zero in \eqref{eq:A-CV-da-RefBus-constr} setting this as the reference node, whereas the remaining voltage angles are declared as free~variables. }

\subsection{Balancing market}

Being close to real-time operation uncertainty realization $s'$ and actual wind power production $W_{js'}, \; \forall j \in \mathcal{J}$ are known. Any energy deviations from the optimal day-ahead schedule $\hat \Phi_{\text D} = \{\hat p_i, \forall i;\allowbreak \hat w_{j}, \forall j; \hat \delta_n, \forall n;\allowbreak \hat f_{\ell},\allowbreak \forall \ell \}$ must be contained using proper re-dispatch actions that respect the reserve procurement schedule $\hat \Phi_{\text R}$.
To determine the re-dispatch actions that minimize the balancing cost, the balancing market is cleared based on the following optimization problem: 
\begin{subequations}
	\label{mod:A-CV-rt}
	\begin{flalign}
		&\underset{ \Phi_{\text B}^{s'} }{\min} \quad  \sum_{i \in \mathcal I}{C}_{i}\left({p}_{is'}^{+}-{p}_{is'}^{-}\right)+\sum_{n \in \mathcal N} {C}^{\text{sh}}l^{\text{sh}}_{ns'} \label{eq:A-CV-rt-Obj}
	\end{flalign}
	\begin{flalign}
		&\mathrm{s. t.}\nonumber\\ & \sum_{i \in \mathcal{M}_n^{\mathcal I}} \left(p_{is'}^{+}-p_{is'}^{-} \right) + l^{\text{sh}}_{n s'} + \sum_{j \in \mathcal{M}_n^{\mathcal J}}\left(W_{j s'}-\hat{w}_j-w^{\text{spill}}_{j s'}\right)
		+ \sum_{\ell \in \mathcal{L}} A_{\ell n} \left( \hat{f}_{\ell} - {f}_{\ell s'} \right) = 0, \ \forall n\in\mathcal N,   \label{eq:A-CV-rt-bal-constr} \\	
		& 0 \le p^+_{is'} \leq  \hat{r}_{i}^{+}, \quad \forall i\in\mathcal I, \quad 0 \le p^-_{is'} \leq  \hat{r}_{i}^{-} , \quad \forall i\in\mathcal I, \label{eq:A-CV-rt-conv-res-dn-max}\\
		& 0 \le l^{\text{sh}}_{n s'} \leq D_n, \quad \forall n\in\mathcal N,  \quad 0 \le w^{\text{spill}}_{j s'} \leq W_{j s'}, \quad \forall j\in\mathcal J,  \label{eq:A-CV-rt-wind-spill-max}\\
		& {f}_{\ell s'} = B_{\ell} \sum_{n\in\mathcal N} A_{\ell n} {\delta}_{n s'}, \quad \forall \ell \in \mathcal L, \label{eq:A-CV-rt-flowAC-def}\\ &- T_{\ell} \le {f}_{\ell s'} \le T_{\ell}, \quad \forall \ell \in \mathcal L, \label{eq:A-CV-rt-flowAC-constr}\\
		& f_{\ell s'} = \hat f_\ell,\quad\forall \ell \in \cup_{e\in\mathcal E^-(\chi)} \Lambda_e,\label{eq:A-CV-rt-tieline-1}\\ &{\delta}_{1 s'} = 0, \quad \delta_{n s'} \; \text{free}, \quad \forall n \in \mathcal{N}, \label{eq:A-CV-rt-tieline}
	\end{flalign}
\end{subequations}
where $\Phi_{\text B}^{s'}\allowbreak =\allowbreak \{ {p}^{+}_{i s'},\allowbreak{p}^{-}_{is'}, \allowbreak\forall i; \allowbreak w^{\text{spill}}_{js'},\allowbreak \forall j; {l}^{\text{sh}}_{ns'},\allowbreak {\delta}_{ns'}, \forall n; {f}_{\ell s'},\forall \ell \}$ is the set of variables. The objective is the cost of re-dispatch actions, that is, reserve activation and load shedding. {Up and down re-dispatch {actions} have the same {cost} as their day-ahead {energy market counterpart, under the assumption that} the reserve market {price} is enough to compensate for the opportunity cost from withdrawing capacity from the day-ahead stage, see the assumption in~\cref{sec:seqassm}.} Equality constraints \eqref{eq:A-CV-rt-bal-constr} ensure that all the nodes remain in balance after the re-dispatch of generation and any necessary wind power curtailment or load shedding. Constraints \eqref{eq:A-CV-rt-conv-res-dn-max} ensure that upward and downward reserve deployment respects the corresponding procured quantities. The upper bounds on load shedding and power spillage are set equal to the nodal demand and the realized wind power production by constraints \eqref{eq:A-CV-rt-wind-spill-max}. Real-time  power flows are first modeled in \eqref{eq:A-CV-rt-flowAC-def} and then restricted by the transmission capacity limits in \eqref{eq:A-CV-rt-flowAC-constr}. 

{Constraints \eqref{eq:A-CV-rt-tieline-1}, where $\mathcal E^-(\chi)=\{e\in\mathcal E \rvert  \chi_e=0\}$ denotes the set of inter-area links with no existing cross-border agreement across them, ensure that if $\chi_e=0$,  the real-time flows on the tie lines are fixed to their day-ahead values. In the existing market framework, the balance responsible parties (in our case, the area operators) are entitled to maintain their scheduled day-ahead net positions in the real-time market~\cite{ent14}, \cite[Article 17]{eur17a}. Since our day-ahead market is modeled by a nodal market and $\hat{f}_\ell$ are already well-defined, we translate this regulation as preventing any reserve sharing or imbalance netting during real-time operation across any line within link $e$ if $\chi_e=0$. For this constraint/requirement, see \cite[\S 2.3]{elia20}, \cite[\S 4.A.7]{solisanalysis}.} Node~$1$ is again the reference node in~\eqref{eq:A-CV-rt-tieline}.
\looseness=-1
\section{Appendix}

{\subsection{Zonal day-ahead market models}\label{app:zonal}
The European electricity market exhibits discrepancies between the day-ahead and real-time representations of the physical system in the electricity market. Day-ahead markets schedule consumption and production using a zonal representation of the underlying nodal electricity network. Such zonal aggregations of the grid allows market participants to trade freely withing each zone and to export/import energy to/from other zones up to certain flow limitations. For the ease of notation, the zones of the day-ahead market are defined to be the same as the areas defined for the reserve market. In the remainder, this appendix presents the models for
the current day-ahead practice in Europe, which is a zonal hybrid where the flows are constrained by the available transmission capacity (ATC) or a flow-based domain. We show how our methods can be readily adapted to these models.

\subsubsection*{Zonal day-ahead market model with ATC}
In zonal electricity markets with ATC, the flow limitations are imposed on the exchanges between neighboring areas. In other words, areas correspond to vertices in a transportation network where every pair of areas connected by a transmission line in the grid are connected by an edge as in Figure~\ref{fig:notation}. The flows through the edges are then limited by the ATCs~\citep{es05}. The zonal market clearing with ATC is formulated as:
\begin{subequations}
	\begin{flalign}
		&\underset{\Phi_{\text D}^{\text{ATC}}}{\min} \quad  \sum_{i\in \mathcal I}{C}_{i}{p}_{i} 
	\end{flalign}
	\begin{flalign}	
		\mathrm{s. t.}\ & \sum_{j \in \mathcal{M}_a^{\mathcal J}} w_j + \sum_{i\in \mathcal{M}_a^{\mathcal I}}{p}_{i} +\sum_{e\in\mathcal E} \mathcal{H}(e,a) f_{e}   = \sum_{n \in\mathcal{M}_a^{\mathcal N}}D_n, \quad \forall a\in \mathcal A,  \\
		&  \hat{r}_{i}^{-} \le p_i \leq \ {P}_{i}- \hat{r}_{i}^{+},   \quad \forall i\in\mathcal I,\quad 0 \le w_j \leq \overline{{W}}_{j},  \quad \forall j\in\mathcal J, \\
		&  -(1- \chi_{e}) \ {ATC}_{e}^- \le f_{e} \le (1- \chi_{e}) \ {ATC}_{e}^+, \quad \forall e \in \mathcal E,\label{eq:minramatc}
	\end{flalign}
\end{subequations}
where ${ATC}_{e}^-$ and ${ATC}_{e}^+$ are the backward and forward exchange limits defined over the link $e\in\mathcal E$. These ATC limits are computed by TSOs respecting the minimum remaining available margin (minRAM) rule of $70\%$ following the requirements of the Clean Energy Package~\citep[Article 16(8a-8b)]{eur19},~\citep[\S 2]{elia20}, and they are supposed to reflect the maximum energy that can be transferred from one zone to the other. Their computation is outside the scope of this part, we kindly refer to~\citep{jensen2017cost}. As an alternative to \eqref{eq:minramatc}, the ATC computation respecting the minRAM rule may also be implemented after the reservation of a transmission share for reserves exchange is implemented to the actual exchange limits $T_e$ of a nodal representation.}

{Notice that the objective function is identical to that of the nodal electricity market, whereas the balance constraints are imposed over each area instead of over each node. Since this problem does not consider the real network, the optimal day-ahead schedule of the market above ${\hat\Phi_{\text D}^{\text{ATC}}}$  might be infeasible for the real network in \eqref{mod:A-DA}. Re-dispatch measures modifying ${\hat\Phi_{\text D}^{\text{ATC}}}$ might then be required in order to recover implementable schedules even when there is no forecast uncertainty. Hence, using this model could potentially increase the need for load shedding and wind curtailment. Replacing the day-ahead market with the above model has implications also for the other models. In the reserve capacity market of~\eqref{mod:A-Res}, the bounds on the reserve exchange between two areas should be replaced with $\chi_{e} {ATC}_{e}^-$ and $\chi_{e} {ATC}_{e}^+$. Moreover, the balancing market of~\eqref{mod:A-CV-rt} should prevent reserve sharing or imbalance netting across any link~$e$ with $\chi_e=0$ (instead of all lines within this link~$e$), since the above market model defines only the total day-ahead link flows but it does not define the individual day-ahead line flows. Finally, the preemptive transmission allocation model of the following chapter can easily  be updated with these three models, since they are all given by linear programs when transmission capacity allocations are fixed. }

{\subsubsection*{Zonal day-ahead market model with flow-based domain}
Zonal electricity market model with a flow-based (FB) domain tries to find a middle ground between nodal model and the ATC model above~\citep{solisanalysis}. The main idea is to approximate the flow on each line as $$f_\ell\sim f_{\ell}^0+\sum_{a\in\mathcal A}PTDF_{\ell,a}({NP}_a-{NP}_a^0),$$ where $f_{\ell}^0$ is the flow through line $\ell$ on a base case, ${NP}_a$ is the net position of area $a$ and ${NP}_a^0$ is the net position on the base case, and $PTDF_{\ell,a}$ are area-to-line power-transfer-distribution-factors computed by TSOs. The zonal market clearing with FB domain is formulated as:
\begin{subequations}\label{eq:zonfb}
	\begin{flalign}
		&\underset{\Phi_{\text D}^{\text{FB}}}{\min} \quad  \sum_{i\in \mathcal I}{C}_{i}{p}_{i} 
	\end{flalign}
	\begin{flalign}	
		&\mathrm{s. t.}\nonumber\\ & \sum_{j \in \mathcal{M}_a^{\mathcal J}} w_j + \sum_{i\in \mathcal{M}_a^{\mathcal I}}{p}_{i} -\sum_{n \in\mathcal{M}_a^{\mathcal N}}D_n   = {NP}_a, \quad \forall a\in \mathcal A,\\
		&  \hat{r}_{i}^{-} \le p_i \leq \ {P}_{i}- \hat{r}_{i}^{+},   \quad \forall i\in\mathcal I,\quad 0 \le w_j \leq \overline{{W}}_{j},  \quad \forall j\in\mathcal J, \\
		&  -(1- \chi_{\ell}) \ T_{\ell} \le f_{\ell}^0+\sum_{a\in\mathcal A}{PTDF}_{\ell,a}({NP}_a-{NP}_a^0) \le (1- \chi_{\ell}) \ T_\ell, \quad \forall \ell \in \mathcal L, \label{eq:linlimzon}\\
		&\sum_{a\in\AC} {NP}_a=0.
	\end{flalign}
\end{subequations}
The last two sets of constraints impose export and import limitations directly on the configuration of net positions of areas, defining what is called the FB domain. As is the case with the ATCs,  the optimal day-ahead schedule of the market above might be infeasible for the real network in~\eqref{mod:A-DA}. Re-dispatch measures modifying ${\hat\Phi_{\text D}^{\text{FB}}}$ might then be required in order to recover implementable schedules even when there is no forecast uncertainty. Replacing the day-ahead market with the above model again has implications also for the other models. Note that the limits $T_\ell$ are computed by TSOs respecting the minimum remaining available margin (minRAM) rule of $70\%$ following the requirements of the Clean Energy Package~\citep[Article 16(8a-8b)]{eur19},~\citep[\S 2]{elia20}, and they could be different from those used in~\eqref{mod:A-Res}. Their computation is outside the scope of this part, we kindly refer to~\citep[\S 2]{elia20}. In this case, for the market in~\eqref{mod:A-Res}, the bounds on the reserve exchange between two areas should be replaced with the new values ${T}_e = \sum_{\ell \in \Lambda_{e}} \ {T}_{\ell},$ for all $e\in\mathcal E$. Moreover, instead of preventing reserve sharing or imbalance netting across links, the balancing market of~\eqref{mod:A-CV-rt} can enforce that day-ahead zonal net positions are maintained in the real-time market for all $a\in\{a'\,|\, \chi_e=0,\,\forall e: \mathcal{H}(e,a')\neq0\}$, since the solution to the above market model defines ${NP}_a$. Finally, the preemptive transmission allocation model of the following chapter can again be easily updated with these three models, since they are all given by linear programs when transmission capacity allocations are fixed.}

{We now provide a brief discussion on the parameter choices. Significant attention has been dedicated towards understanding how discretionary parameters determined by TSOs affect the outcome of the zonal market with FB domain. These parameters include the selection of a base case, but also the determination of $PTDF_{\ell,a}$ requires additional parameters called Generation Shift Keys (GSKs) for disaggregating zonal injections into nodal injections. Morover, the actual methodology is even more complex than \eqref{eq:zonfb}. For instance, a link has a limit as in \eqref{eq:linlimzon} only if it is considered to be critical, e.g, if $PTDF_{\ell,a}$ is larger than $5\%$~\citep[\S 2]{elia20}~\citep[\S 3.2.2.1]{solisanalysis}. \cite{mltw13} show that the same system can lead to very different market outcomes depending on all of these choices. Defining such parameters is outside the scope of this part, we kindly refer to~\citep{solisanalysis}.}

\chapter{Transmission capacity allocations for cross-border balancing}\label{sec:transcapalloc}
In this chapter, we discuss  the  issues  closely connected to reserve exchanges and we then formulate the preemptive transmission allocation model.
\section{Possible coordination schemes and transmission allocation arrangements}

The transition to an integrated balancing market requires several organizational changes to the prevailing operational model, in which reserves are procured and deployed on an intra-area basis.
A prerequisite for the establishment of a well-functioning balancing framework is the standardization of the rules and
products as well as the definition of transparent mechanisms that will facilitate the cooperation among the TSOs \cite{hobbs2005more}. Below, we outline the main coordination schemes and transmission allocation arrangements as defined in the current European regulation~\cite{EC1}.

Inter-area reserve procurement can be organized as a reserve exchange scheme {and}/{or} as a reserve sharing agreement. Implementing the former scheme, regional TSOs can procure balancing capacity resources located in adjacent areas in order to meet their own area reserve requirements. 
Since the reserve requirements of each area remain unchanged, this coordination setup requires limited organizational changes, as it basically reallocates the reserve quantities towards areas with lower procurement costs. This setup is the focus of our studies in this part. To improve also the dimensioning efficiency of the procurement process, a reserve sharing agreement allows a TSO to use available reserve capacity from adjacent TSOs. An implied prerequisite for this arrangement would be that the definition of regional reserve requirements is performed jointly by all TSOs that participate in the sharing agreement. 

In terms of coordination during reserves activation, the main organizational setups are the so-called imbalance netting and exchange of balancing energy. The first setup pertains to the out-of-market inter-area exchange of imbalances with opposite sign, thus preventing the counteracting activation of balancing resources and reducing the total balancing energy volumes. 
In turn, the exchange of balancing energy enables the system-wide least-cost activation of reserves through a common merit-order list to meet the net imbalance of the joint TSO area. This improves the supply efficiency of balancing energy, at the expense of more extensive coordination requirements.

The establishment of any cross-border reserve procurement scheme requires the reservation of a certain share of the inter-area transmission capacity from the day-ahead market for the reserves and their activation~\cite[Art. (14)]{eur19}. Such a reservation increases the cost of the day-ahead market, but in return decreases the costs for the reserve market and the balancing market.

Before we describe the attributes of any specific transmission allocation mechanism, let us provide an illustrative example for the resulting total cost from all three stages. This example highlights the seams issues pertaining to the ex-ante definition of transmission allocation between two neighboring areas.
\cref{fig:plotta} shows the expected system cost, that is, the sum of reserve procurement, day-ahead energy and expected balancing costs, as a function of the share of transmission capacity $\chi$ that is allocated to inter-area reserves trading. The data for this two-area power system is provided in the appendix in~\cref{app:two_area_ex}, and the models are those that were discussed in the previous chapter. {We can observe that the efficiency of an integrated market, in terms of expected system cost defined above, is highly susceptible to the portion~$\chi$ of transmission capacity removed from the day-ahead market.} Moreover, its optimal value minimizing the cost changes significantly under different levels of wind power penetrations.
Even from this simple example, it becomes apparent that there exists an optimal allocation to be made, which however may dynamically vary depending on generation, load and system uncertainties. 
This in turn asks for a systematic method to optimally define~$\chi$, accounting for the market dynamics and the uncertainty involved in the operation of the power system.
\looseness=-1

\begin{figure}[h]
	\centering
	\begin{tikzpicture}[scale=1.02, every node/.style={scale=0.755}]
			\draw (4,1) circle (.35cm) node {\LARGE $a_1$};
			\draw (4,-1) circle (.35cm) node {\LARGE $a_2$};
			\draw[-,line width=.25mm] (4,.65) -- (4,-.65) node  at (4.4, -0.1) {\LARGE $\chi$};
			\node at (-.85,0) {\includegraphics[width=.6\textwidth, trim={3.4cm 9.7cm 3.9cm 9.8cm},clip]{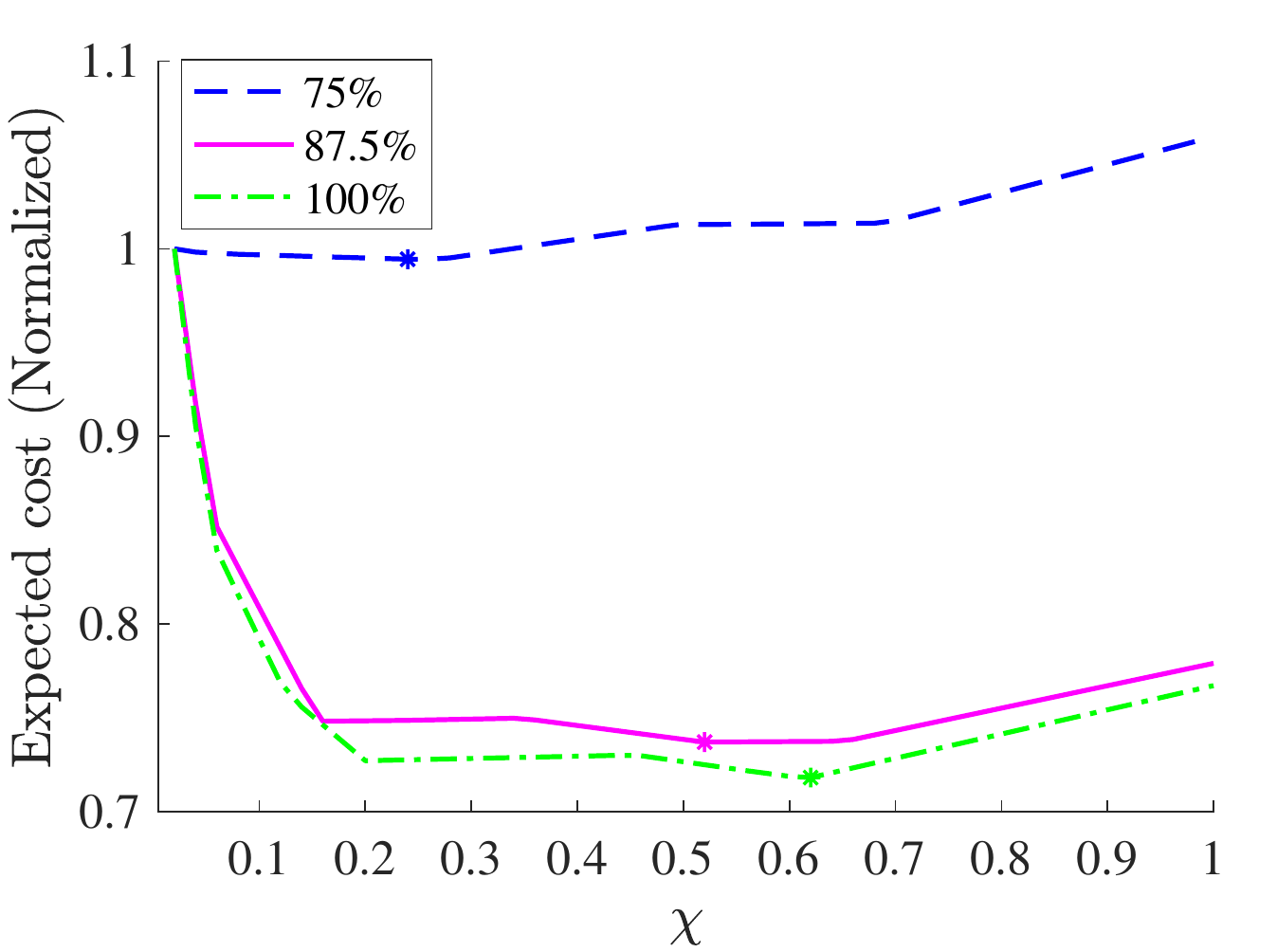}};
			\fill[fill=white] (-2.93,2.25) rectangle (-2.07,1.4);
			\node at (-2.45,1.58) {\footnotesize $312\; \text{MW}$};
			\node at (-2.45,1.87) {\footnotesize $273\; \text{MW}$};
			\node at (-2.45,2.15) {\footnotesize $234\; \text{MW}$};
	\end{tikzpicture}\vspace{-.5cm}
	\caption{Expected operation cost as a function of transmission capacity allocated to inter-area reserves trading under different levels of wind power penetration (in MW). Asterisk symbols correspond to the optimal values of transmission capacity allocations.}\label{fig:plotta}
\end{figure}
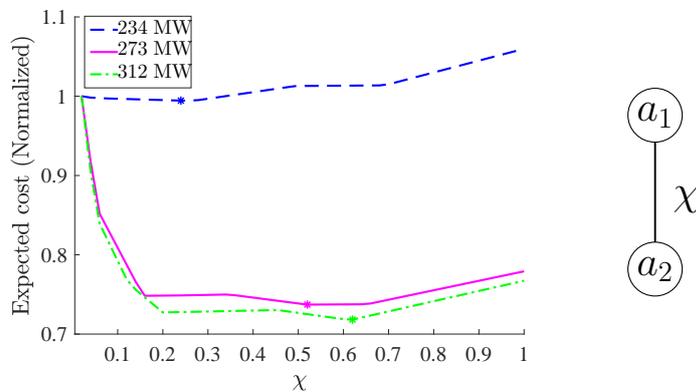

In this part, we focus on the prevailing market-based mechanism for the allocation of cross-border transmission capacity from day-ahead energy to reserves. According to this methodology,
a share of inter-area transmission capacity is set aside from day-ahead energy for reserves based on the comparison of the market value of cross-zonal capacity for the exchange of balancing capacity or sharing of reserves and the market value of cross-zonal capacity for the exchange of day-ahead energy. {This methodology can attain a reasonable allocation efficiency while having practical applicability within the current framework, albeit it still incurs the inherent drawbacks of the sequential market structure regarding the deterministic view of uncertainty and the separation of energy and reserve services.} We refer the interested reader to \cite{EC1} for discussions on alternative transmission allocation mechanisms.

\section{Preemptive transmission allocation model}
In this section, we first describe the preemptive transmission allocation model that was initially proposed in~\cite{delikaraoglou2018optimal} as a market-based mechanism. This work, however, defines the preemptive model in a more general framework, which allows us to consider coalitional deviations and in turn define the necessary benefits that support the solution proposed by the preemptive model. The motivation for considering coalitional deviations originates from the Clean Energy Package regulation which states that an application for a methodology of allocating cross-border capacity to reserves can be filed by even two neighboring operators~\cite[Art. (14)]{eur19}. 

The preemptive transmission allocation model can be perceived as a decision-support tool, which aims at defining the optimal shares of transmission capacity for inter-area trading of energy and reserves. Being fully aligned with the existing sequential market structure, the preemptive model is essentially a market-based allocation process that is performed prior to the reserve capacity and day-ahead energy markets to find the optimal transmission allocations $\{\hat\chi_e, \forall e\}$ that minimize the expected system cost, see~\cref{fig:struct} for a schematic representation. It is worth mentioning that the coordinated reserve exchange would require a transfer of some responsibilities (e.g., transmission capacity computations) to European bodies, even though some TSOs might be hesitant to assign some of their autonomy to a central authority.

\begin{figure}[h]
	\centering
	\begin{tikzpicture}[scale=0.79, every node/.style={scale=0.52}]
	\draw[dashed, draw=black] (-5,-2.5) rectangle (-1.35,0.95);
	\draw[draw=black,rounded corners] (-4.8,-2.3) rectangle (-1.55,0.75) node at (-3.15,-0.4) {\Large Preemptive model} node at (-3.15,-1) {\Large Problem~\eqref{mod:B-Pre}};
	\draw[draw=black,rounded corners] (.7,-0.7) rectangle (3.95,0.75) node at (2.35,0.4) {\Large Res. cap. market} node at (2.35,-0.2) {\Large Problem~\eqref{mod:A-Res}};
	\draw[draw=black,rounded corners] (6.1,-1.3) rectangle (9.35,0.2)
	node at (7.75,-0.2) {\Large Day-ahead market} node at (7.75,-.8) {\Large Problem~\eqref{mod:A-DA}};
	\draw[draw=black,rounded corners] (11.7,-2.3) rectangle (14.95,0.75) node at (13.35,-0.1) {\Large Balancing market} node at (13.35,-0.65) {\Large Problem~\eqref{mod:A-CV-rt}} node at (13.35,-1.15) {\Large for scenario $s'$};
	\draw[->,line width=.1mm] (-1.55,0) -- (.7,0) node at (-.4,-.5) {\LARGE $\{\hat\chi_e, \forall e\}$};
	\draw[->,line width=.1mm] (3.95,-0) -- (6.1,-0) node at (5.1,-.5) {\LARGE $\left(\begin{subarray}{c} \hat{r}^{+}_{i},\hat{r}^{-}_{i}, \\  \hat{r}_{e}^{+}, \hat{r}_{e}^{-} \end{subarray} \right)$};
	\draw[->,line width=.1mm] (9.35,-0.75) -- (11.7,-.75) node at (10.7,-1.25) {\LARGE $\left(\begin{subarray}{c} \hat{p}_{i}, \hat{w}_{j}, \\  \hat{\delta}_n, \hat{f}_{\ell} \end{subarray} \right)$};
	\draw[->,line width=.1mm] (-1.55,-1.2) -- (6.1,-1.2) node at (2.75,-1.6) {\LARGE $\{1-\hat\chi_e, \forall e\}$};
	\draw[->,line width=.1mm] (3.95,.5) -- (11.7,.5) node at (10.7,0) {\LARGE $\left(\begin{subarray}{c} \hat{r}^{+}_{i},\hat{r}^{-}_{i}, \\  \hat{r}_{e}^{+}, \hat{r}_{e}^{-} \end{subarray} \right)$};
	\draw[->,line width=.1mm] (-1.55,-2) -- (11.7,-2) node at (10.4,-2.35) {\Large $\CC\subseteq\AC$};
	\end{tikzpicture}\vspace{-.2cm}
	\caption{Schematic representation of preemptive transmission allocation model}\label{fig:struct}
\end{figure}
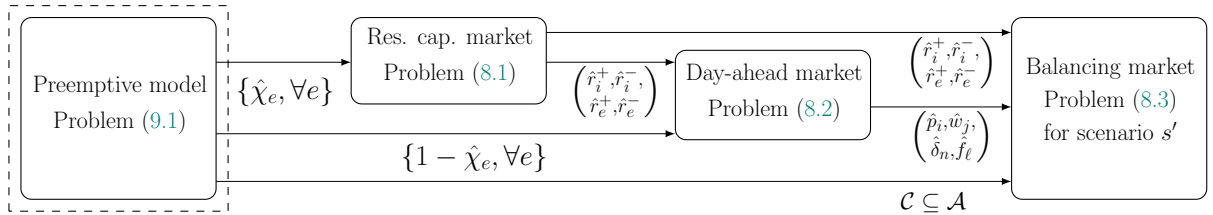

In this work, the focus is on the establishment of coalitional agreements at the reserve procurement stage among a set of areas $\CC\subseteq\AC$. In the balancing market, we assume that all areas that participate in the coalition, exchange balancing energy in a perfectly coordinated setup.
 As a result, real-time tie-line flows are treated as free variables, allowing for deviations from the day-ahead schedule. 
 The coalition-dependent preemptive model is given by:
\begingroup
\allowdisplaybreaks
\begin{subequations} \label{mod:B-Pre}
\begin{align}
 {J(\CC)}&=\,\underset{ \Phi_{\text{PR}} }{\min}\quad \sum_{i \in \mathcal I}\left(C^{+}_{i} r^{+}_{i} + C^{-}_i r^{-}_{i}\right) + \sum_{i\in \mathcal I}{C}_{i}{p}_{i}\nonumber \\ &\quad\quad\quad\quad\quad\quad\quad\quad+ \sum_{s\in\mathcal S} \pi_s \Big[ \sum_{i \in \mathcal I}{C}_{i}\left({p}_{is}^{+}-{p}_{is}^{-}\right)+\sum_{n \in \mathcal N} {C}^{\text{sh}}l^{\text{sh}}_{ns} \Big]  \label{eq:B-pre-obj}
\end{align}
\begin{align}
\mathrm{s. t.}\ & \;\;\;  0\leq\chi_{e}'\leq 1,\quad \forall e\in\mathcal E_{\CC},\label{eq:B-pre-ta}\\
& \;\;\;  \chi_{e}'=\chi_{e},\quad \forall e\in\mathcal E\setminus\mathcal E_{\CC},\label{eq:B-pre-ta-fixed}\\
& \;\;\;  \text{Constraints} \;\; \eqref{eq:A-CV-rt-bal-constr} -\eqref{eq:A-CV-rt-flowAC-constr}\text{ and }{\delta}_{1 s} = 0, \ \delta_{n s} \; \text{free}, \ \forall n \in \mathcal{N},\quad \forall s\in \mathcal S,\label{eq:B-A-CV}\\
& \;\;\;  {{f}_{\ell s} = {f}_{\ell},\quad \forall \ell\in\cup_{e\in\mathcal E^-(\chi,\CC)} \Lambda_{e},\quad\forall s\in \mathcal S},\label{eq:B-A-CV-flow} \\
& \;\; \left(\begin{subarray}{c} r^{+}_{i},r^{-}_{i}, \\  r_{e}^{+}, r_{e}^{-} \end{subarray} \right)  \in \text{arg}
	\left\{\!\begin{aligned}
	&  \underset{\Phi_{\text{R}}}{\min} \quad \eqref{eq:A-Res-obj} \quad \mathrm{s. t.}\ \text{constraints \eqref{eq:A-Res-req-up} - \eqref{eq:A-Res-X2}}
	\end{aligned}\right\}, \label{eq:LLR} \\
	& \;\;\; \left(\begin{subarray}{c} p_{i}, w_{j}, \\  \delta_n, f_{\ell} \end{subarray} \right)  \in \text{arg} \;
	\left\{\!\begin{aligned}
	&  \underset{\Phi_{\text{D}}}{\min} \quad \eqref{eq:DA-Obj} \quad \mathrm{s. t.}\
	\text{constraints \eqref{eq:A-CV-da-bal-constr} - \eqref{eq:A-CV-da-RefBus-constr}}
	\end{aligned}\right\}, \label{eq:LLD} 
\end{align}
\end{subequations}
\endgroup
where $\Phi_{\text{PR}}\allowbreak=\allowbreak\{\chi_e',\allowbreak \forall e \cup \Phi_{\text R} \cup \Phi_{\text D} \cup  \Phi_{\text B}^{s}, \allowbreak\forall s\}$ is the set of primal optimization variables. For the sake of {brevity}, the Lagrange multipliers of the lower-level optimization problems are omitted here, but the complete set of Karush-Kuhn-Tucker (KKT) conditions of problems \eqref{mod:A-Res} and \eqref{mod:A-DA} are listed in the appendix in \cref{app:a}. Unless stated otherwise, the preemptive model refers to problem~\eqref{mod:B-Pre} associated with $J(\AC)$ and the corresponding optimal transmission allocation is denoted by $\{\hat\chi_e, \forall e\}$.

Model \eqref{mod:B-Pre} is a stochastic optimization problem, since wind power is described by a finite set of scenarios~$\mathcal{S}$, with $W_{js}$ being the realization of stochastic generation of farm $j$ in scenario $s$ and $\pi_s$ the corresponding probability. The objective is the expected system cost according to the sequential structure described in \cref{sec:seq}.
Constraint~\eqref{eq:B-pre-ta} bounds the share of transmission capacity ${\chi}'_{\textcolor{blue}{e}}\in[0,\,1]^{\mathcal E}$ allocated to reserve exchange for links $e\in\mathcal E_{\CC}$, where $\mathcal E_{\CC}\allowbreak=\allowbreak\{e\,|\, \mathcal{H}(e,a)=0,\,\forall a\in \AC\setminus\CC\}$ is the set of links among only the areas in the coalition $\CC$. The transmission capacity of the remaining links $e\in\mathcal E\setminus\mathcal E_{\CC}$, that is, links that are connected to areas that are not members of coalition $\mathcal{C}$, is fixed according to existing cross-border agreements ${\chi_e}\in[0,\,1]^{\mathcal E}$ in~\eqref{eq:B-pre-ta-fixed}. 
Constraints~\eqref{eq:B-A-CV} ensure feasibility of re-dispatch actions for each scenario, whereas constraint \eqref{eq:B-A-CV-flow} restricts the real-time tie-line flows to the respective day-ahead values, for links $$e \in \mathcal E^-(\chi,\CC)=\{e\in\mathcal E \rvert  \chi_e=0,\ \text{and}\ \exists a\in\AC\setminus\CC\ \text{such}\ \text{that}\ \mathcal{H}(e,a)\neq0\},$$ which is the set of links that are, at least on one end, connected to an area that is not in the coalition~$\CC$ and does not have an existing cross-border agreement for reserves exchange, that is, $\chi_e=0$. In other words, constraint~\eqref{eq:B-A-CV-flow} prevents any imbalance netting or exchange of balancing energy between areas that are not members of coalition $\CC$. Constraints~\eqref{eq:B-A-CV-flow} are removed from problem~\eqref{mod:B-Pre} associated with $J(\AC)$, since $\mathcal E(\chi,\AC)=\emptyset$ for any~$\chi$. 
\looseness=-1

The lower-level problems \eqref{eq:LLR} and \eqref{eq:LLD} are identical to models \eqref{mod:A-Res} and \eqref{mod:A-DA} implementing the shares of transmission capacities in~${\chi}'$. Having this bilevel model \eqref{mod:B-Pre} ensures by construction that the reserve capacity, day-ahead and balancing markets are cleared in consecutive and independent auctions. This structure allows the definition of $\{\hat\chi_e, \forall e\}$ anticipating the impact of these parameters in all subsequent trading floors. From a computational perspective, to obtain a solvable instance of the model \eqref{mod:B-Pre}, we can equivalently replace the lower-level problems \eqref{eq:LLR} and \eqref{eq:LLD} by the KKT conditions, given that \eqref{eq:LLR} and \eqref{eq:LLD} are linear programs. 
The resulting problem is a single-level mathematical program with equilibrium constraints (MPEC) that involves the complementary slackness constraints, which can be transformed into a mixed-integer linear program (MILP) using disjunctive constraints. We refer to the appendix in \cref{app:a}. for the KKT conditions of \eqref{mod:A-Res} and \eqref{mod:A-DA}. 
\looseness=-1

Regarding the structure of the bilevel model \eqref{mod:B-Pre}, in contrast to the reserve capacity and the day-ahead markets, the balancing market is modeled in the upper-level by the last term of the objective~\eqref{eq:B-pre-obj} and constraints \eqref{eq:B-A-CV}-\eqref{eq:B-A-CV-flow}. 
The proposed structure leverages the fact that
 the variables of the balancing market in~\eqref{mod:A-CV-rt} do not enforce any restriction on the upper-level variables $\{ \chi_e',\,\forall e\}$ and also they do not impact the lower-level problems \eqref{eq:LLR}~and~\eqref{eq:LLD}. Following this observation, our formulation reduces the computational complexity of the final MILP, since it avoids the integer reformulation of the balancing market complementarity conditions for each scenario. {We provided a mathematical explanation in the appendix in \cref{app:a}, and for similar applications of bilevel programming, the interested reader is kindly referred to~\cite{pineda2016capacity,morales2014electricity}.}

Having defined the main properties of the preemptive model, the following comments are in order.
Define the sum of the costs~\eqref{eq:A-Res-obj},~\eqref{eq:DA-Obj} and~\eqref{eq:A-CV-rt-Obj}, as $J^{s'}(\emptyset)$ and $J^{s'}(\AC)$, when $\chi_e, \forall e$ are fixed to the existing cross-border arrangements and to the optimal $\{\hat\chi_e, \forall e\}$ from the preemptive model, respectively. From an economic intuition, it follows that the preemptive model reduces the total expected cost since the establishment of broader coalitions enlarges the pool of available reserves and balancing resources, and enables the more efficient allocation of available generation capacity between these services.
This can be mathematically stated as: $$J(\AC)=\mathbb E_s[J^s(\AC)] \leq J(\emptyset) =\mathbb E_s[J^s(\emptyset)],$$ where $\mathbb{E}_s[\cdot] $ is the expectation calculated over the scenario set $\mathcal{S}$.
Using a similar reasoning, it follows that $J$ is a nonincreasing function with respect to the number of areas participating in the coalition, that is, $J(\CC)\leq J(\hat \CC)$ for all $\hat \CC\subseteq\CC.$ (Note that the preemptive model does not guarantee $J^s(\CC)\leq J^s(\hat \CC),$ $\forall\hat \CC\subseteq\CC$ for each scenario independently.) 
It becomes apparent that the implementation of the preemptive model results in a different {system} cost than the one under the current sequential market. Next, we discuss cost allocation in this new sequential market.

\section{Cost allocations under the preemptive transmission allocation model}\label{sec:marketbased}

The preemptive model \eqref{mod:B-Pre} implements a centralized transmission allocation mechanism, under the implicit assumption that all areas are willing to accept the $\{\hat\chi_e, \forall e\}$ solution that by construction minimizes the system-wide expected cost. However, this model does not suggest an area-specific cost allocation that guarantees sufficient benefits for all areas to remain in the grand coalition~$\AC$.\footnote{{As it is introduced earlier in this part and also defined in~\cite{kristiansen2018mechanism}, benefits are the change in total operational cost allocated to a particular area after all three market stages are cleared. The operational cost is equivalent to minus the social welfare (with inelastic demand) and thus it is reasonable to define positive benefits as the reduction of the total operational cost.}} As an alternative methodology, the task of setting the transmission shares and allocating the resulting costs could be accomplished through establishing a new market, cleared before the reserve market, in which regional operators (by also consulting their market participants) would place their bids/offers for the reservation of inter-area transmission capacities, akin to the decision variable of the preemptive model, $\{\chi_e', \forall e\}$. This new market would constitute an ideal benchmark of the market-based allocation process described in \cite{EC1}, implementing a \textit{complete} market for transmission allocations in which capacities would be traded based on bids/offers that reflect the valuations from regional operators.

Since deriving such valuations might be a hard problem for operators, in our studies, we follow another path to promote the formation of stable coalitions for the exchange of reserves. 
Our approach builds an ex-post benefit allocation mechanism on top of the preemptive model, aiming to realize the necessary conditions that accomplish the coordination requirements of this model, without any new marketplace.
In the remainder, we outline the concepts related to benefit allocations for the preemptive model and we discuss the desirable properties that we want to achieve.

Let $J^s_a(\emptyset)$ denote the cost allocated to area~$a$ in scenario~$s$ in the existing sequential market. As previously discussed, the current implementation of the sequential market provides a cost allocation method that satisfies budget balance under every scenario, that is, $J^s(\emptyset)=\sum_{a\in\mathcal A}J^s_a(\emptyset)$. 
The implementation of the preemptive model requires a new method to allocate costs to the areas that participate in this arrangement. This task can equivalently be viewed as allocating benefits based on the change in the total cost as a discount or a mark-up on the original cost allocation of each area defined by $J^s_a(\emptyset)$. 
{While choosing these benefits, our main goal is to ensure that all areas in~$\AC$ are willing to use the preemptive model as a decision-support tool, since otherwise some areas may opt for having their own reserve {exchange} agreement following~\cite[(14)]{eur19}.} In addition, we should aim to form coalitions as large as possible in order to achieve the highest reduction in the expected system cost. To achieve these, we will treat the preemptive model as a coalitional game, which allows us to approach the benefit allocation problem in two ways. First, we can allocate the expected cost reduction, $J(\emptyset)-J(\AC)\geq 0$, to all areas as benefits. Allocating benefits this way achieves budget balance in expectation, which implies that there is no deficit or surplus if the preemptive model is used repeatedly and the uncertainty modeling is accurate enough.\footnote{In practice, the scenario set is inevitably an approximation to the real world. There are various results showing asymptotic guarantees for convex optimization models as long as the scenario set is rich enough~\cite{birge2011introduction}} However, this method does not guarantee that the resulting allocation satisfies budget balance in every scenario, thus requiring a large financial reserve to buffer the fluctuations in the budget in case of surplus or deficit for some realizations.
The second approach is to allocate the scenario-specific cost variation, $J^s(\emptyset)-J^s(\AC)$. This would guarantee budget balance for every scenario.  {However}, the participating areas would {collect} benefits that vary under scenarios, {possibly raising risk considerations}. 

{As a remark, the benefit allocation framework studied in this thesis/part defines these monetary quantities on an area level. They provide each area with an idealized total cost allocation, which would be minus the sum of three terms, that is, the consumers'  and generators' surplus pertaining to that area and the congestion rents collected by the corresponding area operator. We highlight that the methods proposed in the following chapters do not readily define the payment rules for the new sequential market such that these idealized total cost allocations (or the total available surpluses) are distributed on a market participant level, which should be the second step in this analysis (this direction will be discussed in our future work). However, it is possible to provide a guarantee on the cost recovery property of market participants. Later in our work, by picking nonnegative benefits, we in fact guarantee that the available surplus for each area increases. This implies that it is possible to define payment rules to achieve cost recovery, moreover, it is also possible to improve each generator's and consumer's surplus compared to their values in the existing sequential market. {If the preemptive model is used and all market stages are cleared with their new unit prices given by the new Lagrange multipliers, then we can obtain cost recovery on a market participant level. However, such an approach cannot provide any guarantees on an area level at all. Defining side payments to ensure that the cost allocation of each area is close what is suggested by the benefit allocation methods is part of our future research directions.} } 

\section{Appendix}

\subsection{Data for two-area example}\label{app:two_area_ex}
For this example, we removed area 3 from the example we provide in detail in~\cref{sec:illu_ex} while all the other parameters are kept unchanged.
The different levels of wind power penetration are modeled by changing the installed capacities of the wind power plants $j_3$ and $j_6$. $312$ MW corresponds to $120$ and $192$ MW, respectively. $273$ MW corresponds to $105$ and $168$ MW, whereas $234$ MW corresponds to $90$ and $144$ MW. Normalization is done by dividing the expected operation cost by the cost under the same wind power penetration but with $\chi=\epsilon$ where $\epsilon>0$ is a small positive number.
\subsection{KKT conditions}\label{app:a}
We provide the complete set of KKT conditions for the reserve and day-ahead markets in~\eqref{eq:LLR} and~\eqref{eq:LLD}, that appear in the lower level of the stochastic bilevel optimization in \eqref{mod:B-Pre}. The dual multipliers of inequality constraints are listed to the right of the complementarity relationships denoted by $\perp$. For the equality constraints, the dual multipliers are listed after a colon.

The KKT conditions for the reserve market in~\eqref{eq:LLR} are:
\begin{align*}
		& 0 \le  {R}^{+}_{i}- r_{i}^{+}\perp \mu_i^{\text{R}^{+}}\geq 0, \quad \forall i,\\
		& 0 \le  {R}^{-}_{i}-r_{i}^{-}\perp \mu_i^{\text{R}^{-}}\geq 0, \quad \forall i, \\
		 &0 \le \sum_{i \in \mathcal{M}_{a}^{\mathcal I}} r_{i}^{+} + \sum_{e\in\mathcal E} \mathcal{H}(e,a) r_{e}^{+} - RR_a^{+}\perp \mu_a^{\text{RR}^{+}}\geq0,\quad \forall a, \\
	    &0 \le \sum_{i \in \mathcal{M}_{a}^{\mathcal I}} r_{i}^{-} + \sum_{e\in\mathcal E} \mathcal{H}(e,a) r_{e}^{-} - RR_a^{-}\perp \mu_a^{\text{RR}^{-}}\geq0,\quad \forall a, \\
		& 0 \le r_{e}^{+}+\chi_{e}' T_{e}\perp \zeta_e^{\text{L}^{+}}\geq0,\quad \forall e, \\
		& 0 \le  \chi_{e}' T_{e} - r_{e}^{+}\perp \zeta_e^{\text{U}^{+}}\geq0,\quad \forall e, \\
		& 0 \le r_{e}^{-}+\chi_{e}' T_{e}\perp \zeta_e^{\text{L}^{-}}\geq0,\quad \forall e, \\
		& 0 \le \chi_{e}' T_{e}-r_{e}^{-}\perp \zeta_e^{\text{U}^{-}}\geq0,\  \forall e,\\
		&0\le C_i^++\mu_i^{\text{R}^{+}}-\sum_{a: i \in \mathcal{M}_{a}^{\mathcal I}}\mu_a^{\text{RR}^{+}}\perp r_{i}^{+}\geq 0,\quad \forall i, \\
		&0\le C_i^-+\mu_i^{\text{R}^{-}}-\sum_{a: i \in \mathcal{M}_{a}^{\mathcal I}}\mu_a^{\text{RR}^{-}}\perp r_{i}^{-}\geq 0,\quad \forall i, \\
		&\sum_{a} \mu_a^{\text{RR}^{+}}\mathcal{H}(e,a) + \zeta_e^{\text{L}^{+}} - \zeta_e^{\text{U}^{+}} = 0,\quad \forall e,\\
		&\sum_{a} \mu_a^{\text{RR}^{-}}\mathcal{H}(e,a) + \zeta_e^{\text{L}^{-}} - \zeta_e^{\text{U}^{-}} = 0,\quad \forall e.
\end{align*}

The KKT conditions for the day-ahead market in~\eqref{eq:LLD} are:
\begin{align*}
    		& \sum_{j \in \mathcal{M}_n^{\mathcal J}} w_j + \sum_{i\in \mathcal{M}_n^{\mathcal I}}{p}_{i} - \sum_{\ell \in \mathcal L} A_{\ell n} f_{\ell}   = D_n: \lambda_n\ \text{free},\quad \forall n,\\
		&  0 \le p_i-{r}_{i}^{-}\perp\mu_i^{\text{PL}}\geq 0,   \quad \forall i, \\
		&  0 \leq  {P}_{i}- {r}_{i}^{+}-p_i\perp\mu_i^{\text{PU}}\geq 0,   \  \forall i, \\
		& 0 \leq \overline{{W}}_{j}-w_j\perp\mu_j^{\text{WU}}\geq 0,  \quad \forall j,\\
		& f_{\ell} = B_{\ell} \sum_{n\in\mathcal N} A_{\ell n} \delta_n: \lambda_\ell^\text{F}\ \text{free},\quad \forall \ell, \\
		& 0 \le f_{\ell}+(1- \chi_{\ell}') \ {T}_{\ell}\perp\zeta_\ell^{\text{L}}\geq0,\quad \forall \ell ,\\
		&0 \le (1- \chi_{\ell}') \ {T}_{\ell}-f_{\ell}\perp\zeta_\ell^{\text{U}}\geq0,\quad \forall \ell,\\	
		& \delta_1 = 0: \lambda^\text{REF}\ \text{free},\\
		& C_i + \sum_{n:i\in \mathcal{M}_n^{\mathcal I}}\lambda_n-\mu_i^{\text{PL}}+\mu_i^{\text{PU}}=0,\quad \forall i,\\
		&0\leq  \sum_{n:j \in \mathcal{M}_n^{\mathcal J}}\lambda_n +\mu_j^{\text{WU}}\perp w_j\geq 0, \quad \forall j,\\
		&\sum_{n}A_{\ell n}\lambda_n-\lambda_\ell^\text{F}-\zeta_\ell^{\text{L}}+\zeta_\ell^{\text{U}}=0, \quad \forall \ell,\\
		&\sum_{\ell}\lambda_\ell^\text{F}B_{\ell}  A_{\ell n}=0, \quad \forall n\neq 1,\\
		&\sum_{\ell}\lambda_\ell^\text{F}B_{\ell} A_{\ell n}-\lambda^\text{REF}=0, \ n=1.
\end{align*}
Note that the conditions above involve the solutions of the reserve market ${r}_{i}^{+}$ and ${r}_{i}^{-}$, and the transmission allocations $\chi'$ from the optimization variables of preemptive model~\eqref{mod:B-Pre}.

\subsubsection*{Mathematical explanation for the upper-level balancing market formulation}

{As it is previously discussed, an alternative but an equivalent formulation of model~\eqref{mod:B-Pre} could be obtained by including the balancing market as a lower-level problem. However, the proposed structure leverages the fact that the variables of the balancing market in~\eqref{mod:A-CV-rt} do not enforce any restriction on the upper-level variables $\{ \chi_e',\,\forall e\}$ and also they do not impact the lower-level reserve and day-ahead market problems. Moreover, the last term of the objective function in~\eqref{mod:B-Pre}, which relates to the balancing market, is practically a copy of the objective function of model~\eqref{mod:A-CV-rt} for {each} $s\in\mathcal{S}$. These observations combined imply that the KKT conditions of model~\eqref{mod:A-CV-rt}, other than the primal feasibility conditions included in~\eqref{mod:B-Pre} are redundant. Similar formulations for bilevel programming in the~electricity markets can also be found in~\citep{pineda2016capacity,morales2014electricity,dvorkin2018setting,jensen2017cost}.}
\chapter{Coalitional game theory framework}\label{sec:3}
In this chapter, we bring in the preliminaries for benefit allocation mechanisms and {their} desirable properties. We then review {existing} mechanisms from the literature and discuss whether they {attain} these properties. {To facilitate the exposition}, this chapter treats a general coalitional game. 

\section{Preliminaries}
A coalitional game is {defined} by a set of players and {the so-called} coalitional value function, that maps from the subsets of players to the values, that is, the total benefits created by these players~\citep{osborne1994course,peleg2007introduction}. In the preemptive model, the set of players are given by the set of areas $\AC$\footnote{The players involved in this game are required to be areas as a whole (country or region). This includes consumers and generators pertaining to that area and area operators (and potentially the transmission owners). This is because the transmission capacity allocated to the reserve exchange affects the incentive structure of all these market participants.}, whereas the coalitional value function $v:2^{\AC}\rightarrow \mathbb{R}$ can be defined 
either as the expected cost reduction achieved, that is, $$\bar v(\CC) = J(\emptyset) - J(\CC),$$ for all $\CC\subseteq \AC$ or based on the resulting change in the cost of the realized scenario $s\in\mathcal{S}$, that is, $$v^s(\CC) = J^s(\emptyset) - J^s(\CC),$$ for all $\CC\subseteq \AC$.
Clearly, {it holds that} $\bar v(\CC)=\mathbb{E}_s[v^s(\CC)].$ Later, we will see that these functions yield different structures for the game. In the remainder of this section, we study a generic~$v$ for the preemptive model satisfying $v(\CC)=0$, for all $|\CC|\leq 1$. This assumption holds since coordination is not possible in the preemptive model without the participation of at least two adjacent areas.

Given the coalitional value function $v$, \textit{a benefit allocation mechanism} defines the benefit received by each area $a\in\AC$ with $\beta_a(v)\in \R$. The cost allocated to area $a$ under the preemptive transmission allocation model would then be given by $J^s_a(\AC)=J^s_a(\emptyset)-\beta_a(v)$. Depending on its sign, the benefit can be considered as a discount or a mark-up on the original cost allocation. 

{When designing benefit allocation mechanisms, there are three }fundamental properties we want to guarantee, namely, efficiency, individual rationality, and stability. 
A benefit allocation $\beta(v) = \{\beta_a(v)\}_{a\in\mathcal A}\in\mathbb{R}^{\mathcal A}$ is \textit{efficient} if {the whole} value created by the grand coalition, {that is, $\CC = \AC$}, is allocated to the {member}-areas, that is, $\sum_{a\in\AC}\beta_a(v)=v(\AC)$.\footnote{In coalitional games, efficiency is also often referred to as budget-balance. For clarity, this part of the thesis uses efficiency for the benefit allocation, and the term budget-balance is reserved for the cost allocation.} A benefit allocation ensures \textit{individual rationality} if all areas {obtain} nonnegative benefits, that is, $\beta_a(v)\geq 0,$ for all $a\in\AC$. If this property does not hold, the coordination arrangement would yield increased costs for some areas. As a result, these areas may decide not to participate in the preemptive model. Finally, a benefit allocation attains \textit{stability} (in other words, group rationality) if it eliminates the benefit improvements of the areas from forming {sub-}coalitions, that is, $\nexists \mathcal C\subset\AC$ such that $v(\CC)>\sum_{a\in\CC}\beta_a(v)$. This last property is crucial for the preemptive model, since otherwise some areas may opt for having their own reserve {exchange} agreement by excluding the remaining areas. This coincides with our aforementioned goal of ensuring that all areas participate in the preemptive model.

In coalitional game theory, these properties are known to be attained if the benefit allocation lies in the \textit{core}\footnote{To avoid any confusion, we point out that this is a slightly different definition to the core when compared to~\cref{def:core_def} in \cref{part:1}. Here in $K_{\text{Core}}(v)$, we consider deviations from all subsets of players, whereas in~\cref{def:core_def} we assume that all deviations have to include the central operator. A way to combine these definitions would be to extend the domain of the function $-J$ in \eqref{eq:33} as $\hat v$ by assigning $0$ cost to all deviations that do not include the central operator. In this case, we would obtain $\hat v(\{0\})=-\infty$, where $\{0\}$ denotes the deviation set involving only the central operator. Thus, this extension does not satisfy our requirement in \cref{part:2}:  $\hat v(\CC)=0$, for all $|\CC|\leq 1$. Moreover, for this extension, $\hat v(\CC)$ may not be well-defined also for other subsets: $\hat v(\CC)=-\infty$. Such instances are not present in any of the coalitional games considered in \cref{part:2}. Finally, $K_{\text{Core}}(v)$ can also be considered as the intersection of the core in~\cref{def:core_def} ($-J=v$) with $u_0=0$ (without any monotonicity requirement).} {defined as} $\beta(v) \in K_{\text{Core}}(v)$, where
$$K_\text{Core}(v)=\{\beta\in \mathbb{R}^\AC\,|\, \sum_{a\in\AC}\beta_a=v(\AC),\ \sum_{a\in\CC}\beta_a\geq v(\CC),\ \forall \CC\subset\AC\}.$$ 
In this definition, the equality constraint ensures efficiency, while inequality constraints guarantee stability, that is, there is no subset of areas $\CC\subset\AC$ that can yield higher total benefits for its members compared to the benefit allocation under the grand coalition.
The inequality constraints also include $\beta_a(v)\geq v(a)=0$ for all $a\in\AC$.\footnote{For the sake of simplicity, singleton sets are denoted by $a$ instead of $\{a\}.$} This restriction ensures individual rationality. 

The core is a closed polytope involving $2^{|\AC|}$ linear constraints. 
This polytope is nonempty if and only if the coalitional game is balanced~\citep{shapley1967balanced}. Such settings include the cases in which the coalitional value function exhibits supermodularity\footnote{We remind the reader that supermodularity is attained if for any set the participation of an area results in a larger value increment when compared to the subsets of the set under consideration, that is, $v(\CC\cup\{a\})-v(\CC)\geq v(\CC'\cup\{a\})-v(\CC')$, $\forall a\notin\CC,\CC'\subset\CC\subseteq\AC,$ see \cref{def:supms}.}~\citep{shapley1971cores} and the cases in which the coalitional value function can be modeled by a concave exchange economy~\citep{shapley1969market}, a linear production game~\citep{owen1975core} or a risk-sharing game~\citep{csoka2009stable}. In their most general form, the coalitional value functions in these works are given by an optimization problem minimizing a convex objective subject to linear constraints. In the problem at hand, coalitional value functions are associated with solutions to the general non-convex optimization problem \eqref{mod:B-Pre}. As a result, previous works on the nonemptiness of the core are not applicable to our setup.

In case the core is empty, we need to devise a method to approximate a core allocation. To this end, we bring in the {notion of} strong $\epsilon$-core, defined in \cite{shapley1966quasi} as 
$$K_\text{Core}(v,\epsilon)=\{\beta\in \mathbb{R}^\AC\,|\, \sum_{a\in\AC}\beta_a=v(\AC),\ \sum_{a\in\CC}\beta_a\geq v(\CC)-\epsilon,\ \forall \CC\subset\AC\}.$$
This definition can be interpreted as follows. If organizing a coalitional deviation entails an additional cost of $\epsilon\in\mathbb{R}$, coalition values would be given by $v(\CC)-\epsilon$ for all $\CC\neq\AC$. Then, the resulting core would correspond to the strong $\epsilon$-core. For $\epsilon=0$, we retrieve the original core definition, that is, $K_\text{Core}(v,0)=K_\text{Core}(v)$. Let $\epsilon^*(v)$ be the critical value of $\epsilon$ such that the strong $\epsilon$-core is nonempty, {which is mathematically defined as} $$\epsilon^*(v)=\min\{\epsilon\,|\,K_\text{Core}(v,\epsilon)\not=\emptyset\}.$$ The value $\epsilon^*(v)$ is guaranteed to be finite for any function~$v$ and the set $K_\text{Core}(v,\epsilon^*(v))$ is called the \textit{least-core}~\cite{maschler1979geometric}. 
Let the excess of a coalition be defined by $\theta(v,\beta,\CC)=v(\CC)-\sum_{a\in\CC}\beta_a,$ for any nonempty $\CC\subset\AC$. In other words, the set $K_\text{Core}(v,\epsilon^*(v))$ is the set of all efficient benefit allocations minimizing the maximum excess.
If the core is empty, the maximum excess is the maximum violation of a stability constraint. This implies that the least-core achieves an approximate stability property. As a remark, the least-core relaxes also the inequality constraints corresponding to singleton sets $\beta_a\geq v(a)-\epsilon^*(v)=-\epsilon^*(v)$ for all $a\in\AC$ since $\epsilon^*(v)> 0$, and hence it yields approximate individual rationality. Finally, if the core is not empty, we have $\epsilon^*(v)\leq 0$ and the least-core is a subset of the core. 

With the discussion above, we conclude that whenever the core is empty, we can use the least-core to achieve the second best outcome available, that is, a benefit allocation which is efficient, approximately individually rational and approximately stable. Observe that there are generally many points to choose from the least-core (or the core if it is nonempty) achieving the same fundamental properties. In this case, it could be desirable to require additional intuitively acceptable properties to pick a unique benefit allocation. Later, we revisit this idea in our proposed methods. 

Apart from the aforementioned fundamental properties that pertain to the economic side of the problem, computational tractability is also a practical concern, considering that we may need the complete list of coalition values $v(\CC)$ for all $\CC\subseteq\AC$ to fully describe the core and the least-core.
For the coalitional games arising from the preemptive model, each coalition value requires another solution to MILP in~\eqref{mod:B-Pre}, which is NP-hard in general. 
Hence, our goal is to find a core or a least-core benefit allocation that can be computed with limited queries to the coalitional value function. 

Next, we briefly review two benefit allocation mechanisms that are widely used in the literature. 

\section{Shapley value}
The benefit assigned by the Shapley value is given by
$$\beta_a^{\text{Shapley}}(v) = \sum_{\CC\subseteq\AC} \frac{(|\CC|-1)!(|\AC|-|\CC|)!}{|\AC|!}(v(\CC)-v(\CC\setminus a)).$$
This benefit is the average of the marginal contribution of the area $a$ under all coalitions, considering also all possible orderings of areas. 
The Shapley value results in an efficient benefit allocation. Individual rationality is also satisfied if the coalitional value function is nondecreasing, since the marginal contributions would be nonnegative. On the other hand, the Shapley value is guaranteed to lie in the core only when the coalitional value function is supermodular. This is a restrictive condition that is not applicable to our problem. In addition, when the core is empty, the Shapley value does not necessarily lie in the least-core, making it incompatible with the fundamental properties we desire~\cite{maschler1979geometric}. 
In terms of the computational performance, the calculation of the Shapley value requires the exhaustive enumeration of coalition values $v(\CC)$ for all $\CC\subseteq\AC$. Finally, it should be noted that the Shapley value is the unique efficient benefit allocation that satisfies dummy player, symmetry, and additivity properties simultaneously.
Dummy player property requires $\beta_a=0$ for all $a$ for which $v(\CC)-v(\CC\setminus a)=0$ for all~$\CC\subseteq\AC$. In other words, an area incapable of contributing to any coalition~$\CC$ ends up with zero benefits. Next, we show the relation between the previously discussed properties and the dummy player property.

\begin{proposition}\label{prop:dummy} For the core and the least-core, we have,\begin{enumerate}
		\item[(i)] if $a'$ satisfies $v(\AC)-v(\AC\setminus a')=0$, then $K_\text{Core}(v)\subset\{\beta\,|\,\beta_{a'}=0\}$,\item[(ii)]  if $a'$ satisfies $v(\CC)-v(\CC\setminus a')=0$ for all~$\CC\subseteq\AC$, then $K_\text{Core}(v,\epsilon^*(v))\subset\{\beta\,|\,\beta_{a'}=0\}$.
		\end{enumerate}
\end{proposition}

{This proposition provides a missing link in the comparisons of the Shapley value, the core, and the least-core in a generic coalitional game.}
This result shows that the core attains a more restrictive version of the dummy player property, that is, an area incapable of contributing to the set $\AC$ ends up with zero benefits. Finally, the least-core attains the dummy player property in the same way that it is defined for the Shapley value.
The proof and the discussions on symmetry and additivity are relegated to the appendix in \cref{app:a2}. 
\looseness=-1

\section{Nucleolus allocation}

 Among all efficient benefit allocations, the nucleolus allocation is the unique benefit allocation that minimizes the excesses of all coalitions in a lexicographic manner~\cite{schmeidler1969nucleolus}. Nucleolus allocation lies in the least-core and hence attains the desirable economic properties. 
 
 In terms of practical implementation, the lexicographic minimization is computationally demanding in the general case. Nucleolus allocation can be computed by solving a sequence of ${\mathcal{O}(|\AC|)}$ linear programs with constraint sets that are parametrized versions of the core $K_{\text{Core}}(v)$, see~\cite{kopelowitz1967computation,fromen1997reducing}.
 However, each linear program requires the complete list of coalition values. In case the coalition values are given implicitly by the objective value of a single linear optimization problem with constraints depending on the participants of the coalition, the work by~\cite{hallefjord1995computing} proposes using constraint generation algorithms. 
In this approach, ${\mathcal{O}(|\AC|)}$ linear programs are solved by ${\mathcal{O}(|\AC|)}$ constraint generation algorithms that iteratively generates coalitional values on demand. Nevertheless, we may still need to generate all possible coalition values~\cite{hallefjord1995computing,kimms2012approximate}. When the number of areas is large, this approach involving the execution of the constraint generation algorithm $\mathcal{O}(|\AC|)$ times becomes computationally prohibitive for our application.
     
     {For the sake of completeness,}~the appendix in \cref{app:a2-2} provides the mathematical definition for the nucleolus allocation and its comparison with the Shapley value.
 
\section{Appendix}
\subsection{Unique properties of the Shapley value}\label{app:a2}
In this appendix, we analyze each of the unique properties of the Shapley value and how they relate to the fundamental properties associated with the core and the least-core.
    \subsubsection{Dummy player} Dummy player property requires $\beta_a=0$ for all $a$ such that $v(\CC)-v(\CC\setminus a)=0$ for all~$\CC\subseteq\AC$. In other words, an area incapable of contributing to any coalition~$\CC$ ends up with zero benefits. We now reiterate \cref{prop:dummy} and then prove the two claims.
\begin{proposition}\label{prop:dummy2} For the core and the least-core, we have, 
	\begin{itemize}
		\item[(i)] 	If $a'$ satisfies $v(\AC)-v(\AC\setminus a')=0$, then $K_\text{Core}(v)\subset\{\beta\,|\,\beta_{a'}=0\}$,
		\item[(ii)]	If $a'$ satisfies $v(\CC)-v(\CC\setminus a')=0$ for all~$\CC\subseteq\AC$, then $K_\text{Core}(v,\epsilon^*(v))\subset\{\beta\,|\,\beta_{a'}=0\}$.
	\end{itemize}
\end{proposition}
\begin{proof}
	(i) Assume core is nonempty, since otherwise the proof is trivial. Combining the equality constraint with the inequality constraint corresponding to $\AC\setminus a'$, we obtain $\beta_{a'}\leq v(\AC)-v(\AC\setminus a')=0$ for any $\beta\in K_\text{Core}(v)$. Combining this with $K_\text{Core}(v)\subset\mathbb{R}^\AC_+$ gives us $\beta_{a'}=0$ for any $\beta\in K_\text{Core}(v)$.
	
	(ii) Assume core is empty, $\epsilon^*(v)>0$, since otherwise part~(i) concludes that $$K_\text{Core}(v,\epsilon^*(v))\subseteq K_\text{Core}(v)\subset\{\beta\,|\,\beta_{a'}=0\}.$$ Next, we prove by contradiction that first $\beta_{a'}>0$ is not possible and then $\beta_{a'}<0$ is not possible.
	
	Let $\hat \beta \in K_\text{Core}(v,\epsilon^*(v))$ be a benefit allocation with $\hat \beta_{a'}>0$. We now show that there exists $\epsilon <\epsilon^*(v)$ such that $K_\text{Core}(v,\epsilon)\neq \emptyset$. This would contradict the definition of the least-core. 
	
	For any $\CC\ni a'$, we have $\sum_{a\in\CC}\hat\beta_a>\sum_{a\in\CC\setminus a'}\hat\beta_a\geq v(\CC\setminus a')-\epsilon^*(v)=v(\CC)-\epsilon^*(v)$. Notice that we can always find a small positive number $\delta$ such that  $\sum_{a\in\CC}\hat\beta_a - (|\AC|-|\CC|)\delta> v(\CC)-\epsilon^*(v)+\delta$ holds for any $\CC\ni a'$. Next, we show that $K_\text{Core}(v,\epsilon^*(v)-\delta)$ is nonempty for this particular choice. 
	
	Define $\bar{\beta}$ such that $\bar\beta_a=\hat \beta_a + \delta$ for all $a\neq a'$and $\bar{\beta}_{a'}=\hat{\beta}_{a'}-(|\AC|-1)\delta$. This new allocation $\bar{\beta}$ clearly satisfies the equality constraint in $K_\text{Core}(v,\epsilon^*(v)-\delta)$. For inequality constraints $\CC\ni a'$, we have $\sum_{a\in\CC}\bar\beta_a=\sum_{a\in\CC}\hat\beta_a - (|\AC|-|\CC|)\delta> v(\CC)-\epsilon^*(v)+\delta$, where the strict inequality follows from the definition of $\delta$. For inequality constraints $\CC\not\ni a'$, we have $\sum_{a\in\CC}\bar\beta_a\geq\sum_{a\in\CC}\hat\beta_a+\delta\geq v(\CC)-\epsilon^*(v)+\delta$. Hence, $\bar \beta \in K_\text{Core}(v,\epsilon^*(v)-\delta)$, in other words, $K_\text{Core}(v,\epsilon^*(v)-\delta)\neq\emptyset.$ This contradicts $K_\text{Core}(v,\epsilon^*(v))$ being the least-core. Hence, $\hat \beta_{a'}\not>0$. 
	
	Next, let $\hat \beta \in K_\text{Core}(v,\epsilon^*(v))$ be a benefit allocation with $\hat \beta_{a'}<0$. We again show that there exists $\epsilon <\epsilon^*(v)$ such that $K_\text{Core}(v,\epsilon)\neq \emptyset$. 
	
	Since $\hat \beta \in K_\text{Core}(v,\epsilon^*(v))$ and $v(a')=0$, we have $0>\beta_{a'}\geq v(a') -\epsilon^*(v)=-\epsilon^*(v)$. Notice that, for any $\CC\not\ni a'$, we have $\sum_{a\in\CC}\hat\beta_a\geq v(\CC\cup a')-\epsilon^*(v)-\beta_{a'}$ by adding and subtracting $\beta_{a'}$, and by using the fact that $\epsilon^*(v)>0$ for the special case corresponding to $\CC\cup a'=\AC$. Since we have $v(\CC\cup a')=v(\CC)$ and $\beta_{a'}<0$, we obtain $\sum_{a\in\CC}\hat\beta_a> v(\CC)-\epsilon^*(v)$. 
	Notice that we can always find a small positive number $\delta$ such that  $\sum_{a\in\CC}\hat\beta_a - |\CC|\delta> v(\CC)-\epsilon^*(v)+\delta$ holds for any $\CC\not\ni a'$. Next, we show that $K_\text{Core}(v,\epsilon^*(v)-\delta)$ is nonempty for this particular choice. 
	
	Define $\bar{\beta}$ such that $\bar\beta_a=\hat \beta_a - \delta$ for all $a\neq a'$and $\bar{\beta}_{a'}=\hat{\beta}_{a'}+(|\AC|-1)\delta$. This new allocation $\bar{\beta}$ clearly satisfies the equality constraint in $K_\text{Core}(v,\epsilon^*(v)-\delta)$. For inequality constraints $\CC\not\ni a'$, we have $\sum_{a\in\CC}\bar\beta_a=\sum_{a\in\CC}\hat\beta_a - |\CC|\delta> v(\CC)-\epsilon^*(v)+\delta$, where the strict inequality follows from the definition of $\delta$. For inequality constraints $\CC\ni a'$, we have $\sum_{a\in\CC}\bar\beta_a\geq\sum_{a\in\CC}\hat\beta_a+\delta\geq v(\CC)-\epsilon^*(v)+\delta$. Hence, $\bar \beta \in K_\text{Core}(v,\epsilon^*(v)-\delta)$, in other words, $K_\text{Core}(v,\epsilon^*(v)-\delta)\neq\emptyset.$ This contradicts $K_\text{Core}(v,\epsilon^*(v))$ being the least-core.  Hence, $\hat \beta_{a'}\not<0$. This concludes that $\hat \beta_{a'}=0$.
\end{proof}

The proposition above provides a missing link in the comparisons of the Shapley value, the core, and the least-core in a generic coalitional game. It shows that the core attains a more restrictive version of the dummy player property, that is, $\beta_a=0$ for all $a$ such that $v(\AC)-v(\AC\setminus a)=0$. In other words, an area incapable of contributing to the set of all areas $\AC$ ends up with zero benefits. Finally, the least-core attains the dummy player property in the same way that it is defined for the Shapley value.

\subsubsection*{Symmetry} Symmetry property is achieved if the benefit allocations of two areas are the same whenever their marginal contributions to any coalition~$\CC$ are the same. It can be verified that this property does not hold for every benefit allocation from the core and the least-core. However, it is always possible to find a benefit allocation satisfying the symmetry property in any nonempty strong $\epsilon$-core (and hence both in the core and the least-core) since the linear inequality constraints imposed by the convex polytope $K_\text{Core}(v,\epsilon)$ on two such area benefits are identical. Notice that the symmetry property is computationally hard to check since it would require evaluating the function $v$ for all coalitions. We can instead aim for a more restrictive version of the symmetry property by considering only the marginal contributions to the set of all areas $\AC$. This stronger condition would be computationally tractable to check. 

\subsubsection*{Additivity} Additivity property is given by $\beta(\hat v + v)=\beta(\hat v)+\beta(v)$ for all $\hat v,v:2^\AC\rightarrow \mathbb{R}$. As it is discussed in~\citep{osborne1994course}, this property is mathematically convenient but hard to argue for since the sum of coalitional value functions is in general considered to induce an unrelated coalitional game. As a remark, additivity further implies that the benefit allocation of any player responds monotonically to changes in the coalition value $v(\CC)$ (a positive change if $\CC$ contains the area). In the general case, this property cannot be achieved by any benefit allocation chosen from both the core and the least-core, see the discussions and the counter examples provided in the work of \cite{young1985monotonic,megiddo1974nonmonotonicity}.
	
\subsection{Definitions and discussions for the nucleolus allocation}\label{app:a2-2}

In this section, we provide the mathematical definition for the nucleolus allocation, compare it with the Shapley value, and discuss additional aspects of its computation and lexicographic minimization property.

Denote the excesses as $\theta(v,\beta,\CC)=v(\CC)-\sum_{a\in\CC}\beta_a$ for any nonempty $\CC\subset\AC$. Let $\theta(v,\beta)\in\mathbb{R}^{2^{|\AC|}-2}$ be the vector whose entries are the excesses but arranged in a nonincreasing order. Given two such ordered vectors $x,y\in\mathbb R^{n_0}$, with $x<_{\text L}y$ we mean that $x$ is lexicographically smaller than $y$, that is, there exists an index $\nu_0\leq n_0$ such that $x_\nu = y_\nu$ for all $\nu<\nu_0$, and $x_{\nu_0}<y_{\nu_0}$. Let $\X\subset\mathbb{R}^{\AC}$ denote a set of benefit allocations that we are interested in. Then, the nucleolus of the set $\X$ is the benefit allocation that minimizes the excess of all coalitions in a lexicographic manner among all benefit allocations from the set $\X$.
Specifically, the nucleolus of the set $\X$, $\beta^{\text{Nuc}}(v,\X)\in\X$,  is defined by
 \begin{equation*}
 \theta(v,\beta^{\text{Nuc}}(v,\X))<_{\text L}\theta(v,\hat\beta),\ \forall \hat \beta\in\X\ \text{s.t.}\ \hat \beta\neq\beta^{\text{Nuc}}(v,\X).
 \end{equation*}

 In the literature, the nucleolus is generally defined with respect to two sets. Let $\X_{\text{BB}}=\{\beta\in\R^{\AC}|\sum_{a\in\AC}\beta_a=v(\AC)\}$ be the set of all efficient benefit allocations.
 The nucleolus of the set $\X_{\text{BB}}$, $\beta^{\text{Nuc}}(v,\X_{\text{BB}})$, was introduced in~\cite{sobolev1975characterization}, and is also called the prenucleolus.
 This allocation always exists for any coalitional game, and it is unique. Moreover, it lies in the least-core since its definition can be regarded as a stronger version of minimizing the maximum excess among all efficient benefit allocations~\citep{maschler1979geometric}. Hence, it satisfies efficiency, approximate stability and approximate individual rationality.
 
 On the other hand, let $\X_{\text{BB,IR}}=\{\beta\in\R^{\AC}_+|\sum_{a\in\AC}\beta_a=v(\AC)\}$ be the set of all efficient and individually rational benefit allocations.
 The nucleolus of the set $\X_{\text{BB,IR}}$, $\beta^{\text{Nuc}}(v,\X_{\text{BB,IR}})$, was introduced in \citep{schmeidler1969nucleolus}. This allocation is again unique, but it exists if and only if $v(\AC)\geq 0$. It is guaranteed to lie in the least-core only when the coalitional value function~$v$ is nondecreasing, since then it was proven that this allocation coincides with $\beta^{\text{Nuc}}(v,\X_{\text{BB}})$ in~\citep{aumann1985game}.
 
 In comparison with the Shapley value, both nucleolus allocations above are consistent with the symmetry property but they do not satisfy the additivity property in general, see~\citep{maschler1979geometric}. For the special class of convex/supermodular graph games, these two nucleolus allocations coincide with the Shapley value~\citep{deng1994complexity}. For our work, the nucleolus will refer to the nucleolus of the set $\X_{\text{BB}}$, since this allocation always exists for a general coalitional game, and it is always a unique allocation from the least-core.
 
 In terms of computational approaches, note that there are also methods to compute the nucleolus by solving a single linear program involving either $4^{|\AC|}$ constraints~\citep{owen1974note} or $2^{|\AC|}!$ constraints~\citep{kohlberg1972nucleolus}. To the best of our knowledge, there are no iterative approaches applicable, and these two methods necessitate the complete evaluation of the coalitional value function. 
 
 Finally, as discussed in~\citep{maschler1992general}, the lexicographic minimization property attained by the nucleolus allocation may not be relevant to the needs of every application, and it may even be considered hard to grasp in many cases. There could be more intuitively acceptable properties that yield a unique point. For instance, a prominent example is the work by \cite{young1982cost} suggesting to allocate benefits in a water supply project in proportion to the population and the total demand of an area.
\chapter{Benefit allocation mechanisms for preemptive transmission allocation model}\label{sec:4}

In this chapter, the first benefit allocation mechanism, which is an ex-ante process with respect to the uncertainty realization, employs as coalitional value function the expected cost reduction. The second mechanism is an ex-post process that can be applied only when the scenario is unveiled, since it uses as coalitional value function the scenario-specific cost variation.
\vspace{.5cm}
\section{Benefit allocations for expected cost reduction}\label{sec:4a}
\vspace{.2cm}

{For the ex-ante allocation mechanism, }the coalitional value function $\bar v(\CC) = J(\emptyset) - J(\CC)\geq 0$ for all $\CC\subseteq \AC$ is nondecreasing, since $J$ is nonincreasing.
Given the function~$\bar v$, an efficient benefit allocation, $\sum_{a\in\AC} \beta_a(\bar v)=\bar v(\AC)$, would result in a cost allocation that is budget-balanced in expectation, since $$J(\AC)= J(\emptyset) - \bar v(\AC) = \mathbb{E}_s\big[\sum_{a\in\AC}J^s_a(\emptyset)\big] - \sum_{a\in\AC} \beta_a(\bar v)= \mathbb{E}_s\big[\sum_{a\in\AC}J^s_a(\AC)\big].$$

While designing a benefit allocation mechanism, our goal is to achieve the three fundamental properties, that is, efficiency, individual  rationality and stability, associated with the core $K_{\text{Core}}(\bar v)$. However, as already mentioned, the previous results on the nonemptiness of the core are not applicable to our problem.

The following condition is applicable to some specialized instances of~$\bar v$ above.

\begin{proposition}\label{prop:corempt}
	The core, $K_\text{Core}(\bar v)$, is nonempty if there exists an area $a'\in\AC$ such that $\bar v(\AC\setminus a')=0$.
\end{proposition}

The proof is relegated to the appendix in \cref{app:a3}. 
Note that this condition can only be attained in specialized instances of the preemptive  model. For instance, in the case of a star graph $(\mathcal A,\mathcal E)$, the central area would satisfy this condition, since it is indispensable for enabling any reserve exchange.\footnote{Notice that this result has connections with the nonemptiness of the core in \cref{def:core_def} in \cref{part:1}, because the central operator was deemed indispensable for a secure grid operation via the transmission network.}
However, in a general graph, the core could potentially be empty and we focus on this case in the illustrative example provided in~\cref{sec:emptycore}. 

In case of an empty core, our goal is to achieve {a least-core solution}, {which can be perceived as} the second best outcome in our context. Other than approximating the stability property, the least-core also approximates the individual rationality property by relaxing the inequality constraints for the singleton sets, that is, $\beta_a\geq \bar v(a)-\epsilon^*(\bar v)=-\epsilon^*(\bar v)$ for all $a\in\AC$. The following proposition shows that the least-core is individually rational for the coalitional game given by~$\bar v$.

\begin{proposition}\label{prop:leastcore_volp}
	$K_\text{Core}(\bar v,\epsilon^*(\bar v))$ lies in $\mathbb{R}^\AC_+$.
\end{proposition}

The proof is relegated to the appendix in \cref{app:a4}. It relies on the observation that 
whenever the coalitional value function is given by a stochastic bilevel program
any least-core allocation violating the individual rationality would imply the existence of an $\epsilon <\epsilon^*$, such that $K_\text{Core}(\bar v,\epsilon)$ is nonempty, contradicting the definition of the least-core. Thus, we can use the least-core to achieve efficiency, individual rationality and approximate stability, whenever the core is empty.

For this coalitional game, the Shapley value satisfies efficiency and individual rationality, but stability (or approximate stability) and computational tractability are not attained. We provide an example for the stability violation of the Shapley value in~\cref{sec:shapno}. The nucleolus allocation, on the other hand, lies in the least-core and it satisfies efficiency, individual rationality and approximate stability.

Based on these discussions, we propose a least-core-selecting mechanism:
\begin{equation}
\label{eq:ba_opt}
\underset{  \epsilon,\, \beta}{\min}\ \epsilon\quad \mathrm{s. t.}\ \ \epsilon\geq 0,\ \beta\in K_{\text{Core}}(\bar v,\epsilon).
\end{equation}
Let $\hat{\epsilon}$ denote the optimal value of $\epsilon$ for this problem. If the core is empty, we have $\hat\epsilon=\epsilon^*(\bar v)>0$ and problem~\eqref{eq:ba_opt} finds a least-core allocation. On the other hand, if the core is nonempty, we have $\hat\epsilon=0$ and problem~\eqref{eq:ba_opt} finds instead a core benefit allocation, which attains properties of efficiency, individual rationality and  stability.

The nucleolus allocation always forms an optimal solution pair with $\hat{\epsilon}$ to problem~\eqref{eq:ba_opt}, since it lies in the least-core. In fact, there are in general many optimal solutions to this problem. To this end, we will propose an additional criterion for tie-breaking purposes.

Let $\beta^{\text{c}}$ be a desirable and a fair benefit allocation that is easy to compute but not necessarily in the core or in the least-core. An example could be the marginal contribution of each area $\beta^{\text{m}}:\ \beta_a^{\text{m}} = \bar v(\AC)-\bar v(\AC\setminus a)$ for all $a\in\AC$ which requires $|\AC|+2$ calls to problem~\eqref{mod:B-Pre}. 
Receiving the marginal contribution can be regarded as a fair outcome.\footnote{This allocation coincides with the Vickrey-Clarke-Groves mechanism. \cref{part:1} discussed how this allocation ensures that truthfully reporting the preferences is a dominant strategy Nash equilibrium in an auction setting.} This allocation satisfies individual rationality, dummy player and symmetry properties. However, it is generally not efficient, and not stable.
Another example could be $\beta^{\text{eq}}=(\bar v(\AC)/|\AC|)\bm {1}^\top$ which assigns equal importance to each area. This choice requires two calls to problem~\eqref{mod:B-Pre} and it satisfies efficiency and individual rationality. However, this allocation also violates stability. 

Starting from such a desirable benefit allocation, we can solve the following problem 
\begin{equation}
\label{eq:ba_opt_tie}
\underset{\beta}{\min}\ ||\beta-\beta^{\text{c}}||_2^2\quad \mathrm{s. t.}\ \beta\in K_{\text{Core}}(\bar v, \hat \epsilon),
\end{equation}to obtain a unique benefit allocation for problem~\eqref{eq:ba_opt}.
The uniqueness follows from having a strictly convex objective ($2$-norm).
Let $\hat{\beta}(\bar v,\beta^{\text{c}})$ denote the optimal value of $\beta$ in problem~\eqref{eq:ba_opt_tie}. We define the benefit allocation $\hat{\beta}(\bar v,\beta^{\text{c}})$ as the\textit{ least-core-selecting mechanism}. {This allocation} achieves economic properties of the least-core, and also the core if the core is nonempty, while approximating an additional criterion defined by $\beta^{\text{c}}$. For instance, if the marginal contribution $\beta^{\text{m}}$ is chosen, problem~\eqref{eq:ba_opt_tie} would pick the allocation $\hat{\beta}(\bar v,\beta^{\text{m}})$ approximating the fairness of the marginal contribution. 

Characterizing the constraint sets of problems~\eqref{eq:ba_opt}~and~\eqref{eq:ba_opt_tie} still requires exponentially many solutions to \eqref{mod:B-Pre}. Next, we show that~\eqref{eq:ba_opt}~and~\eqref{eq:ba_opt_tie} can be solved by a single constraint generation algorithm.

\subsubsection*{Constraint generation algorithm}

Here, we describe the steps of the constraint generation algorithm at iteration $k\geq 1$. Let $\mathcal{F}^k\subset 2^{\AC}$ denote the family of coalitions for which we have already generated the coalition values. The algorithm first obtains a candidate solution by solving a relaxed version of problem~\eqref{eq:ba_opt} as follows:
\begin{equation}
\label{mod:candidate_k}
\underset{  \epsilon,\, \beta}{\min}\ \epsilon\quad
\mathrm{s. t.}\ \epsilon\geq 0,\quad \beta\geq0,\quad\sum_{a\in\AC}\beta_a=\bar v(\AC),\quad \sum_{a\in\CC}\beta_a\geq \bar v(\CC)-\epsilon,\quad \forall \CC\in\mathcal{F}^k. 
\end{equation}
Let the optimal solution be denoted by $\epsilon^{k}$, clearly, $\hat\epsilon\geq\epsilon^{k}.$
Solve the following as a tie-breaker:\begin{equation}
	\label{mod:candidate_k_unique}
\underset{\beta}{\min}\ ||\beta-\beta^{\text{c}}||_2^2\quad
\mathrm{s. t.}\ \beta\geq0,\quad\sum_{a\in\AC}\beta_a=\bar v(\AC),
\quad\sum_{a\in\CC}\beta_a\geq \bar v(\CC)-\epsilon^k,\quad \forall \CC\in\mathcal{F}^k.
\end{equation}
Denote the optimal benefit allocation for problem~\eqref{mod:candidate_k_unique} by~$\beta^{k}$. This allocation would form an optimal solution pair to problem~\eqref{mod:candidate_k} with~$\epsilon^{k}$. 
In principle, $\mathcal{F}^1$ can potentially be chosen as an empty set, by setting $\epsilon^{1}$  equal to zero and removing the last set of constraints in~\eqref{mod:candidate_k_unique}.

Given a candidate allocation $\beta^k$, we can then generate the coalition with the maximum stability violation by solving the following problem, which treats $\beta^k$ as a fixed parameter:
\begin{flalign}\label{eq:ba_findviolation}
	&\underset{ \CC\subseteq\AC}{\max}\ \bar v(\CC)-\sum_{a\in\CC}\beta_a^k .
\end{flalign}
Denote the optimal solution by $\CC^k$, and the optimal value by~$\eta^k$.
Using the fact that the coalitional value function~$\bar v$ is given implicitly by the MILP version of the preemptive model~\eqref{mod:B-Pre}, we can show that problem \eqref{eq:ba_findviolation} has an equivalent MILP reformulation, which can be solved efficiently by off-the-shelf optimization solvers. Using this approach, we eliminate the need for evaluating $\bar v(\CC)$ for all $\CC\subseteq \AC$ to solve problem~\eqref{eq:ba_findviolation}.
This MILP reformulation is relegated to the appendix in~\cref{app:a6b}. 
Finally, this problem generates the coalition value $\bar v(\CC^k)=\eta^k+\sum_{a\in\CC^k}\beta_a^k,$ and we add $\CC^k$ to $\mathcal{F}^k$ for the next iteration.

In order to define a stopping criterion ensuring that the iterative solution of problems~\eqref{mod:candidate_k},~\eqref{mod:candidate_k_unique} and~\eqref{eq:ba_findviolation} converges to the optimal solution of problem~\eqref{eq:ba_opt_tie}, we need the following two observations. First, if $\epsilon^k>0$ in problem~\eqref{mod:candidate_k}, then there exists a set $\CC\in\mathcal{F}^k$ such that $\bar v(\CC)-\sum_{a\in\CC}\beta_a^k=\epsilon^k$, which implies that $\eta^k\geq\epsilon^k$. On the other hand, if $\epsilon^k=0$, by setting $\CC=\AC$ we can show that $\eta^k\geq\bar v(\AC)-\sum_{a\in\AC}\beta_a^k=0=\epsilon^k$. Based on these remarks, we have $\eta^k\geq\epsilon^k$ for any iteration $k$. 

The iterative solution of problems~\eqref{mod:candidate_k},~\eqref{mod:candidate_k_unique}~and~\eqref{eq:ba_findviolation} terminates when $\eta^k=\epsilon^k$. In this case,
 problem \eqref{eq:ba_findviolation} provides a certificate that the pair $(\epsilon^k,\beta^k)$ is a feasible solution to problem~\eqref{eq:ba_opt}. Note that this pair is also optimal to a relaxed version of problem~\eqref{eq:ba_opt}, given by problem~\eqref{mod:candidate_k}. This concludes that $(\epsilon^k,\beta^k)$ is optimal for problem~\eqref{eq:ba_opt}.
Observing that $\epsilon^k=\hat \epsilon$ and using a similar reasoning, we conclude that $\beta^k$ is the optimal solution to problem~\eqref{eq:ba_opt_tie}. Hence, the algorithm converges. 

In the intermediate solution points of the iterative process, we have $\eta^k>\epsilon^k$ and consequently $\CC^k\notin \mathcal{F}^k$ according to the optimality of~$\epsilon^k$ for problem~\eqref{mod:candidate_k}. We then extend the family of generated coalitions by $\mathcal{F}^{k+1}=\mathcal{F}^k\cup\CC^k$ until convergence is achieved.
Since there are finitely many coalitions to be generated, the algorithm converges after a finite number of iterations. As a remark, the algorithm would not generate any set from $\mathcal{F}^1$, the full set, the empty set, and the singleton sets (since we enforce $\beta\geq0$ in problem~\eqref{mod:candidate_k}). In practice, even when there are many areas, the algorithm requires the generation of only several coalition values. We show this in the numerical results.

We summarize this iterative algorithm in \cref{alg:const_gen}.

 \begin{algorithm}[h]
 	\caption{Constraint Generation Algorithm for the Least-Core-Selecting Mechanism}
	\begin{algorithmic}[1]\label{alg:const_gen}
 		\renewcommand{\algorithmicrequire}{\textbf{Initialize:} Compute $\bar v(\AC)$,}
 		\REQUIRE 
 		and $\beta^{\text{c}}$, set $k=0$, $\eta_0=1$, $\epsilon_0=0$, initialize $\mathcal{F}^1$ (e.g., $\mathcal{F}^1=\emptyset$). \\
 		\WHILE {$\eta^{k}> \epsilon^k$}
 	\STATE Update $k = k+1$.
 		\STATE Obtain $\epsilon^{k}$ and $\beta^{k}$ by solving~\eqref{mod:candidate_k} and~\eqref{mod:candidate_k_unique}\STATE Obtain $\eta^{k}$, $\CC^k$, and $\bar v(\CC^k)$ by solving~\eqref{eq:ba_findviolation}.
		\STATE Set $\mathcal{F}^{k+1}=\mathcal{F}^k\cup\CC^k$.
 		\ENDWHILE 
 		\RETURN $\beta^k=\hat{\beta}(\bar v,\beta^{\text{c}})$.
 	\end{algorithmic}
 \end{algorithm}

\section{Benefit allocations per scenario}

Allocating benefits for the expected cost reduction does not guarantee that the resulting cost allocation satisfies budget-balance in every scenario. Having a surplus or a deficit might be undesirable, since this may necessitate a large financial reserve to buffer the fluctuations in the budget. To address this issue, here we focus on the allocation of {the scenario-specific cost variation}, $J^{s}(\emptyset)-J^{s}(\AC)$.
The coalitional value function in this case is given by $v^{s}(\CC) = J^{s}(\emptyset) - J^{s}(\CC),$ for all $\CC\subset \AC$. Observe that the set function $v^{s}$ is not necessarily nondecreasing, while it can also map to negative reals, since the preemptive transmission allocation model does not guarantee that $J^{s}(\CC)\leq J^{s}(\emptyset)$ holds. Given the function $v^{s}$, an efficient benefit allocation mechanism, $\sum_{a\in\AC} \beta_a(v^{s})=v^{s}(\AC)$, would result in a cost allocation that is budget-balanced in scenario~${s}$, since $$J^{s}(\AC)= J^{s}(\emptyset) - v^{s}(\AC) = \sum_{a\in\AC}J^{s}_a(\emptyset) - \sum_{a\in\AC} \beta_a(v^{s})= \sum_{a\in\AC}J^{s}_a(\AC).$$

Aiming at establishing a per-scenario benefit allocation, our goal now is to achieve the properties of the scenario-specific core $K_\text{Core}(v^{s})$. However, neither the previous results nor~\cref{prop:corempt} apply to this core to prove that it is nonempty 
as it can be affirmed by the following result.
\begin{proposition}\label{prop:empty}
	$K_\text{Core}(v^{s})$ is empty if there exists $\CC\subset\AC$ such that $v^{s}(\AC)< v^{s}(\CC)$.
\end{proposition}

The proof is relegated to the appendix in~\cref{app:a7}. In practice, this condition above would prevent the formation of the grand coalition $\AC$,
as shown in the example of~\cref{sec:shapno}. 
The coalition value $v^{s}(\AC)$ being negative is a special case of~\cref{prop:empty}, since we would then have $v^{s}(\AC)< v^{s}(a)=0$ for all $a\in \AC$. 

We see that it may not be realistic to achieve all three fundamental properties, and we should instead aim for the least-core $K_\text{Core}(v^s,\epsilon^*(v^s)).$ Note that in this case~\cref{prop:leastcore_volp} is not applicable and the least-core would instead achieve efficiency, approximate individual rationality, and approximate stability.

For the coalitional game arising from the function~$v^s$, the Shapley value satisfies efficiency, but individual rationality, stability, and computational tractability are not attained. On the other hand, the nucleolus allocation provides a least-core allocation. Note that, in contrast to the expected coalitional value function~$\bar v$, the function~$v^s$ is not implicitly given by an optimization problem. Instead, it is an ex-post calculation from the sequential electricity market after the uncertainty realization. As a result, the value function~$v^s$ is not amenable to a constraint generation approach. Thus, we look at an alternative approach that can be computed in a computationally tractable manner. This approach will extend our results from~\cref{sec:4a}, showing that any efficient benefit allocation for the expected cost reduction gives rise to an efficient scenario-specific benefit allocation that results in budget-balance in every scenario. 

Let $\beta(\bar v)\in\mathbb{R}^{\AC}_+$ be an efficient individually rational benefit allocation for the expected cost reduction, computed prior to the uncertainty realization. We then define the following scenario-specific benefit allocation, $\beta(v^{s},\beta(\bar v))\in\mathbb{R}^{\AC}$,
\begin{equation}\label{eq:scen_spec}\beta_{a}(v^{s},\beta(\bar v))=\dfrac{\beta_a(\bar v)}{\bar v(\AC)} v^{s}(\AC),\quad \forall a\in\AC.\end{equation}
The benefit $\beta_{a}(v^{s},\beta(\bar v))$ for each area $a$ is computed based on an ex-post computation of $v^{s}(\AC)$ for the specific uncertainty realization $s$. Given $\beta(\bar v)$, this definition does not require any further solutions to the preemptive transmission allocation model or the sequential market. The term ${\beta_a(\bar v)}/{\bar v(\AC)}\in[0,1]$ can be considered as a percentage share of profits/losses depending on the sign of $v^s(\AC)$. (This also holds for any other weighting from the $|\AC|$-simplex.)
Notice that since $\sum_{a\in\AC}\beta_{a}(v^s,\beta(\bar v))=v^s(\AC)$, the efficiency property holds. Moreover, {having} $\mathbb E_s [\beta_{a}(v^s,\beta(\bar v))]=\beta_a(\bar v)$ implies that the scenario-specific benefit allocation $\beta(v^s,\beta(\bar v))$ satisfies in expectation the other fundamental properties of the original benefit allocation~$\beta(\bar v)$. 

Given the above reasoning, we propose a \textit{scenario-specific least-core-selecting mechanism}, which builds upon the least-core-selecting benefit allocation mechanism from problems~\eqref{eq:ba_opt}~and~\eqref{eq:ba_opt_tie} to define $\beta(v^{s},\hat{\beta}(\bar v,\beta^{\text{c}}))\in\mathbb{R}^{\AC}$ according to the procedure above. 
We have previously showed that the allocation $\hat{\beta}(\bar v,\beta^{\text{c}})$ satisfies individual rationality and approximate stability, while enabling a tractable computation via a constraint generation algorithm. In a similar vein, the scenario-specific version $\beta(v^s,\hat{\beta}(\bar v,\beta^{\text{c}}))$ satisfies individual rationality and approximate stability in expectation, while still enabling a tractable computation. We illustrate this approach in~\cref{sec:CaseStudies}. As a remark, it is possible to use the Shapley value and the nucleolus allocation in a similar manner. The comparisons of these mechanisms in the previous section would remain unchanged.

{In contrast to the least-core-selecting mechanism $\hat{\beta}(\bar v,\beta^{\text{c}})$, the scenario-specific benefit allocation $\beta(v^s,\hat{\beta}(\bar v,\beta^{\text{c}}))$ would lie in the least-core only in expectation and thus the benefits vary under different scenarios. This implies, in turn, that coalition member areas are now exposed to risk. In case some of these areas are risk-averse endowed with a risk measure, we cannot guarantee that each area's risk adjusted benefits lie in the least core, which can potentially hamper their willingness to participate. Defining ways to incorporate risk measures is part of our future research directions.}
	
	{Note that scenario-specific benefit allocations are implemented only after obtaining the perfect knowledge of the uncertainty realization, that is, when the actual system imbalance is known. This information is available for the balancing market clearing, which is consistent with the existing market models. Notice that if this realization is not within the scenario set, that is, $s^{{*}}\notin \mathcal{S}$, we can still compute the benefits defined in equation~\eqref{eq:scen_spec}. This is true because we can compute $v^{s^{{*}}}(\AC)=J^{s^{{*}}}(\emptyset)-J^{s^{{*}}}(\AC)$, where $J^{s^{{*}}}(\emptyset)$ and $J^{s^{{*}}}(\AC)$ are the market costs computed from the reserve, day-ahead, and balancing markets. These computed benefits would still incentivize the areas to participate in the preemptive model, if all the areas agree in advance that the scenario set describes how they perceive the distribution of the uncertainty {at the stage of transmission allocation process, which takes place before the reserve market clearing}. The desirable properties of the least core would hold if the expectation is taken with respect to the scenario set, however, these properties may not hold with respect to the true distribution.}
\section{Appendix}
\subsection{Proof of~\cref{prop:corempt}}\label{app:a3}

	Define $\hat\beta\in\mathbb{R}^{\AC}$ such that $\hat\beta_{a'}=\bar v(\AC)$, and $\hat\beta_{a}=0$ for all $a\neq a'.$ We prove by showing that this allocation lies in the core $K_\text{Core}(\bar v)$. The equality constraint in $K_\text{Core}(\bar v)$ is satisfied by definition. Notice that the condition given in the proposition, $\bar v( \AC\setminus a')=0$, implies that $\bar v(\CC\setminus a')=0$ for all $\CC\subset\AC.$ Using this, inequality constraints are divided into two sets of constraints as $\sum_{a\in\CC}\hat\beta_a\geq \bar v(\CC)$, for all $\CC\ni a'$, and $\sum_{a\in\CC}\hat\beta_a\geq 0$, for all $\CC\not\ni a'$. The first set of inequalities are satisfied, since $\bar v(\AC)\geq \bar v(\CC)$, for all $\CC\ni a'$. The second set of inequalities are satisfied, since $\hat\beta\in\mathbb{R}^{\AC}_+$. This concludes that $\hat\beta\in K_\text{Core}(\bar v)$, and hence the core is nonempty. \hfill\QEDA

According to the proof above, an area that satisfies the condition in~\cref{prop:corempt} has a right to veto any coalitional deviation that does not include it. The proof constructs the benefit allocation that assigns all the expected cost reduction to this area and then shows that the remaining areas do not have any incentives to form coalitions.
\subsection{Proof of~\cref{prop:leastcore_volp}}\label{app:a4}

This proof is an application of the proof method in \citep[Theorem 2.7]{maschler1979geometric} by taking into account that our coalitional value function is defined by $\bar v(\CC) = J(\emptyset) - J(\CC)\geq 0,$ for all $\CC\subset \AC$. We prove by contradiction. Let $\hat \beta \in K_\text{Core}(\bar v,\epsilon^*(\bar v))$ be a benefit allocation with $\hat \beta_{a'}<0$. In this case, we show that there exists $\epsilon <\epsilon^*(\bar v)$ such that $K_\text{Core}(\bar v,\epsilon)\neq \emptyset$. This would contradict the definition of the least-core.
	
	Since $\hat \beta \in K_\text{Core}(\bar v,\epsilon^*(\bar v))$ and $\bar v(a')=J(\emptyset) - J(a')=0$, we have $0>\beta_{a'}\geq \bar v(a') -\epsilon^*(\bar v)=-\epsilon^*(\bar v)$. Notice that, for any $\CC\not\ni a'$, we have $\sum_{a\in\CC}\hat\beta_a\geq \bar v(\CC\cup a')-\epsilon^*(\bar v)-\beta_{a'}$ by adding and subtracting $\beta_{a'}$, and by using the fact that $\epsilon^*(\bar v)>0$ for the case corresponding to $\CC\cup a'=\AC$. Since $\bar v$ is nondecreasing and $\beta_{a'}<0$, we obtain $\sum_{a\in\CC}\hat\beta_a> \bar v(\CC)-\epsilon^*(\bar v)$. 
	
	Notice that we can always find a small positive number $\delta$ such that  $\sum_{a\in\CC}\hat\beta_a - |\CC|\delta> \bar v(\CC)-\epsilon^*(\bar v)+\delta$ holds for any $\CC\not\ni a'$. Next, we show that $K_\text{Core}(\bar v,\epsilon^*(\bar v)-\delta)$ is nonempty for this particular choice. Define $\bar{\beta}$ such that $\bar\beta_a=\hat \beta_a - \delta$ for all $a\neq a'$and $\bar{\beta}_{a'}=\hat{\beta}_{a'}+(|\AC|-1)\delta$. This new allocation $\bar{\beta}$ clearly satisfies the equality constraint in $K_\text{Core}(\bar v,\epsilon^*(\bar v)-\delta)$. For inequality constraints $\CC\not\ni a'$, we have $\sum_{a\in\CC}\bar\beta_a=\sum_{a\in\CC}\hat\beta_a - |\CC|\delta> \bar v(\CC)-\epsilon^*(\bar v)+\delta$, where the strict inequality follows from the definition of $\delta$. For inequality constraints $\CC\ni a'$, we have $\sum_{a\in\CC}\bar\beta_a\geq\sum_{a\in\CC}\hat\beta_a+\delta\geq \bar v(\CC)-\epsilon^*(\bar v)+\delta$. Hence, $\bar \beta \in K_\text{Core}(\bar v,\epsilon^*(\bar v)-\delta)$, in other words, $K_\text{Core}(\bar v,\epsilon^*(\bar v)-\delta)\neq\emptyset.$ This contradicts $K_\text{Core}(\bar v,\epsilon^*(\bar v))$ being the least-core. \hfill\QEDA
	
\subsection{MILP reformulation of problem \eqref{eq:ba_findviolation}}\label{app:a6b}

 Since we have~$\bar v(\CC)=J(\emptyset)-J(\CC)$, optimal solutions to problem \eqref{eq:ba_findviolation} coincide with the ones to the following problem:
\begin{flalign}\label{eq:ba_findviolation_ref}
	&\underset{ \CC\subseteq\AC}{\min}\ J(\CC)+\sum_{a\in\CC}\beta_a^k .
\end{flalign}
The problem above is given by the following MILP:
\begingroup
\allowdisplaybreaks
\begin{subequations} \label{mod:B-Pre-k}
\begin{align}
& \bar J(\beta^k)=\,\underset{ \Phi_{\text{PR}}^k }{\min}\quad \sum_{i \in \mathcal I}\left(C^{+}_{i} r^{+}_{i} + C^{-}_i r^{-}_{i}\right) + \sum_{i\in \mathcal I}{C}_{i}{p}_{i} \nonumber\\
& \hspace{1.5cm}+ \sum_{s\in\mathcal S} \pi_s \Big[ \sum_{i \in \mathcal I}{C}_{i}\left({p}_{is}^{+}-{p}_{is}^{-}\right)+\sum_{n \in \mathcal N} {C}^{\text{sh}}l^{\text{sh}}_{ns} \Big]+\sum_{a\in\CC}b_a\beta_a^k \label{eq:B-pre-k-obj}
\end{align}
$\mathrm{s. t.}$
\begin{align}
 &\hspace{-.25cm} b_a\in\{0,1\},\quad\forall a\in \AC,\label{eq:preempt_bin}\\
		&\hspace{-.25cm}  (1-b_{a_r(e)})\chi_e \le \chi_{e}' \le (1-b_{a_r(e)})\chi_e+b_{a_r(e)}\ \text{and}\nonumber\\
		&\hspace{-.25cm}  (1-b_{a_s(e)})\chi_e \le \chi_{e}' \le (1-b_{a_s(e)})\chi_e+b_{a_s(e)},  \forall e\in\mathcal E,\label{eq:preempt_bin_ta}\\
&\hspace{-.25cm}  \text{Constraints} \;\; \eqref{eq:A-CV-rt-bal-constr} -\eqref{eq:A-CV-rt-flowAC-constr}\text{ and }{\delta}_{1 s} = 0, \ \delta_{n s} \; \text{free}, \ \forall n \in \mathcal{N},\quad \forall s\in \mathcal S,\label{eq:B-pre-k-CV}\\
	&\hspace{-.25cm} - b_{a_r(\ell)}M  \le f_{\ell s}-{f}_{\ell} \le b_{a_r(\ell)} M \ \text{and}\nonumber\\
	&\hspace{-.25cm} - b_{a_s(\ell)}M  \le f_{\ell s}-{f}_{\ell} \le b_{a_s(\ell)} M ,\quad\forall \ell\in\cup_{e\in\mathcal E(\chi)}\Lambda_{e},\quad \forall s\in \mathcal S,\label{eq:preempt_resact}\\
&\hspace{-.25cm}  \text{Constraints} \;\; \eqref{eq:LLR}\  \text{and}\ \eqref{eq:LLD}.
\end{align}
\end{subequations}
\endgroup
where $M$ is a large positive number, $\Phi_{\text{PR}}^k\allowbreak=\allowbreak\{b_a,\allowbreak \forall a \cup\chi_e',\allowbreak \forall e \cup \Phi_{\text R} \cup \Phi_{\text D} \cup  \Phi_{\text B}^{s}, \allowbreak\forall s\}$  is the set of primal optimization variables.
For the sake of brevity, the KKT conditions and the Lagrange multipliers of the lower-level optimization problems are omitted, see~\cref{app:a}.

In problem~\eqref{mod:B-Pre-k}, the parameter ${\beta}_a^k$ can be considered as the activation fee of area~$a$ for participating in the preemptive model. Given such activation fees, problem~\eqref{mod:B-Pre-k} then finds the optimal set of participants for the preemptive model. Notice that this problem involves $|\AC|$ more binary variables than problem~\eqref{mod:B-Pre}. In~\cref{sec:IEEE_RTS}, the numerical case studies illustrate that the computation times are still similar for both problems. Let $\{\hat{b}_a, \forall a\}$ denote the optimal binary solution. Finally, we have $\CC^k=\{a\in\AC\,|\,\hat{b}_a=1\}$ and $\bar v(\CC^k)=J(\emptyset)-(\bar J(\beta^k)-\sum_{a\in\CC^k}\beta^k_a)$. 

\subsection{Proof of \cref{prop:empty}}\label{app:a7}
	Assume $\beta\in K_\text{Core}(v^s)$. We now prove that this yields a contradiction. Notice that $K_\text{Core}(v^s)\subset\mathbb{R}_+^{\AC}$ since $\beta_{a}\geq v^s(a)=0.$ The scenario-specific core also implies $$\sum_{a\in \AC\setminus \CC}\beta_{a}\leq v^s(\AC)- v^s(\CC)<0.$$ The inequality above follows from combining $\sum_{a\in\AC}\beta_{a}=v^s(\AC)$ and $\sum_{a\in\CC}\beta_{a}\geq v^s(\CC)$. We obtained a contradiction. \hfill\QEDA
	
\chapter{Numerical case studies}
\label{sec:CaseStudies} 
In this chapter, we first use an illustrative three-area nine-bus system to provide and discuss benefit allocation mechanisms under different system configurations and stochastic renewable in-feed. We then apply the models and the benefit allocation mechanisms in a more realistic case study. 
In our numerics, all problems are solved with GUROBI 7.5~\cite{gurobi} called through MATLAB on a computer equipped with 32 GB RAM and a 4.0 GHz Intel i7 processor. 
\section{Illustrative three-area examples}\label{sec:illu_ex}
We describe a base model, which will be subject to several modifications in the system configuration and penetration of stochastic renewables to discuss the resulting changes in the benefit allocations described in \cref{sec:3} and~\cref{sec:4}.

We consider the nine-bus system depicted ine~\cref{fig:three_area} which comprises three areas. The intra-area transmission network consists of AC lines with capacity and reactance equal to $100$ MW and $0.13$ p.u., respectively. The four tie lines between areas~$1$ and~$2$, and between areas~$2$ and~$3$ are AC lines with capacity of $20$ MW, and reactance of $0.13$ p.u. each. 
\looseness=-1
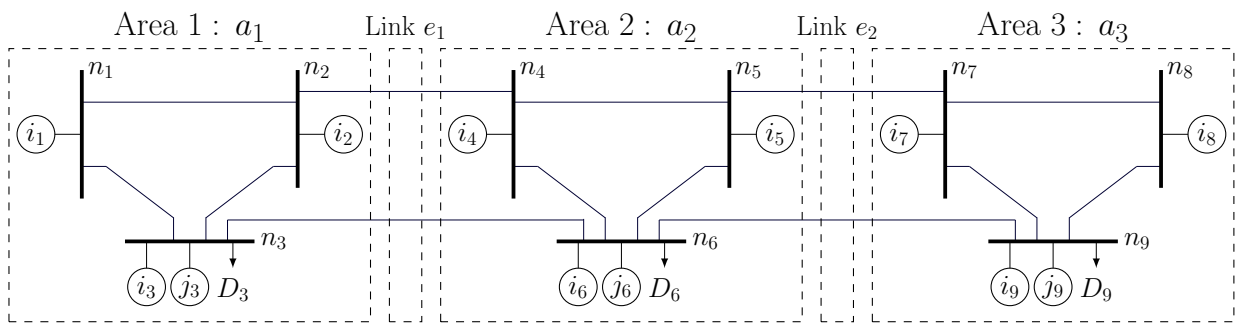
\begin{figure*}[h]
	\centering
	\begin{tikzpicture}[scale=0.71, every node/.style={scale=0.5}]
\draw[-,black!80!blue,line width=.1mm] (+0.7,1) -- (+0.7,1.4);
\draw[-,black!80!blue,line width=.1mm] (0.7,1.4) -- (7.3,1.4);
\draw[-,black!80!blue,line width=.1mm] (7.3,1) -- (7.3,1.42);
\draw[-,black!80!blue,line width=.1mm] (+8.7,1) -- (+8.7,1.4);
\draw[-,black!80!blue,line width=.1mm] (8.7,1.4) -- (15.3,1.4);
\draw[-,black!80!blue,line width=.1mm] (15.3,1) -- (15.3,1.4);
\draw[-,black!80!blue,line width=.1mm] (+0.3,1) -- (+0.3,1.45);
\draw[-,black!80!blue,line width=.1mm] (0.3,1.44) -- (1.55,2.4);
\draw[-,black!80!blue,line width=.1mm] (1.55,2.4) -- (2,2.4);
\draw[-,black!80!blue,line width=.1mm] (+8.3,1) -- (+8.3,1.45);
\draw[-,black!80!blue,line width=.1mm] (8.3,1.44) -- (9.55,2.4);
\draw[-,black!80!blue,line width=.1mm] (9.55,2.4) -- (10,2.4);
\draw[-,black!80!blue,line width=.1mm] (+16.3,1) -- (+16.3,1.45);
\draw[-,black!80!blue,line width=.1mm] (16.3,1.44) -- (17.55,2.4);
\draw[-,black!80!blue,line width=.1mm] (17.55,2.4) -- (18,2.4);
\draw[-,black!80!blue,line width=.1mm] (-0.3,1) -- (-0.3,1.45);
\draw[-,black!80!blue,line width=.1mm] (-0.3,1.44) -- (-1.55,2.4);
\draw[-,black!80!blue,line width=.1mm] (-1.55,2.4) -- (-2,2.4);
\draw[-,black!80!blue,line width=.1mm] (7.7,1) -- (7.7,1.45);
\draw[-,black!80!blue,line width=.1mm] (7.7,1.44) -- (6.45,2.4);
\draw[-,black!80!blue,line width=.1mm] (6.45,2.4) -- (6,2.4);
\draw[-,black!80!blue,line width=.1mm] (15.7,1) -- (15.7,1.45);
\draw[-,black!80!blue,line width=.1mm] (15.7,1.44) -- (14.45,2.4);
\draw[-,black!80!blue,line width=.1mm] (14.45,2.4) -- (14,2.4);
\draw[-,black!80!blue,line width=.1mm] (-2,3.6) -- (2,3.6);
\draw[-,black!80!blue,line width=.1mm] (6,3.6) -- (10,3.6);
\draw[-,black!80!blue,line width=.1mm] (14,3.6) -- (18,3.6);
\draw[-,black!80!blue,line width=.1mm] (2,3.8) -- (6,3.8);
\draw[-,black!80!blue,line width=.1mm] (10,3.8) -- (14,3.8);
\draw[-,line width=.5mm] (-1.2,1) -- (1.2,1) node[anchor=west]  {\LARGE $n_3$};
\draw[-,line width=.5mm] (6.8,1) -- (9.2,1) node[anchor=west]  {\LARGE $n_6$};
\draw[-,line width=.5mm] (14.8,1) -- (17.2,1) node[anchor=west]  {\LARGE $n_9$};
\draw[-,line width=.5mm] (-2,1.8) -- (-2,4.2) node[anchor=west]  {\LARGE $n_1$};
\draw[-,line width=.5mm] (6,1.8) -- (6,4.2) node[anchor=west]  {\LARGE $n_4$};
\draw[-,line width=.5mm] (14,1.8) -- (14,4.2) node[anchor=west]  {\LARGE $n_7$};
\draw[-,line width=.5mm] (2,2) -- (2,4.2) node[anchor=west]  {\LARGE $n_2$};
\draw[-,line width=.5mm] (10,2) -- (10,4.2) node[anchor=west]  {\LARGE $n_5$};
\draw[-,line width=.5mm] (18,2) -- (18,4.2) node[anchor=west]  {\LARGE $n_8$};
\draw[<-,line width=.1mm] (0.8,.5) -- (0.8,1) node at (0.8,0.1) {\LARGE $D_3$};
\draw[<-,line width=.1mm] (8.8,.5) -- (8.8,1) node at (8.8,0.1) {\LARGE $D_6$};
\draw[<-,line width=.1mm] (16.8,.5) -- (16.8,1) node at (16.8,.1) {\LARGE $D_9$};
\draw[-,line width=.1mm] (-0.8,.5) -- (-0.8,1);
\draw (-0.8,.15) circle (.35cm) node {\LARGE $i_3$};
\draw[-,line width=.1mm] (7.2,.5) -- (7.2,1);
\draw (7.2,.15) circle (.35cm) node {\LARGE $i_6$};
\draw[-,line width=.1mm] (15.2,.5) -- (15.2,1);
\draw (15.2,.15) circle (.35cm) node {\LARGE $i_9$};
\draw[-,line width=.1mm] (-2.5,3) -- (-2,3);
\draw (-2.85,3) circle (.35cm) node {\LARGE $i_1$};
\draw[-,line width=.1mm] (5.5,3) -- (6,3);
\draw (5.15,3) circle (.35cm) node {\LARGE $i_4$};
\draw[-,line width=.1mm] (13.5,3) -- (14,3);
\draw (13.15,3) circle (.35cm) node {\LARGE $i_7$};
\draw[-,line width=.1mm] (2.5,3) -- (2,3);
\draw (2.85,3) circle (.35cm) node {\LARGE $i_2$};
\draw[-,line width=.1mm] (10.5,3) -- (10,3);
\draw (10.85,3) circle (.35cm) node {\LARGE $i_5$};
\draw[-,line width=.1mm] (18.5,3) -- (18,3);
\draw (18.85,3) circle (.35cm) node {\LARGE $i_8$};
\draw[-,line width=.1mm] (0,.5) -- (0,1);
\draw (0,.15) circle (.35cm) node {\LARGE $j_3$};
\draw[-,line width=.1mm] (8,.5) -- (8,1);
\draw (8,.15) circle (.35cm) node {\LARGE $j_6$};
\draw[-,line width=.1mm] (16,.5) -- (16,1);
\draw (16,.15) circle (.35cm) node {\LARGE $j_9$};
	\draw[dashed,draw=black] (-3.35,-.5) rectangle (3.35,4.6) node at (0, 5) {\huge $\text{Area}\ 1:\, a_1$};
	\draw[dashed,draw=black] (4.65,-.5) rectangle (11.35,4.6) node at (8, 5) {\huge $\text{Area}\ 2:\, a_2$};
	\draw[dashed,draw=black] (12.65,-.5) rectangle (19.35,4.6) node at (16, 5) {\huge $\text{Area}\ 3:\, a_3$};
	\draw[dashed,draw=black] (3.7,-0.5) rectangle (4.3,4.6) node at (4, 5) {\LARGE $\text{Link}\ e_1$};
	\draw[dashed,draw=black] (11.7,-0.5) rectangle (12.3,4.6) node at (12, 5) {\LARGE $\text{Link}\ e_2$};
	\end{tikzpicture}
	\caption{Nine-node three-area interconnected power system}\label{fig:three_area}
\end{figure*}
\newpage
The day-ahead price offers and the generation capacities of conventional units are provided in~\cref{tab:three_generators}. Units $i_1$, $i_4$, and $i_7$ are inflexible, that is, these units cannot change their generation level during real-time operation, while all remaining units are flexible offering half of their capacity for upward and downward reserves provision at a cost equal to $10\%$ of their day-ahead energy offer~$C$. The cost of load shedding $C^\text{sh}$ is equal to $1000$\euro$/$MWh for the inelastic electricity demands $D_3=220$ MW, $D_6=190$ MW, and $D_9=220$ MW. 

\begin{table}[h]
\centering
\caption{Generator data}
{\begin{tabular}{|c||c|c|c|c|c|c|c|c|c|}
\hline
 Unit &$i_1$ &$i_2$&$i_3$&$i_4$&$i_5$&$i_6$&$i_7$&$i_8$&$i_9$ \\ \hline
 $C$ (\euro$/$MWh) & $20$ &$30$ &$40$ &$30$ &$40$ &$50$ &$25$ &$35$ &$45$ \\ \hline
 $P^\text{cap}$ (MW) & $120$ &$50$ &$50$ &$120$ &$50$ &$50$ &$120$ &$50$ &$50$ \\ \hline
 Flexible & No & Yes & Yes & No & Yes & Yes & No & Yes & Yes \\ \hline
\end{tabular}}\label{tab:three_generators}
\end{table}

In addition, there are three wind power plants, $j_3,$ $j_6,$ and $j_9$, with installed capacities $50$, $80$, and $50$ MW, respectively. The stochastic wind power is modeled using two scenarios, $s_1$ and $s_2$, listed in~\cref{tab:three_wind} with probability of occurrence $0.6$ and $0.4$, respectively. The expected wind power production $\overline{W}_j$ for $j_3$ is equal to $42$ MW, for $j_6$ is equal to $70.4$ MW, and for $j_9$ is equal to $42$ MW. Wind power price offers and subsequently the wind power spillage costs are considered to be zero.

\begin{table}[h]
\centering
\caption{Wind scenarios as percentage of the nominal value of the power plant }
{\begin{tabular}{|c||c|c|c|}
\hline
 Wind power plant & $j_3$ &$j_6$ & $j_9$ \\ \hline
 Scenario $s_1$&$1$ &$0.8$&$1$ \\ \hline
 Scenario $s_2$&$0.6$ &$1$&$0.6$ \\ \hline
\end{tabular}}\label{tab:three_wind}
\end{table}

Following the prevailing approach in which regional capacity markets are cleared separately, we set the percentage of transmission capacity allocated to reserves equal to $\chi=0$. Reserve requirements for each area are calculated based on the probabilistic forecasts, such that the largest negative and positive deviations from the expected wind power production foreseen in the scenario set are covered by domestic resources. For instance, the upward and downward area reserve requirements for area $1$ are calculated~as, $RR^+_{a_1}  = \overline{W}_{j_3} - \min\{W_{j_3s_1},W_{j_3s_2}\}= 12\, \text{MW},$ and $RR^-_{a_1}  = \max\{W_{j_3s_1},W_{j_3s_2}\}-\overline{W}_{j_3} = 8\, \text{MW}.$
For the other areas, these values are given as $RR^+_{a_2}  =6.4\, \text{MW}$, $RR^-_{a_2}  =9.6\, \text{MW}$, $RR^+_{a_3}  =12\, \text{MW}$, $RR^-_{a_3}  =8\, \text{MW}$.

The market costs and transmission allocations resulting from the preemptive model are provided in~\cref{tab:three_costs}. The preemptive model reallocated transmission resources from the day-ahead energy trading to the reserve capacity trading, increasing the costs in the day-ahead market. This reallocation yields an expected system cost of $13{,}238.0$\euro, which translates to $25.9\%$ reduction compared to the cost of $17{,}871.2$\euro\ from the existing sequential market. Under the existing setup with $\chi=0$, the uncertainty realization $s_2$ leads to significant load shedding in the balancing stage. In this scenario, even though we have enough reserve capacity, we are not able to deploy it due to network congestion. This problem is avoided by enabling reserve exchange when the preemptive model is implemented. Quantities assigned to each generator at all trading floors are provided in~the appendix in~\cref{app:market_outcome}. 
\looseness=-1

\begin{table}[h]
\centering
\caption{Comparison of market costs (in \euro)}
{\begin{tabular}{|c||c|c|}
\hline
 Model & Existing Seq. Market & Preemptive Model \\ \hline
 $[\chi_{e_1},\,\chi_{e_2}]$ & $[0,\,0]$ &$[0,\,0.0592]$\\ \hline 
 Reserve capacity cost &$194.0$ &$191.6$\\ \hline 
 Day-ahead cost &$13{,}087.2$ &$13{,}120.2$\\ \hline 
 Balancing cost in $s_1$ &$1{,}150.0$ &$-410.7$\\ \hline 
 Balancing cost in $s_2$ &$9{,}750.0$ &$431.5$\\ \hline 
 Total cost in $s_1$ &$14{,}431.2$ &$12{,}901.2$\\ \hline 
 Total cost in $s_2$ &$23{,}031.2$ &$13{,}743.4$\\ \hline 
\end{tabular}}\label{tab:three_costs}
\end{table}

We now provide a budget-balanced cost allocation method for the existing sequential market. For this method, we assume that all three trading floors are cleared by marginal pricing mechanisms (zonal prices for the reserve capacities, nodal prices for the day-ahead and balancing energy services), albeit, similar methods can be applied also to other payment mechanisms. 
This method assigns producer and consumer surpluses, and congestion rents of the intra-area lines to their corresponding areas, and divides the congestion rents of the tie lines equally between the adjacent areas, see~\cite{kristiansen2018mechanism}. budget-balance holds since the market cost is given by the opposite of the sum of producer and consumer surpluses, and congestion rent for each trading floor.  These values are summarized in~the appendix in~\cref{app:market_outcome}. We refer to~\cref{tab:budget_out} for the resulting cost allocations. Area~1 is allocated a large cost in scenario $s_2$ because of the load shedding in node~3.
\looseness=-1

\begin{table}[h]
\centering
\caption{Cost allocation for each area in the existing sequential market (in \euro)}
{\begin{tabular}{|c||c|c|c|}
\hline
 Areas & Area $1$  & Area $2$ & Area $3$\\ \hline
 $J^{s_1}_a(\emptyset)$& $4{,}348.4$  & $9{,}853.8$ & $229.0$\\ \hline
 $J^{s_2}_a(\emptyset)$& $16{,}348.4$  & $3{,}453.8$ & $3{,}229.0$\\ \hline
\end{tabular}}\label{tab:budget_out}
\end{table}


\subsection{Comparison of the different benefit allocations}\label{sec:shapno}

Benefit allocation mechanisms for the expected cost reduction are provided in~\cref{fig:three_benefit}. The core is nonempty since area 2 satisfies the veto condition in~\cref{prop:corempt}. Marginal contribution benefit allocation $\beta^\text{m}$ is not in the core since it is not efficient. We provide this allocation because it can be regarded as a fair outcome. Observe that the marginal contributions of areas 1 and 2 are larger than that of area 3. This is because area 1 has low cost generators and area 2 is indispensable for any coordination considering that in the current network configuration, in which areas 1 and 3 are not directly interconnected, area 2 has to act as an intermediary for any reserves exchange.  

\begin{figure}[h]
	\centering
	\begin{tikzpicture}[scale=1, every node/.style={scale=0.35}]
			\node at (0,0) {\includegraphics[width=2\textwidth]{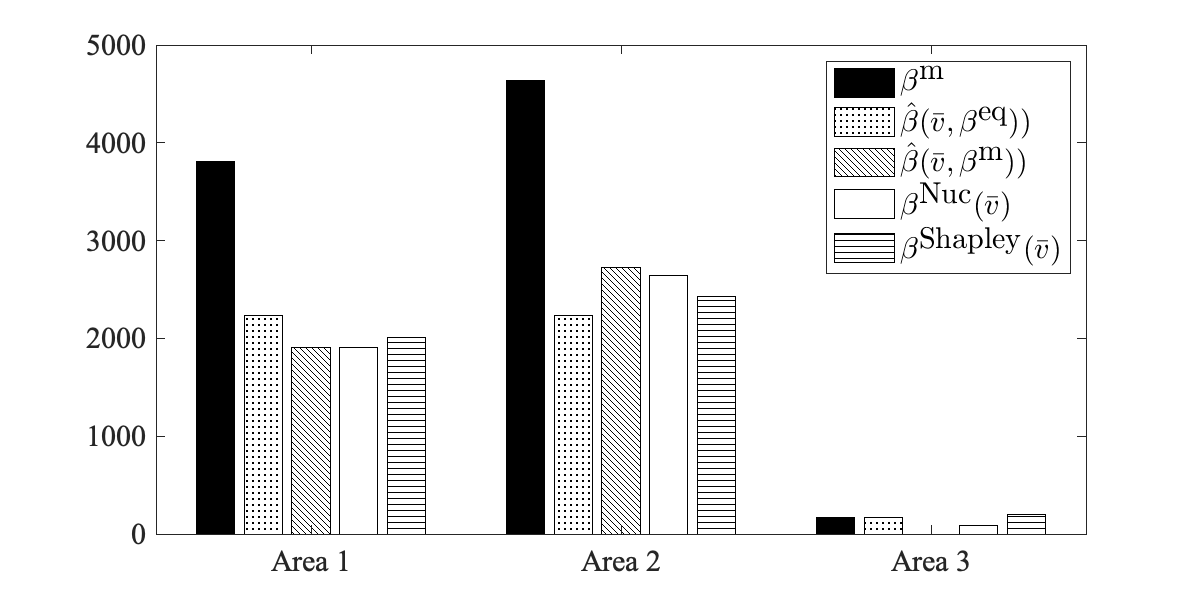}};
	\end{tikzpicture}
	\caption{Benefit allocations for the three-area system (in \euro)}\label{fig:three_benefit}
\end{figure}

The Shapley value $\beta^\text{Shapley}$ is not in the core. Among the core constraints, combining $\sum_{a\in\AC}\beta_a=v(\AC)$ with $\sum_{a\in\AC\setminus\hat{a}}\beta_a\geq v(\AC\setminus\hat{a})$ implies that $\beta_{\hat a}\leq\beta^\text{m}_{\hat a}=v(\AC)-v(\AC\setminus\hat{a})$, or equivalently, no area can receive more than its marginal contribution in the core. This condition is violated for the Shapley value assigned to area 3. The coalitional value function is also not supermodular, since
$$\bar v(\{1,2,3\})-\bar v(\{1,2\})\not\geq \bar v(\{2,3\})-\bar v(\{2\})\implies
    4{,}633.1-4{,}460.5 \not\geq 826.8-0.$$ On the other hand, the nucleolus $\beta^\text{Nuc}$ is in the core, however, the lexicographic minimization results in allocating benefits to area~3. We later see that there is a core allocation that better approximates the marginal contribution in terms of minimizing the Euclidean distance by allocating no benefits to~area~3.

Finally, we employ our approach approximating two different criteria, that is, marginal contribution and equal shares, with corresponding allocations being denoted as $\hat{\beta}(\bar{v},\beta^{\text{m}}))$ and $\hat{\beta}(\bar{v},\beta^{\text{eq}}))$, respectively. These two outcomes are different from each other, and they approximate their respective fairness consideration in an effective manner. This criterion should be decided either by the regulator or it should be based on the consensus of participating areas. In the following, we will approximate the marginal contribution, since similar discussions can be made for any other criteria.
\looseness=-1

Next, we study the budget-balance per scenario for the cost allocation in the preemptive model. For all efficient benefit allocations of the expected cost reduction, that is, all methods except the marginal contribution allocation, the budget  $\sum_{a\in\AC}J^s_a(\AC)-J^{s}(\AC)$ remains unchanged. In scenario $s_1$, there is a deficit of $3{,}103.1$\euro, whereas in scenario $s_2$ there is a surplus of $4{,}654.7$\euro, thus budget-balance is obtained in expectation. 

In the coalitional game arising from the scenario-specific cost variation, despite that $K_\text{Core}(\bar v)$ is nonempty, the core $K_\text{Core}(v^{s_1})$ is empty, since the condition in~\cref{prop:empty} is satisfied by \begin{equation*}\begin{split}
v^{s_1}(\{1,2\})&=J^{s_1}(\emptyset) - J^{s_1}(\{1,2\})\\ &= 14{,}431.2-12{,}884.6\\ &> 14{,}431.2-12{,}901.2 \\&= J^{s_1}(\emptyset) - J^{s_1}(\{1,2,3\})=v^{s_1}(\{1,2,3\}).
\end{split}
\end{equation*}
For scenario $s_2$, this condition is not satisfied and $K_\text{Core}(v^{s_2})$ is nonempty, since the coalitional game is supermodular. 

To address the budget-balance, we now employ {the proposed  scenario-specific least-core-selecting mechanism.} The scenario-specific allocations generated by $\hat{\beta}(\bar v,\beta^{\text{m}})$ for the expected cost reduction are given by $\beta(v^{s_1},\hat{\beta}(\bar v,\beta^{\text{m}}))=[628.5,\, 901.5,\, 0]^\top$~and~$\beta(v^{s_2},\allowbreak\hat{\beta}(\bar v,\beta^{\text{m}}))=[3{,}815.2,\, 5{,}472.6,\, 0]^\top$. These allocations result in a budget-balanced cost allocation under both scenarios, since they sum up to the scenario-specific cost variations in~\cref{tab:three_costs}. 

We now illustrate {the scenario-specific least-core-selecting mechanism} in an out-of-sample wind scenario $s_3$ as follows. In this scenario, the wind power plant productions of $j_3$, $j_6$, and $j_9$ are given by $0.7$, $0.8$, and $1$, respectively. All the other parameters remain unchanged. We computed $J^{s_3}(\emptyset)=18{,}428.7$ and $J^{s_3}(\{1,2,3\})=13{,}394.4$. The scenario-specific allocations generated by $\hat{\beta}(\bar v,\beta^{\text{m}})$ are given by $\beta(v^{s_3},\hat{\beta}(\bar v,\beta^{\text{m}}))=[2{,}068.0,\, 2{,}966.3,\, 0]^\top$. Clearly, this benefit allocation results in a budget-balanced cost allocation even if this scenario was not included in the scenario set. However, since our scenario set is not modeling the uncertainty exactly, we cannot guarantee that the properties pertaining to the least core are still satisfied in expectation. In the remainder, we focus our efforts on the game arising from the expected cost reduction, since we can always map the benefit allocations to the scenario-specific case using our proposed approach. 

{The computational comparison for the three benefit allocation mechanisms for the expected cost reduction can be summarized as follows. The Shapley value is computed in $39.6$ seconds, whereas the  nucleolus is computed in $41.3$ seconds. On the other hand, the least-core allocations for marginal contribution and equal shares are computed in $6.3$ and $6.7$ seconds, respectively. Following our previous discussions, the computation time difference between our methods and the others will be even more significant when there are more areas. Because of this reason, the computational methods will be studied and discussed in detail for the larger realistic case study.}
\subsection{Impact of the network topology on benefit allocations}\label{sec:three_area_netw_top}
In this example, our goal is to illustrate how the network topology and the specific location of each area in the electricity network affects the benefits allocated to this area. To this end, we modify our base model by removing the wind generator from area 2 and by changing all the units in area 2 to be inflexible. Benefit allocations for the expected cost reduction are provided in~\cref{fig:three_removeflex}. We highlight that removing a zero marginal cost wind generator increased the costs globally. In this setup, area 2 continues to receive the highest benefits under every allocation mechanism, even though this area is not capable of directly participating in any reserve exchange. Nevertheless, the role of area 2 in this network arrangement is instrumental, since it enables the coordination between areas 1 and 3 because of its location in the area graph.

\begin{figure}[h]
	\centering
	\begin{tikzpicture}[scale=1, every node/.style={scale=0.35}]
			\node at (0,0) {\includegraphics[width=2\textwidth]{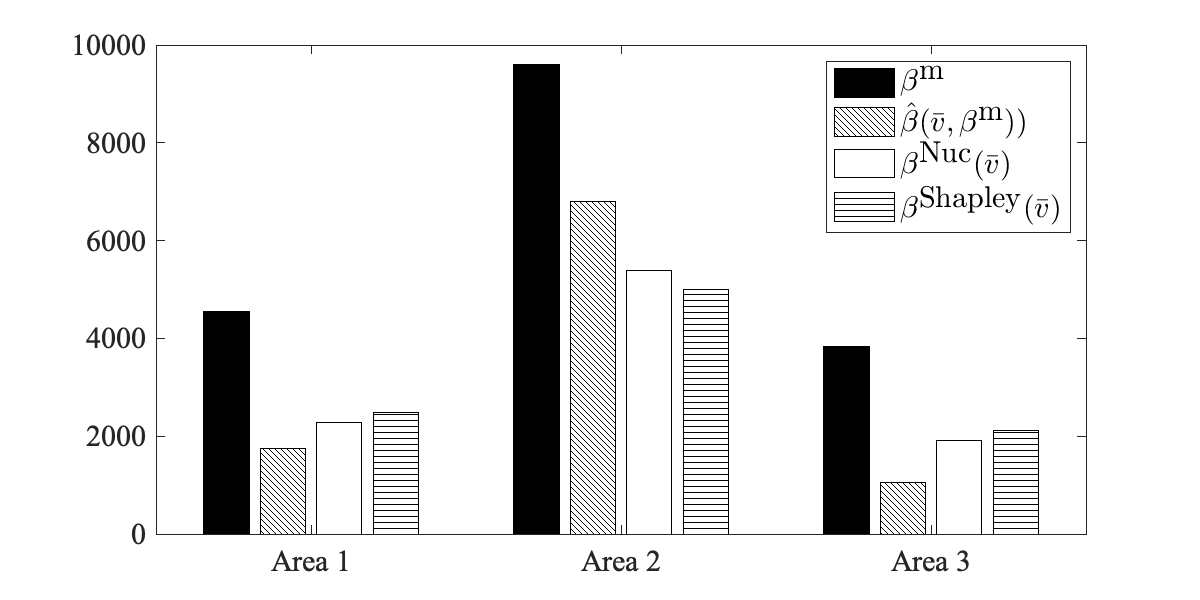}};
	\end{tikzpicture}
	\caption{Benefit allocations after removing the flexibility and the uncertainty from area 2 (in \euro)}\label{fig:three_removeflex}
\end{figure}

As a remark, the core is nonempty, since the condition in~\cref{prop:corempt} is satisfied. The Shapley value is also in the core, since we verified that the coalition value function is supermodular.

\subsection{Impact of the uncertainties on benefit allocations}\label{sec:threearea_windprof}
Here, we aim to assess the impact of the spatial correlation of the wind power forecast errors on the outcome of the different benefit allocation mechanisms that we consider in this work. In order to eliminate the impact of the network topology (cf.~\cref{sec:three_area_netw_top}), we connect areas 1 and 3 via two AC lines. The first connects nodes 1 and 8, and the second connects nodes 3 and 9, each with transmission capacity of $20$ MW, and reactance of $0.13$ p.u. The area graph is not a star anymore, and~\cref{prop:corempt} is not applicable. However, we verified that the core is still nonempty. 

The resulting benefit allocations for the expected cost reduction are provided in~\cref{fig:three_connected}, which shows that areas~1 and~2 receive most of the benefit under every allocation mechanism. This outcome can be explained considering that these areas have complementary wind power production scenarios, that is, the corresponding wind power scenarios exhibit negative correlation. Moreover, area~1 has low cost generation.
Finally, notice that the total benefits are greater than the ones from the example in~\cref{sec:shapno}. This follows since compared to~\cref{sec:shapno} expected system cost is increased by $49$\% ($26{,}687.9$\euro) in the existing sequential market due to additional network dependencies, whereas this cost is decreased by $0.01$\% ($13{,}161.0$\euro) in the preemptive model.
\begin{figure}[h]
	\centering
	\begin{tikzpicture}[scale=1, every node/.style={scale=0.35}]
			\node at (0,0) {\includegraphics[width=2\textwidth]{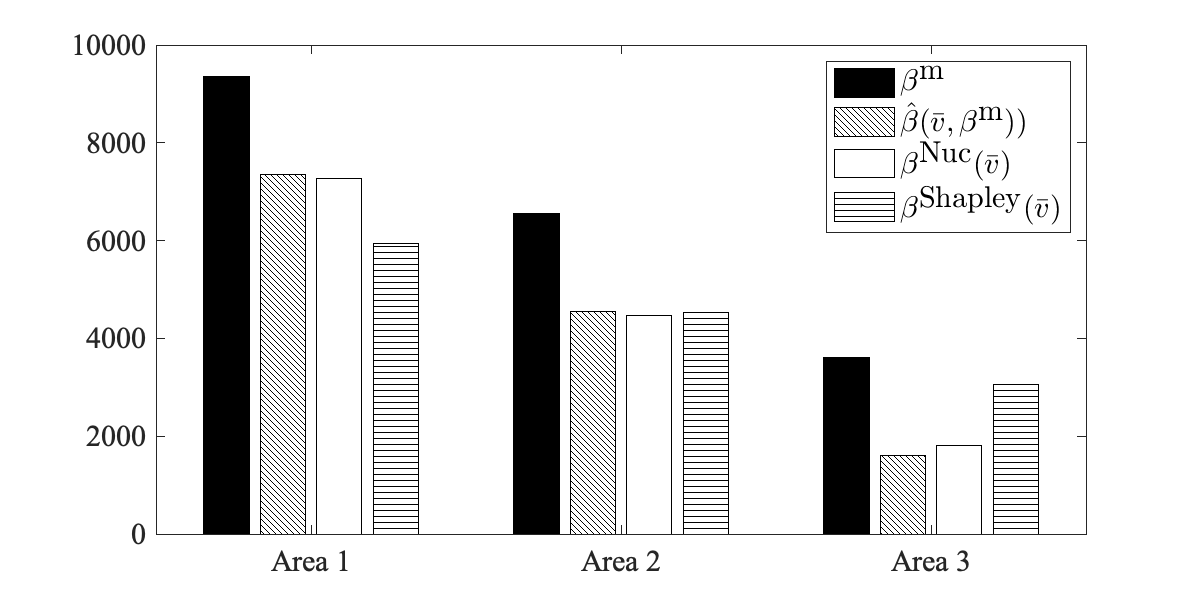}};
	\end{tikzpicture}
	\caption{Benefit allocations after connecting areas 1 and 3 (in \euro)}\label{fig:three_connected}
\end{figure}

\subsection{Benefit allocations in the case of an empty core}\label{sec:emptycore}
A natural question that arises in the context of this work is how the different benefit allocation mechanisms perform when we have an empty core, which can occur when the condition in~\cref{prop:corempt} is not satisfied.
To this end, we modify the example in~\cref{sec:threearea_windprof} by changing the wind scenarios. 

The stochastic wind power generation is modeled using two scenarios, $s_1$ and $s_2$ with probability of occurrence $0.8$ and $0.2$. We have $1$ and $0.8$ for $j_3$, $0.4$ and $1$ for $j_6$, $0.4$ and $1$ for $j_9$ as the percentages of the nominal values of the plants, respectively. Hence, the corresponding expected wind power productions for $j_3$ is equal to $48$ MW, for $j_6$ is equal to $41.6$ MW, and for $j_9$ is equal to $26$ MW. The reserve requirements are recomputed accordingly. Since the uncertainty is significantly increased, we allow the units $i_1$, $i_4$, and $i_7$ to be flexible in order to ensure feasibility.

The resulting benefit allocations for the expected cost reduction are provided in~\cref{fig:three_empty}.
We observe that the nucleolus and the least-core-selecting benefit allocation coincide. For both allocations, the maximum violation of a stability constraint is given by $\epsilon^*=924.9$\euro, where $\epsilon^*(\bar v)=\min\{\epsilon\,|\,K_\text{Core}(\bar v,\epsilon)\not=\emptyset\}$.
On the other hand, the maximum stability violation for the Shapley value is $2{,}752.0$\euro. In other words, if the Shapley value is utilized, there are $3$ times the profits to be made by not participating in the preemptive model compared to the case implementing a least-core allocation. We see that all benefit allocation mechanisms allocated the most benefits to area 1, since it has low cost generation and also its wind profile complements the wind profiles of areas 2 and 3.
\looseness=-1

\begin{figure}[h]
	\centering
	\begin{tikzpicture}[scale=1, every node/.style={scale=0.8}]
			\node at (0,0) {\includegraphics[width=.92\textwidth]{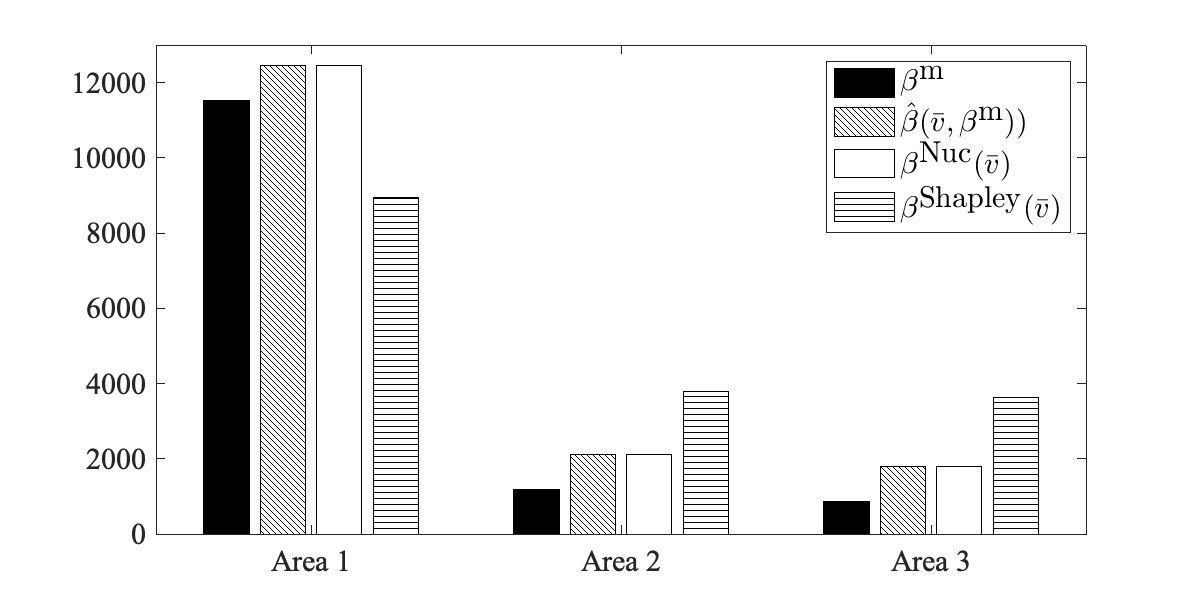}};
	\end{tikzpicture}
	\caption{Benefit allocations in case the core is empty (in \euro)}\label{fig:three_empty}
\end{figure}
\section{Case study based on the IEEE RTS}\label{sec:IEEE_RTS}
We now consider a six-area power system that is based on the modernized version of the IEEE Reliability Test System (RTS) presented in \cite{pandzicunit}. 
The definitions of the areas correspond to the ones proposed by~\cite{dvorkin2018setting} and \cite{jensen2017cost}, and they are provided in~the appendix in~\cref{app:IEEERTS96_area}.  The nodal positions, types, generation capacities and offers from conventional and wind power generators, and transmission line parameters are provided in the online repository~\cite{karacadata}. 

Due to their limited flexibility, nuclear, coal and integrated gasification combined cycle (IGCC) units do not provide any reserves. On the other hand, open and combined cycle gas turbines (OCGT and CCGT) offer $50\%$ of their capacity for upward and downward reserves at a cost equal to $20\%$ of their day-ahead energy offer. {Wind power production is modeled using a set of 10 equiprobable scenarios obtained from \cite{bukhsh2017data}. This scenario set is originally generated according to the methodology explained in \cite{papaefthymiou2008modeling}, and it captures the spatial correlation of forecast errors over the different wind farm locations, see also~\cref{fig:rts_corr} and system dataset in~\cite{karacadata}. In our case study, areas $2$, $4$, and $5$ are assumed to be close to each other, and hence the corresponding wind power production exhibits higher correlation. Evaluation of this aspect for different areas will be shown to be a useful predictor for the benefit allocation methods.} The demand is inelastic with the large cost of load shedding $1{,}000$\euro$/$MWh. In the existing sequential market, the percentage of transmission capacity allocated to reserves exchange is set to $\chi=0$, while the reserve requirements are calculated according to the methodology discussed in~\cref{sec:illu_ex}. 

\begin{figure}[h]
    \centering
	\begin{tikzpicture}[scale=1, every node/.style={scale=1}]
			\node at (0,0) {\includegraphics[width=.6\textwidth]{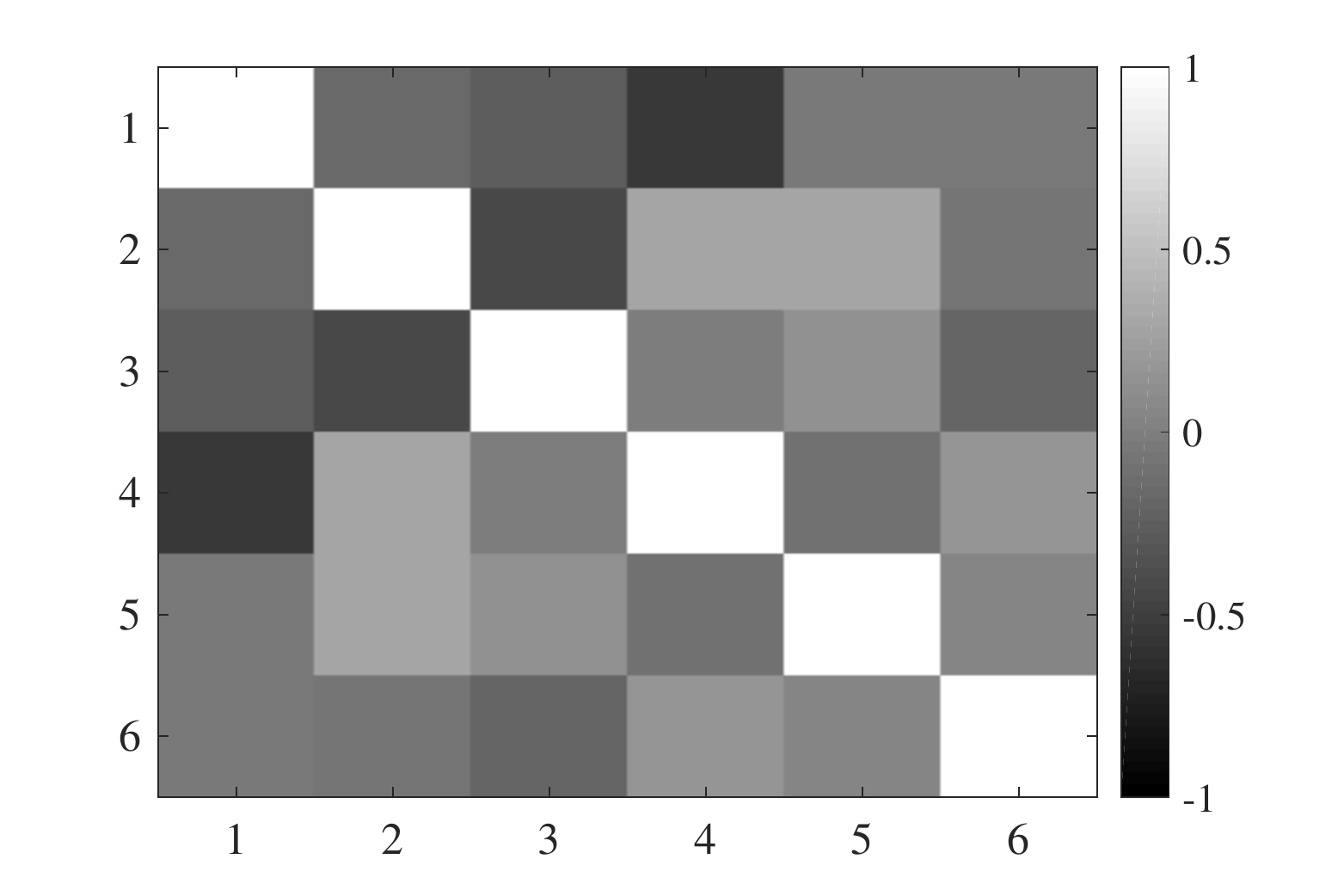}};
	\end{tikzpicture}
	\caption{The Pearson correlation coefficients for the wind profiles of all areas}\label{fig:rts_corr}
\end{figure}

\cref{tab:rts_costs} compares the costs and transmission allocations resulting from the existing market with $\chi=0$ and the preemptive model where $\chi$ is a decision variable. The preemptive model yields an expected cost of $87{,}594.8$\euro, which translates to $2.3\%$ reduction compared to $89{,}696.4$\euro\ from the existing market.
This can be explained by $2{,}097.6$\euro\ reduction in the expected balancing cost obtained by eliminating load shedding. Using the approach in~\cref{sec:illu_ex}, we provide the expected values for a budget-balanced cost allocation for the existing sequential market in~\cref{tab:rts_budget_out}.

\begin{table}[h]
\centering
\caption{Comparison of market costs (in \euro)}
{\begin{tabular}{|c||c|c|}
\hline
 Model & Existing Seq. Market & Preemptive Model \\ \hline
 $\chi=[\chi_{e_1},\ldots,\chi_{e_7}]$ & $[0,\,0,\,0,\,0,\,0,\,0,\,0]$ &$[0,\,0.359,\,0,\,0.038,\,0,\,0,\,0]$\\ \hline 
 Reserve capacity cost & $2{,}392.5$ & $2{,}389.4$\\ \hline 
 Day-ahead cost & $90{,}734.2$& $90{,}733.3$\\ \hline 
 {Expected} balancing cost  &$-3{,}430.3$& $-5{,}527.9$\\ \hline 
 {Expected} cost & $89{,}696.4$ & $87{,}594.8$\\ \hline 
\end{tabular}}\label{tab:rts_costs}
\end{table}

\begin{table}[h]
\centering
\caption{Expected cost allocation for each area in the existing sequential market (in \euro)}
{\begin{tabular}{|c||c|c|c|c|c|c|}
\hline
 Areas & Area $1$  & Area $2$ & Area $3$ &Area $4$  & Area $5$ & Area $6$\\ \hline
 $\mathbb{E}_s[J^{s}_a(\emptyset)]$ &$8{,}966.1$ & $25{,}252.6$    & $10{,}085.2$ & $10{,}208.2$ &  $25{,}151.1$ &  $10{,}033.2$\\ \hline
\end{tabular}}\label{tab:rts_budget_out}
\end{table}

The results of the different benefit allocation mechanisms for the expected cost reduction are provided in~\cref{fig:rts_benefit}. 
We verified that the core is nonempty even though the condition in~\cref{prop:corempt} is not satisfied. Marginal contribution benefit allocation is not in the core since it is not efficient, whereas the Shapley value is not in the core since areas 3, 5, and~6 receive more than their marginal contributions. The nucleolus and the least-core-selecting benefit allocation mechanisms result in core allocations. Notice that our approach provides a different benefit allocation depending on the criteria considered. The nucleolus allocation is not consistent with the marginal contribution allocation since it allocates more benefits to area~4 compared to area~1. 

All mechanisms allocated the most benefits to area 2, since it has a central role by being well-connected in the area graph. On the other hand, areas 1 and 4 are also allocated a significant amount, since they are the two largest areas with wind profiles complementing each other as it is shown in~\cref{fig:rts_corr}.

\begin{figure}[h]
	\centering
	\begin{tikzpicture}[scale=1, every node/.style={scale=0.4}]
			\node at (0,0) {\includegraphics[width=2\textwidth]{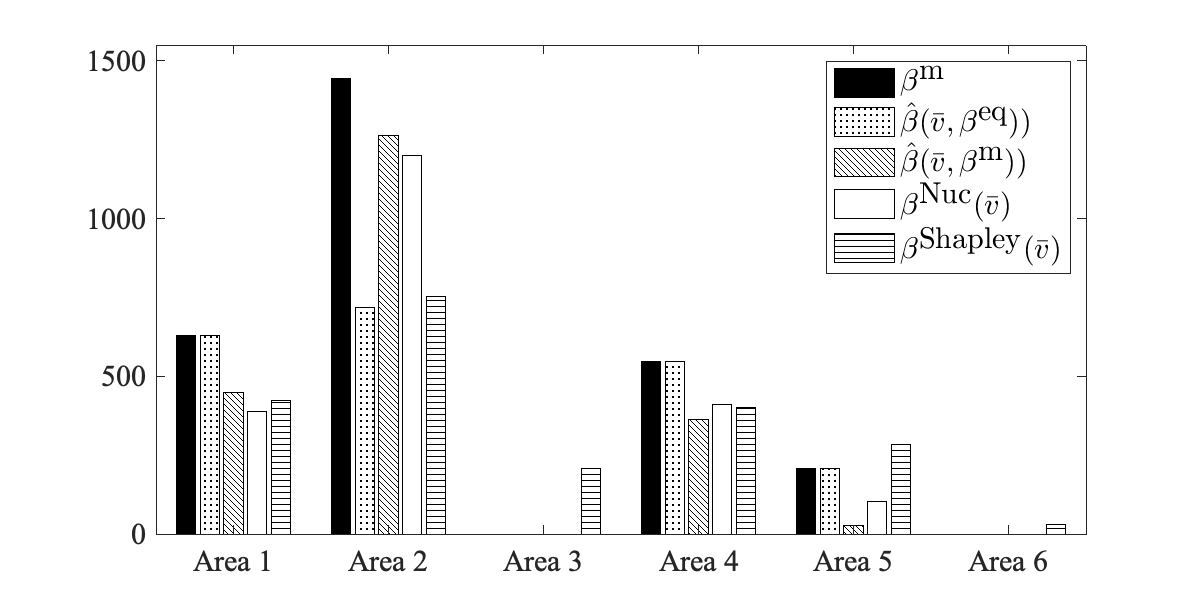}};
	\end{tikzpicture}
	\caption{Benefit allocations for the IEEE RTS case study (in \euro)}\label{fig:rts_benefit}
\end{figure}

We now provide a discussion on the computational comparison for the different benefit allocation mechanisms. The coalitions $J(\emptyset)$ and $J(\AC)$ are precomputed to obtain $v(\AC)$, in $19.4$ and $35.4$ seconds, respectively. 
The calculation of the marginal contribution allocation, which involves solving the preemptive model~\eqref{mod:B-Pre} for coalitions $\{\AC\setminus\{a\}\}_{a\in\AC}$, requires $119.7$ seconds. 
The Shapley value requires solving the preemptive model~\eqref{mod:B-Pre} for all coalitions except the singleton sets, the empty set, and the full set, that is, $2^6-6-1-1=58$ coalitions, and the resulting computational time is $1{,}264.6$ seconds. 
The least-core-selecting mechanism with the marginal contribution criteria requires only a single iteration from the constraint generation algorithm, which takes $84.8$ seconds. {This constraint generation algorithm converges fast, since in this case the algorithm starts with an initial family of coalitions ($\mathcal{F}^1$) given by the coalitions that were used to compute the marginal contribution allocation.} 
On the other hand, the least-core-selecting mechanism with the equal shares criteria requires four iterations from the constraint generation algorithm, which takes $150.4$ seconds. Notice that, in this case, the initial family of coalitions is an empty set. 
{Finally, using the method proposed in~\cite{kopelowitz1967computation,fromen1997reducing}, the nucleolus is computed by solving $15$ linear programs sequentially to find $21$ coalitional equality constraints that fully describe the nucleolus allocation. However, this method scales exponentially with the number of areas considered in the application, since it needs the complete list of coalition values. This computation takes $1{,}266.2$ seconds.}\footnote{As an alternative, the iterative method in~\cite{hallefjord1995computing} would require running $15$ separate constraint generation algorithms, {increasing significantly the computational time compared to} the least-core-selecting mechanism, since each algorithm run requires at least one iteration of constraint generation.} 

Next, we provide two modifications to the IEEE RTS case study.

\subsection{Impact of the wind power penetration levels on benefit allocations}

The wind power penetration level of an area is defined as the ratio between the expected wind power production and the total demand of that area. In this example, we change the level of wind power penetration for area 1. The default value is given by $30\%$ from the previous section. The resulting benefits allocations for area~1 are provided in~\cref{fig:rts_benefit_2}. For all efficient benefit allocation mechanisms, observe that the benefits initially increase and then decrease with the wind power penetration level. 

It is generally hard to anticipate such changes in the benefits, since the wind power generation has two impacts acting in opposite directions. {On the one hand}, it has null production cost bringing in low cost energy to the coalition. {On the other hand}, it increases the need for the reserve and balancing services. As a remark, in this study, the core is empty for the levels $15\%$ and $22.5\%$.

\begin{figure}[h]
	\centering
	\begin{tikzpicture}[scale=1, every node/.style={scale=0.35}]
			\node at (0,0) {\includegraphics[width=2\textwidth]{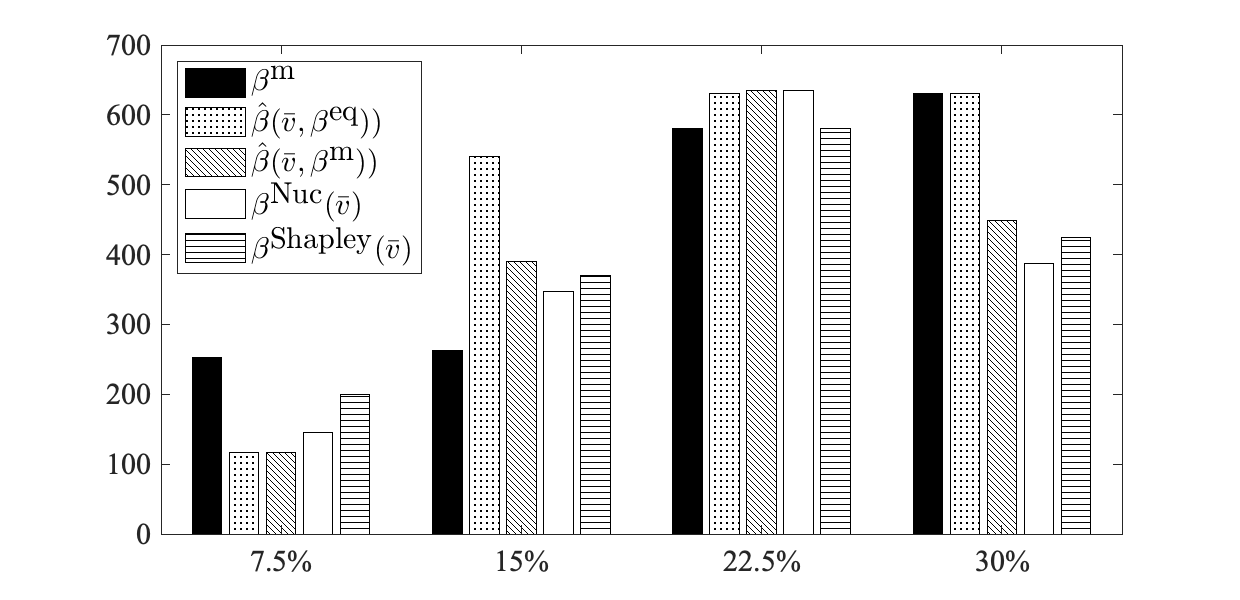}};
	\end{tikzpicture}
	\caption{Benefit allocations of area 1 for the different levels of wind power penetration in area 1 (in \euro)}\label{fig:rts_benefit_2}
\end{figure}

\subsection{Impact of the available flexibility on benefit allocations}
In order to assess the impact of {available flexibility}, we {double} the capacities of all flexible generators in area 4. The changes in the benefits allocations are provided in~\cref{fig:rts_benefit_3}, denoted by $\Delta \beta$. The figure shows that the benefits allocated to area~4 are increased under every allocation mechanism compared to~\cref{fig:rts_benefit}. All mechanisms account for the flexibility offered by the generators in area~4 both in the reserve capacity and the balancing markets. Note that an increase in the capacity of flexible generators reduces also the day-ahead system cost both with or without the implementation of the preemptive model.

Finally, predicting the changes in the benefits of the other areas is quite difficult. We observe that this additional flexibility replaced the flexibility in areas 1 and 2, while increasing the contributions of areas 3 and 5 in subcoalitions.

\begin{figure}[h]
	\centering
	\begin{tikzpicture}[scale=1, every node/.style={scale=0.35}]
			\node at (0,0) {\includegraphics[width=2.2\textwidth]{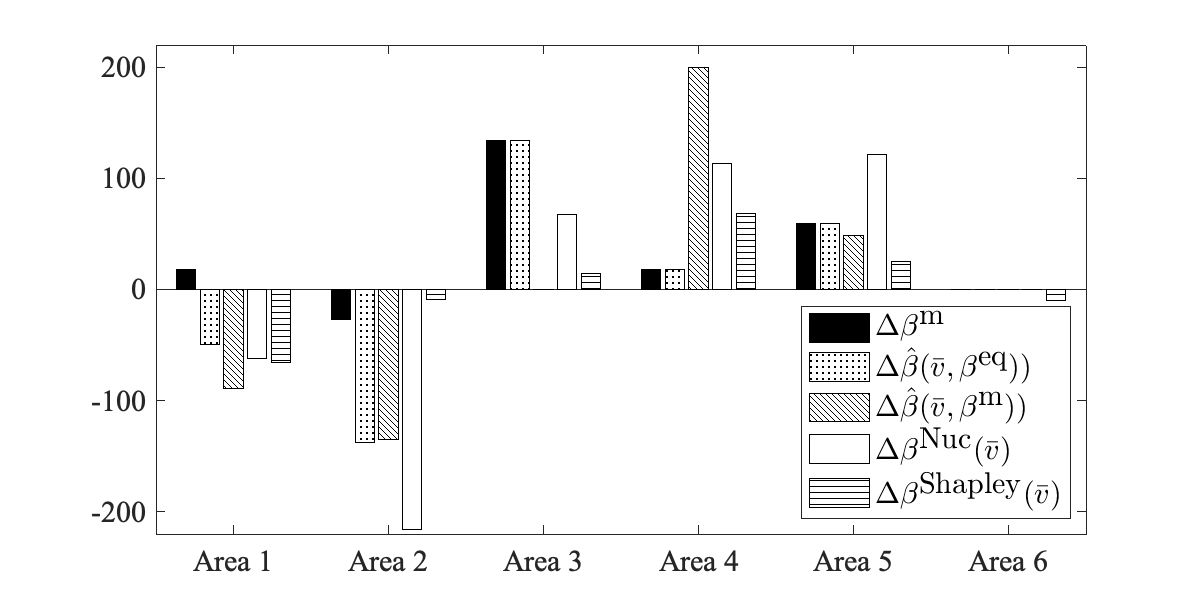}};
	\end{tikzpicture}
	\caption{The changes in benefit allocations after increasing the flexibility in area 4 (in \euro)}\label{fig:rts_benefit_3}
\end{figure}

\section{Appendix}

\subsection{Market outcomes for three-area example}\label{app:market_outcome}	

The market outcomes for the existing sequential market model and the preemptive model are provided in~\cref{tab:three_marketout}. For the preemptive model, the values of $r^-$ and $r^+$ shown in parenthesis indicate the amount of reserves destined to meet the requirements of the neighboring areas. Notice that inflexible generators and wind power generators are not capable of providing any reserves. In the day-ahead stage, the quantities correspond to $p_i$ and $w_j$ for conventional units and wind power generators, respectively. Because of the null production costs, wind power generators are always utilized at their full capacity given by the expected value of the stochastic process. In the balancing stage, the quantities correspond to the changes in the production levels with respect to the ones assigned at the day-ahead stage. Specifically, they are given by $({p}_{is}^{+}-{p}_{is}^{-})$ and $(W_{j s}-{w}_j-w^{\text{spill}}_{j s})$ for conventional units and wind power generators, respectively. It can be verified that there is $31$ MW and $15.6$ MW curtailment in the wind power production in the existing market model for scenarios $s_1$ and $s_2$. This is completely eliminated when the preemptive model is used. 

Observe that the implementation of the preemptive model results in a higher reserve allocation for the lower cost generator at node $8$. Consequently, the day-ahead market takes these new reserve quantities into account while deciding on the day-ahead energy quantities. Notice that there is $9$ MW load shedding in scenario $s_2$ in the existing market model. The preemptive model ensures that we do not resort to any costly load shedding in both scenarios.

\begin{table*}[h]\footnotesize
\centering
\caption{Market outcomes for the existing sequential market and the preemptive model (in MW)}
\begin{tabular}{|c||c|c|c|c|c||c|c|c|c|c|}
\hline
 Models &  \multicolumn{5}{c||}{Existing market model} & \multicolumn{5}{c|}{Preemptive model}\\ \hline
 $\substack{\text{Trad.}\\\text{floors}}$ &  \multicolumn{2}{c|}{Reserve} &  Day-ahead & \multicolumn{2}{c||}{Balancing}&  \multicolumn{2}{c|}{Reserve} &  Day-ahead & \multicolumn{2}{c|}{Balancing}\\ \hline
  &  $r^-$&  $r^+$ &   & $s_1$ & $s_2$ &$r^-$&  $r^+$ &   & $s_1$ & $s_2$   \\ \hline
   $i_1$ &$0$  &$0$ &$120$&$0$ &$0$ &$0$ &$0$&$120$ &$0$ &$0$ \\ \hline
   $i_2$ &$8$  &$12$ &$38$&$3$ &$3$ &$8$ &$12$&$38$ &$7.7$ &$7.2$ \\ \hline
   $i_3$ &$0$  &$0$ &$0$  &$0$ &$0$ &$0$ &$0$ &$0$ &$0$ &$0$ \\ \hline
   $i_4$ &$0$  &$0$ &$120$&$0$ &$0$ &$0$ &$0$&$120$ &$0$ &$0$ \\ \hline
   $i_5$ &$9.6$&$6.4$&$33$&$6$ &$6$ &$7.2$ &$4$ &$31.1$ &$-7.2$ &$-7.2$ \\ \hline
   $i_6$ &$0$  &$0$ &$0$  &$0$ &$0$ &$0$ &$0$&$0$ &$0$ &$0$ \\ \hline
   $i_7$ &$0$  &$0$ &$120$&$0$ &$0$ &$0$ &$0$&$120$ &$0$ &$0$ \\ \hline
   $i_8$ &$8$&$12$ &$38$&$12$&$12$&$10.4\,(2.4)$&$14.4\,(2.4)$&$35.6$ &$-10.1$ &$14.4$ \\ \hline
   $i_9$ &$0$  &$0$ &$6.6$&$0$ &$0$ &$0$   &$0$ &$10.9$ &$0$ &$0$ \\ \hline
    $j_3$ &$0$ &$0$ &$42$ & $-3$  &$-12$ &$0$ &$0$ &$42$ &$8$ &$-12$ \\ \hline
   $j_6$ &$0$  &$0$ &$70.4$&$-6.4$&$-6$&$0$ &$0$ &$70.4$ &$-6.4$ &$9.6$ \\ \hline
   $j_9$ &$0$  &$0$ &$42$ & $-12$ &$-12$ &$0$ &$0$ &$42$ &$8$ &$-12$ \\ \hline
\end{tabular}\label{tab:three_marketout}
\end{table*}

In~\cref{tab:three_jempty}, we provide the consumer and producer surpluses (CS and PS), and the congestion rents (CR) allocated to each area under both scenarios in the existing sequential market. Using these allocations for all trading floors, we defined a budget-balanced cost allocation in the numerical case studies.

\begin{table}[h]
\centering
\caption{Cost allocations in all trading floors for the existing sequential market (in \euro)}
\begin{tabular}{|c||c|c|c|}
\hline
 Areas & Area $1$  & Area $2$ & Area $3$\\ \hline
 CS for~(1) & $-60.0$  & $-64.0$ & $-70.0$\\ \hline
 PS for~(1)& $0$  & $0$ & $0$\\ \hline
 CR for~(1)& $0$  &$0$ & $0$\\ \hline
 CS for~(2) & $-8{,}448.0$  & $-7{,}239.0$ & $-9{,}900.0$\\ \hline
 PS for~(2)& $4{,}159.6$  & $3{,}750.2$ & $4{,}392.0$\\ \hline
 CR for~(2)& $0$  & $99.0$ & $99.0$\\ \hline
  CS for~(3) in~$s_1$ & $0$  & $0$ & $0$\\ \hline
 PS for~(3) in~$s_1$& $0$  & $-6{,}400.0$ & $5{,}250.0$\\ \hline
 CR for~(3) in~$s_1$& $0$  & $0$ & $0$\\ \hline
  CS for~(3) in~$s_2$ & $0$  & $0$ & $0$\\ \hline
 PS for~(3) in~$s_2$& $-12{,}000.0$  & $0$ & $2{,}250.0$\\ \hline
 CR for~(3) in~$s_2$ & $0$  & $0$ & $0$\\ \hline
 $J^{s_1}_a(\emptyset)$& $4{,}348.4$  & $9{,}853.8$ & $229.0$\\ \hline
 $J^{s_2}_a(\emptyset)$& $16{,}348.4$  & $3{,}453.8$ & $3{,}229.0$\\ \hline
\end{tabular}\label{tab:three_jempty}
\end{table}

\subsection{IEEE RTS layout}\label{app:IEEERTS96_area}

The layout can be found in~\cref{fig:IEEERTS96_area}. The area graph is illustrated in~\cref{fig:IEEERTS96_area_graph}.

\input{Fig-IEEE-24.tex}

\begin{figure}[h]
	\centering
	\begin{tikzpicture}[scale=1.4, every node/.style={scale=0.8}]
    \draw[-,black!80!blue,line width=.1mm] (0.35,0) -- (1.15,0) node[above=.2cm, left=.2cm] {\large $e_1$};
    \draw[-,black!80!blue,line width=.1mm] (0,0.35) -- (0,1.15) node[above=-.4cm, left=-1.8cm] {\large $e_3$};
    \draw[-,black!80!blue,line width=.1mm] (1.5,0.35) -- (1.5,1.15) node[above=-.6cm, left=-.6cm] {\large $e_4$};
    \draw[-,black!80!blue,line width=.1mm] (1.24,0.26) -- (0.26,1.24) node[above=-.7cm, left=.4cm] {\large $e_2$};
    \draw[-,black!80!blue,line width=.1mm] (1.85,0) -- (2.65,0) node[above=.2cm, left=.2cm] {\large $e_5$};
    \draw[-,black!80!blue,line width=.1mm] (1.85,1.5) -- (2.65,1.5) node[above=.2cm, left=.2cm] {\large $e_7$};
    \draw[-,black!80!blue,line width=.1mm] (0.35,1.5) -- (1.15,1.5) node[above=.2cm, left=.2cm] {\large ${e_6}$};
    \draw (0,0) circle (.35cm) node {\LARGE $a_1$};
    \draw (1.5,0) circle (.35cm) node {\LARGE $a_2$};
    \draw (3,0) circle (.35cm) node {\LARGE $a_3$};
    \draw (0,1.5) circle (.35cm) node {\LARGE $a_4$};
    \draw (1.5,1.5) circle (.35cm) node {\LARGE $a_5$};
    \draw (3,1.5) circle (.35cm) node {\LARGE $a_6$};
    \node at (5,.8) {%
  \begin{tabular}{|c||c|}
\hline
 Link & Capacity \\ \hline
 $T_{e_1}$& $1{,}900$MW \\ \hline
 $T_{e_2}$& $200$MW \\ \hline
 $T_{e_3}$& $500$MW \\ \hline
 $T_{e_4}$& $200$MW \\ \hline
 $T_{e_5}$& $700$MW \\ \hline
 $T_{e_6}$& $1{,}900$MW \\ \hline
 $T_{e_7}$& $700$MW \\ \hline
\end{tabular}};
	\end{tikzpicture}
	\caption{The area graph for the IEEE RTS}\label{fig:IEEERTS96_area_graph}
\end{figure}
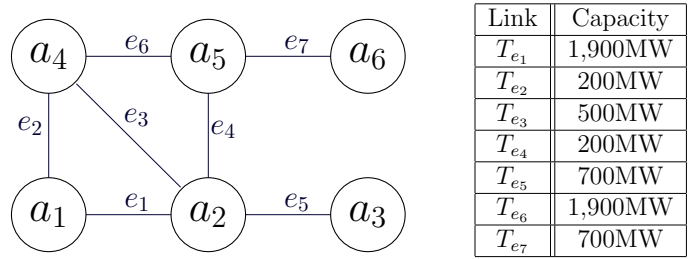


\part{Conclusion}
\chapter{Conclusions and outlook}
\section{\cref{part:1}}
In the first part of the thesis, our goal was to study mechanism design to achieve an efficient outcome in an electricity market setting, which involves continuous goods (e.g., electrical power), second stage costs, and general nonlinear constraints (instead of simple constraints, e.g., availability of a fixed number of items in multi-item auctions).

For the general class of electricity market problems, the dominant-strategy incentive-compatible VCG mechanism is susceptible to collusion and shill bidding. Motivated by this problem, we derived three different conditions under which collusion and shill bidding are not profitable, and hence the VCG mechanism is coalition-proof. Since these conditions are restrictive and they may not capture the constrained optimization problem under consideration, we investigated core-selecting mechanisms for their coalition-proofness.
We showed that the well-established LMP mechanism is core-selecting, and hence coalition-proof. This result was an implication of a stronger result we proved, core-selecting mechanisms are the exact mechanisms that ensure the existence of a competitive equilibrium in linear/nonlinear prices. In contrast to the LMP mechanism, we showed that core-selecting mechanisms are applicable to a broad class of markets with nonconvex bids and constraints. We then characterized core-selecting mechanisms that can approximate dominant-strategy incentive-compatibility without the price-taking assumption. In an exchange market setting, we proved that core-selecting mechanisms are also budget-balanced. Our results were verified in several case studies based on realistic electricity market~models. 

\subsection{Further research directions}

Some of the future research directions are removed, and will not be included until the end of 2021.

\subsubsection*{Relaxing the conditions of \cref{thm:conditions_on_bids} and \cref{thm:re2_proof}}
Invoking theoretical studies on the strong-substitute condition of~\cite{milgrom2007substitute}, we can potentially relax the separable convexity requirement of \cref{thm:re2_proof} on the bid functions to supermodularity and component-wise convexity, see the discussions in \cref{footnote:extension} of \cref{sec:p1_3}. However, this conclusion requires further investigation. 
\subsubsection*{Budget surplus reallocation}
When applied to exchanges and two-sided markets, one would need to reallocate budget surplus of core-selecting mechanisms to provide correct investment~signals for transmission capacity expansion. As it is pointed out at the end of~\cref{sec:444}, this topic remains to be investigated.

\subsubsection*{Extensions to single-settlement two-stage stochastic market models}

As several recent works show, for example, \cite{bouffard2005market,wong2007pricing,morales2009economic,pritchard2010single,zakeri2018pricing,exizidis2019incentive} to name a few, a scenario-based stochastic market can attain perfect temporal coordination using a probabilistic description of renewables' forecast errors. This approach enables co-optimizing the day-ahead and the real-time markets in a single stochastic optimization, as it is discussed in \cref{footnote:stoch} of \cref{sec:p1_2}. However, the stochastic dispatch model comes also with some practical drawbacks.
Unlike the existing market frameworks, it is unable to ensure simultaneously individual rationality and budget-balance (for every scenario), see~\cite{kazempour2018stochastic,pritchard2010single,zakeri2018pricing}. In addition, it requires ensuring the agreement of all the participants on the scenarios and the probabilities associated with them. Participants however might not share a unified view of the probability distribution governing the uncertainty. A deviation from this assumption may arise for example in case the participants are risk-averse~\cite{gerard2018risk,philpott2016equilibrium,kazempour2016effects,ralph2011pricing,ralph2015risk,martin2014stochastic}. To address the first problem, it is essential to study what core-selecting mechanisms can offer by extending on the works of~\cite{pritchard2010single,zakeri2018pricing}. For the second problem, it would be beneficial to incorporate risk measures and risk trading into the stochastic market design frameworks. One concrete idea is to exploit the works on risk-sharing in financial market literature and insurance mathematics~\cite{denault2001coherent,chen2017stable,philpott2016equilibrium,ralph2015risk,vespermann2020risk}.

\vspace{.75cm}

\section{\cref{part:2}}
In the second part of the thesis, our goal was to propose a coalitional game-theoretic approach to enable coordinated balancing and reserve exchange in a European level and to gain technical and economical insights about such a process.

To achieve this goal, we formulated a coalition-dependent preemptive transmission allocation model that defines the optimal inter-area transmission capacity allocation {between energy and reserves} for a given {set of participating areas}. 
We then accompanied this model with benefit allocation mechanisms such that all {coalition members} have sufficient benefits to accept the transmission allocation solution proposed by our method. We formulated the coalitional game both as an ex-ante and as an ex-post process with respect to the uncertainty realization. {We} showed that the former results in suitable mechanisms with budget-balance in expectation, whereas the latter results in suitable mechanisms with budget-balance in every {realization}. 
Applying the prevailing benefit allocations to a larger case study, we showed that these existing methods are unable to find a benefit allocation with minimal stability violation (that is, attaining the least-core property) within a reasonable computational time frame. 
To address this issue for both coalitional games, we proposed the least-core-selecting benefit allocation mechanism and we formulated an iterative constraint generation algorithm for its efficient computation. 
Considering that this work aimed to contribute to the ongoing discussion towards the design of the transmission allocation model, our benefit allocation mechanism can be adapted to different plausible fairness criteria that may be imposed by the regulatory authorities, moving towards the full integration of the balancing markets.
 The numerical results illustrated several crucial factors (e.g., flexibility, network structure, wind correlations) that are essential in driving the benefit allocation in different configurations.

\subsection{Further research directions}

\subsubsection*{Privacy concerns}
Even though fair cost/benefit allocations are proposed with our studies in~\cref{part:2}, all of these methods require private regional information (such as generator bids, demand profiles, forecast scenarios) and an agreement on the estimates of a considerable number of parameters. Currently, the level of transparency remains as a major concern in the transmission capacity calculations for day-ahead markets~\cite{eurel20}. Recent regulations and studies strongly advocate that the market should be developing towards more and more coordinated and transparent operation in terms of all parameters affecting the market outcome. 

On the other hand, the aforementioned alternative market-forming approach would establish a new trading floor in which each regional operator can submit its own valuation to a market organizer (we kindly refer to the discussions in~\cref{sec:marketbased}). This methodology can potentially preserve the privacy of each operator, but it remains to be investigated. Utilizing tools from graph theory and mechanism design, it would be interesting to analyze problems over networks of continuous public choices that would originate from shared transmission resources. 

Related to this direction, another interesting research direction is to explore the development of decentralized computational schemes for general stochastic bilevel programs (such as the one given by problem \eqref{mod:B-Pre} of the preemptive model) that can preserve the privacy of the area operators with minimal exchange of intra-area information.
For the current approach provided in~\cref{part:2}, it would also be beneficial to study the impact of strategic behaviour of the area operators and the generators on the transmission allocations and the resulting benefit allocation mechanisms. 

\subsubsection*{Modelling-related analyses}
In case there are scenarios that are not included in the forecast, our preemptive transmission allocation in~\cref{part:2} can also potentially be suboptimal, and it is hard to provide an out-of-sample analysis for the suboptimality of nonconvex stochastic optimization problems~\cite{birge2011introduction}.
However, it would be interesting to provide a numerical analysis for the impact of out-of-sample uncertainty realizations on the optimality of the preemptive model.

It is also expected that the initial/pilot coordinated balancing platforms would also be zonal as it is the case today in the European day-ahead energy markets~\cite{IGCC,dominguez2019reserve}. Hence, it would provide an additional guide for the future integrated European markets, if we can incorporate the potential congestion from zonal balancing platforms into the preemptive model.

\subsubsection*{Extensions to market participant-level benefits}

As we have discussed in~\cref{sec:marketbased}, our studies define the monetary quantities on an operator level, and do not specify the payment rules for the markets such that these costs (or discounts) are distributed on a market participant level. An open problem is extending \cref{part:2} by specifying payments for the new markets such that the benefits of implementing cross-border trading is distributed on a market participant level. Similar to what is suggested also in the studies of~\cite{kristiansen2018mechanism}, the goal here can potentially be achieved by defining the correct set of side payments complementing the locational (nodal or zonal) prices. However, special consideration is needed since such payments can also be gamed by the participants. This research direction warrants further exploration.


\addcontentsline{toc}{part}{Bibliography}
\printbibliography
\end{document}